\documentclass[10pt,a4paper,openright,twoside]{scrbook}

\def\ltsim{\raise 2pt \hbox {$<$} \kern-1.1em \lower 4pt \hbox {$\sim$}}
\def\ltapprox{\raise 2pt \hbox {$<$} \kern-1.1em \lower 5pt \hbox {$\approx$}}
\def\gtsim{\raise 2pt \hbox {$>$} \kern-1.1em \lower 4pt \hbox {$\sim$}}
\def\gtapprox{\raise 2pt \hbox {$>$} \kern-1.1em \lower 5pt \hbox {$\approx$}}

\usepackage{graphicx}
\usepackage{epsfig}
\usepackage{psfrag}
\usepackage[english]{babel}
\usepackage{amsmath}
\usepackage{amssymb}
\usepackage{amsthm}
\usepackage{mathrsfs}
\usepackage{amsfonts}
\usepackage{longtable}
\usepackage{multirow}
\usepackage[all]{xy}
\usepackage{fancyhdr}
\usepackage{verbatim}
\usepackage{setspace}
\usepackage{rotating}
\usepackage{hyperref}

\newtheorem{theorem}{Theorem}[section]

\theoremstyle{remark}

\begin{document}

\begin{titlepage}
\pagestyle{empty}

\begin{center}

 { \Huge{\textsc{ Permutation Orbifolds   \vfill
                  in   \vfill
                  Conformal Field Theories   \vfill
                  and      \vfill
                  String Theory} \vfill}}

\vfill
\vfill
\vfill
\vfill
\vfill

{ \Huge{ \textit{Michele Maio} } }

\end{center}

\newpage

This work has been accomplished at the Dutch Institute for Subatomic Physics (Nikhef) in Amsterdam and has been financially supported by the Foundation for Fundamental Research of Matter (FOM) which is part of the Dutch Organization  for Scientific Research (NWO).
\vfill
\noindent Cover page: Marine Landscape; author: Peter van Buren.\\
This work has been printed by Ipskamp Drukkers.\\
ISBN 978-94-91211-93-5.
\vfill
\copyright Michele Maio 2011.\\
All rights reserved. Without limiting the rights under copyright reserved above, no part of this book may be reproduced, stored in or introduced into a retrieval system or transmitted, in any form or by any means (electronic, mechanical, photocopying, recording or otherwise) without the written permission of both the copyright owner and the author of the book.

\cleardoublepage

{\doublespacing
\begin{center}
{\LARGE \sc  Permutation Orbifolds   \vfill
                  in   \vfill
                  Conformal Field Theories   \vfill
                  and   \vfill
                  String Theory  \vfill
}
\end{center}
\vfill
\begin{center}
{\large \sc Een wetenschappelijke proeve op het gebied van de Natuurwetenschappen, Wiskunde en Informatica}
\end{center}
\vfill
\begin{center}
{\Large \sc Proefschrift}
\end{center}
\vfill
\begin{center}
{\large \sc
ter verkrijging van de graad van doctor\\
aan de Radboud Universiteit Nijmegen\\
op gezag van de rector magnificus\\
prof. mr. S.C.J.J. Kortmann,\\
volgens besluit van het college van decanen\\
in het openbaar te verdedigen op woensdag 5 oktober 2011\\
om 10:30 uur}
\end{center}
\begin{center}
door\\
{\large \sc Michele Maio\\
geboren op 22 Maart 1981\\
te Avellino (Itali\"{e})}
\end{center}
}

\vfill

\newpage

\begin{tabular}{ll}
{\sc \Large Promotor:} & {\sc \Large Prof. dr. A.N.J.J. Schellekens}\\
 &\\
 &\\
{\sc \Large Manuscriptcommissie} & {\sc \Large Prof. dr. R.H.P. Kleiss}\\
 & {\sc \Large Prof. dr. C. Schweigert} (Universit\"{a}t Hamburg)\\
 & {\sc \Large Prof. dr. E.P. Verlinde} (Universiteit van Amsterdam)\\
\end{tabular}

\vfill
\newpage

\cleardoublepage

\end{titlepage}

\large
\newpage

\baselineskip 4ex

\newcommand{\ltae}{\raisebox{-0.6ex}{$\,\stackrel
{\raisebox{-.2ex}{$\textstyle <$}}{\sim}\,$}}
\newcommand{\gtae}{\raisebox{-0.6ex}{$\,\stackrel
{\raisebox{-.2ex}{$\textstyle >$}}{\sim}\,$}}

\baselineskip 4ex

\cleardoublepage
\pagenumbering{Roman}

\begin{titlepage}
\thispagestyle{empty}
\raggedleft
\large

\textit{Chemile, to you again, for the last time.}\par
\small \textit{-M.}\par

\vfill

\large \textsc{Learn all the rules and}\par
\large \textsc{then break some of them}\par
\textit{Nepalese Tantra}

\mbox{}
\clearpage{\pagestyle{empty}\cleardoublepage}
\end{titlepage}

\cleardoublepage

\pagestyle{empty}

\begin{center}
{\huge \bf Publications}
\end{center}

\vfill

\begin{center}
This Ph.D. thesis is the outcome of three years of research carried out at the National Institute for Subatomic Physics (Nikhef) in Amsterdam (The Netherlands) in the field of Theoretical Physics.
It is based on the following publications:
\end{center}

\vfill

\begin{enumerate}
\item
M. Maio and A. N. Schellekens,\\
{\it Permutation Orbifolds of Heterotic Gepner Models},\\
Nucl. Phys. B {\bf 848} (2011) 594-628 [arXiv: 1102.5293 [hep-th]];
\item
M. Maio and A. N. Schellekens,\\
{\it Permutation Orbifolds of N=2 Supersymmetric Minimal Models},\\
Nucl. Phys. B {\bf 845} (2011) 212-245 [arXiv: 1011.0934 [hep-th]];
\item 
M. Maio and A. N. Schellekens,\\
{\it Formula for Fixed Point Resolution Matrix of Permutation Orbifolds},\\
Nucl. Phys. B {\bf 830} (2010) 116-152 [arXiv: 0911.1901 [hep-th]];
\item
M. Maio and A. N. Schellekens,\\
{\it Complete Analysis of Extensions of $D(n)_1$ Permutation Orbifolds},\\
Nucl. Phys. B {\bf 826} (2010) 511-521 [arXiv: 0907.3053 [hep-th]];
\item
M. Maio and A. N. Schellekens,\\
{\it Fixed Point Resolution in Extensions of Permutation Orbifolds},\\
Nucl. Phys. B {\bf 821} (2009) 577-606 [arXiv: 0905.1632 [hep-th]].
\end{enumerate}

\cleardoublepage

\phantomsection
\addcontentsline{toc}{chapter}{Table of Contents}
\tableofcontents
\cleardoublepage

\phantomsection
\addcontentsline{toc}{chapter}{List of Tables}
\listoftables
\cleardoublepage

\phantomsection
\addcontentsline{toc}{chapter}{List of Figures}
\listoffigures
\cleardoublepage

\pagestyle{headings}

\pagenumbering{arabic}

\chapter{Introduction}

{\flushright
{\small 
\textit{O voi che siete in piccioletta barca,}\par
\textit{desiderosi d'ascoltar, seguiti}\par
\textit{dietro al mio legno che cantando varca,}\par
\textit{tornate a riveder li vostri liti:}\par
\textit{non vi mettete in pelago, ch\'e, forse,}\par
\textit{perdendo me, rimarreste smarriti.}\par
\emph{(Dante, Div. Comm.)}\par
}
}

String Theory has enjoyed a growing interest and has attracted the attention of scientists over the last twenty years because it is a leading candidate for deriving all the four interactions from a single framework.

The Standard Model, built in the seventies as a theory of point-like particles, is the best working model that we have at our disposal at the moment for electro-magnetic, strong and weak interactions, but it is not completely satisfactory. First, because gravity is left out: in fact, there is a huge incompatibility between quantum mechanics and general relativity, due to the fact that their union results in a non-renormalizable theory, and this makes the inclusion of gravity impossible. Secondly, the Standard Model has too many free parameters that have to be determined empirically and no-one knows why, for example, the gauge group is what it is.

String Theory addresses both these problems. First of all, it includes quantum gravity in a consistent way, where General Relativity is re-obtained as a low-energy approximation. Secondly, it does not have any free dimensionless parameter (there is only one dimensionful parameter, the tension of the string or equivalently the string constant $\alpha'$, which sets the scale for the theory). The Standard Model parameters are still not determined, but reinterpreted as vacuum expectation values (v.e.v.'s) of several ``moduli'' fields. These fields specify couplings and background and are not fixed by the theory, since by definition they have a flat potential (assuming Supersymmetry, see below). One of them is the dilaton field whose expectation value determines the string coupling constant $g_s$, which enters the calculations of loop corrections as Feynman-like diagrams. Moreover, also the Standard Model gauge group, as it appears at low energies, is not fixed by the full theory.

However, the theory has a very serious problem, namely the presence of extra dimensions. This implies the existence of other dimensions besides the four that we observe in our spacetime. Within String Theory, spacetime is predicted to be ten dimensional. So, where are the extra dimensions and why do we not experience them? The reason is that they are probably curled up in some compact manifold of the size of the Planck length and hence too small to be detected, at least at the present.

String Theory's main constituents are not point particles but one-dimensional extended objects called strings. Actually this is not quite correct, because the theory is much richer: besides strings, it includes also any sort of $p$-branes, i.e. $p$-dimensional spatial membranes, which have their own dynamics.

Another striking feature is that there exist several equivalent ways of describing the same theory, each representation having its own name (see figure\footnote{Figure taken from the website $http://wordassociation1.net/symmetry.html$.} \ref{mtheory_pic}). They describe different ``\emph{corners} of our world'' and are related by an intricate web of dualities.
\begin{figure} [ht]
\centering
\includegraphics[scale=0.55]{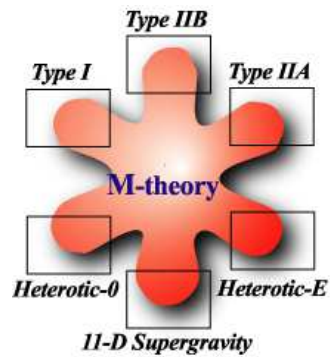}
\caption{M-Theory moduli space.}
\label{mtheory_pic}
\end{figure} 
Just to give some examples, Type $IIA$ and Type $IIB$ theories are \emph{T-dual} of each other, meaning that Type $IIA$ theory on a circle of radius $R$ is equivalent to Type $IIB$ theory on a circle of radius ${\alpha'/R}$. Analogously, $E_8\times E_8$ Heterotic theory is \emph{S-dual} to $SO(32)$ Heterotic theory, in the sense that $E_8\times E_8$ Heterotic theory at coupling $g_s$ is equivalent to $SO(32)$ Heterotic theory at coupling $1/g_s$; similarly, Type $IIB$ is self-dual under $S$-duality. 

An important feature of String Theory is Supersymmetry. Among other things, Supersymmetry implies the existence of additional matter: to each already-existing particle Supersymmetry associates a supersymmetric partner, whose spin differs by one half from the spin of that particle. Hence, each bosonic (fermionic) particle has a fermionic (bosonic) superpartner. 
Supersymmetry is important in String Theory for several reasons.
First of all, dark matter. Dark matter seems to exist in the universe and appears to require weakly interacting massive particles. Supersymmetric partners provide us with suitable dark matter candidates. 
Secondly, the hierarchy problem. In a quantum field theory, the Higgs mass diverges quadratically, making it hard to explain why it is actually so small. Supersymmetry instead allows us to cancel quadratic divergences in the calculation of loop corrections for the Higgs mass. These quadratic divergences originate from loop diagrams where fermions run in the loop. With Supersymmetry extra diagrams need to be considered, where also the bosonic partners of the fermions run in the loop, thus contributing with a minus sign to the total amplitude. The final divergence is only logarithmic and can be easily dealt with renormalization.
Thirdly, coupling unification. In supersymmetric extensions of the Standard Model, the superpartners contribute also to the beta function of the electromagnetic, strong and weak coupling constants, modifying their runnings such that at very high energy (of order $10^{16}$ GeV) they have the same value and hence are unified. Even if it does not have to be this way, this is often considered an extremely attractive feature of Supersymmetry.
Finally, non-physical tachyons. The construction of string spectra often produces tachyons. Supersymmetry helps in projecting out tachyons from the particle spectrum. Nevertheless, there are examples (e.g. the $O(16)\times O(16)$ heterotic string \cite{AlvarezGaume:1986jb,Dixon:1986iz}) with no Supersymmetry and also with no tachyons. 
Despite all these nice features of Supersymmetry, our world, in the way we experience it, is not supersymmetric and hence Supersymmetry must be broken.

The applications of String Theory extend in many directions. There are phenomenological directions, such as the construction of a supersymmetric Standard Model, with the inclusion of gravity and supersymmetry breaking at the TeV scale; there are highly theoretical directions related to \emph{the} possible formulation of the theory; there are connections with gauge theories and the $AdS/CFT$ correspondence; there are interesting applications to black holes, which represent a theoretical laboratory to test any quantum theory of gravity, with the inclusion of both quantum mechanics and general relativity, reproducing the original setup of the early universe, when gravity was as strong as the other forces.

In this thesis our main focus will be on mathematical aspects, in particular Conformal Field Theories (CFT's), and on the phenomenology of String Theory. These two topics are indeed closely connected. 
When we talk about phenomenology we are asking the question whether and how a ten-dimensional theory can reproduce a four-dimensional model at low energies with the right properties. It is by now clear to most people in the field that there does not exist a unique answer to this question: very many models can be constructed which possess the correct number of families and the correct gauge group, at least in the vicinity of the Standard Model.

The idea that only one way existed to obtain the Standard Model has been already given up long time ago. The reason for that is the huge amount of possibilities that are available in building four-dimensional string theories. This is what is known as the landscape. It seems unreasonable that only one out of maybe-infinitely many constructions would do the job. It is instead more reasonable to expect that there are many four dimensional models with Standard-Model-like features in the landscape. Then the correct question to ask in this case would not be which particular model is the real model, but rather how rare and how frequent certain properties (e.g. family number, gauge group, etc.) are. It would definitely be disappointing if it turns out that we live in the least probable universe!

The first problem one has to deal with is getting rid of the six extra dimensions. The standard geometric approach is to consider compactifications on ``small'' six-dimensional manifolds which preserve some supersymmetry. These manifolds are not completely arbitrary, but constrained by supersymmetry to be of a special type, the so-called Calabi-Yau manifolds \cite{Candelas:1985en}. By changing the compactification, the four-dimensional physics changes as well. However, in this approach, the family number is related to topological properties of the Calabi-Yau (in particular, its Euler number), which is normally much larger than three. Also, the typical gauge groups are too big and contain the standard model gauge group as a subgroup. Moreover, in terms of generating four-dimensional spectra, the geometric approach does not go very far.

Moduli fields are related to deformations of the Calabi-Yau manifold, controlling its size and shape.  Sometimes, for particular values of the parameters, which are v.e.v.'s of the moduli fields, the geometric description has an equivalent formulation in terms of Conformal Field Theory. It is already remarkable that the interacting CFT at those points can be solved exactly. In some ways the CFT approach is more general than the geometric one. The extra spatial dimensions are related to the central charge of the CFT and, when treated in this perspective, they do not need to admit a geometric interpretation at all. The power of CFT manifests itself when one builds four-dimensional theories. Through the formalism of simple-current extensions, a huge number of modular invariant partition functions (MIPF's), and hence spectra, can be built for any given CFT. Each of these so-called ``simple-current invariants'' gives rise to a spectrum with a given number of families and gauge group, whose likelihood within the landscape can be studied statistically. We will see how this is done in detail towards the end.

\section{This thesis}
In this thesis we consider the CFT approach to String Theory. As already mentioned, simple-current invariants will be the main tool. These are partition functions that exist because the CFT has very special fields, called simple currents, in its spectrum. Sometimes these simple currents admit ``fixed points''. Then the CFT built out of extensions has non-trivial modular matrices that are not known. The problem of determining these matrices is called the ``fixed point resolution'' \cite{Schellekens:1990xy}. We will define both simple currents and fixed points in the main chapters.

More precisely, we study permutations of identical CFT's and their orbifolds \cite{Klemm:1990df}, limiting ourselves to the order-two case. We address the problem of extensions of these permutations and resolve the fixed points of its simple current. This is a purely mathematical problem, but with interesting physical implications. As an application, we apply our results to string model building of four-dimensional spectra.

The structure of this thesis is as follows. We have divided it into four parts. Part I includes the first three chapters, Part II the following two. 
\begin{itemize}
\item Part I deals with two-dimensional Conformal Field Theories and in particular we define permutation orbifolds, simple-current extensions and fixed points.

\item Part II deals with applications of our results to String Theory and addresses the problem of constructing four-dimensional models using extensions of the permutation orbifold.

\item Part III summarizes our conclusions and contains discussions about additional research directions and future possible work.

\item Part IV contains some technical appendices. All the material that would have spoiled the readibility of the work has been collected here.
\end{itemize}
In chapter \ref{paper1} we introduce the subject of permutation orbifolds in conformal field theories. We establish our notation and define the problem. Simple currents arising in the orbifold have a very special structure: they are diagonal representations of the simple currents in the mother theory. In addition, they always have fixed points of various kind. We study the fixed point resolution for those currents and derive explicit expressions for the ``$S^J$'' matrices in some particular examples. Specifically, we address the problem when the mother theory is a current algebra of $SU(2)_k$ and $SO(N)_1$. These specific cases are interesting in their own right, since they involve very non-trivial tricks that will eventually lead to the final answer. \\
In chapter \ref{paper2} we give more examples of ``$S^J$'' matrices. In particular, we consider spinor currents of the $D(n)_1$ series, which have integer spin when $n$ is multiple of four and half-integer when $n$ is even but not a multiple of four. The main tool here is triality of $SO(8)$. Although half-integer spin currents cannot be used to extend the chiral algebra, when combined with other half-integer spin currents (for example in a CFT built as a tensor product of several blocks) they can give rise to integer-spin currents where the fixed point resolution becomes an issue. \\
In chapter \ref{paper3} we find a general formula for the resolution of fixed points in extensions of permutation orbifolds by its (half-)integer-spin simple current. This formula is based on an ansatz that we are able to infer from the examples studied in the two previous chapters. We check that our ansatz makes sense, namely that it gives a unitary and modular invariant $S$ matrix. We also compute the fusion rules for several conformal field theories, including cases with a huge number of primary fields, and find non-negative integer coefficients. We conclude that our ansatz provides us with a very robust formula for solving the fixed point problem in extended orbifolds. \\
In chapter \ref{paper4} we make a first move towards string theory. In the back of our minds we are thinking about Gepner models, hence we study here permutations of $N=2$ superconformal minimal models. We combine permutations and extensions and find a very interesting mathematical structure relating various conformal field theories. In particular, it turns out that the supersymmetric version of the $N=2$ orbifold is obtained by extending the non-supersymmetric orbifold by a very specific simple current. Moreover, we will see that in the supersymmetric orbifold the chiral extension transforms some fields into simple currents. This was not expected a priori. Hence these new currents will be called ``exceptional''. They have completely different origin from all the currents encountered so far and admit sometimes fixed points. The resolution of those fixed points is still an open problem. \\
In chapter \ref{paper5} we are finally able to study permutations in heterotic Gepner models. Our permutations will be of order two only. The spectra obtained with our CFT approach fully agree with those that were previously known in the literature. However, the power of simple currents in conformal field theory manifest itself at this point by making it possible to generate a huge number of four-dimensional modular invariant partition functions. Since standard Gepner models are not expected to produce a significant number of three-family models, we apply the so-called lifting procedure to them in order to make three families more and more frequent in this kind of four-dimensional string theory constructions. \\
In chapter \ref{conclusions} we conclude with some remarks and discussions about possible related work. \\
In the appendix we collect all the supporting material (e.g. tables, theorems) that is relevant but would have slowed down the reading of the manuscript.

Throughout this thesis, we consider mostly $\mathbb{Z}_2$ permutation orbifolds. Hence, often we will refer to it simply as the permutation orbifold, unless clearly stated otherwise.

\section{Notation}
In this section we summarize the notation that we use throughout this work about permutation orbifolds, $N=2$ minimal models and their permutations, Gepner models, simple current extensions.
\begin{itemize}
\item \underline{Permutation orbifold}\\
In the permutation orbifold $(\mathcal{A}\times\mathcal{A})/\mathbb{Z}_2$ various kinds of fields arise from the fields in the mother theory $\mathcal{A}$. We denote them as follows:
\begin{itemize}
\item diagonal: $(i,\psi)$, $\psi=0,\,1$,
\item off-diagonal: $\langle i,j\rangle$, $i\neq j$,
\item twisted: $\widehat{(i,\psi)}$, $\psi=0,\,1$,
\end{itemize}
with $i,\,j\in\mathcal{A}$.
In particular, the so-called \textit{un-orbifold current}, which is
\begin{itemize}
\item $(0,1)$, anti-symmetric representation of the identity,
\end{itemize}
is immediately relevant, since the extension by this field un-does the permutation orbifold and gives back the tensor product CFT.\\
The orbifold $S$ matrix was derived by Borisov, Halpern and Schweigert \cite{Borisov:1997nc}: we will often call it $S^{BHS}$.

\item \underline{$N=2$ minimal models}\\
$N=2$ superconformal minimal models are rational CFT's, fully specified by their ``level'' $k$, which fixes both the central charge $c=\frac{3k}{k+2}$ and the field content. Their primary fields are labelled by the multi-index 
\begin{equation}
\phi_{l,m,s}\equiv(l,m,s)\,,
\nonumber
\end{equation}
where
\begin{itemize}
\item $l=0,\dots,k$ is an $SU(2)_k$ label;
\item $m=-k+1,\dots,k+2$ is a $U(1)_{2(k+2)}$ label;
\item $s=-1,\dots,2$ is a $U(1)_{4}$ label.
\end{itemize}
Moreover, these labels satisfy a given field identification and obey a given constraint:
\begin{itemize}
\item $(l,m,s)\sim(k-l,m+k+2,s+2)$,
\item $l+m+s=0$ mod $2$.
\end{itemize}
Very special $N=2$ fields are:
\begin{itemize}
\item $0\equiv(0,0,0)$, identity;
\item $T_F\equiv(0,0,2)$, world-sheet supercurrent;
\item $S_F\equiv(0,1,1)$, spectral flow operator.
\end{itemize}

\item \underline{Permutations of $N=2$ minimal models}\\
In the study of permutation of $N=2$ minimal models a few other fields become important:
\begin{itemize}
\item $(T_F,0)$, symmetric representation of the world-sheet supercurrent: the extension by this current gives a non-supersymmetric CFT;
\item $(T_F,1)$, anti-symmetric representation of the world-sheet supercurrent: the extension by this current gives the super-symmetric orbifold;
\item $\langle 0,T_F\rangle$, the world-sheet supercurrent: it is a fixed point of both $(T_F,\psi)$ and it splits in two fields in those extensions;
\item $(S_F,0)$, the symmetric representation of the spectral-flow operator: it is used to impose space-time supersymmetry in the permuted Gepner model.
\end{itemize}

\item \underline{Gepner models}\\
Gepner models are tensor products of the space-time $SO(10)$ factor times $r$ internal $N=2$ minimal models, plus extensions by the subgroup generated by the space-time supercurrent and the world-sheet supercurrents
\begin{itemize}
\item $S_{\rm st}\otimes(S_F)^r$,
\item $V_{\rm st}\otimes (0 \otimes \dots \otimes T_{F,i} \otimes \dots \otimes 0)$, $i=1,\dots r$.
\end{itemize}
$S_{\rm st}$ and $V_{\rm st}$ are the spinor and vector representations of $SO(10)$, $T_{F,i}$ is the world-sheet supercurrent of the $i^{\rm th}$ internal $N=2$ factor. \\
Gepner models are conveniently labelled by their levels; a hat on one of the $k$'s denotes a lift on that factor:
\begin{itemize}
\item $(k_1,\dots,k_i,\dots,k_r)$,
\item $(k_1,\dots,\hat{k}_i,\dots,k_r)$.
\end{itemize}
Permuted Gepner models are denoted by brackets (a hat for the lifts):
\begin{itemize}
\item $(k_1,\dots,\langle k_i,k_i\rangle,\dots,k_r)$,
\item $(k_1,\dots,\langle k_i,k_i\rangle,\dots,\hat{k}_j,\dots,k_r)$.
\end{itemize}

\item \underline{Simple current extensions}\\
Simple currents will be generically denoted by $J$, unless we are talking about specific currents, in which case they will be denoted by their own names. Similarly, fixed points are generically denoted by $f$.\\
Quantities in theories $\tilde{\mathcal{A}}$ extended by simple currents are normally called by the same name they have in the original theory $\mathcal{A}$, but in addition they carry a tilde. For example, if $S$ is the S matrix of some CFT, then $\tilde{S}$ is the S matrix of the extended CFT.
\end{itemize}

\cleardoublepage

\part{CONFORMAL FIELD THEORY}
\chapter*{About Part I}

{\flushright
{\small 
\textit{I may climb perhaps to no great heights,}\par
\textit{but I will climb alone.}\par
\emph{(E. Rostand, Cyrano de Bergerac)}\par
}
}

Part I deals with two-dimensional Conformal Field Theories \cite{Belavin:1984vu}. CFT is in principle an independent subject in its own right, which shares many applications in other areas of Physics, from Condensed Matter to Quantum Information. Two-dimensional conformal systems are very special, because only in two dimensions the conformal group admits an infinite-dimensional algebra whose generators are the Virasoro operators. Supersymmetric CFT extensions contain the Virasoro algebra as a sub-algebra and can be treated similarly to non-supersymmetric CFT's. The existence of this well-defined mathematical structure allows us to split the theory in two (almost independent) sectors, one holomorphic (right-movers) and one anti-holomorphic (left-movers). Modular invariance of the partition function puts additional constraints on which left-moving representations can couple to which right-moving ones. 

Modular invariance means that the one-loop partition function is invariant under reparameterizations of the torus. Topologically different tori are characterized by inequivalent values of the modulus $\tau$, where inequivalent means that two values $\tau_1$ and $\tau_2$ are not related by an $SL(2,\mathbb{R})$ transformation, $\tau\rightarrow\frac{a\tau+b}{c\tau+d}$ ($ad-bc=1$).
Geometrically, the modular generators interchange the two fundamental cycles ($S$ transformation: $\tau\rightarrow -\frac{1}{\tau}$) or act as Dehn twists ($T$ transformation: $\tau\rightarrow \tau +1$) of the torus. Algebraically, the generators act on the characters of the theory. A character is defined as a trace over the full Hilbert space generated by the conformal algebra, which in the simplest case contains only the Virasoro operators:
\begin{equation}
\chi_i(\tau)=Tr_{\mathcal{H}_i} e^{2\pi i\tau(L_0-\frac{c}{24})}\,.
\end{equation}
The characters summarize all the information about the full representation, i.e. not just single primary fields but also their descendants, and suitable combinations of characters define a partition function. The generators $S$ and $T$ of the modular group act as a matrix representation on the characters:
\begin{equation}
\chi_i(-\frac{1}{\tau})=\sum_j S_{ij} \chi_j(\tau)\,,\qquad
\chi_i(\tau+1)=\sum_j T_{ij} \chi_j(\tau)
\end{equation}
where $T_{ij}=e^{2\pi i(h_i-\frac{c}{24})}\delta_{ij}$ is a diagonal matrix of phases depending on the weights $h_i$ of the representations of the CFT and $S$ is a symmetric and unitary matrix satisfying the constraints $(ST)^3=S^2$.

The $S$ matrix is a fundamental object in a CFT, because it determines the fusion rules of two representations
\begin{equation}
(i)\cdot(j)=\sum_{k}N_{ij}^{\phantom{ij}k} (k)\,,
\end{equation}
with positive-integer coefficients $N_{ij}^{\phantom{ij}k}$, via the Verlinde formula \cite{Verlinde:1988sn}. Some fields have simple fusion with any other field in the theory and they are called ``simple currents'' \cite{Schellekens:1990xy}. The word current is used to characterize these special fields, because they can be regarded as additional generators, which in turn can be used to enlarge the conformal algebra and define a new extended conformal field theory. Simple currents are probably the most powerful tool available in a CFT. The reason is that to each simple current one can associate a modular invariant partition function. In practical models the number of these currents can be huge and as a consequence the number of spectra that can be constructed is huge too. In a CFT integer-spin simple currents are mostly relevant, since fractional-spin simple currents act as automorphism of the chiral algebra, permuting the characters while preserving the fusion rules, so we will not consider them in this work.

Sometimes a simple current leaves a representation fixed. When this happens, the fixed representation is called a ``fixed point'' of the current. From the MIPF corresponding to a given current, one can organize characters into orbits of that current and define an ``extended CFT'', where the extension is provided by the simple current. Generically some fields will be projected out in the extension, but others may appear corresponding to resolved fixed points.

It is not always easy to infer the modular matrices of the new extended theory in terms of those of the original theory. In particular, if the current has got fixed points, then one has to go through a non-trivial formalism to be able to write down the $S$ matrix (on the contrary, the $T$ matrix is always trivially determined). The reason is that the fixed points get ``split'' in the extensions, in the sense that each of them generates many fields with identical characters on which the action of the $S$ matrix is ambiguous. This formalism involves a set of ``$S^J$'' matrices which can be used to parameterize the full $S$ matrix. These matrices are model dependent and need to be determined case by case. They are already known for Wess-Zumino-Witten (WZW) models, for coset theories and their extensions. The next case to consider is the permutation orbifold and it is addressed here.

Consider a generic CFT and take the tensor product with itself. The tensor product theory has got a manifest $\mathbb{Z}_2$ symmetry which interchanges the two factors. We call the theory where this symmetry has been modded out from the tensor product the ``permutation orbifold''. In this thesis we only consider $\mathbb{Z}_2$ permutation orbifolds. Both the spectrum and the modular matrices have been known for quite some time, but the formalism of simple-current extensions was missing until a couple of years ago. In fact, the reason is that the permutation orbifolds admits simple currents in its spectrum and those simple currents have fixed points. Hence, the set of $S^J$ matrices was needed in order to compute the full $S$ matrix. This was a highly non-trivial task, but finally we are now able to present the answer, in the form of an ansatz, for the $S^J$ matrices of the permutation orbifold for all its simple currents. 

The formula appearing at the end of Part I is very powerful and it works perfectly (in the sense of satisfying some very stringent constraints and giving positive-integer fusion coefficients), even for very non-trivial rational CFT's with a huge number of primary fields. 

\chapter{The Permutation Orbifold}
\label{paper1}

{\flushright
{\small 
\textit{Does this Aleph exist in the heart of a stone?}\par
\textit{Did I see it there in the cellar when I saw all things,}\par
\textit{and have I now forgotten it?}\par
\textit{Our minds are porous and forgetfulness seeps in;}\par
\textit{I myself am distorting and losing,}\par
\textit{under the wearing away of the years, the face of Beatriz.}\par
\textit{(J. L. Borges, El Aleph)}\par
}
}

\section{Introduction}
In this chapter we study the fixed point resolution in simple-current extensions of two-dimensional conformal field theories (CFT's) \cite{Belavin:1984vu}. CFT's are very well established tools not only within String Theory, but also in other systems such as Condensed Matter and Quantum Information, hence representing an independent field of study in their own right. 

Symmetries play a crucial role. A CFT is by definition built on conformal symmetries, which in two dimensions are generated by an infinite-dimensional algebra, which in the simplest case is just the Virasoro algebra, but it becomes larger when additional generators are included, as in the case of $N=1$ or $N=2$ super-Virasoro algebra. 

In this work we will consider additional symmetries. The first one is the permutation symmetry. Such a symmetry is present when a CFT is made out of tensor products of smaller CFT's and when there are at least two identical factors in the product that can be permuted. The theory that remains after that the permutation symmetry has been modded out is called the ``permutation orbifold''. 

The other symmetry that we will consider is more subtle \cite{Schellekens:1989am,Intriligator:1989zw}. It exists when the CFT admits simple currents, namely fields with simple fusion rules:
\begin{equation}
(J)\cdot (i)= (Ji)\,.
\end{equation}
The word ``simple'' refers to the fact that the fusion of the current $J$ with any other field $\phi_i\equiv i$ contains only one term $Ji$ on the r.h.s. The word ``current'' refers to the role of this field as a symmetry generator. Simple currents form a cyclic Abelian group under fusion multiplication, sometimes called the center of the CFT. We will normally consider \textit{rational} CFT's, which by definition have a finite number of fields. Acting by powers of $J$ allows us to organize fields into orbits $(i,Ji,J^2i,\dots,J^{N-1}i)$, where $N$ is the order of the current, i.e. $J^N=0$ (we denote the identity field by $0$). One can also define a charge associated to the current $J$: it is the monodromy charge $Q_J(i)$ that a field $i$ carries. By definition:
\begin{equation}
Q_J(i)=h_J+h_i-h_{Ji} \quad {\rm mod}\,\,\mathbb{Z}\,,
\end{equation}
$h_i$ being the weight of the primary field $i$. 
The quantity $e^{2\pi i Q_J(i)}$ can be regarded as a symmetry generator. In order to mod out this symmetry from the theory, one has to keep only states which are invariant under this generator, namely states with integer monodromy charge, project out everything else and finally add the twisted sector. The modded-out theory contains the integer-monodromy orbits as primary fields and is often referred to as the ``extended'' conformal field theory, because the algebra has been enlarged by the inclusion of the current generator. 

In this chapter we are going to combine both the simple current and the permutation symmetries in order to study extensions of the permutation orbifold. The generic set up is as follows. We  start with a given CFT, take the tensor product of $\lambda$ copies of it and mod out by the cyclic symmetry $Z_\lambda$, which generates the full permutation group $S_\lambda$. The field content of such cyclic orbifold theories was worked out already long ago by Klemm and Schmidt \cite{Klemm:1990df} who were able to read off the twisted fields using modular invariance. Later, Borisov, Halpen and Schweigert \cite{Borisov:1997nc} introduced an orbifold induction procedure, providing a systematic construction of cyclic orbifolds, including their twisted sector, and determining orbifold characters and, in the $\lambda=2$ case, their modular transformation properties. Generalizations to arbitrary permutation groups were done by Bantay \cite{Bantay:1997ek,Bantay:1999us}.

Extensions with integer spin simple currents \cite{Schellekens:1989am,Intriligator:1989zw} are essential tools in conformal field theories (see \cite{Schellekens:1990xy} for a review). In string theory, they appear when it is needed to make projections (e.g. GSO projection) or implement constraints (such as world-sheet supersymmetry constraints, or the so-called $\beta$-constraints in Gepner models \cite{Gepner:1987vz,Gepner:1987qi}, which impose world-sheet and space-time supersymmetry). Simple current extensions are also used to implement field identification in coset models \cite{Gepner:1989jq,Schellekens:1989uf}.

The modular $S$ and $T$ matrices of the extended theory can be easily derived from those of the original theory if all the orbits generated by the current $J$ have length strictly larger than one. Length-one orbits, denoted by $f$, are fixed points of $J$, namely $J\cdot f=f$. Fixed points exist only for currents with integer or half-integer spin. For integer-spin currents, fixed points are kept in the extension. In the modular invariant partition function (MIPF), the fixed point contribution always comes with an overall multiplicative factor, typically as
\begin{equation}
N_f \sum_f \bar{\chi}_f(\bar{\tau}) \chi_f(\tau)\,.
\end{equation}
The factor $N_f$ is interpreted as the number of fields $f_\alpha\equiv(f,\alpha)$, with $\alpha=1,\dots,N_f$, all having identical characters, in which $f$ is resolved. This means that in the extended theory the single field $f$ splits up into $N_f$ fields $f_\alpha$. The resolved fields $f_\alpha$ contribute to the partition function as
\begin{equation}
\label{paper1: fp mipf contribution}
\sum_\alpha \bar{\chi}_{f,\alpha}(\bar{\tau}) \chi_{f,\alpha}(\tau)\,,
\qquad \chi_{f,\alpha}(\tau)=m_\alpha \chi_f(\tau)\,,
\qquad \sum_\alpha (m_\alpha)^2=N_f\,.
\end{equation}
However, since there is a priori no information on how the modular matrix $S$ acts on the label $\alpha$, it will be generically undetermined. In literature, this problem is known as the fixed point resolution. When this is the case, the knowledge of the full $S$ matrix is parametrized by a set of  ``$S^J$'' matrices \cite{Fuchs:1996dd}, one for each simple current $J$: knowing all the $S^J$ matrices amounts to knowing the $S$ matrix of the extended theory. Fixed points can also appear for half-integer spin currents, and the corresponding matrices $S^J$ are important when these currents are combined to form integer spin currents. Furthermore, simple current fixed points and their resolution matrices are essential ingredients for determining the boundary coefficients in a large 
class of rational CFT's \cite{Bianchi:1991rd,Fuchs:2000cm}.

The determination of fixed point matrices $S^J$ was first considered in \cite{Schellekens:1989uf}. There an empirical approach was used, based on the information that these matrices must satisfy modular group properties. Hence an {\it ansatz} could be guessed in some simple cases from the known fixed point spectrum. These {\it ans\"atze} were proved and extended in \cite{Fuchs:1995zr}. Starting from these results, the $S^J$ matrices are now known in many cases, such as for WZW models \cite{Schellekens:1990xy,Schellekens:1999yg} and coset models \cite{Schellekens:1989uf}. 

Here we would like to determine the set of $S^J$ matrices for cyclic permutation orbifolds. In this work we will restrict ourselves to $Z_2$ permutation orbifolds of an original CFT and to order-two simple currents. We will manage to determine the $S^J$ matrices in a few, but interesting, cases, namely for the integer-spin currents of the $SU(2)_2$ WZW model and for the $B(n)_1$ and $D(n)_1$ series. The method we use is based on the fact that the extensions corresponding to these cases are CFT's whose $S$ matrix can also be obtained by other means and hence it is already known. However, even though strictly speaking the $S^J$ matrices are not needed to construct the $S$ matrix of these extension, the result still provides important new information. In particular, we expect that the solutions we present here for an infinite series of special examples will give insights into the general case, and, as we will see in chapter \ref{paper3}, will lead to a universal {\it ansatz} that can be checked explicitly.

This chapter is organized as follows.\\
In section \ref{Definition of the problem} we define the problem that we would like to address, namely the resolution of the fixed points in extensions of permutation orbifolds.\\
Before going into the details of the problem, in section \ref{Simple currents of the orbifold CFT} we study a bit more systematically the structure of simple currents and corresponding fixed points in orbifold CFT's. In particular, we will see which simple currents and fixed points can arise in the orbifold theory and how they are related to the simple currents and fixed points of the mother theory. This is an application of \cite{Borisov:1997nc}.\\
Section \ref{Example SU(2)k} provides an example where the mother theory is $SU(2)_k$.\\
Next we move to the main problem, i.e. the fixed point resolution in extensions of permutation orbifolds. We present the results in section \ref{Fixed point resolution in SU(2)k orbifolds} and section \ref{Fixed point resolution in SO(N)1 orbifolds} for $SU(2)_1$ and $SO(N)_1$. We say something about arbitrary level $k$ as well.\\
Our analysis of these special cases give crucial hints to determine the general formula, valid for any CFT.
The solution to the general problem will be given in chapter \ref{paper3}. Finally, we would like to remark again that in this work we will mostly be concerned with $\mathbb{Z}_2$ permutations ($\lambda=2$) and order-two currents ($J^2=0$).

The content of this chapter is based on \cite{Maio:2009kb}.

\section{The problem}
\label{Definition of the problem}
Given a certain CFT $\mathcal{A}$, we would like to look at the orbifold theory with $\lambda=2$: 
\begin{equation}
 \mathcal{A}_{\rm perm} \equiv (\mathcal{A}\times \mathcal{A})/\mathbb{Z}_2\,.
\end{equation}
Modding out by $\mathbb{Z}_2$ means that the spectrum must contain fields that are symmetric under the interchange of the two factors. This theory admits an untwisted and a twisted sector. The untwisted fields are those combinations of the original tensor product fields that are invariant under this flipping symmetry. Their weights are simply given by the sum of the two weights of each single factor. Twisted fields are required by modular invariance. In general, for any field $\phi_i$ in the original CFT $\mathcal{A}$, there are exactly $\lambda$ twisted fields in the orbifold theory, labelled by $\psi=0,1,\dots,\lambda-1$.

If there is any integer or half-integer spin simple current in $\mathcal{A}$, it gives rise to an integer spin simple current in the orbifold CFT, which can be used to extend $\mathcal{A}_{\rm perm}$. In the extension, some fields are projected out while the remaining organize themselves into orbits of the current. Typically untwisted and twisted fields do not mix among themselves. As far as the new spectrum is concerned, these orbits become the new fields of the extended orbifold CFT, but we do not normally know the new $S$ matrix. From now on we will write $\tilde{S}$ with a tilde to denote the $S$ matrix of the extended theory.

If there are no fixed points, i.e. orbits of length one, the $S$ matrix of the extended theory, $\tilde{S}$, is simply given by the $S$ matrix of the unextended theory (in case of permutation orbifolds it is the BHS $S$ matrix given in \cite{Borisov:1997nc}) multiplied by the order of the extending simple current. Unfortunately, often this is not the case: normally there will be fixed points and the extended $S$ matrix cannot be easily determined.

Using the formalism developed in \cite{Fuchs:1996dd}, we can trade our ignorance about $\tilde{S}$ with a set of matrices $S^J$, one for every simple current $J$, according to the formula
\begin{equation}
\label{main formula for f.p. resolution}
\tilde{S}_{(a,i)(b,j)}=\frac{|G|}{\sqrt{|U_a||S_a||U_b||S_b|}}\sum_{J\in G}\Psi_i(J) S^J_{ab} \Psi_j(J)^{\star}\,,
\end{equation}
These $S^J_{ab}$'s are non-zero only if both $a$ and $b$ are fixed points. This equation can be viewed as a Fourier transform and the $S^J$'s as Fourier coefficients of $\tilde{S}$. The prefactor is a group theoretical factor acting as a normalization and the $\Psi_i(J)$'s are the group characters acting as phases. In our calculations, where all the simple currents have order two, the normalization prefactor is $1/2$ and the group characters are just signs. 
As conjectured in \cite{Fuchs:1996dd} and proved in \cite{Bantay:1996jt}, the $S^J$ matrices describe the modular transformation properties of the one-point function on the torus with the insertion of the simple current $J(z)$.
Unitarity and modular invariance of $\tilde{S}$ implies unitarity and modular invariance of the $S^J$'s \cite{Fuchs:1996dd}:
\begin{equation}
\label{modular}
S^J\cdot (S^J)^\dagger=1 \qquad (S^J\cdot T^J)^3=(S^J)^2\,.
\end{equation}
Here $T^J$ denotes the $T$ matrix of the unextended theory restricted to the fixed points of $J$.

In this way, the problem of finding $\tilde{S}$ is equivalent to the problem of finding the set of matrices $S^J$. 

The matrices $S^J$ are restricted not only by modular invariance and unitarity, but also by the condition that the full matrix 
$\tilde{S}_{(a,i)(b,j)}$ acts on a set of characters with positive integer coefficients, that the Verlinde formula \cite{Verlinde:1988sn} yields non-negative
integer coefficients and that there is a corresponding set of fusing and braiding matrices that satisfy all hexagon and pentagon
identities. In other words, all the usual conditions of rational conformal
field theory should be  satisfied. However, all these additional constraints are very hard to check, and modular invariance and unitarity
are very restrictive already. Experience so far suggests that for generic formulas ({\it i.e.} formulas valid for an entire class, as opposed
to special solutions valid only for a single RCFT) this is sufficient. We do not know any general results concerning the uniqueness
of the solutions to (\ref{modular}), but there is at least one obvious, and irrelevant ambiguity. If $S^J$ satisfies (\ref{modular}), clearly
$U^{\dagger} S^J U$ satisfies it for any unitary matrix $U$ that commutes with $T$. Since we are aiming for a generic solution, we
may assume that $T$ is non-degenerate; accidental degeneracies in specific cases cannot affect a generic formula. This reduces $U$
to a diagonal matrix of phases. The matrix $\tilde{S}_{(a,i)(b,j)}$ must be symmetric, and this has implications for the symmetry of the
matrix $S^J$. In particular, if $J$ is of order 2 (the case considered here), the matrix $S^J$ must be 
symmetric itself \cite{Fuchs:1996dd}. This requirement reduces $U$ to a diagonal matrix of signs. These signs are irrelevant:
they simply correspond to a relabeling of the two components of each resolved fixed point field. Note that the matrix $S$ itself
also satisfies (\ref{modular}), but here there is no such ambiguity: $S$ acts on positive characters, and any non-trivial sign choice
would affect the positivity of $S_{0i}$. However, $S^J$ acts on {\it differences} of characters, and hence satisfies no such 
restrictions. 

In this chapter we want to address exactly this problem, but in the case of permutation orbifolds. Suppose we know (and we do!) the $S$ matrix of the orbifold theory, then extend it by any of its simple currents; what is the matrix $\tilde{S}$ of the new extended theory? Equivalently, given the fact that there will be fixed points in the extension, what are the matrices $S^J$ for all the integer spin simple currents $J$? Hence, we are dealing with the fixed point resolution in extensions of permutation orbifolds.

\section{Currents of $\mathcal{A}_{\rm perm}$}
\label{Simple currents of the orbifold CFT}
Consider a CFT $\mathcal{A}$ which admits a set of integer-spin simple currents $J$. This means that the $S$ matrix satisfies the sufficient and necessary condition \cite{Dijkgraaf:1988tf} $S_{J0}=S_{00}$, where $0$ denotes the identity field of $\mathcal{A}$. Every CFT has at least one simple current, namely the identity. Here we would like to determine the simple currents of the orbifold theory $\mathcal{A}_{\rm perm}$. The only thing we need is the orbifold $S$ matrix given by BHS \cite{Borisov:1997nc}. 
Recall that $\mathcal{A}_{\rm perm}$ has different kinds of fields: untwisted (which are of diagonal or off-diagonal type) and twisted and that the identity field of the orbifold theory is the symmetric representation of the identity ``$0$'' of the original CFT, here denoted by $(0,0)$.

It is probably useful to recall the BHS $S$ matrix. The convention for the orbifold fields is as follows. Orbifold twisted fields carry a hat: $\widehat{(i,\psi)}$; off-diagonal fields are denoted by $\langle i,j\rangle$, with $i\neq j$; diagonal fields by $(i,\psi)$. Here $i,j$ are fields of the mother theory and $\psi=0,\dots,\lambda-1$. 
The untwisted fields are those combinations of the original tensor product fields that are invariant under this flipping symmetry. Their weights are simply given by the sum of the two weights of each single factor. There are two kinds of untwisted fields:
\begin{itemize}
\item \textit{diagonal}, $(i,\chi)$, with $\chi=0,\,1$, corresponding to the combination $\phi_i\otimes\phi'_i + (-1)^\chi \phi'_i\otimes\phi_i$, where $\phi'_i$ denotes the first non-vanishing descendant of the $\mathcal{A}-$field $\phi_i$, ($\chi=0$ for the symmetric and $\chi=1$ for the anti-symmetric representations);
\item \textit{off-diagonal}, $\langle m,n\rangle$, with $m<n$, corresponding to the combination $\phi_m\otimes\phi_n + \phi_n\otimes\phi_m$.
\end{itemize}
Twisted fields are required by modular invariance \cite{Klemm:1990df}. In general, for any field $\phi_i$ in $\mathcal{A}$, there are two twisted fields in the orbifold theory, labelled by $\chi=0,1$. We denote twisted fields by $\widehat{(i,\chi)}$. 
The typical weights of the fields are:
\begin{itemize}
\item $h_{(i,\chi)}=2 h_i$
\item $h_{\langle i,j\rangle}=h_i+h_j$
\item $h\widehat{(i,\chi)}=\frac{h_i}{2}+\frac{c}{24}\frac{(\lambda^2-1)}{\lambda}+\frac{\chi}{\lambda}$
\end{itemize}
for diagonal, off-diagonal and twisted representations. Here, $h_i\equiv h_{\phi_i}$, $c$ is the central charge of $\mathcal{A}$ and $\lambda=2$.
Sometimes it can happen that the naive ground state has dimension zero: then one must go to its first non-vanishing descendant whose weight is incremented by integers. 

\def\ket#1{|#1\rangle}

There are two possible reasons why a ``naive" ground state dimension might vanish, so that the actual ground state weight is larger by some integer value. If a ground state $i$ has dimension one, the naive dimension of $(i,1)$ vanishes. The one has to go to the first non-vanishing excited state of $i$. Similarly, the conformal weight of an excited twist field $(\chi=1)$ is
larger than that of the unexcited one $(\chi=0)$ by half an integer, unless some odd excitations of the ground state vanish. In CFT, every state $| \phi_i\rangle$, except the vacuum, always has an excited state $L_{-1} | \phi_i\rangle$. Furthermore, in $N=2$ CFT's even the
vacuum has an excited state $J_{-1} \ket{0}$. Therefore, in $N=2$ permutation orbifolds, the conformal weights of all ground states is equal to the typical values given above, except when a state $\ket{i}$ has ground state dimension 1. Then the conformal weight is larger by one unit.

The orbifold $S$ matrix for $\lambda=2$ was derived by Borisov, Halpern and Schweigert \cite{Borisov:1997nc} and reads:
\begin{eqnarray}
\label{BHS off-diag}
S_{\langle i,j\rangle\langle p,q\rangle}&=&S_{ip}\,S_{jq}+S_{iq}\,S_{jp} \nonumber\\
S_{\langle i,j\rangle(p,\psi)}&=&S_{ip}\,S_{jp} \nonumber\\
S_{\langle i,j\rangle\widehat{(p,\psi)}}&=&0
\end{eqnarray}
\begin{eqnarray}
\label{BHS diag}
S_{(i,\psi)(j,\chi)}&=&\frac{1}{2}\,S_{ij}\,S_{ij} \nonumber\\
S_{(i,\psi)\widehat{(p,\chi)}}&=&\frac{1}{2}\,e^{2\pi i\psi/2} \,S_{ip}
\end{eqnarray}
\begin{eqnarray}
\label{BHS twisted}
S_{\widehat{(p,\psi)}\widehat{(q,\chi)}}&=&\frac{1}{2}\,e^{2\pi i(\psi+\chi)/2} \,P_{ip}
\end{eqnarray}
where the $P$ matrix is defined by $P=\sqrt{T}ST^2S\sqrt{T}$, as first introduced in \cite{Bianchi:1990yu}. Sometimes we will write $S^{BHS}$ to refer to the orbifold $S$ matrix.

\subsection{Simple currents}
Let us start with the off-diagonal fields of the orbifold and ask if any of them can be a simple current. If $i$ and $j$ are two arbitrary fields of the original CFT $\mathcal{A}$ and $\langle i,j\rangle$ the corresponding off-diagonal field in the orbifold, in order for the latter to be a simple current we have to demand that its $S^{BHS}$ matrix satisfies
\begin{equation}
 S_{\langle i,j\rangle(0,0)}=S_{(0,0)(0,0)}
\end{equation}
which, upon using BHS formula, amounts to satisfying the constraint
\begin{equation}
 S_{i0} S_{j0}=\frac{1}{2} S_{00} S_{00}
\end{equation}
for the $S$ matrix of the original CFT $\mathcal{A}$. This relation is never satisfied because of the constraint $S_{i0}\geq S_{00}$, which holds for unitary CFT's. Consequently there are no simple currents coming from off-diagonal fields.

Let us do the same analysis for twisted fields. Twisted fields are denoted by $\widehat{(k,\psi)}$, where $k$ is a field in $\mathcal{A}$ and $\psi=0,\,1$. Now the constraint 
 \begin{equation}
 S_{\widehat{(k,\psi)}(0,0)}=S_{(0,0)(0,0)}
\end{equation}
translates into
\begin{equation}
 \frac{1}{2}S_{k0}=\frac{1}{2} S_{00} S_{00}\,.
\end{equation}
This is also never satisfied, because of the same unitarity constraints as before. Once again there are no simple currents coming from twisted fields. 

Finally let us study the more interesting situation of diagonal fields as simple currents. A diagonal field is denoted by $(i,\psi)$, where $i$ is a field in $\mathcal{A}$ and $\psi=0,\,1$ corresponding respectively to symmetric and anti-symmetric representation. Here the constraint
\begin{equation}
 S_{(i,\psi)(0,0)}=S_{(0,0)(0,0)}
\end{equation}
gives
\begin{equation}
 \frac{1}{2}S_{i0} S_{i0}=\frac{1}{2} S_{00} S_{00}\,,
\end{equation}
which is satisfied if and only if $i$ is a simple current. 

Hence we conclude that, despite the fact that the existence of simple currents in the orbifold theory is in general related to the $S$ matrix of the original CFT, there always exist definite simple currents in the orbifold theory: they are the symmetric and anti-symmetric representations of those diagonal fields corresponding to the simple currents of the original theory. In particular, since in $\mathcal{A}$ there is at least one simple current, namely the identity, in $\mathcal{A}_{\rm perm}$ there will be at least two, namely $(0,0)$ (trivial, because it plays the role of the identity) and $(0,1)$. The latter, known as the un-orbifold current for reasons that will become clear later on, will turn out to play a crucial role.

We will soon see that this pattern is respected for $SU(2)_k$ WZW models. They admit one integer-spin simple current (the identity) for $k$ odd and two (one of which is again the identity) integer-spin simple currents for $k$ even. Consequently, we will always find $(0,0)$ and $(0,1)$ as orbifold simple currents when $k$ is odd; when $k$ is even, there will be two additional ones denoted by $(k,0)$ and $(k,1)$.

\subsection{Fixed points}
Given our simple currents of the $\mathcal{A}_{\rm perm}$ theory, hereafter denoted by $(J,\psi)$ with $J$ a simple current of $\mathcal{A}$ and $\psi=0\,,1$, we now move on to study the structure of their fixed points. For this purpose, the correct strategy is to compute the fusion coefficients. 

\subsubsection{Twisted sector}
Let us start from the twisted sector. For twisted fixed points we have to demand that
\begin{equation}
 N_{(J,\phi)\widehat{(f,\psi)}}^{\phantom{(J,\phi)\widehat{(f,\psi)}}\widehat{(f,\psi)}}=1\,.
\end{equation}
On the other hand, if $N$ is an arbitrary field of the orbifold theory, in terms of the $S$ and $P$ matrix of the original theory we have
\begin{eqnarray}
N_{(J,\phi)\widehat{(f,\psi)}}^{\phantom{(J,\phi)\widehat{(f,\psi)}}\widehat{(f,\psi)}} 
&=&  \sum_{N} 
\frac{S_{(J,\phi)N}S_{\widehat{(f,\psi)}N}S_{\phantom{\dagger}N}^{\dagger\phantom{N}\widehat{(f,\psi)}}}{S_{(0,0)N}}=
\nonumber\\  &=&  
\sum_{\langle p,q\rangle}
\frac{S_{(J,\phi)\langle p,q\rangle} S_{\widehat{(f,\psi)}\langle p,q\rangle} S_{\phantom{\dagger}\langle p,q\rangle}^{\dagger\phantom{\langle p,q\rangle}\widehat{(f,\psi)}}}{S_{(0,0)\langle p,q\rangle}}+
 \nonumber\\ &+&  
 \sum_{(j,\chi)}
\frac{S_{(J,\phi)(j,\chi)}S_{\widehat{(f,\psi)}(j,\chi)}S_{\phantom{\dagger}(j,\chi)}^{\dagger\phantom{(j,\chi)}\widehat{(f,\psi)}}}{S_{(0,0)(j,\chi)}}+
 \nonumber\\ &+& 
 \sum_{\widehat{(p,\chi)}}
\frac{S_{(J,\phi)\widehat{(p,\chi)}}S_{\widehat{(f,\psi)}\widehat{(p,\chi)}}S_{\phantom{\dagger}\widehat{(p,\chi)}}^{\dagger\phantom{\widehat{(p,\chi)}}\widehat{(f,\psi)}}}{S_{(0,0)\widehat{(p,\chi)}}}=
 \nonumber\\ 
&=& (BHS)= \nonumber\\ 
&=& \frac{1}{2} \sum_j 
 \left[ \frac{(S_{Jj})^2 }{(S_{0j})^2} S_{fj} S_{\phantom{\dagger}j}^{\dagger\phantom{j}f}+
 e^{i\pi\phi} \frac{S_{Jj}P_{fj} P_{\phantom{\dagger}j}^{\dagger\phantom{j}f}}{S_{0j}}  \right]\,.
\end{eqnarray}
More in general one has
\begin{equation}
 N_{(J,\phi)\widehat{(f,\psi)}}^{\phantom{(J,\phi)\widehat{(f,\psi)}}\widehat{(f',\psi')}} =
 \frac{1}{2} \sum_j 
 \left[ \frac{(S_{Jj})^2 }{(S_{0j})^2} S_{fj} S_{\phantom{\dagger}j}^{\dagger\phantom{j}f'}+
 e^{i\pi(\phi+\psi-\psi')} \frac{S_{Jj}P_{fj} P_{\phantom{\dagger}j}^{\dagger\phantom{j}f'}}{S_{0j}}  \right]\,.
\end{equation}

It is important to remember that here we want $\widehat{(f,\psi)}$ to be a fixed point of $(J,\phi)$, i.e.
\begin{equation}
 N_{(J,\phi)\widehat{(f,\psi)}}^{\phantom{(J,\phi)\widehat{(f,\psi)}}\widehat{(f',\psi')}} = 
 \delta_f^{f'} \delta_\psi^{\psi'}\,.
\end{equation}
By itself, $f$ does not have to be a fixed point of $J$ in the original theory.
For an arbitrary field $i$, the following is true \cite{Intriligator:1989zw,Schellekens:1989dq}:
\begin{equation}
\label{simple-current-fundamental-relation}
 \frac{S_{Ji}}{S_{0i}}=e^{{2 \pi i (h_J +h_i -h_{J\cdot i})}}\,.
\end{equation}
In the exponent, we recognize the monodromy charge $Q_J(i)$ of $i$ with respect to $J$:
\begin{equation}
Q_J(i)=h_J+h_i-h_{J\cdot i} \,\,\,\,{\rm mod}\,\,\mathbb{Z}\,.
\end{equation}
Now use formula (\ref{simple-current-fundamental-relation}) in the first sum. In the following, we will restrict ourselves to order-2 simple currents. Because of the square and the fact that the monodromy charge of $j$ is either integer of half-integer\footnote{For order-2 simple currents.}, the exponent cancels out. Then we are left with $S$ times $S^\dagger$, which gives $\delta_f^{f'}$. 

We need to be more careful with the second piece, which involves the integer-valued \cite{Pradisi:1995qy,Gannon:1999is} $Y_{Jf}^{\phantom{Jf}f'}$-tensor. Our constraint reads then
\begin{equation}
 \delta_f^{f'} \delta_\psi^{\psi'} = \frac{1}{2}\delta_f^{f'} + 
 e^{i\pi(\phi+\psi-\psi')}\frac{1}{2}Y_{Jf}^{\phantom{Jf}f'}\,,
\end{equation}
which reduces either to
\begin{equation}
 e^{i\pi\phi}Y_{Jf}^{\phantom{Jf}f'}=\delta_f^{f'}\qquad (\psi=\psi')\,,
\end{equation}
when $\psi=\psi'$, or to
\begin{equation}
 e^{i\pi(\phi+\psi-\psi')}Y_{Jf}^{\phantom{Jf}f'}=-\delta_f^{f'}\qquad (\psi\neq\psi')\,,
\end{equation}
when $\psi\neq\psi'$. 
Since we are considering currents with order 2, we can simplify the minus sign on the r.h.s. with $e^{i\pi(\psi-\psi')}$ on the l.h.s., thus re-obtaining the same expression of the case $\psi=\psi'$ for our constraint, which explicitly reads:
\begin{equation}
\label{Paper1: SPP/S}
 e^{i\pi\phi} \sum_j \frac{S_{Jj}P_{fj} P_{\phantom{\dagger}j}^{\dagger\phantom{j}f'}}{S_{0j}}=\delta_f^{f'}\,.
\end{equation}
In order to solve it, let us study for the moment the equation:
\begin{equation}
 \sum_j x_j P_{fj} P_{\phantom{\dagger}j}^{\dagger\phantom{j}f'}=\delta_f^{f'}\,,
\end{equation}
for some $x_j$. Define a vector $v_f$ with components
\begin{equation}
 (v_f)_j:=x_j P_{fj}\,.
\end{equation}
Then we have
\begin{equation}
 \sum_j (v_f)_j P_{\phantom{\dagger}j}^{\dagger\phantom{j}f'}=\delta_f^{f'}\,.
\end{equation}
The vector $v_f$ is then orthogonal to all the columns of the matrix $P$, except for the column $f$ with which it has unit scalar product. Since $P$ is unitary, this implies that
\begin{equation}
 (v_f)_j=P_{fj}\,,
\end{equation}
which by definition yields\footnote{A shorter derivation is the following. Consider a diagonal matrix $X$ whose diagonal entries are $x_j$. Then the constraint in matrix form is: $PXP^\dagger=1$. Recalling that $PP^\dagger=1$ by unitarity, one can write $P(X-1)P^\dagger=0$, which gives the solution $X=1$.}
\begin{equation}
 x_j=1 \qquad \forall j\,.
\end{equation}

Going back to our situation where $x_j=e^{i\pi\phi} S_{Jj}/S_{0j}$, we arrive at the final form of our constraint:
\begin{equation}
 e^{i\pi\phi}S_{Jj}=S_{0j}\,.
\end{equation}
Let us first notice that when $J$ is the identity, there is no news, since this constraint is either trivially satisfied (for $\psi=0$ all the twisted fields are fixed points of the identity) or impossible (for $\psi=1$ there are no fixed points coming from the twisted sector). When instead $J$ is not the identity, we find that $\widehat{(p,\chi)}$ is a fixed point of $(J,\psi)$ in the following cases (according to (\ref{simple-current-fundamental-relation})):
\begin{itemize}
\item if $\psi=0$, when $p$ has integer monodromy charge with respect to $J$, i.e. $Q_J(p)=0$;
\item if $\psi=1$, when $p$ has half-integer monodromy charge with respect to $J$, i.e. $Q_J(p)=\frac{1}{2}$.
\end{itemize}
These conditions hold for integer-spin currents. Generalized expressions will be needed for currents with half-integer spin. We will give them later.

\subsubsection{Off-diagonal fields}
Similar arguments apply for the untwisted sector. Starting with off-diagonal fixed points one has
\begin{eqnarray}
 N_{(J,\phi)\langle p,q\rangle}^{\phantom{(J,\phi)\langle p,q\rangle}\langle p,q\rangle} 
 &=& \sum_N 
 \frac{S_{(J,\phi)N}S_{\langle p,q\rangle N} 
S_{\phantom{\dagger}N}^{\dagger\phantom{N}\langle p,q\rangle}}{S_{(0,0)N}}=
 \nonumber\\  &=& 
 \sum_{\langle i,j\rangle}
 \frac{S_{(J,\phi)\langle i,j\rangle} S_{\langle p,q\rangle\langle i,j\rangle} 
S_{\phantom{\dagger}\langle i,j\rangle}^{\dagger\phantom{\langle i,j\rangle}\langle p,q\rangle}}{S_{(0,0)\langle i,j\rangle}}+
 \nonumber\\ &+&  
 \sum_{(i,\psi)}
\frac{S_{(J,\phi)(i,\psi)}S_{\langle p,q\rangle(i,\psi)} S_{\phantom{\dagger}(i,\psi)}^{\dagger\phantom{(i,\psi)}\langle p,q\rangle}}{S_{(0,0)(i,\psi)}}+
 \nonumber\\ &+&  
 \sum_{\widehat{(i,\psi)}}
\frac{S_{(J,\phi)\widehat{(i,\psi)}}S_{\langle p,q\rangle\widehat{(i,\psi)}} S_{\phantom{\dagger}\widehat{(i,\psi)}}^{\dagger\phantom{\widehat{(i,\psi)}}\langle p,q\rangle}}{S_{(0,0)\widehat{(i,\psi)}}}=
 \nonumber\\ &=&   (BHS) =\nonumber\\ &=&
 N_{Jp}^{\phantom{Jp}p} N_{Jq}^{\phantom{Jq}q} + N_{Jp}^{\phantom{Jp}q} N_{Jq}^{\phantom{Jq}p}\,.
\end{eqnarray}
This must be equal to $1$. Moreover $N_{ij}^{\phantom{ij}k}$ are positive integers. Hence we have two possibilities:
\begin{itemize}
 \item either
 \begin{equation}
\left\{
 \begin{array}{l}
   N_{Jp}^{\phantom{Jp}p} = N_{Jq}^{\phantom{Jq}q} = 1 
   \qquad \Rightarrow \,\,p \,\,\& \,\,q {\rm \,\,are\,\, fixed \,\,points\,\, of\,\,} J\\ 
   N_{Jp}^{\phantom{Jp}q} = N_{Jq}^{\phantom{Jq}p} = 0 
 \end{array}
\right.
\end{equation}
\item or
 \begin{equation}
 \left\{
  \begin{array}{l}
   N_{Jp}^{\phantom{Jp}p} = N_{Jq}^{\phantom{Jq}q} = 0 \\ 
   N_{Jp}^{\phantom{Jp}q} = N_{Jq}^{\phantom{Jq}p} = 1 
   \qquad \Rightarrow \,\,p \,\,\& \,\,q {\rm \,\,are\,\, in\,\,the\,\,same\,\,} J{\rm -orbit, i.e. \,\,} p=Jq
  \end{array}
 \right.
 \end{equation}
\end{itemize}

\subsubsection{Diagonal fields}
For diagonal fixed points one has
\begin{eqnarray}
 N_{(J,\phi)(i,\psi)}^{\phantom{(J,\phi)(i,\psi)}(i,\psi)} 
 &=& \sum_N 
 \frac{S_{(J,\phi)N}S_{(i,\psi)N}S_{\phantom{\dagger}N}^{\dagger\phantom{N}(i,\psi)}}{S_{(0,0)N}}=
 \nonumber\\  &=& 
 \sum_{\langle p,q\rangle}
\frac{S_{(J,\phi)\langle p,q\rangle}S_{(i,\psi)\langle p,q\rangle} S_{\phantom{\dagger}\langle p,q\rangle}^{\dagger\phantom{(i,\psi)}(i,\psi)}}{S_{(0,0)\langle p,q\rangle}}+
 \nonumber\\ &+&  
 \sum_{(j,\chi)}
\frac{S_{(J,\phi)(j,\chi)} S_{(i,\psi)(j,\chi)}S_{\phantom{\dagger}(j,\chi)}^{\dagger\phantom{(j,\chi)}(i,\psi)}}{S_{(0,0)(j,\chi)}}+
 \nonumber\\ &+&  
 \sum_{\widehat{(j,\chi)}}
\frac{S_{(J,\phi)\widehat{(j,\chi)}} S_{(i,\psi)\widehat{(j,\chi)}}S_{\phantom{\dagger}\widehat{(j,\chi)}}^{\dagger\phantom{\widehat{(j,\chi)}}(i,\psi)}}{S_{(0,0)\widehat{(j,\chi)}}}=
 \nonumber\\ &=&   (BHS) =\nonumber\\ &=&
 \frac{1}{2} N_{Ji}^{\phantom{Ji}i}(N_{Ji}^{\phantom{Ji}i} + e^{i\pi\phi})\,.
\end{eqnarray}
Again we must demand 
\begin{equation}
 N_{(J,\phi)(i,\psi)}^{\phantom{(J,\phi)(i,\psi)}(i,\psi)} =1\,;
\end{equation}
then the only solution is when\footnote{We can exclude the other possibility $\phi=1$ and $N_{Ji}^{\phantom{Ji}i}=2$, because $J$ is a simple current.} $N_{Ji}^{\phantom{Ji}i}=1$, i.e. $i$ is a fixed point of $J$, and $\phi=0$, i.e. these fixed points appear only for the symmetric diagonal representation of the simple current.

\section{Example: $SU(2)_k$}
\label{Example SU(2)k}
Here we consider some examples of the previous general theory. We take our CFT to be an $SU(2)_k$ WZW model and work out spectrum and fusion rules of the orbifold theory.

Let us recall a few facts about affine Lie algebras \cite{Knizhnik:1984nr,Gepner:1986wi}. In an affine Lie algebra with group $G$, the weights of the highest weight representations $\lambda$ are given by
\begin{equation}
h(\lambda)=\frac{\frac{1}{2}C(\lambda)}{k+g}\,,
\end{equation}
where $C(\lambda)$ denotes the quadratic Casimir eigenvalue, $g$ is the dual Coxeter number (equal to half the Casimir of the adjoint representation) and $k$ is the level. The central charge is
\begin{equation}
c(G,k)=\frac{k\,{\rm dim}\,G}{k+g}
\end{equation}
and the matrix element is
\begin{equation}
S(\lambda,\mu)=const \cdot \sum_w \epsilon(w)\exp\left(-\frac{2\pi i}{k+g}(w(\lambda+\delta),\mu+\delta)\right)\,.
\end{equation}
Here the sum is over all the elements of the Weyl group and $\epsilon$ is the determinant of $w$. The normalization constant is fixed by unitarity and the requirement $S_{00}>0$.

Now we can apply these general pieces of information to our $SU(2)_k$ models (and later to $B(n)_1$ and $D(n)_1$ series).

\subsection{Generalities about $SU(2)_k$ WZW model}
In the $SU(2)_k$ theory, the level $k$ specifies both the central charge
\begin{equation}
 c=\frac{3 k}{k+2}
\end{equation}
and the spectrum of the primary fields through their weights
\begin{equation}
 h_{2j} =\frac{j (j+1)}{k+2},\qquad 2j=0,1,\dots k.
\end{equation}
Moreover, the field corresponding to the last value $2j=k$ is a simple current\footnote{Note that $j$ is either integer or half-integer. An equivalent notation is to set $l=2j$, with $l=0,\dots,k$, and hence $h_l=\frac{l(l+2)}{4(k+2)}$.} of order two, the fusion being:
\begin{equation}
 (k)\times (2j) = (k-2j).
\end{equation}
Its weight is $h_{2j=k}=\frac{k}{4}$. This is integer or half-integer if $k$ is even. Furthermore, in the latter case, there is also a fixed point, given by the median value $2j=\frac{k}{2}$:
\begin{equation}
 (k)\times (\frac{k}{2}) = (\frac{k}{2}).
\end{equation}
There are no fixed points for odd $k$.

We can label these $k+1$ fields using their value of $j$. It will be convenient to call them
\begin{equation}
 \{\phi_{2j}\}=\{\phi_0 \,,\quad \phi_1\,,\quad\dots\quad\,,\quad \phi_k\,\}. 
\end{equation}
The $S$ matrix is given by \cite{Cappelli:1986hf}
\begin{equation}
 S_{2j,\,2m}=\sqrt{\frac{2}{k+2}}\sin{\left[\frac{\pi}{k+2}(2j+1)(2m+1)\right]}.
\end{equation}

\subsection{$SU(2)_k \otimes SU(2)_k /\mathbb{Z}_2$ Orbifold: field spectrum}
Now let us consider the orbifold theory at some particular level $k$. The notation we will be using is as follows. First of all we need to distinguish the three types of fields in the orbifold theory: diagonal, off-diagonal and twisted fields.

Diagonal fields are generated by taking the symmetric tensor product of each field in the original theory with itself or the antisymmetric tensor product with the same field with its first non-vanishing descendant. Hence there are $2(k+1)$ diagonal fields, that will be denoted as:
\begin{equation}
 (2j,\psi)\qquad\psi=0,\,1
\end{equation}
with $2j=0,1,\dots k$. Here $\psi=0$ ($\psi=1$) labels the symmetric (anti-symmetric) representation.
These fields have weights
\begin{equation}
 h_{(2j,\psi)}=2\,\frac{j(j+1)}{k+2}+\delta_{2j,0}\delta_{\psi,1}.
\end{equation}
The factor $2$ in front comes from the sum of weights of the fields appearing in the tensor product. In the anti-symmetric representation ($\psi=1$) of the identity ($2j=0$), one has to include the contribution to the weight coming from the Virasoro operators $L_{-1}$. The ground state is degenerate with dimension three due to the three $SU(2)$ generators.

Off-diagonal fields are obtained by taking the symmetric tensor product of each field in the original theory with a different field. Hence there are $\frac{k(k+1)}{2}$ off-diagonal fields, that will be denoted as:
\begin{equation}
 \langle\phi_{2i},\phi_{2j}\rangle\qquad 2i<2j.
\end{equation}
These fields have weights
\begin{equation}
 h_{\langle\phi_{2i},\,\phi_{2j}\rangle}=\frac{i(i+1)}{k+2}+\frac{j(j+1)}{k+2},
\end{equation}
which is simply the sum of the weights of the fields in the tensor product.

Twisted fields of any permutation orbifold theory were described in \cite{Klemm:1990df}. After adapting their result to our $\mathbb{Z}_2$ orbifold, we find that there are two twisted fields associated to each primary of the original theory. Hence there are $2(k+1)$ twisted fields, that will be denoted as:
\begin{equation}
 \widehat{(2j,\psi)}\qquad \psi=0,\,1,
\end{equation}
with $2j=0,1,\dots k$ as usual. Their weights are given by:
\begin{equation}
 h_{\widehat{(2j,\psi)}}=\frac{1}{2}\left[\frac{j(j+1)}{k+2}+\psi\right]+\frac{3k}{16(k+2)}.
\end{equation}

The next step is to compute the $S$ matrix for this orbifold theory using the BHS formulas (\ref{BHS off-diag}, \ref{BHS diag}, \ref{BHS twisted}). Using the Verlinde formula \cite{Verlinde:1988sn} we will then be able to compute the fusion rules, which will allow us to look for simple currents in the orbifold theory.

\subsection{$SU(2)_k \otimes SU(2)_k /\mathbb{Z}_2$ Orbifold: currents and fixed points}
\label{su2 fixed point table}
From the results corresponding to a few values of $k$, we can determine important generalizations for arbitrary $k$.\\
First of all, for all $k$ there is at least one non-trivial integer spin simple current, namely $(0,1)$ with $h=1$, whose fixed points are all the off-diagonal fields. Their number is $\binom{k+1}{2}=\frac{k(k+1)}{2}$.\\
In addition, if $k$ is even, there are  other two integer spin simple currents\footnote{These are actually the only ones with integer spin.}. They are the symmetric and anti-symmetric diagonal fields corresponding to the last value $2j=k$: $(k,0)$ and $(k,1)$, both with $h=\frac{k}{2}$. This reflects the general structure of the $SU(2)_k$ simple currents. Their fixed points are also easily determined. For the current $(k,0)$ they come from diagonal, off-diagonal and twisted fields according to some rules which are given below, while those of $(k,1)$ come only from off-diagonal and twisted fields.

Summarizing:

\begin{tabular}{l l}
&\\
 Simple current & Fixed point\\
\hline
$(0,1)$, \, $h=1$ & \textit{all the $\frac{k(k+1)}{2}$ off-diagonal fields}\\
$(k,0)$, \, $h=\frac{k}{2}$ & \textit{$2$ diag. + $\frac{k}{2}$ off-diag. + $(k+2)$ twisted fields}\\
$(k,1)$, \, $h=\frac{k}{2}$ & \textit{$\frac{k}{2}$ off-diag. + $k$ twisted fields}\\
&\\
\end{tabular}

The rule to construct the fixed points of the additional simple currents when $k$ is even is as follows. 

The diagonal fields appearing as fixed points of $(k,0)$ are always the two fields in the middle: $(\frac{k}{2},0)$ and $(\frac{k}{2},1)$. These are $\frac{k}{2}$ and have weights
\begin{equation}
 h_{(\frac{k}{2},0)}=h_{(\frac{k}{2},1)}=\frac{1}{8}\frac{k(k+4)}{k+2}\,.
\end{equation}

The off-diagonal fields appearing as fixed points are the same for both the two additional currents and are given by the fields $\langle\phi_{2i} ,\phi_{k-2i}\rangle$, i.e. the fields $2i$ and $k-2i$ belong to the same orbit under $J\equiv \phi_k$. The weights of these off-diagonal fixed points are:
\begin{equation}
 h_{\langle\phi_{2i} ,\phi_{k-2i}\rangle}=\frac{1}{k+2}\left[i^2+\left(\frac{k}{2}-i\right)^2+\frac{k}{2}\right],
\end{equation}
with $2i=0,1,\dots,k$.

The fixed points coming from the twisted sector are ``complementary'' for the two additional simple currents, in the sense that $(k,0)$ has $\widehat{(4j,\psi)}$, $\psi=0,\,1$ and $2j=0,1,\dots,k$, as fixed points\footnote{Explicitly, these fixed points are $\widehat{(0,\psi)},\,\,\widehat{(2,\psi)},\,\,\widehat{(4,\psi)},\,\dots,\,\widehat{(k,\psi)}$, $\psi=0,\,1$, with the first argument even. In total, there are $k+2$ of them.}, while $(k,1)$ has $\widehat{(4j+1,\psi)}$, $\psi=0,\,1$ and $2j=0,1,\dots,k-1$, as fixed points\footnote{Explicitly, these fixed points are $\widehat{(1,\psi)},\,\,\widehat{(3,\psi)},\,\,\widehat{(5,\psi)},\,\dots,\,\widehat{(k-1,\psi)}$, $\psi=0,\,1$, with the first argument odd. In total, there are $k$ of them.}. Their weights are:
\begin{equation}
h_{\widehat{(4j,\psi)}}=\frac{1}{2}\left[\frac{2j(2j+1)}{k+2}+\psi\right]+\frac{3}{16(k+2)}
\end{equation}
and
\begin{equation}
h_{\widehat{(4j+1,\psi)}}=\frac{1}{2}\left[\frac{1}{k+2}\left(2j+\frac{1}{2}\right)\left(2j+\frac{1}{2}+1\right)+\psi\right]+\frac{3}{16(k+2)}
\end{equation}
for $\widehat{(4j,\psi)}$ and $\widehat{(4j+1,\psi)}$ respectively.

\subsection{Fixed point resolution in $SU(2)_k$ orbifolds}
\label{Fixed point resolution in SU(2)k orbifolds}
We would like to determine the $S^J$ matrices corresponding to the simple currents given above using formula (\ref{main formula for f.p. resolution}) which relates them to the $S$ matrix of the extended theory via the group characters $\Psi_i(J)$. As we will now explain, we know what the $S^J$ matrix is in the case $J\equiv (0,1)$. It is given by an expression analogous to the off-diagonal/off-diagonal BHS $S$ matrix, but with a minus (instead of the plus) sign. This is a fortunate situation because the current $J\equiv (0,1)$ is omnipresent, since it appears for all values of the level $k$. The other two currents that appear occasionally are slightly more complicated since they involve twisted fields.

\subsection{$S^J$ matrices}
\subsubsection{$S^J$ matrix for $J\equiv (0,1)$}
The general procedure when we make an extension via integer spin simple currents is as follows: keep states that are invariant under the symmetry generated by the current, namely those with integer monodromy charge w.r.t. $J$, and organize fields into orbits. Fixed points are particular orbits: orbits with length one. 

Consider the current $J\equiv (0,1)$ of order $2$. The extension projects out the twisted fields, since they are all non-local w.r.t. this current. Only untwisted fields are left, both diagonal and off-diagonal. Off-diagonal fields are fixed points of $(0,1)$, so they get doubled by the extension, while diagonal fields group themselves into orbits of length two containing symmetric and anti-symmetric representation of each original field. It is interesting to see that the resulting theory is equal to the tensor product $SU(2)_k\otimes SU(2)_k$. What happens is the following. The length-two orbits come from diagonal fields and correspond to fields $\phi_{2i}\otimes \phi_{2i}$ of the tensor product, while the two fields coming from the fixed points correspond to $\phi_{2i}\otimes \phi_{2j}$ and $\phi_{2j}\otimes \phi_{2i}$ (with $2i \neq 2j$) of the tensor product. The weights indeed match exactly. So in the end we have the result:
\begin{equation}
\left(\mathcal{A}\otimes \mathcal{A}/\mathbb{Z}_2 \right)_{(0,1)}=\mathcal{A}\otimes \mathcal{A} ·
\end{equation}
The subscript $(0,1)$ means that we are taking the extension by the $(0,1)$ current.
This result is not limited to $\mathcal{A}=SU(2)_k$, but is true for any rational CFT. The reason is that this simple current extension is in fact the inverse of the permutation orbifold procedure. This justifies the name of \textit{un-orbifold} current to denote the field $(0,1)$. The argument follows from the
fact that the permutation orbifold splits the original chiral algebra in a symmetric and an anti-symmetric part,
and the representation space of the current $(0,1)$ is precisely the latter. By extending the chiral algebra with
this current we re-constitute the original chiral algebra of $\mathcal{A}\otimes \mathcal{A}$. This result extends
straightforwardly to the other representations, and of course the twisted field must be projected out, since by construction 
they are non-local with respect to $\mathcal{A}\otimes \mathcal{A}$.

Resolving the fixed points is equivalent to finding a set of $S^J$ matrices such that
\begin{equation}
\label{main formula for f.p. resolution 2}
\tilde{S}_{(a,i)(b,j)}=\frac{|G|}{\sqrt{|U_a||S_a||U_b||S_b|}}\sum_{J\in G}\Psi_i(J) S^J_{ab} \Psi_j(J)^{\star}\,,
\end{equation}
where $\tilde{S}$ is the full extended $S$ matrix, $a$ and $b$ denote the fixed points of $J$, while $i$ and $j$ the fields into which the fixed points are resolved. For $J\equiv (0,1)$ we know that the extended theory is the tensor product theory, whose $S$ matrix is the tensor product of the $S$ matrices of the two factors. When we extend w.r.t. $(0,1)$, only two terms contribute on the r.h.s., namely $S^0\equiv S^{BHS}$ and $S^J$. The indices $a$ and $b$ run over the off-diagonal fields. Hence it is natural to write down the following ansatz for $S^J$ for $J=(0,1)$:
\begin{equation}
\label{SJ_(offdiag.-offdiag.)}
S^J_{\langle mn\rangle\langle pq\rangle}=S_{mp}S_{nq}-S_{mq}S_{np}\,.
\end{equation}
This is unitary and satisfies the modular constraint $(S^JT^J)^3=(S^J)^2$. Here $S_{mp}$ is the $S$ matrix of the original theory\footnote{As an exercise, one could try to write this $S^J$ matrix explicitly for $k=2$. With our conventional choice for the labels of the fields, it turns out to be numerically equal to minus the $S$ matrix of the original $SU(2)_2$ theory isomorphic to the Ising model: $S^J=-S_{SU(2)_2}$.\label{footnote_SJ}}. Note that there is an apparent sign ambiguity: the matrix elements
depend on the labelling of the off-diagonal fields, because the field $\langle p,q\rangle$ might just as well have been labelled $\langle q,p\rangle$. According to our previous discussion, this is irrelevant, since it merely amounts to a basis choice among the two split fields originating from $\langle p,q\rangle$. It is easy to check that the matrix $\tilde S$ computed with  (\ref{main formula for f.p. resolution}) is indeed the one of the tensor product, {\it i.e.} $S_{mp}S_{nq}$.

\subsubsection{$S^J$ matrix for $J\equiv (k,0)$}
The order-$2$ current $J\equiv (k,0)$ arises only when $k$ is even, so in this subsection we will restrict to such values. The first thing we need to do is to determine the orbits of the current, since they become the fields of the extended theory.

Either by looking at explicit low values of $k$ or by general arguments, one can observe a few facts about orbits of $J\equiv (k,0)$.\\
First, form the diagonal sector, $J$ couples symmetric (anti-symmetric) representation of a field $\phi_{2j}$ with symmetric (anti-symmetric) representation of its image $J\cdot \phi_{2j}=\phi_{k-2j}$ into length-2 orbits. In particular, the field $(\frac{k}{2},0)$ can couple only to itself, hence it must be a fixed point. Similarly for the field $(\frac{k}{2},1)$. So, there are exactly $k$ length-2 orbits and two fixed points coming from diagonal fields.\\
Secondly, from the off-diagonal sector, only $\langle\phi_{2i},\phi_{2j}\rangle$ with $2i$ and $2j$ either both even or both odd  survive the projection, because only those have a well-defined monodromy charge. Moreover, $J$ couples the field $\langle\phi_{2i},\phi_{2j}\rangle$ with its image $J\cdot\langle\phi_{2i},\phi_{2j}\rangle=\langle\phi_{k-2i},\phi_{k-2j}\rangle$. In particular, fields of the form $\langle\phi_{2j},\phi_{k-2j}\rangle$ must be fixed points. There are $\frac{1}{2}\left((\frac{k}{2})^2-\frac{k}{2}\right)$ length-2 orbits and $\frac{k}{2}$ fixed points coming from off-diagonal fields. In this formula, we divide by $2$ because generically fields are coupled into orbits. The contribution within brackets comes from the number of off-diagonal fields that are not projected out minus the number of off-diagonal fixed points.\\
Finally, there are no orbits coming from the twisted sector, but only $k+2$ fixed points.

Putting everything together, the theory extended by $J\equiv (k,0)$ has $3k+8$ fixed points (i.e. twice the number given in section \ref{su2 fixed point table}) plus $\frac{k(k+6)}{8}$ length-2 orbits.

Here an ansatz for $S^J$ is at this stage unknown for generic values of the level $k$. However, we have worked out the simpler case $k=2$, which is closely related to the Ising model. We will discuss it shortly.

\subsubsection{$S^J$ matrix for $J\equiv (k,1)$}
Also in this case $k$ must be even in order for the current $J\equiv (k,1)$ to be present. The orbit structure here is, \textit{mutatis mutandis}, analogous to the previous one.\\
From the diagonal sector, $J$ couples symmetric (anti-symmetric) representation of a field $\phi_{2j}$ with anti-symmetric (symmetric) representation of its image $J\cdot \phi_{2j}=\phi_{k-2j}$ into length-2 orbits. In particular, the fields $(\frac{k}{2},0)$ and $(\frac{k}{2},1)$ must couple to each other, contributing an additional orbit. There are exactly $k+1$ length-2 orbits and no fixed points coming from diagonal fields.\\
From the off-diagonal sector, one has the same length-2 orbits as for the previous case above. So there are again $\frac{1}{2}\left((\frac{k}{2})^2-\frac{k}{2}\right)$ orbits and $\frac{k}{2}$ fixed points coming from off-diagonal fields.\\
As above, there are no orbits coming from the twisted sector, but only $k$ fixed points.

Putting everything together, the theory extended by $J\equiv (k,1)$ has $3k$ fixed points (i.e. twice the number as given in section \ref{su2 fixed point table}) plus $\frac{k(k+6)}{8}+1$ length-2 orbits.

Also here an ansatz for $S^J$ is at this stage unknown, except for the case $k=2$, given below.

\subsection{$S^J$ matrices for $k=2$}
\label{SJ of SU2 level 2}
The case $k=2$ is particularly simple to analyze, because the matrices involved are relatively small, but it is also very interesting, because it gives us a lot of insights. 

First of all, as we have already remarked in footnote \ref{footnote_SJ},
\begin{equation}
S^{J\equiv (0,1)}=-S_{SU(2)_2}\,,
\end{equation}
resolving the three fixed points of the current $(0,1)$ (see table \ref{table S^J1_k=2}). It is important to remark here that the form of the $S^J$ matrix depends very much on the choice of the labels for the mother CFT: once we reshuffle the labeling of the original $SU(2)_2$ spectrum, the $S^J$ does not simply change by a reshuffling of its rows and columns since some entries can drastically change as well.
\begin{table}[ht]
\caption{Fixed point Resolution: Matrix $S^{J\equiv (0,1)}$}
\centering
\begin{tabular}{c|c c c}
\hline \hline\\
$S^{J\equiv (0,1)}$ & $\langle\phi_0,\phi_1\rangle$ & $\langle\phi_0,\phi_2\rangle$ & $\langle\phi_1,\phi_2\rangle$ \\ 
\hline &&&\\
$\langle\phi_0,\phi_1\rangle$   & $-\frac{1}{2}$        & $-\frac{\sqrt{2}}{2}$  & $-\frac{1}{2}$         \\
$\langle\phi_0,\phi_2\rangle$   & $-\frac{\sqrt{2}}{2}$ & $0$                    & $\frac{\sqrt{2}}{2}$   \\
$\langle\phi_1,\phi_2\rangle$   & $-\frac{1}{2}$        & $\frac{\sqrt{2}}{2}$   & $-\frac{1}{2}$
\end{tabular}
\label{table S^J1_k=2}
\end{table}

By numerical checks of unitarity and modular properties\footnote{Namely, one checks that $S^J$ satisfies $S^J (S^J)^\dagger=1$ and $(S^J T^J)^3=(S^J)^2$.}, one can guess the $S^J$ matrix of the third current $(2,1)$:
\begin{equation}
S^{J\equiv (2,1)}=-S_{SU(2)_2}\,.
\end{equation}
This is numerically equal to the previous one if we order the fixed point fields according to their conformal weights  in the same way as for the first current (see table \ref{table S^J3_k=2}). Indeed, the origin of this equality is that these two extensions are isomorphic to each other, having their fixed points and orbits equal weights.
\begin{table}[ht]
\caption{Fixed point Resolution: Matrix $S^{J\equiv (2,1)}$}
\centering
\begin{tabular}{c|c c c}
\hline \hline\\
$S^{J\equiv (2,1)}$ & $\widehat{(1,0)}$ & $\langle\phi_0,\phi_2\rangle$ & $\widehat{(1,1)}$ \\ 
\hline &&&\\
$\widehat{(1,0)}$   & $-\frac{1}{2}$        & $-\frac{\sqrt{2}}{2}$  & $-\frac{1}{2}$         \\
$\langle\phi_0,\phi_2\rangle$   & $-\frac{\sqrt{2}}{2}$ & $0$                    & $\frac{\sqrt{2}}{2}$   \\
$\widehat{(1,1)}$   & $-\frac{1}{2}$        & $\frac{\sqrt{2}}{2}$   & $-\frac{1}{2}$
\end{tabular}
\label{table S^J3_k=2}
\end{table}

It is a bit more complicated to determine the $S^J$ matrix of the second current $(2,0)$. We would like to use the main formula (\ref{main formula for f.p. resolution 2}) where we need the $S$ matrix of the extended theory. Observe that the extended theory has 16 primaries, of which $2 \times 7$ come from the seven fixed points of $J$, all with known conformal weights. Moreover, it also has central charge $c\leq 3$. There are not many options one has to consider. Indeed, one can show that the extended theory coincides with the tensor product theory $SU(3)_1 \times U(1)_{48}$ extended by a particular integer spin simple current of order three. We denote it here by $(1,16)$. It has no fixed points and its $S$ matrix is known. Explicitly:
\begin{equation}
(SU(2)_2\times SU(2)_2/\mathbb{Z}_2)_{(2,0)}=(SU(3)_1\times U(1)_{48})_{(1,16)}.
\end{equation}
Using (\ref{main formula for f.p. resolution 2}), we can now determine the unknown $S^{J\equiv (2,0)}$ by brute-force calculation. The result is given in table \ref{table S^J2_k=2} (one can find more details in the original paper \cite{Maio:2009kb}).
\begin{table}[ht]
\caption{Fixed point Resolution: Matrix $S^{J\equiv (2,0)}$}
\centering
\begin{tabular}{c|c c c c c c c}
\hline \hline\\
$S^{J\equiv (2,0)}$ & $(1,0)$ & $(1,1)$ & $\langle\phi_0,\phi_2\rangle$ & $\widehat{(0,0)}$ & $\widehat{(0,1)}$ & $\widehat{(2,0)}$ & $\widehat{(2,1)}$ \\ 
\hline &&&\\
$(1,0)$             & $2 i a$  & $2 i a$  & 0        & $2 i b$  & $-2 i b$ & $-2 i b$ & $2 i b$ \\
$(1,1)$             & $2 i a$  & $2 i a$  & 0        & $-2 i b$ & $2 i b$  & $2 i b$  & $-2 i b$\\
$\langle\phi_0,\phi_2\rangle$   & 0        & 0        & 0        & $2 i a$  & $-2 i a$ & $2 i a$  & $-2 i a$\\
$\widehat{(0,0)}$   & $2 i b$  & $-2 i b$ & $2 i a$  & $-2 i d$ & $-2 i d$ & $2 i c$  & $2 i c$ \\
$\widehat{(0,1)}$   & $-2 i b$ & $2 i b$  & $-2 i a$ & $-2 i d$ & $-2 i d$ & $2 i c$  & $2 i c$\\
$\widehat{(2,0)}$   & $-2 i b$ & $2 i b$  & $2 i a$  & $2 i c$  & $2 i c$  & $2 i d$  & $2 i d$\\
$\widehat{(2,1)}$   & $2 i b$  & $-2 i b$ & $-2 i a$ & $2 i c$  & $2 i c$  & $2 i d$  & $2 i d$
\end{tabular}
\label{table S^J2_k=2}
\end{table}
The numbers $a,\,b,\,c,\,d$ above are given by: $a=\frac{1}{4}$, $b=\frac{1}{4\sqrt{2}}$, $c=\frac{\sqrt{2-\sqrt{2}}}{8}$, $d=\frac{\sqrt{2+\sqrt{2}}}{8}$. One can check that the matrix above is unitary, modular invariant and produces sensible fusion coefficients.

A few remarks are in order. First, it is interesting to observe that the numbers $a$ and $b$ are related to the $S$ matrix of the original $SU(2)_2$ CFT, while $c$ and $d$ come from the corresponding $P$ matrix, $P=T^{1/2}ST^2ST^{1/2}$.

Second, this matrix is not the only possible one. There in fact exists a few other
consistent\footnote{I.e. unitary, modular invariant and producing non-negative integer fusion coefficients.} possibilities for $S^J$ where some entries have different sign, due to other sign conventions in (\ref{main formula for f.p. resolution 2}) for the split fixed points..

\section{Example: $SO(N)_1$}
\label{Fixed point resolution in SO(N)1 orbifolds}
Another interesting example of fixed point resolution that we have worked out is the $SO(N)_1$ permutation orbifold. This is a relatively straightforward case since we know the extended theories of all of its integer spin simple current extensions. In fact, they can be derived from the same arguments given in section \ref{SJ of SU2 level 2} for the $SU(2)_2$ permutation orbifold. In the easier cases, the $S^J$ matrix can be computed using (\ref{SJ_(offdiag.-offdiag.)}), since the extension of the orbifold theory gives back the tensor product theory (or a theory isomorphic to it); in more complicated situations, the $S^J$ matrix can be derived from (\ref{main formula for f.p. resolution 2}) and the knowledge of the full, i.e. extended, $S$ matrix via the embedding that we have mentioned before. This embedding works as follows:
\begin{equation}
\label{SO embedded in SU}
\xymatrix{ 
 SO(N)_{perm} \ar[r] \ar[dr]_{ext} & SO(2N) \\
   & SU(N)\times U(1)  \ar[u]_{ext'}  }
\end{equation}
i.e. the extension of the permutation orbifold gives $SU(N)\times U(1)$ whose extension (with another particular current) is $SO(2N)$, the group where the permutation orbifold is embedded.

Let us remind the reader a few facts about these two CFT's \cite{Knizhnik:1984nr,Gepner:1986wi}. The $U(1)_R$ CFT at radius $R$ has central charge $c=1$, $R$ primary fields labelled by $u=0,1,\dots,R-1$ with weight
\begin{equation}
h_u=\frac{u^2}{2R}\,\,{\rm mod}\, \mathbb{Z}.
\end{equation}
Its $S$ matrix and corresponding fusion rules are given by
\begin{eqnarray}
S_{uu'}=\frac{1}{\sqrt{R}}\,e^{-2\pi i \frac{uu'}{R}},\\
(u)\cdot(u')=(u+u')\,\,{\rm mod}\,R.
\end{eqnarray}
The $SU(N)_1=A(N-1)_1$ CFT has central charge $c=N-1$, $N$ primary fields labelled by $s=0,1,\dots,N-1$ with weight
\begin{equation}
h_s=\frac{s^2(N-1)}{2N}\,\,{\rm mod}\, \mathbb{Z}.
\end{equation}
Its $S$ matrix and corresponding fusion rules are given by
\begin{eqnarray}
S_{ss'}=\frac{1}{\sqrt{N}}\,e^{2\pi i \frac{ss'}{N}},\\
(s)\cdot(s')=(s+s')\,\,{\rm mod}\,N.
\end{eqnarray}

For our study of $SO(N)$ at level one, we only need to determine the level of the $SU(N)$ and the radius of the $U(1)$ factors. After a few trials, it is not difficult to convince ourselves that the level of the $SU(N)$ factor is one and the radius of the $U(1)$ factor is $16 N$, while the integer spin simple current (with order $N$) that we need to extend this product group in order to get $SO(2N)$ is\footnote{It is convenient to label fields in the tensor product by pairs $(s,u)$, with $s$ and $u$ labeling fields of the two factors. Sometimes other labels can be used, e.g. one single label $l$, with $l=s\cdot R + u$ or vice versa $s=l\,{\rm mod}\,R$ and $u=\left[\frac{l}{R}\right]$, squared brackets denoting the integer part.} $(\#,16)$, where the first entry denotes a particular field of the $SU(N)_1$ CFT depending\footnote{E.g. for low values of $N$, $\#=4$.} on the value $N$ and the second entry another particular, but given, field of the $U(1)_{16N}$ CFT. Explicitly,
\begin{equation}
(SO(N)_1\times SO(N)_1/\mathbb{Z}_2)_{\rm ext}=(SU(N)_1\times U(1)_{16N})_{\rm ext'}\,.
\end{equation}
The $S$ matrix of the tensor product theory is simply the tensor product of the two $S$ matrices, $S^{\otimes}_{(s,u)(s',u')}=S_{ss'}S_{uu'}$, while the $S$ matrix of the extended theory, $\tilde{S}$, is the tensor product $S$ matrix multiplied by the order $N$ of the current \cite{Schellekens:1990xy}. Hence the $S$ matrix of the extended tensor product $(SU(N)_1\times U(1)_{16N})_{(\#,16)}$ is:
\begin{equation}
\label{S matrix of ext SUN times U1}
\tilde{S}_{(su)(s'u')}=\frac{1}{4}\,\exp\left\{\frac{2\pi i}{N}\left(ss'-\frac{uu'}{16}\right)\right\}\,,
\end{equation}
where the factor $N$ in the denominator is cancelled by the order $N$ in the numerator. This gives the following fusion rules:
\begin{equation}
 (s,u)\cdot (s',u')=((s+s')\,{\rm mod}\,N,(u+u')\,{\rm mod}\,16N)\,.
\end{equation}

Recall that in the extended theory only certain fields $(s,u)$ appear, namely those with integer monodromy charge with respect to the current $(\#,16)$. It is given by 
\begin{equation}
\label{monodromy in SO(N)}
 Q_{(\#,16)}(s,u)=-\frac{\#\cdot s(N-1)+u}{N} \,{\rm mod}\,\mathbb{Z}\,.
\end{equation}
This allows us to analytically relate the labels $s$ and $u$ of the fields in the extension to the fields in the permutation orbifolds, by comparing the weights of the fields in the permutation, $\{h_{\rm perm}\}$, with the ones in the extension, $h_{s,u}=h_s+h_u$, and choosing $s$ and $u$ such that (\ref{monodromy in SO(N)}) is satisfied. This will be crucial when we use (\ref{main formula for f.p. resolution 2}).

Let us move now to study the fixed point resolution of the $SO(N)_1$ permutation orbifolds, distinguishing the case of $N$ even and $N$ odd.

\subsection{$B(n)_1$ series}
The $B(n)_1=SO(N)_1$, $N=2n+1$, series has central charge $c=\frac{N}{2}$ and three primary fields $\phi_i$ with weight $h_i=0,\frac{1}{2},\frac{N}{16}$ ($i=0,1,2$ respectively). The $S$ matrix is the same as the Ising model, as shown in table \ref{table S_Bn_1}.
\begin{table}[ht]
\caption{$S$ matrix for $B(n)_1$}
\centering
\begin{tabular}{c|c c c}
\hline \hline\\
$S_{B(n)_1}$ & $h=0$ & $h=\frac{1}{2}$ & $h=\frac{N}{16}$ \\ 
\hline &&&\\
$h=0$             & $\frac{1}{2}$        & $\frac{1}{2}$         & $\frac{\sqrt{2}}{2}$    \\
$h=\frac{1}{2}$   & $\frac{1}{2}$        & $\frac{1}{2}$         & $-\frac{\sqrt{2}}{2}$   \\
$h=\frac{N}{16}$  & $\frac{\sqrt{2}}{2}$ & $-\frac{\sqrt{2}}{2}$ & $0$
\end{tabular}
\label{table S_Bn_1}
\end{table}

The $B(n)_1$ series has two simple currents\footnote{And only two, because $N$ is odd. This will be different for the $D(n)_1$ series.}, namely the fields with $h_0=0$ (the identity) and $h_1=\frac{1}{2}$. 
In the tensor product they give rise to integer spin simple currents and can both be used to extend the permutation orbifold. Hence, according to our notation, $(B(n)_1)_{\rm perm}$ has four integer spin simple currents arising from the symmetric and anti-symmetric representations of $\phi_0$ and $\phi_1$. Explicitly they are: $(0,0)$, $(0,1)$, $(1,0)$ and $(1,1)$. This situation is very similar to the one already studied in section \ref{SJ of SU2 level 2}. The extension w.r.t. the identity $(0,0)$ is trivial. The extension w.r.t. the current $(0,1)$ projects out all the twisted sector and gives back the tensor product theory $B(n)_1\times B(n)_1$; the fixed points are all the three off-diagonal fields ($h_{\langle0,1\rangle}=\frac{1}{2}$, $h_{\langle0,2\rangle}=\frac{N}{16}$, $h_{\langle1,2\rangle}=\frac{N}{16}+\frac{1}{2}$) and hence the corresponding $S^J$, with $J=(0,1)$, is given by (\ref{SJ_(offdiag.-offdiag.)}). \\
Also easy is the extension w.r.t. the current $(1,1)$: it is indeed isomorphic to the previous one. The fixed points are the off-diagonal field $\langle\phi_0,\phi_1\rangle$ ($h=\frac{1}{2}$) and the two twisted fields coming from $\phi_2$ (with $h=\frac{N}{16}$ and $\frac{N}{16}+\frac{1}{2}$). All their weights are equal to the weights of the fixed points of the current $(0,1)$, hence,  if we label them according to $h$, the $S^J$ matrix for the  current $(1,1)$ is numerically the same as for $(0,1)$. \\
A bit more involved is the $S^J$ matrix for the current $(1,0)$. For this, we need to use the main formula (\ref{main formula for f.p. resolution 2}).

\subsubsection{$(B(n)_1)_{\rm perm}$ $S^J$ matrix for $J=(1,0)$}
There are seven fixed points for the current $J=(1,0)$ of the permutation orbifold $(B(n)_1)_{\rm perm}$, coming from all possible sectors. From the  diagonal fields, we have $(2,0)$ and $(2,1)$ (both have $h=\frac{N}{8}$), from the off-diagonal $\langle\phi_0,\phi_1\rangle$ (with $h=\frac{1}{2}$) and from the twisted $\widehat{(0,0)}$ ($h=\frac{N}{32}$), $\widehat{(0,1)}$ ($h=\frac{N}{32}+\frac{1}{2}$), $\widehat{(1,0)}$ ($h=\frac{N+8}{32}$) and $\widehat{(1,1)}$ ($h=\frac{N+8}{32}+\frac{1}{2}$). We know the original $S$ matrix for these fields, given by $S^{BHS}$. We also know the $S$ matrix of the extended theory, $\tilde{S}$ as in (\ref{S matrix of ext SUN times U1}), given by the embedding (\ref{SO embedded in SU}). Hence, to obtain the desired matrix, we can use the simplified version of the main formula (\ref{main formula for f.p. resolution 2}) which reads:
\begin{equation}
\label{main formula for f.p. resolution simplified}
\tilde{S}_{(a,i)(b,j)}=\frac{1}{2}\left[S^{BHS}_{ab}+\Psi_i(J)S^J_{ab}\Psi_j(J)^\star \right]\,.
\end{equation}

Before giving the $S^J$ matrix, there is a very important issue that we should cover first. We mentioned before that the labels of the permutation and those of the extension are different but related. How can we exactly relate them? Recall that in the extension fields are defined by orbits of the current, with all the fields in the same orbit having same weight (modulo integer) and same $S$ matrix (see \cite{Schellekens:1990xy}). Within each orbit in the extended theory, we choose the field with lowest weight as representative of the split fields coming from the fixed point resolution. According to this convention, every fixed point gets split in two fields $(s_1,u_1)$ and $(s_2,u_2)$ given by:
\begin{itemize}
\item ${\rm if \,\,n=3,4,7,8,11,12,\dots} \Leftrightarrow\, 
{\rm if\,\,} \left[\frac{n-1}{2}\right] {\rm is \,\, odd \,\,}$
\begin{eqnarray}
(2,0) \,\, \longrightarrow&  (0,2N) & \&\qquad (0,14N) \nonumber \\
(2,1) \,\, \longrightarrow&  (2,14N+8) & \&\qquad (N-2,2N-8) \nonumber
\end{eqnarray}
\item ${\rm if \,\,n=5,6,9,10,13,14\dots} \Leftrightarrow\, 
{\rm if\,\,} \left[\frac{n-1}{2}\right] {\rm is \,\, even \,\,}$
\begin{eqnarray}
(2,0) \,\, \longrightarrow&  (2,14N+8) & \&\qquad (N-2,2N-8) \nonumber \\
(2,1) \,\, \longrightarrow&  (0,2N) & \&\qquad (0,14N) \nonumber
\end{eqnarray}
\item ${\rm for \,\, all \,\, n}$
\begin{eqnarray}
\langle\phi_0,\phi_1\rangle \,\, \longrightarrow&  (1,4) & \&\qquad (N-1,16N-4) \nonumber \\
\widehat{(0,0)} \,\, \longrightarrow&  (0,N) & \&\qquad (0,15N) \nonumber \\
\widehat{(0,1)} \,\, \longrightarrow&  (2,15N+8) & \&\qquad (N-2,N-8) \nonumber \\
\widehat{(1,0)} \,\, \longrightarrow&  (N-1,N-4) & \&\qquad (1,15N+4) \nonumber \\
\widehat{(1,1)} \,\, \longrightarrow&  (3,12-N) & \&\qquad (1,N+4) \nonumber
\end{eqnarray}
\end{itemize}
This table also fixes the order of which field we call ``split field $1$'' and ``split field $2$''. We must use fields only from the first set or only from the second set when computing $S^J$. Both the two sets will give the same result, but we cannot choose field representative randomly without losing unitarity and/or modular invariance. It is interesting to check that the orbits corresponding to the two split fields are ``conjugate'' of each other, in the sense that $s_1+s_2 =0$ mod $N$ and $u_1+u_2=0$ mod $16 N$.

The $S^J$ matrix is now given below. It is expressed in terms of the $S$ and $P$ matrices\footnote{The $P$ matrix is $P=T^{1/2}ST^2ST^{1/2}$ and for the $B(n)_1$ series reads:
\begin{equation}
 P=\left(
\begin{array}{ccc}
\cos\left(\frac{\pi N}{8}\right) & \sin\left(\frac{\pi N}{8}\right)  & 0 \\
\sin\left(\frac{\pi N}{8}\right) & -\cos\left(\frac{\pi N}{8}\right) & 0 \\
0                                & 0                                 & 1 
\end{array}
\right)\,,\nonumber
\end{equation}
where $N=2n+1$.}
of the mother $B(n)_1$ theory; also a sign $\epsilon$ appears, depending on the value of $N=2n+1$, $\epsilon=(-1)^{\left[\frac{n-1}{2}\right]}$, square brackets denoting the integer part. We have checked that it is unitary ($S^J (S^J)^{\dagger}=1$), modular invariant ($(S^J T^J)^3=-1=(S^J)^2$) and it gives correct fusion coefficients.

\begin{eqnarray}
S^J_{(2,0)(2,0)} &=& -\frac{1}{2}\,i^N \nonumber \\
S^J_{(2,0)(2,1)} &=& \frac{1}{2}\,i^N \nonumber \\
S^J_{(2,0)\langle\phi_0,\phi_1\rangle} &=& -\frac{1}{2} -S_{20}\,S_{21}
                        \,\,=\,\, 0 \nonumber \\
S^J_{(2,0)\widehat{(0,0)}} &=& -\epsilon \,\frac{1}{2} \,e^{\frac{\epsilon\pi iN}{4}} -\frac{1}{2}\,S_{20}
                        \,\,=\,\, -i\,\, \frac{1}{2} \sin \left(\frac{\pi N}{4}\right) \nonumber \\
S^J_{(2,0)\widehat{(0,1)}} &=& -\epsilon \,\frac{1}{2} \,e^{-\frac{\epsilon\pi iN}{4}} -\frac{1}{2}\,S_{20}
                        \,\,=\,\, i\,\, \frac{1}{2} \sin \left(\frac{\pi N}{4}\right) \nonumber \\
S^J_{(2,0)\widehat{(1,0)}} &=& \epsilon \,\frac{1}{2} \,e^{\frac{\epsilon\pi iN}{4}} -\frac{1}{2}\,S_{21}
                        \,\,=\,\, i\,\, \frac{1}{2} \sin \left(\frac{\pi N}{4}\right) \nonumber \\
S^J_{(2,0)\widehat{(1,1)}} &=& \epsilon \,\frac{1}{2} \,e^{-\frac{\epsilon\pi iN}{4}} -\frac{1}{2}\,S_{21}
                        \,\,=\,\, -i\,\, \frac{1}{2} \sin \left(\frac{\pi N}{4}\right) \nonumber
\end{eqnarray}
\begin{eqnarray}
S^J_{(2,1)(2,1)} &=& -\frac{1}{2}\,i^N \nonumber \\
S^J_{(2,1)\langle\phi_0,\phi_1\rangle} &=& -\frac{1}{2} -S_{20}\,S_{21}
                        \,\,=\,\, 0 \nonumber \\
S^J_{(2,1)\widehat{(0,0)}} &=& \epsilon \,\frac{1}{2} \,e^{-\frac{\epsilon\pi iN}{4}} +\frac{1}{2}\,S_{20}
                        \,\,=\,\, -i\,\, \frac{1}{2} \sin \left(\frac{\pi N}{4}\right) \nonumber \\
S^J_{(2,1)\widehat{(0,1)}} &=& \epsilon \,\frac{1}{2} \,e^{\frac{\epsilon\pi iN}{4}} +\frac{1}{2}\,S_{20}
                        \,\,=\,\, i\,\, \frac{1}{2} \sin \left(\frac{\pi N}{4}\right) \nonumber \\
S^J_{(2,1)\widehat{(1,0)}} &=& -\epsilon \,\frac{1}{2} \,e^{-\frac{\epsilon\pi iN}{4}} +\frac{1}{2}\,S_{21}
                        \,\,=\,\, i\,\, \frac{1}{2} \sin \left(\frac{\pi N}{4}\right) \nonumber \\
S^J_{(2,1)\widehat{(1,1)}} &=& -\epsilon \,\frac{1}{2} \,e^{\frac{\epsilon\pi iN}{4}} +\frac{1}{2}\,S_{21}
                        \,\,=\,\, -i\,\, \frac{1}{2} \sin \left(\frac{\pi N}{4}\right) \nonumber
\end{eqnarray}
\begin{eqnarray}
S^J_{\langle\phi_0,\phi_1\rangle\langle\phi_0,\phi_1\rangle} &=& \frac{1}{2} -(S_{00}\,S_{11}+S_{01}\,S_{01}) \,\,=\,\, 0 \nonumber \\
S^J_{\langle\phi_0,\phi_1\rangle\widehat{(0,0)}} &=& -\frac{i}{2} \nonumber \\
S^J_{\langle\phi_0,\phi_1\rangle\widehat{(0,1)}} &=& \frac{i}{2} \nonumber \\
S^J_{\langle\phi_0,\phi_1\rangle\widehat{(1,0)}} &=& -\frac{i}{2} \nonumber \\
S^J_{\langle\phi_0,\phi_1\rangle\widehat{(1,1)}} &=& \frac{i}{2} \nonumber
\end{eqnarray}
\begin{eqnarray}
S^J_{\widehat{(0,0)}\widehat{(0,0)}} &=& \frac{1}{2} \,e^{-\frac{\pi iN}{8}} -\frac{1}{2}\,P_{00} 
                                  \,\,=\,\, -i\,\, \frac{1}{2} \sin \left(\frac{\pi N}{8}\right) \nonumber \\
S^J_{\widehat{(0,0)}\widehat{(0,1)}} &=& -\frac{1}{2} \,e^{\frac{\pi iN}{8}} +\frac{1}{2}\,P_{00} 
                                  \,\,=\,\, -i\,\, \frac{1}{2} \sin \left(\frac{\pi N}{8}\right) \nonumber \\
S^J_{\widehat{(0,0)}\widehat{(1,0)}} &=& \frac{1}{2}\,i \,e^{-\frac{\pi iN}{8}} -\frac{1}{2}\,P_{01} 
                                  \,\,=\,\, i\,\, \frac{1}{2} \cos \left(\frac{\pi N}{8}\right) \nonumber \\
S^J_{\widehat{(0,0)}\widehat{(1,1)}} &=& \frac{1}{2}\,i \,e^{\frac{\pi iN}{8}} +\frac{1}{2}\,P_{01} 
                                  \,\,=\,\, i\,\, \frac{1}{2} \cos \left(\frac{\pi N}{8}\right) \nonumber
\end{eqnarray}
\begin{eqnarray}
S^J_{\widehat{(0,1)}\widehat{(0,1)}} &=& \frac{1}{2} \,e^{-\frac{\pi iN}{8}} -\frac{1}{2}\,P_{00} 
                                  \,\,=\,\, -i\,\, \frac{1}{2} \sin \left(\frac{\pi N}{8}\right) \nonumber \\
S^J_{\widehat{(0,1)}\widehat{(1,0)}} &=& \frac{1}{2}\,i \,e^{\frac{\pi iN}{8}} +\frac{1}{2}\,P_{01} 
                                  \,\,=\,\, i\,\, \frac{1}{2} \cos \left(\frac{\pi N}{8}\right) \nonumber \\
S^J_{\widehat{(0,1)}\widehat{(1,1)}} &=& \frac{1}{2}\,i \,e^{-\frac{\pi iN}{8}} -\frac{1}{2}\,P_{01} 
                                  \,\,=\,\, i\,\, \frac{1}{2} \cos \left(\frac{\pi N}{8}\right) \nonumber
\end{eqnarray}
\begin{eqnarray}
S^J_{\widehat{(1,0)}\widehat{(1,0)}} &=& -\frac{1}{2} \,e^{-\frac{\pi iN}{8}} -\frac{1}{2}\,P_{11} 
                                  \,\,=\,\, i\,\, \frac{1}{2} \sin \left(\frac{\pi N}{8}\right) \nonumber \\
S^J_{\widehat{(1,0)}\widehat{(1,1)}} &=& \frac{1}{2} \,e^{\frac{\pi iN}{8}} +\frac{1}{2}\,P_{11} 
                                  \,\,=\,\, i\,\, \frac{1}{2} \sin \left(\frac{\pi N}{8}\right) \nonumber 
\end{eqnarray}
\begin{eqnarray}
\label{B SJ J=1,0}
S^J_{\widehat{(1,1)}\widehat{(1,1)}} &=& -\frac{1}{2} \,e^{-\frac{\pi iN}{8}} -\frac{1}{2}\,P_{11} 
                                  \,\,=\,\, i\,\, \frac{1}{2} \sin \left(\frac{\pi N}{8}\right) \nonumber \\
&&
\end{eqnarray}

\subsection{$D(n)_1$ series}
The $D(n)_1=SO(N)_1$, $N=2n$, series has central charge $c=\frac{N}{2}$ and four primary fields $\phi_i$ with weight $h_i=0,\frac{N}{16},\frac{1}{2},\frac{N}{16}$ ($i=0,1,2,3$ respectively).
The $S$ matrix is given in table \ref{table S_Dn_1}.
\begin{table}[ht]
\caption{$S$ matrix for $D(n)_1$}
\centering
\begin{tabular}{c|c c c c}
\hline \hline\\
$S_{D(n)_1}$ & $h=0$ &  $h=\frac{N}{16}$ & $h=\frac{1}{2}$ & $h=\frac{N}{16}$ \\ 
\hline &&&\\
$h=0$             & $\frac{1}{2}$        & $\frac{1}{2}$         & $\frac{1}{2}$ &     $\frac{1}{2}$ \\
$h=\frac{N}{16}$  & $\frac{1}{2}$        & $\frac{(-i)^n}{2}$    & $-\frac{1}{2}$ &    $-\frac{(-i)^n}{2}$ \\
$h=\frac{1}{2}$   & $\frac{1}{2}$        & $-\frac{1}{2}$        & $\frac{1}{2}$ &     $-\frac{1}{2}$ \\
$h=\frac{N}{16}$  & $\frac{1}{2}$        & $-\frac{(-i)^n}{2}$   & $-\frac{1}{2}$ &    $\frac{(-i)^n}{2}$
\end{tabular}
\label{table S_Dn_1}
\end{table}

All the four fields of the $D(n)_1$ series are simple currents. 
In the permutation orbifold, they give rise to four integer spin simple currents, namely $(0,0)$, $(0,1)$, $(2,0)$ and $(2,1)$, and to four non-necessarily-integer spin simple current\footnote{For $n$ multiple of $4$, these currents have also integer spin. We will consider them in chapter \ref{paper2}.}, namely $(1,0)$, $(1,1)$, $(3,0)$ and $(3,1)$. In this chapter we focus on the former set, leaving the latter for the next chapter. Again, the current $(0,0)$ gives a trivial extension. The current $(0,1)$ gives back the tensor product $D(n)_1\times D(n)_1$, with the six off-diagonal fields ($h_{\langle0,2\rangle}=\frac{1}{2}$, $h_{\langle1,3\rangle}=\frac{N}{8}$, $h_{\langle0,1\rangle}=\frac{N}{16}$, $h_{\langle1,2\rangle}=\frac{N}{16}+\frac{1}{2}$, $h_{\langle0,3\rangle}=\frac{N}{16}$, $h_{\langle2,3\rangle}=\frac{N}{16}+\frac{1}{2}$) as fixed points; the $S^J$ matrix is again given by (\ref{SJ_(offdiag.-offdiag.)}).\\
The current $(2,1)$ gives a theory isomorphic to the tensor product. Its fixed points are the fields $\langle\phi_0,\phi_2\rangle$ ($h=\frac{1}{2}$), $\langle\phi_1,\phi_3\rangle$ ($h=\frac{N}{8}$), two twisted fields coming from $\phi_1$ (with $h=\frac{N}{16}$ and $\frac{N}{16}+\frac{1}{2}$) and other two from $\phi_3$ (also with $h=\frac{N}{16}$ and $\frac{N}{16}+\frac{1}{2}$), all with same weights as for the off-diagonal fields. The $S^J$ matrix is again equal to the one for $(0,1)$, if the fixed points are ordered suitably according to their weights. \\
As before, it is more difficult to derive the $S^J$ matrix for $J=(2,0)$, for which we need (\ref{main formula for f.p. resolution 2}).

\subsubsection{$(D(n)_1)_{\rm perm}$ $S^J$ matrix for $J=(2,0)$}
There are six fixed points for the current $J=(2,0)$ of the permutation orbifold $(D(n)_1)_{\rm perm}$, coming from off-diagonal and twisted fields. They are: $\langle\phi_0,\phi_2\rangle$ (with $h=\frac{1}{2}$), $\langle\phi_1,\phi_3\rangle$ (with $h=\frac{N}{8}$), $\widehat{(0,0)}$ ($h=\frac{N}{32}$), $\widehat{(0,1)}$ ($h=\frac{N}{32}+\frac{1}{2}$), $\widehat{(2,0)}$ ($h=\frac{N+8}{32}$) and $\widehat{(2,1)}$ ($h=\frac{N+8}{32}+\frac{1}{2}$).

The $S^J$ matrix can be derived following the same procedure as before. We know $\tilde{S}$ and $S^{BHS}$ and we still have (\ref{main formula for f.p. resolution simplified}). We use the same principle as before to choose the orbit representatives according to their minimal weight. The table in this case is:
\begin{itemize}
\item ${\rm if \,\,n\,\,is\,\,odd}$
\begin{eqnarray}
\langle\phi_1,\phi_3\rangle \,\, \longrightarrow&  (0,2N) & \&\qquad (0,14N) \nonumber
\end{eqnarray}
\item ${\rm if \,\,n\,\,is\,\,even}$
\begin{eqnarray}
\langle\phi_1,\phi_3\rangle \,\, \longrightarrow&  (1,14N+4) & \&\qquad (3,14N+12) \nonumber
\end{eqnarray}
\item ${\rm for \,\, all \,\, n}$
\begin{eqnarray}
\langle\phi_0,\phi_2\rangle \,\, \longrightarrow&  (1,4) & \&\qquad (N-1,16N-4) \nonumber \\
\widehat{(0,0)} \,\, \longrightarrow&  (0,N) & \&\qquad (0,15N) \nonumber \\
\widehat{(0,1)} \,\, \longrightarrow&  (2,15N+8) & \&\qquad (N-2,N-8) \nonumber \\
\widehat{(2,0)} \,\, \longrightarrow&  (N-1,N-4) & \&\qquad (1,15N+4) \nonumber \\
\widehat{(2,1)} \,\, \longrightarrow&  (N-1,15-N) & \&\qquad (1,N+4) \nonumber \\
&&\nonumber
\end{eqnarray}
\end{itemize}
This fixes our order of ``split field 1'' and ``split field 2''. We must use fields only from the first set or only from the second set as before. Orbits corresponding to these two split fields are conjugates of each other.

The $S^J$ matrix is given below. It depends on the original $S$ and $P$ matrices\footnote{The $P$ matrix for the $D(n)_1$ series is:
\begin{equation}
 P=\left(
\begin{array}{cccc}
\cos\left(\frac{\pi N}{8}\right) & 0 & \sin\left(\frac{\pi N}{8}\right)  & 0 \\
0 & e^{-\frac{i\pi N}{8}}\cos\left(\frac{\pi N}{8}\right)
& 0 & i\,\, e^{-\frac{i\pi N}{8}}\sin\left(\frac{\pi N}{8}\right) \\
\sin\left(\frac{\pi N}{8}\right) & 0 & -\cos\left(\frac{\pi N}{8}\right) & 0 \\
0 &  i\,\, e^{-\frac{i\pi N}{8}}\sin\left(\frac{\pi N}{8}\right)
& 0 & e^{-\frac{i\pi N}{8}}\cos\left(\frac{\pi N}{8}\right)
\end{array}
\right)\,. \nonumber
\end{equation}}
of the $D(n)_1$ theory. We have defined the quantity $r=n {\rm \,\,mod\,\,}2=n-2\left[\frac{n}{2}\right]$, which is $0$ if $n$ is even and $1$ if $n$ is odd. We recall that here $N=2n$. We have checked that it is unitary ($S^J (S^J)^{\dagger}=1$), modular invariant ($(S^J T^J)^3=-1=(S^J)^2$) and gives correct fusion coefficients.

\begin{eqnarray}
S^J_{\langle\phi_0,\phi_2\rangle\langle\phi_0,\phi_2\rangle} &=& \frac{1}{2} -(S_{00}\,S_{22}+S_{02}\,S_{02}) \,\,=\,\, 0 \nonumber \\
S^J_{\langle\phi_0,\phi_2\rangle\langle\phi_1,\phi_3\rangle} &=& \frac{1}{2} -(S_{01}\,S_{23}+S_{03}\,S_{21}) \,\,=\,\, 0 \nonumber \\
S^J_{\langle\phi_0,\phi_2\rangle\widehat{(0,0)}} &=& -\frac{i}{2} \nonumber \\
S^J_{\langle\phi_0,\phi_2\rangle\widehat{(0,1)}} &=& \frac{i}{2} \nonumber \\
S^J_{\langle\phi_0,\phi_2\rangle\widehat{(2,0)}} &=& -\frac{i}{2} \nonumber \\
S^J_{\langle\phi_0,\phi_2\rangle\widehat{(2,1)}} &=& \frac{i}{2} \nonumber
\end{eqnarray}
\begin{eqnarray}
S^J_{\langle\phi_1,\phi_3\rangle\langle\phi_1,\phi_3\rangle} &=& \frac{1}{2}\,\,i^N -(S_{11}\,S_{33}+S_{13}\,S_{13}) \,\,=\,\, 0 \nonumber \\
S^J_{\langle\phi_1,\phi_3\rangle\widehat{(0,0)}} &=& -\frac{1}{2}\,\,i^{n+\delta_{r,0}} \nonumber \\
S^J_{\langle\phi_1,\phi_3\rangle\widehat{(0,1)}} &=& \frac{1}{2}\,\,i^{n+\delta_{r,0}} \nonumber \\
S^J_{\langle\phi_1,\phi_3\rangle\widehat{(2,0)}} &=& \frac{1}{2}\,\,i^{n+\delta_{r,0}} \nonumber \\
S^J_{\langle\phi_1,\phi_3\rangle\widehat{(2,1)}} &=& -\frac{1}{2}\,\,i^{n+\delta_{r,0}} \nonumber
\end{eqnarray}
\begin{eqnarray}
S^J_{\widehat{(0,0)}\widehat{(0,0)}} &=& \frac{1}{2}\,\, e^{-\frac{\pi iN}{8}}- \frac{1}{2}\,\,P_{00}                                                    \,\,=\,\, -i\,\,\frac{1}{2}\,\sin\left(\frac{\pi N}{8}\right) \nonumber \\
S^J_{\widehat{(0,0)}\widehat{(0,1)}} &=& -\frac{1}{2}\,\, e^{\frac{\pi iN}{8}}+ \frac{1}{2}\,\,P_{00}                                                    \,\,=\,\, -i\,\,\frac{1}{2}\,\sin\left(\frac{\pi N}{8}\right) \nonumber \\
S^J_{\widehat{(0,0)}\widehat{(2,0)}} &=& \frac{1}{2}\,\,i\,\, e^{-\frac{\pi iN}{8}}- \frac{1}{2}\,\,P_{20}                                                    \,\,=\,\, i\,\,\frac{1}{2}\,\cos\left(\frac{\pi N}{8}\right) \nonumber \\
S^J_{\widehat{(0,0)}\widehat{(2,1)}} &=& \frac{1}{2}\,\,i\,\, e^{\frac{\pi iN}{8}}+ \frac{1}{2}\,\,P_{20}                                                    \,\,=\,\, i\,\,\frac{1}{2}\,\cos\left(\frac{\pi N}{8}\right) \nonumber
\end{eqnarray}
\begin{eqnarray}
S^J_{\widehat{(0,1)}\widehat{(0,1)}} &=& \frac{1}{2}\,\, e^{-\frac{\pi iN}{8}}- \frac{1}{2}\,\,P_{00}                                                    \,\,=\,\, -i\,\,\frac{1}{2}\,\sin\left(\frac{\pi N}{8}\right) \nonumber \\
S^J_{\widehat{(0,1)}\widehat{(2,0)}} &=& \frac{1}{2}\,\,i\,\, e^{\frac{\pi iN}{8}}+ \frac{1}{2}\,\,P_{20}                                                    \,\,=\,\, i\,\,\frac{1}{2}\,\cos\left(\frac{\pi N}{8}\right) \nonumber \\
S^J_{\widehat{(0,1)}\widehat{(2,1)}} &=& \frac{1}{2}\,\,i\,\, e^{-\frac{\pi iN}{8}}- \frac{1}{2}\,\,P_{20}                                                    \,\,=\,\, i\,\,\frac{1}{2}\,\cos\left(\frac{\pi N}{8}\right) \nonumber
\end{eqnarray}
\begin{eqnarray}
S^J_{\widehat{(2,0)}\widehat{(2,0)}} &=& -\frac{1}{2}\,\, e^{-\frac{\pi iN}{8}}- \frac{1}{2}\,\,P_{22}                                                    \,\,=\,\, i\,\,\frac{1}{2}\,\sin\left(\frac{\pi N}{8}\right) \nonumber \\
S^J_{\widehat{(2,0)}\widehat{(2,1)}} &=& \frac{1}{2}\,\, e^{\frac{\pi iN}{8}}+ \frac{1}{2}\,\,P_{22}                                                    \,\,=\,\, i\,\,\frac{1}{2}\,\sin\left(\frac{\pi N}{8}\right) \nonumber
\end{eqnarray}
\begin{eqnarray}
\label{D SJ J=2,0}
S^J_{\widehat{(2,1)}\widehat{(2,1)}} &=& -\frac{1}{2}\,\, e^{-\frac{\pi iN}{8}}- \frac{1}{2}\,\,P_{22}                                                    \,\,=\,\, i\,\,\frac{1}{2}\,\sin\left(\frac{\pi N}{8}\right) \nonumber \\
&&
\end{eqnarray}

\section{Conclusion}
In this chapter we have studied the simple current and fixed point structure of permutation orbifolds and we have asked the question of resolving fixed points in these extensions. We have not done it in general but only for the $SU(2)_2$ orbifolds and for the $B(n)_1$ and $D(n)_1$ series. The main results were presented in sections \ref{Fixed point resolution in SU(2)k orbifolds} and \ref{Fixed point resolution in SO(N)1 orbifolds}. Besides the particular expressions for the $S^J$ matrices in those cases, we have also showed the existence of a very special current, the un-orbifold current $J=(0,1)$, which is always present in the permutation orbifold and whose action is to un-do the orbifold giving back the initial tensor product. 

At this point we still have plenty of open questions. First we would like to solve the problem in full generality by giving a sensible ansatz for the $S^J$ matrix for any arbitrary CFT. We expect that this ansatz should depend neither on the particular CFT nor on the particular current used in the extension. The results for the special cases considered here give some hints about such a  general formula. 
This problem will be addressed and solved in chapter \ref{paper3}, where we will see how to make an educated guess for the $S^J$ matrices in full generality and how this guess can then be checked.

Secondly, the two $SO(N)_1$ series are interesting since they appear in the numerator of the coset CFT defining $N=2$ minimal models. However, once we have the general formula, it will contain the $SO(N)_1$ series as a particular example and we will not have to worry anymore about the specific details derived here. For example, in this chapter we have looked only at spin-1 currents of permutation orbifolds (apart from the special case $(0,1)$), but we will see in chapter \ref{paper3} how to generalize the results to arbitrary-spin currents. These currents will be relevant to impose world-sheet supersymmetry on the permutation orbifold of two identical $N=2$ minimal models. Moreover, using extensions by these higher-spin currents, we should be able to derive a ``super-BHS" formula for permutation orbifolds
of supersymmetric RCFT's. This will be the subject of chapter \ref{paper4}.

\chapter{Finishing the $D(n)_1$ orbifolds}
\label{paper2}

{\flushright
{\small 
\textit{For a moment, nothing happened.}\par
\textit{Then, after a second or so, nothing continued to happen.}\par
\textit{(D. Adams, So Long, and Thanks for All the Fish)}\par
}
}

\section{Introduction}
In the previous chapter we studied the structure of order-two simple currents in permutation orbifolds in two-dimensional conformal field theories \cite{Belavin:1984vu}. The main tool was the BHS $S$ matrix for the permutation orbifold \cite{Klemm:1990df,Borisov:1997nc}. We have seen that in general simple currents can only be generated from diagonal fields that correspond to simple currents in the mother theory, while their fixed points can come from both the untwisted (diagonal and off-diagonal) and the twisted sector. We have also considered extensions of the permutation orbifold and their fixed point resolution. 

Modular transformation matrices of simple current extensions  
\cite{Schellekens:1989am,Intriligator:1989zw,Schellekens:1989dq} are often
quite non-trivial due to fixed points \cite{Schellekens:1999yg,Schellekens:1989uf}. 
So far we have been able to derive the $S$ matrices corresponding to extensions in the case of $SU(2)_2$, $B(n)_1$ and $D(n)_1$ WZW models \cite{Knizhnik:1984nr,Gepner:1986wi}. The procedure was described in the previous chapter (see also original paper \cite{Maio:2009kb}). This was completely done for the first two models but only partially for the $D(n)_1$. In fact, we provided the $S$ matrix for the  integer spin simple currents that exist for any value of $n$, but sometimes additional currents appear in the $D(n)_1$ model whose fixed points must be resolved as well, in order to use them as extensions. Generically fixed points can arise for integer spin and half-integer spin simple currents \cite{Schellekens:1990xy}. We will see that for particular ranks of $D(n)_1$ there are addition currents whose fixed points must be resolved. In this chapter we address those additional problems, providing a complete picture for the fixed point resolution in $D(n)_1$ permutation orbifold.

Explicitly, there are two interesting situations where fixed points can occur and that we have not studied so far, both with even rank $n$. When $n$ is multiple of four, $n=4p$ with $p\in \mathbb{Z}$, there are additional integer-spin simple currents coming from the two spinor representations of the $D(n)_1$ WZW model. The spinor fields have weight $h=\frac{n}{8}$ and their symmetric and anti-symmetric representations in the $D(n)_1$ permutation orbifold have weight $h=\frac{n}{4}$. Similarly, when $n=4p+2$, the same two spinor currents generate half-integer spin simple currents in the $D(n)_1$ permutation orbifold. 
Although the latter cannot be used to extend the chiral algebra, they can be used in combination with half-integer spin currents of
another factor in a tensor product. For example, one may tensor the permutation orbifold with an Ising model, and consider 
the product of the half-integer spin current of the $D(n)_1$ permutation orbifold and the Ising spin field. This is not just of academic interest.
Extended tensor products of rational conformal field theories are an important tool in explicit four-dimensional string constructions, and
in the  vast majority of cases one encounters fixed points. For this reason the fixed point resolution matrices we determine
here and in the previous chapter have a range of applicability far beyond the special cases used here to 
determine them. 

There is no known algorithm for determining these matrices in generic rational CFT's, even if their matrix $S$ is known. In the previous chapter we made use of the fact that the extension currents
had spin 1 and led to
identifiable CFT's. This method will not work here except in the special case of $D(4)_1$, where the spinor currents of the permutation orbifold have spin 1.
In that case one can make use of triality of $SO(8)$ to determine the missing fixed point resolution matrices. Although triality does not extend
to larger ranks, it turns out that in the other cases the fixed point spectrum is sufficiently similar to allow us to generalize to $D(n)_1$, for any $n$.

The plan of the chapter is as follows.\\
In section \ref{D4p permutation orbifolds} we describe the $D(4p)_1$ permutation orbifolds extended by the two spinor currents and resolve the fixed points. In the special case $p=1$ we 
use triality of $SO(8)$ to determine the  set of $S^J$ matrices. From the case $p=1$ is indeed possible to generalize the result to arbitrary values of $p$.\\
In section \ref{D4p+2 permutation orbifolds} we repeat the procedure for $D(4p+2)_1$ permutation orbifolds. We can be fast here since very few changes are sufficient to write down consistent $S^J$ matrices.

The content of this chapter is based on \cite{Maio:2009cy}.

\section{$D(4p)_1$ orbifolds}
\label{D4p permutation orbifolds}
We start with the $D(n)_1$ WZW model as mother theory and focus on the spinor currents that for even rank $n$ can have (half-)integer spin.
Let us fix our notation. The $D(n)_1=SO(N)_1$, $N=2n$, series has central charge $c=n=\frac{N}{2}$ and four primary fields $\phi_i$ with weight $h_i=0,\frac{N}{16},\frac{1}{2},\frac{N}{16}$ ($i=0,1,2,3$ respectively).
The $S$ matrix is given in table \ref{table S_Dn_1}.


All the four fields of the $D(n)_1$ series are simple currents. 
In the permutation orbifold, they give rise to four integer-spin simple currents, namely $(0,0)$, $(0,1)$, $(2,0)$ and $(2,1)$, and to four non-necessarily-integer-spin simple currents, namely $(1,0)$, $(1,1)$, $(3,0)$ and $(3,1)$. For $n$ multiple of four, the latter currents have also integer spin. In this chapter we want to study precisely these currents, coming from the spinor representations $i=1,3$ of the $D(n)_1$ model.

There are already a few observations that we can make. First of all, there exists an automorphism that exchanges the fields $\phi_1$ and $\phi_3$. This will have the consequence that the permutation theories extended by the currents $(1,0)$ and $(3,0)$ will be isomorphic\footnote{The fields $\phi_1$ and $\phi_3$ also have same $P$-matrix entries. In fact, the $P$ matrix for $n=4p$ is
\begin{equation}
P=\left(
\begin{array}{cccc}
(-1)^p &  0  &    0       & 0 \\
0      &  1  &    0       & 0 \\
0      &  0  & (-1)^{p+1} & 0 \\
0      &  0  &    0       & 1
\end{array}
\right)\,. \nonumber
\end{equation}
We recall that the $P$ matrix, $P=\sqrt{T}ST^2S\sqrt{T}$, first introduced in \cite{Pradisi:1995qy}, enters the BHS formulas \cite{Borisov:1997nc} for the $S$ matrix of the permutation orbifold in the twisted sector.} (the fields having same weights and the two theories having equal central charge); this holds as well as for the extensions by $(1,1)$ and $(3,1)$. Secondly, when $n$ is multiple of four, i.e. $n=4p$ with $p\in\mathbb{Z}$, the $S$ matrix of the mother $D(n)_1$ theory is the same for every $p$. This will have the consequence that the fusion rules of these currents in the permutation orbifolds are the same for every value of $p$. Putting these two observations together, we conclude that for $n=4p$ there will be only two universal $S^J$ matrices to determine\footnote{They will in general depend on $p$ through a phase in order to satisfy modular invariance, since the $T$ matrix depends on $p$.}.

Let us illustrate these points with the explicit construction. Consider\footnote{The case $n=4$, that we will consider extensively later, is very interesting since it corresponds to $SO(8)_1$ where, due to triality, three out of four fields have equal weight.} the case with arbitrary $n=4 p$. The $D(n)_1$ weights are then $h=0,\frac{n}{8},\frac{1}{2},\frac{n}{8}$ and the orbit structure under the additional orbifold integer-spin simple currents (all with $h=\frac{n}{4}=p$) is as follows.

\begin{tabular}{l l l l}
&&&\\
$J\equiv(1,0)$ & \underline{Fixed points} &\phantom{$J\equiv(1,0)$}& \underline{Length-2 orbits}\\
& $\langle\phi_0,\phi_1\rangle$, \,$h=\frac{n}{8}$              && $\Big( (0,0),(1,0) \Big)$,\, $h=0$\\
& $\langle\phi_2,\phi_3\rangle$, \,$h=\frac{n}{8}+\frac{1}{2}$  && $\Big( (0,1),(1,1) \Big)$,\, $h=1$\\
& $\widehat{(0,0)}$,\, $h=\frac{n}{16}$             && $\Big( (2,0),(3,0) \Big)$,\, $h=1$\\
& $\widehat{(0,1)}$,\, $h=\frac{n}{16}+\frac{1}{2}$ && $\Big( (2,1),(3,1) \Big)$,\, $h=1$\\
& $\widehat{(1,0)}$,\, $h=\frac{n}{8}$              &&\\
& $\widehat{(1,1)}$,\, $h=\frac{n}{8}+\frac{1}{2}$  &&\\
&&&\\
\end{tabular}

\begin{tabular}{l l l l}
$J\equiv(1,1)$ & \underline{Fixed points} &\phantom{$J\equiv(1,1)$}& \underline{Length-2 orbits}\\
& $\langle\phi_0,\phi_1\rangle$, \,$h=\frac{n}{8}$                          && $\Big( (0,0),(1,1) \Big)$,\, $h=0$\\
& $\langle\phi_2,\phi_3\rangle$, \,$h=\frac{n}{8}+\frac{1}{2}$              && $\Big( (0,1),(1,0) \Big)$,\, $h=1$\\
& $\widehat{(2,0)}$,\, $h=\frac{n}{16}+\frac{1}{4}$             && $\Big( (2,0),(3,1) \Big)$,\, $h=1$\\
& $\widehat{(2,1)}$,\, $h=\frac{n}{16}+\frac{1}{4}+\frac{1}{2}$ && $\Big( (2,1),(3,0) \Big)$,\, $h=1$\\
& $\widehat{(3,0)}$,\, $h=\frac{n}{8}$                          &&\\
& $\widehat{(3,1)}$,\, $h=\frac{n}{8}+\frac{1}{2}$              &&\\
&&&\\
\end{tabular}

\begin{tabular}{l l l l}
$J\equiv(3,0)$ & \underline{Fixed points} &\phantom{$J\equiv(3,0)$}& \underline{Length-2 orbits}\\
& $\langle\phi_0,\phi_3\rangle$, \,$h=\frac{n}{8}$                && $\Big( (0,0),(3,0) \Big)$,\, $h=0$\\
& $\langle\phi_1,\phi_2\rangle$, \,$h=\frac{n}{8}+\frac{1}{2}$    && $\Big( (0,1),(3,1) \Big)$,\, $h=1$\\
& $\widehat{(0,0)}$,\, $h=\frac{n}{16}$               && $\Big( (1,0),(2,0) \Big)$,\, $h=1$\\
& $\widehat{(0,1)}$,\, $h=\frac{n}{16}+\frac{1}{2}$    && $\Big( (1,1),(2,1) \Big)$,\, $h=1$\\
& $\widehat{(3,0)}$,\, $h=\frac{n}{8}$                &&\\
& $\widehat{(3,1)}$,\, $h=\frac{n}{8}+\frac{1}{2}$    &&\\
&&&\\
\end{tabular}

\begin{tabular}{l l l l}
$J\equiv(3,1)$ & \underline{Fixed points} &\phantom{$J\equiv(3,1)$}& \underline{Length-2 orbits}\\
& $\langle\phi_0,\phi_3\rangle$, \,$h=\frac{n}{8}$               && $\Big( (0,0),(3,1) \Big)$,\, $h=0$\\
& $\langle\phi_1,\phi_2\rangle$, \,$h=\frac{n}{8}+\frac{1}{2}$   && $\Big( (0,1),(3,0) \Big)$,\, $h=1$\\
& $\widehat{(1,0)}$,\, $h=\frac{n}{8}$                           && $\Big( (1,0),(2,1) \Big)$,\, $h=1$\\
& $\widehat{(1,1)}$,\, $h=\frac{n}{8}+\frac{1}{2}$               && $\Big( (1,1),(2,0) \Big)$,\, $h=1$\\
& $\widehat{(2,0)}$,\, $h=\frac{n}{16}+\frac{1}{4}$              &&\\
& $\widehat{(2,1)}$,\, $h=\frac{n}{16}+\frac{1}{4}+\frac{1}{2}$  &&\\
&&&\\
\end{tabular}

\noindent Note that in going from the fixed points of $(1,\psi)$ to $(3,\psi)$, the fields $\phi_1$ and $\phi_3$ get interchanged: this provides isomorphic sets of fields in the extensions.

The fixed points get split into two fields in the extended permutation orbifold and hence all the theories above admit $2\cdot 6+4=16$ fields. By changing $n=4p$, the weights of the orbits and the ones of the fixed points might change, but there are a few things that remain invariant, namely: 1) the fact that the extension by the current $(1,0)$ (resp. $(1,1)$) is isomorphic (up to field reordering) to the one by $(3,0)$ (resp. $(3,1)$), as it can be seen by looking at the weights of the extended fields; 2) the orbit and fixed-point structure (i.e. the fusion rules of the currents with any other field in the permutation orbifold) remains the same for arbitrary $p$; this has the consequence that we will have to determine only two $S^J$ matrices instead of four.

\subsection{$S^J$ matrices for $D(4p)_1$ permutation orbifolds}
We have already noticed that there are in practice only two $S^J$ matrices to determine for the four above-mentioned integer-spin simple currents. So here we are going to derive $S^{(1,0)}$ and $S^{(1,1)}$; $S^{(3,0)}$ and $S^{(3,1)}$ are equal to the former two, after proper field ordering.

It is instructive to start with the $D(4)_1$ $(p=1)$ case. $SO(8)_1$ is special in the sense that the three non-trivial representations, i.e. the vector ${\bf 8_v}$ and the two spinors ${\bf 8_s}$ and ${\bf 8_c}$, have same weight ($h=\frac{1}{2}$) and same dimension (dim$=8$) and can be mapped into each other. This property of $SO(8)$ is triality. 
Due to triality of $SO(8)$, the extensions by the currents $(1,\psi)$, $(2,\psi)$ and $(3,\psi)$ must produce the same result. The extension by $(2,\psi)$ is already known from chapter \ref{paper1} and, according to our earlier arguments, the extensions by $(1,\psi)$ and $(3,\psi)$ are equal. 

Let us now work out the $S^J$ matrices corresponding to the two integer-spin simple currents $J=(1,0)$ and $J=(1,1)$. The extension by $(1,0)$ of the permutation orbifold is isomorphic to an extension of the tensor product of an $SU(8)$ and a $U(1)$ factor as done in section \ref{Fixed point resolution in SO(N)1 orbifolds}:
\begin{equation}
(D(4)_1 \times D(4)_1/ \mathbb{Z}_2)_{(1,0)}=(SU(8)_1 \times U(1)_{128})_{(4,16)}\,,
\end{equation}
while the extension by $(1,1)$ is isomorphic to the tensor product $D(4)_1 \times D(4)_1$. This is exactly what happened for the already known currents $(2,\psi)$; in fact, due to triality of $SO(8)$, the three theories extended by $(1,\psi)$ $(2,\psi)$ $(3,\psi)$ must be the same.

\subsubsection{\underline{$J=(1,0)$}}
We use the main formula (\ref{main formula for f.p. resolution}), that we repeat here for convenience,
\begin{equation}
\label{MainFormula}
\tilde{S}_{(a,i)(b,j)}=\frac{|G|}{\sqrt{|U_a||S_a||U_b||S_b|}}\sum_{J\in G}\Psi_i(J) S^J_{ab} \Psi_j(J)^{\star}
\end{equation}
to derive the $S^J$ matrix from the knowledge of the extended matrix $\tilde{S}$ and the permutation orbifold matrix $S^{(0,0)}\equiv S^{BHS}$. The prefactor in (\ref{MainFormula}) is a group theoretical factor and the $\Psi_i$'s are the group characters. As already done in the previous chapter, our field convention to distinguish between the two split fixed points is:
\begin{eqnarray}
\langle\phi_0,\phi_1\rangle \,\, \longrightarrow&  (1,4) & \&\qquad (7,124) \nonumber \\
\langle\phi_2,\phi_3\rangle \,\, \longrightarrow&  (1,116) & \&\qquad (3,124) \nonumber \\
\widehat{(0,0)} \,\, \longrightarrow&  (0,120) & \&\qquad (0,8) \nonumber \\
\widehat{(0,1)} \,\, \longrightarrow&  (6,0) & \&\qquad (2,0) \nonumber \\
\widehat{(1,0)} \,\, \longrightarrow&  (7,4) & \&\qquad (1,124) \nonumber \\
\widehat{(1,1)} \,\, \longrightarrow&  (1,12) & \&\qquad (3,4) \nonumber \\
&&\nonumber
\end{eqnarray}
where $(s,u)$ denotes a field in the extended theory ($s\equiv s+8$, $u\equiv u+128$). Here the first entry $s$ is the rank of the anti-symmetric representations of $SU(8)_1$, while the second entry $u$ gives the weights of the $U(1)_R$ representations (in this specific case $R=128$) according to $h_u=\frac{u^2}{2R}$ mod $\mathbb{Z}$.
Observe that field one and field two correspond to complementary orbits. The $S^J$ matrix for the $(D(4)_1\times D(4)_1/\mathbb{Z}_2)_{(1,0)}$ orbifold can be derived as done previously and is given in table \ref{table S^J=10_p=1}. We denote it by $S^J_{D4}$, with $J=(1,0)$, for reasons that will become clear later.
\begin{table}[ht]
\caption{Fixed point Resolution: Matrix $S^{J\equiv (1,0)}_{D4}$}
\centering
\begin{tabular}{c|c c c c c c c}
\hline \hline\\
$S^{J\equiv (1,0)}_{D4}$ & $\langle\phi_0,\phi_1\rangle$ & $\langle\phi_2,\phi_3\rangle$ & $\widehat{(0,0)}$ & $\widehat{(0,1)}$ & $\widehat{(1,0)}$ & $\widehat{(1,1)}$ \\ 
\hline &&&\\
$\langle\phi_0,\phi_1\rangle$   & $0$  & $0$     & $\frac{i}{2}$  & $-\frac{i}{2}$ & $-\frac{i}{2}$ & $-\frac{i}{2}$ \\
$\langle\phi_2,\phi_3\rangle$   & $0$  & $0$     & $\frac{i}{2}$  & $-\frac{i}{2}$ & $\frac{i}{2}$  & $\frac{i}{2}$ \\
$\widehat{(0,0)}$   & $\frac{i}{2}$  & $\frac{i}{2}$  & $0$  & $0$ & $\frac{i}{2}$  & $-\frac{i}{2}$ \\
$\widehat{(0,1)}$   & $-\frac{i}{2}$ & $-\frac{i}{2}$ & $0$  & $0$ & $\frac{i}{2}$  & $-\frac{i}{2}$ \\
$\widehat{(1,0)}$   & $-\frac{i}{2}$ & $\frac{i}{2}$  & $\frac{i}{2}$  & $\frac{i}{2}$  & $0$  & $0$ \\
$\widehat{(1,1)}$   & $-\frac{i}{2}$ & $\frac{i}{2}$  & $-\frac{i}{2}$ & $-\frac{i}{2}$ & $0$  & $0$ \\
\end{tabular}
\label{table S^J=10_p=1}
\end{table}

One can check that this matrix is unitary ($S^J (S^J)^\dagger=1$) and modular invariant ($(S^J)^2=(S^J T^J)^3$, where $T^J$ is the $T$ matrix restricted to the fixed points) and gives non-negative integer fusion coefficients. Moreover, one can see that unitarity and modular invariance are preserved for $p=1$ mod $4$: then this matrix can be used also in these situations.

Observe that rescaling the $S^J$ matrix by a phase does not destroy unitarity but it does affect modular invariance. By a suitable choice of the phase, it is possible to make a modular invariant matrix out of $S^{(1,0)}_{D4}$ valid for all $p$. The correct choice is:
\begin{equation}
\label{S10p}
S^{(1,0)}= (-i)^{p-1} \cdot S^{(1,0)}_{D4}= e^{-\frac{i\pi}{4}(m-2)} \cdot S^{(1,0)}_{D4}
\end{equation}
which will use for any value of $p$. Here $m=2p$ is an even integer such that $D(2m)_1\equiv D(4p)_1$. This is again unitary, modular invariant and gives non-negative integer fusion coefficients.

Let us make a final comment. What happens when we shift $p\rightarrow p+1$? Under this shift, the fixed point weights change differently. In particular, for the current $(1,0)$ the shifts are $h\rightarrow h+\{\frac{1}{2},\frac{1}{2},\frac{1}{4},\frac{1}{4},\frac{1}{2},\frac{1}{2}\}$. The $T^{(1,0)}$ matrix then changes as $T^{(1,0)}\rightarrow e^{-\frac{2\pi i}{3}}\,{\rm diag}(-1,-1,i,i,-1,-1)\cdot T^{(1,0)}$ (the phase in front coming from the central charge), while the $S^{(1,0)}$ takes a phase, $S^{(1,0)}\rightarrow -iS^{(1,0)}$. These changes are such that modular invariance is still preserved for every $p$.

\subsubsection{\underline{$J=(1,1)$}}
For this current, recall that
\begin{equation}
(D(4)_1 \times D(4)_1/ \mathbb{Z}_2)_{(1,1)}\sim D(4)_1 \times D(4)_1\,.
\end{equation}
The split fixed points correspond to fields in the tensor product theory. We choose conventionally the following scheme, but a few other choices are also possible.
\begin{eqnarray}
\langle\phi_0,\phi_1\rangle \,\, \longrightarrow&  \phi_0 \otimes \phi_1  & \&\qquad \phi_1 \otimes \phi_0 \nonumber \\
\langle\phi_2,\phi_3\rangle \,\, \longrightarrow&  \phi_2 \otimes \phi_3  & \&\qquad \phi_3 \otimes \phi_2 \nonumber \\
\widehat{(2,0)} \,\, \longrightarrow&  \phi_0 \otimes \phi_2  & \&\qquad \phi_2 \otimes \phi_0 \nonumber \\
\widehat{(2,1)} \,\, \longrightarrow&  \phi_1 \otimes \phi_3  & \&\qquad \phi_3 \otimes \phi_1 \nonumber \\
\widehat{(3,0)} \,\, \longrightarrow&  \phi_0 \otimes \phi_3  & \&\qquad \phi_3 \otimes \phi_0 \nonumber \\
\widehat{(3,1)} \,\, \longrightarrow&  \phi_1 \otimes \phi_2  & \&\qquad \phi_2 \otimes \phi_1 \nonumber \\
&&\nonumber
\end{eqnarray}
The next step is to compute the $S^J$ matrix for the $(D(4)_1 \times D(4)_1/ \mathbb{Z}_2)_{(1,1)}$ orbifold. We call it again $S^J_{D4}$, with $J=(1,1)$. Our strategy is as follows. We first go to the isomorphic tensor product theory and use
\begin{equation}
S^J_{\langle mn\rangle\langle pq\rangle}=S_{mp}S_{nq}-S_{mq}S_{np}
\end{equation}
as derived in (\ref{SJ_(offdiag.-offdiag.)}) to compute the $S^J$ matrix there and then we go back to the extended permutation orbifold using the field map. We obtain the $S^J$ matrix as in table \ref{table S^J=11_p=1}.
\begin{table}[ht]
\caption{Fixed point Resolution: Matrix $S^{J\equiv (1,1)}_{D4}$}
\centering
\begin{tabular}{c|c c c c c c c}
\hline \hline\\
$S^{J\equiv (1,1)}_{D4}$ & $\langle\phi_0,\phi_1\rangle$ & $\langle\phi_2,\phi_3\rangle$ & $\widehat{(2,0)}$ & $\widehat{(2,1)}$ & $\widehat{(3,0)}$ & $\widehat{(3,1)}$ \\ 
\hline &&&\\
$\langle\phi_0,\phi_1\rangle$   & $0$  & $0$     & $-\frac{1}{2}$ & $-\frac{1}{2}$ & $-\frac{1}{2}$ & $-\frac{1}{2}$ \\
$\langle\phi_2,\phi_3\rangle$   & $0$  & $0$     & $-\frac{1}{2}$ & $-\frac{1}{2}$ & $\frac{1}{2}$  & $\frac{1}{2}$ \\
$\widehat{(2,0)}$   & $-\frac{1}{2}$ & $-\frac{1}{2}$ & $0$  & $0$ & $-\frac{1}{2}$ & $\frac{1}{2}$ \\
$\widehat{(2,1)}$   & $-\frac{1}{2}$ & $-\frac{1}{2}$ & $0$  & $0$ & $\frac{1}{2}$  & $-\frac{1}{2}$ \\
$\widehat{(3,0)}$   & $-\frac{1}{2}$ & $\frac{1}{2}$  & $-\frac{1}{2}$ & $\frac{1}{2}$  & $0$  & $0$ \\
$\widehat{(3,1)}$   & $-\frac{1}{2}$ & $\frac{1}{2}$  & $\frac{1}{2}$  & $-\frac{1}{2}$ & $0$  & $0$ \\
\end{tabular}
\label{table S^J=11_p=1}
\end{table}

The $S^J$ matrix obtained in this way for $(D(4)_1 \times D(4)_1/ \mathbb{Z}_2)_{(1,1)}$ is unitary and modular invariant, so it is a good matrix for the extended theory. Moreover, this $S^J$ matrix is a good (i.e. unitary and modular invariant) matrix also for $p=1$ mod $4$.

In order to make this matrix modular invariant for any $p$, we again multiply by a phase. The choice is the same as before:
\begin{equation}
\label{S11p}
S^{(1,1)}= (-i)^{p-1} \cdot S^{(1,1)}_{D4}= e^{-\frac{i\pi}{4}(m-2)} \cdot S^{(1,1)}_{D4}
\end{equation}
which will use for any value of $p$. This is again unitary, modular invariant and gives non-negative integer fusion coefficients. The shift $n\rightarrow n+16$, corresponding to $p\rightarrow p+4$, changes all the weights by integers, does not change $S^{(1,1)}$, but does change $T^{(1,1)}$ by a phase which is a cubic root of unity, thus preserving modular invariance.

One can check formulas (\ref{S10p}) and (\ref{S11p}) in many explicit examples. For instance, one can see that they have good properties by looking at a few values of $p$, but also considering tensor products like $D(8)_1\times D(12)_1$ or $D(8)_1\times D(16)_1$ and extending with many current combinations $(J_1,J_2)$, where $J_1$ belongs to the first factor and $J_2$ to the second factor. In every example, the fusion rules give non-negative integer coefficients.

\section{$D(4p+2)_1$ orbifolds}
\label{D4p+2 permutation orbifolds}
So far we have not addressed half-integer spin simple currents. They might also admit fixed points that must be resolved in the extended theory. This happens for the $D(n)_1$ permutation orbifolds with $n=4p+2$. In fact, the four currents $(1,\psi)$ and $(3,\psi)$, with $\psi =0,1$, will have weight $h=\frac{2p+1}{2}$ and will admit fixed points. The orbit structure is in this case with $n=4p+2$ very similar to the previous situation with $n=4p$, except for the fact that the twisted fields get reshuffled. The fixed point structure is as follows. Observe that this is very similar to the structure for the previous case $n=4p$.

\begin{tabular}{l l l l}
&&&\\
$J\equiv(1,0)$ & \underline{Fixed points} &$J\equiv(3,0)$& \underline{Fixed points}\\
& $\langle\phi_0,\phi_1\rangle$,\,$h=\frac{n}{8}$               && $\langle\phi_0,\phi_3\rangle$, \,$h=\frac{n}{8}$\\
& $\langle\phi_2,\phi_3\rangle$,\,$h=\frac{n}{8}+\frac{1}{2}$   && $\langle\phi_1,\phi_2\rangle$, \,$h=\frac{n}{8}+\frac{1}{2}$\\
& $\widehat{(2,0)}$,\,$h=\frac{n}{16}+\frac{1}{4}$            && $\widehat{(2,0)}$,\, $h=\frac{n}{16}+\frac{1}{4}$\\
& $\widehat{(2,1)}$,\,$h=\frac{n}{16}+\frac{1}{4}+\frac{1}{2}$&& $\widehat{(2,1)}$,\,$h=\frac{n}{16}+\frac{1}{4}+\frac{1}{2}$\\
& $\widehat{(1,0)}$,\, $h=\frac{n}{8}$             && $\widehat{(3,0)}$,\, $h=\frac{n}{8}$\\
& $\widehat{(1,1)}$,\, $h=\frac{n}{8}+\frac{1}{2}$ && $\widehat{(3,1)}$,\, $h=\frac{n}{8}+\frac{1}{2}$\\
&&&\\
\end{tabular}

\begin{equation}
\label{D 4p+2 scheme}
\end{equation}

\begin{tabular}{l l l l}
&&&\\
$J\equiv(1,1)$ & \underline{Fixed points} &$J\equiv(3,1)$& \underline{Fixed points}\\
& $\langle\phi_0,\phi_1\rangle$, \,$h=\frac{n}{8}$              && $\langle\phi_0,\phi_3\rangle$, \,$h=\frac{n}{8}$\\
& $\langle\phi_2,\phi_3\rangle$, \,$h=\frac{n}{8}+\frac{1}{2}$  && $\langle\phi_1,\phi_2\rangle$, \,$h=\frac{n}{8}+\frac{1}{2}$\\
& $\widehat{(0,0)}$,\, $h=\frac{n}{16}$             && $\widehat{(0,0)}$,\, $h=\frac{n}{16}$\\
& $\widehat{(0,1)}$,\, $h=\frac{n}{16}+\frac{1}{2}$ && $\widehat{(0,1)}$,\, $h=\frac{n}{16}+\frac{1}{2}$\\
& $\widehat{(3,0)}$,\, $h=\frac{n}{8}$              && $\widehat{(1,0)}$,\, $h=\frac{n}{8}$\\
& $\widehat{(3,1)}$,\, $h=\frac{n}{8}+\frac{1}{2}$  && $\widehat{(1,1)}$,\, $h=\frac{n}{8}+\frac{1}{2}$\\
&&&\\
\end{tabular}

Again, the current $(1,0)$ (resp. $(1,1)$) generates the same fixed points as the current $(3,0)$ (resp. $(3,1)$), hence we have to determine only two, instead of four, $S^J$ matrices, since $S^{(1,\psi)}=S^{(3,\psi)}$, with $\psi=0,1$. Actually the study of the previous section helps us a lot, since it is easy to generate unitary and modular invariant matrices out of two matrices \emph{numerically} equal to the two $S^J_{D4}$ matrices of tables \ref{table S^J=10_p=1} and \ref{table S^J=11_p=1} with the fields ordered as above. More tricky is to check that also the fusion coefficients are non-negative integers if these currents are used in chiral algebra extensions (see comment below). 

The more sensible choice is the following. Let us have a closer look at the fixed point structure of the $n=4p$ and the $n=4p+2$ cases. They are very similar, but not quite. The weights of the fixed points of the current $(1,0)$ in the $n=4p$ case have the same expression as the weights of the fixed points of the current $(1,1)$ in the $n=4p+2$ case, and similarly for the $(3,\psi)$ current. So a natural guess for the $S^J$ matrices would involve interchanging the matrices in tables \ref{table S^J=10_p=1} and \ref{table S^J=11_p=1}. Equivalently, symmetric and anti-symmetric representations are interchanged in going from $n=4p$ to $n=4p+2$. Hence, we would expect $S^{(1,0)} \sim S^{(1,1)}_{D4}$ and $S^{(1,1)} \sim S^{(1,0)}_{D4}$. This is indeed the case. The unitary and modular invariant\footnote{Modular invariance reads here: $(S^J)^2=(-1)^p i \cdot 1=(S^JT^J)^3$ for $J=(1,0)$ and $(S^J)^2=(-1)^{p-1} i \cdot 1=(S^JT^J)^3$ for $J=(1,1)$, both with imaginary $(S^J)^2$.} combinations are in fact\rlap:\footnote{Note that in order to use these relations one must order the six fields as indicated above, without paying attention to the actual labelling of the fixed point fields.}
\begin{equation}
\label{S10p2bis}
S^{(1,0)}= e^{-\frac{i\pi}{4}} \cdot (-i)^{p-1} \cdot S^{(1,1)}_{D4}
= e^{-\frac{i\pi}{4}(m-2)} \cdot S^{(1,1)}_{D4}
\end{equation}
and
\begin{equation}
\label{S11p2bis}
S^{(1,1)}= e^{-\frac{i\pi}{4}} \cdot (-i)^{p-1} \cdot S^{(1,0)}_{D4}
= e^{-\frac{i\pi}{4}(m-2)} \cdot S^{(1,0)}_{D4}
\end{equation}
giving also acceptable fusion rules. Here $m=2p+1$ is an odd integer such that $D(2m)_1\equiv D(4p+2)_1$.

There are a few comments that we can make here. The first comment regards the labelling of the matrices just given. We observe that the matrix $S^{(1,0)}$ (resp. $S^{(1,1)}$) contains the same fields as the matrix $S^{(1,1)}_{D4}$ (resp. $S^{(1,0)}_{D4}$) except for the fact that the twisted fields corresponding to the spinors are interchanged (but they still have the same weights). We will then keep the same labels as given in the above scheme (\ref{D 4p+2 scheme}) and in table \ref{table S^J=11_p=1} (resp. table \ref{table S^J=10_p=1}).

The second comment regards the periodicity of the modular matrices. Observe that in (\ref{D 4p+2 scheme}) a shift $n\rightarrow n+16$ (corresponding to $m\rightarrow m+8$ and $p\rightarrow p+4$) changes all the weights by integers, but the $T^J$ matrices will be invariant. Similarly, the $S^J$ matrices are invariant under the same shift $m\rightarrow m+8$. This happened already for the modular matrices in the $n=4p$ case and it happens here again in the $n=4p+2$ case. Hence, it seems that in comparing phases one should consider situations which have the same $p$ mod $4$. On the other hand, in going from $n=4p$ to $n=4p+2$, the $S^J$ formulas are similar, but there is one main difference, namely $S^{(1,0)}_{D4}$ gets interchanged by $S^{(1,1)}_{D4}$ and this is a completely different matrix. The same consideration that we made after (\ref{S10p}) about the shift $p\rightarrow p+1$ can be repeated here.

The last comment regards the fusion coefficients. Note that when we check the fusion rules, we cannot do it directly from the single $D(n)_1$ permutation orbifolds, exactly because the spinor currents have half-integer spin. Instead, we have to tensor the $D(n)_1$ theory with another one which also has half-integer spin simple currents (e.g. Ising model or the $D(n)_1$ model itself, maybe with different values of $n$) such that the tensor product has integer spin simple currents that can be used for the extension: those integer spin currents will then have acceptable fusion coefficients. We have checked that this is indeed the case for tensor products of the permutation orbifold CFT's with the Ising model, and also in extensions of different permutation orbifold CFT's tensored with each other (we have also performed the latter check for $n=4p$, for combinations of integer spin currents).

\section{Conclusion}
In this chapter we have completed the analysis initiated in the previous chapter regarding extensions of $D(n)_1$ permutation orbifolds by additional integer spin simple currents arising when the rank $n$ is multiple of four and by additional half-integer spin simple currents arising when the rank $n$ is even but not multiple of four. In both situations fixed points occur that must be resolved in the extended theory. This means that we have to provide the $S^J$ matrices corresponding to those extra currents $J$. They will allow us to obtain the full $S$ matrix of the extended theory which satisfies all the necessary properties. 

The currents in question are those corresponding to the spinor representations $i=1$ and $i=3$ of $D(n)_1$, both with weight $h=\frac{n}{8}$. In the permutation orbifold they arise from the symmetric and the anti-symmetric representations of the spinors, both with weight $h=\frac{n}{4}$: so they have integer spin for $n=4p$ ($p$ is integer) and half-integer spin for $n=4p+2$. Moreover, they produce pairwise identical extensions of the permutation orbifold, such that there are only two unknown matrices to determine: $S^{(1,\psi)}=S^{(3,\psi)}$ ($\psi=0,1$). The solutions were given in sections \ref{D4p permutation orbifolds} and \ref{D4p+2 permutation orbifolds}. This completely solves the fixed point resolution in extension of $D(n)_1$ permutation orbifold.

There is still more work to do. First of all, we do not have any general expression yet for the $S^J$ matrix in terms of the $S$ (and maybe $P$) matrix of the mother theory. This should be independent of the particular CFT and/or the particular current used to extend the theory. Secondly, it would be interesting to apply these CFT results in String Theory. Suitable candidates appear to be the minimal models of the $N=2$ superconformal algebra, which are the building blocks of Gepner models \cite{Gepner:1987vz,Gepner:1987qi}. We will address the first problem in the next chapter, while string theory applications will be postponed to the second part of this work.

\chapter{The ansatz}
\label{paper3}

{\flushright
{\small 
\textit{I don't know,}\par
\textit{and I would rather not guess.}\par
\emph{(J. R. R. Tolkien, The Lord of the Rings)}\par
}
}

\section{Introduction}
In the first two chapters we have started to study the problem of resolving the fixed points 
\cite{Fuchs:1996dd,Schellekens:1999yg,Schellekens:1989uf} in simple current \cite{Schellekens:1990xy,Schellekens:1989am,Intriligator:1989zw,Schellekens:1989dq} extensions of permutation orbifold \cite{Klemm:1990df,Borisov:1997nc} conformal field theories \cite{Belavin:1984vu}. The aim of this chapter is to give a general solution to this problem, valid for all conformal field theories and all order-two simple currents, going much beyond the specific examples discussed previously. 
The strategy will be to obtain an {\it ansatz} for $S^J$ based on its modular properties. To arrive at this {\it ansatz} we make
use of the following pieces of information:
\begin{itemize}
\item{The BHS $S$ matrix, $S^{BHS}$, of the unextended $\mathbb{Z}_2$ orbifold, derived in \cite{Borisov:1997nc}. This is the matrix $S^J$ for the special case $J=0$, which fixes all fields in the CFT.}
\item{The matrix $S^J$ for the anti-symmetric component of the identity, the so-called un-orbifold current as described in chapter \ref{paper1}. This matrix could be derived because this simple current
undoes the permutation orbifold and gives back the original tensor product.}
\item{The matrix $S^J$ for some cases where $J$ has spin 1. Here we used the fact that the simple current extension can be identified with a known WZW model. This allowed us to determine $S^J$ for the vector current of $SO(N)$ level 1. This was described in chapter \ref{paper1}.}
\item{Using triality in $SO(8)$  this could be generalized to the spinor currents of $SO(8)$ level 1, and 
from there to all spinor currents of $SO(2n)$ level 1, which have very similar modular properties. This was described in chapter \ref{paper2}.}
\end{itemize}

Here we will use these previous works as ``stepping stones" towards a general ansatz, which includes all
the aforementioned results as special cases, and has a far larger range of validity. In particular, the results of the previous chapters were limited to low levels, such as in the permutation orbifold of $B(n)_1$, $D(n)_1$ and $A(1)_k$ (completely for $k=2$ and $k$ odd, partially for $k$ even). By an educated guess, one could very well suspect that this formula would depend on a few quantities of the original or mother CFT $\mathcal{A}$, such as its $S$ matrix, its $P$ matrix, the weight $h_J$ of the simple current $J$, etc. This is the problem that we address and solve in this chapter. The formula which we obtain is valid
for any order-two simple current $J$ of any order-two permutation orbifold. In particular, this extends the foregoing results
for $B(n)$, $D(2n)$ and $A(1)$ to arbitrary level, but it also includes permutation orbifolds of many other WZW models such as $C(n)$, 
$E(7)$, as well as the permutation orbifolds of 
many coset CFT's, such as the $N=0$ and $N=1$ minimal superconformal models and some of the currents of the $N=2$ minimal superconformal
models.   
Not included are fixed points of simple currents of orders larger than two, which occur for example in 
the permutation orbifolds
of $A(2)$ level $3k$, or $D(2n+1)$ for even level.

The plan of this chapter is as follows.\\
Since this chapter contains the main CFT result of this whole work, we would like to make it more or less independent from the previous chapter as well as self-contained, hence we start by fixing our notation and reviewing the construction of the permutation orbifold, its $S^{BHS}$ matrix, together with its simple current and fixed point structure.\\
In section \ref{section_ansatz}, we extend the ansatz to the most general case and comment about its unitarity and modular invariance. The complete proof that our ansatz is actually unitary and modular invariant is not given here, but can be found in the original paper \cite{Maio:2009tg}.

\section{The permutation orbifold}
In this section we review a few facts that will be relevant about permutation orbifolds, already described in chapter \ref{paper1}. 
The $\mathbb{Z}_2$-permutation orbifold 
\begin{equation}
\mathcal{A}_{\rm perm} \equiv (\mathcal{A}\times \mathcal{A})/\mathbb{Z}_2\,,
\end{equation}
by definition, contains fields that are symmetric under the interchange of the two $\mathcal{A}$ factors. Moreover, there is also a twisted sector, as demanded by modular invariance. 
The $S$ matrix of $\mathcal{A}_{\rm perm}$, denoted by $S^{BHS}$, has been already presented in chapter \ref{paper1}, but for convenience reasons we will recall it here:
\begin{subequations}
\label{BHS}
\begin{eqnarray}
S_{\langle mn\rangle \langle pq\rangle}&=&S_{mp}\,S_{nq}+S_{mq}\,S_{np} \\
S_{\langle mn\rangle \widehat{(p,\chi)}}&=&0 \\
S_{\widehat{(p,\phi)}\widehat{(q,\chi)}}&=&\frac{1}{2}\,e^{2\pi i(\phi+\chi)/2} \,P_{pq} \\
S_{(i,\phi)(j,\chi)}&=&\frac{1}{2}\,S_{ij}\,S_{ij} \\
S_{(i,\phi) \langle mn\rangle}&=&S_{im}\,S_{in} \\
S_{(i,\phi)\widehat{(p,\chi)}}&=&\frac{1}{2}\,e^{2\pi i\phi/2} \,S_{ip} \,,
\end{eqnarray}
\end{subequations}
where the $P$ matrix (introduced in \cite{Bianchi:1990yu}) is defined by $P=\sqrt{T}ST^2S\sqrt{T}$.

If there is any integer or half-integer spin simple current in $\mathcal{A}$, it gives rise to an integer spin simple current in $\mathcal{A}_{\rm perm}$, which can be used to extend the orbifold CFT. We can denote the extended permutation orbifold by $\tilde{\mathcal{A}}_{\rm perm}$. In the extension, some fields are projected out while the remaining organize themselves into orbits of the current. Typically untwisted and twisted fields do not mix among themselves. As far as the new spectrum is concerned, we do know that these orbits become the new fields of $\tilde{\mathcal{A}}_{\rm perm}$, but we do not normally know the new $S$ matrix, $\tilde{S}$.

In chapter \ref{paper1}, using the sufficient and necessary condition $S^{BHS}_{J0}=S^{BHS}_{00}$ \cite{Dijkgraaf:1988tf}, it was proved that orbifold simple currents correspond to the symmetric ($\psi=0$) and anti-symmetric ($\psi=1$) representations (namely diagonal fields) of the simple currents in the mother theory $\mathcal{A}$, hence the notation $(J,\psi)$, being $J$ the corresponding simple current in the mother theory. Consequently, one simple current in $\mathcal{A}$ generates two simple currents in $\mathcal{A}_{\rm perm}$. 
The fixed point structure arising in $\mathcal{A}_{\rm perm}$ was also determined in chapter \ref{paper1} for currents with (half-)integer spin. Here we want to consider currents with spin $h_J \in \frac{1}{4}\,\mathbb{Z}_{\rm odd}$ as well. In fact, if $h_J$ is quarter-integer, the resulting permutation orbifold current has half-integer weight, and hence could have fixed points. The generalization is straightforward and involves small changes only for twisted fixed points. In fact, by studying the fusion coefficients, we can show that:
\begin{itemize}
\item diagonal fields: $(i,\phi)$ is a fixed point of $(J,\psi)$ if  $\psi=0$ and if $i$ is a fixed  point of $J$, i.e. $Ji=i$;
\item off-diagonal fields: $\langle m,n\rangle$ is a fixed point of $(J,\psi)$ 
\begin{itemize}
\item either if $m$ and $n$ are both fixed points of $J$, i.e. $Jm=m$ and $Jn=n$,
\item or if $m$ and $n$ are in the same $J$-orbit, i.e. $Jm=n$;
\end{itemize}
\item twisted fields: $\widehat{(p,\phi)}$ is a fixed point of $(J,\psi)$ if $Q_J(p)=\frac{\psi}{2}+2\,h_J\,\,{\rm mod}\,\, \mathbb{Z}$, independently of $\phi$.
\end{itemize}
For the twisted fixed points, the proof can be found in the appendix of the original paper \cite{Maio:2009tg}. Observe that for (half-)integer spin simple currents we can drop the additional $2h_J$ from the monodromy charge.

Also note that there exist diagonal fixed points only for the symmetric representation of the simple current and that the twisted fixed points are determined by $Q_J(p)$, the monodromy charge of $p$ w.r.t. $J$. Moreover, we will often have to distinguish between the two types of fixed points coming from the off-diagonal sector: for obvious reasons, we will call them \textit{fixed-point-like off-diagonal fields} and \textit{orbit-like off-diagonal fields} respectively in the two cases.

\section{The general ansatz}
\label{section_ansatz}
Here we give the most general ansatz for the fixed-point resolution matrices $S^{(J,\psi)}$ of a $(J,\psi)$-extended permutation orbifold. It reads:
\begin{subequations}
\label{ansatz with fixed points}
\begin{eqnarray}
S^{(J,\psi)}_{\langle m,n\rangle\langle p,q\rangle}&=&S^J_{mp}\,S^J_{nq}+(-1)^\psi S^J_{mq}\,S^J_{np} \\
S^{(J,\psi)}_{\langle m,n\rangle\widehat{(p,\chi)}}&=&
\left\{
\begin{array}{cl}
0 & {\,\,\rm if\,\,}   J\cdot m=m\\
A\,S_{mp} & {\,\,\rm if\,\,}   J\cdot m=n
\end{array}
\right. \\
S^{(J,\psi)}_{\widehat{(p,\phi)}\widehat{(q,\chi)}}&=&
B\,\frac{1}{2}\,e^{i\pi\hat{Q}_J(p)}\,P_{Jp,q}\,e^{i\pi(\phi+\chi)} \\
S^{(J,\psi)}_{(i,\phi)(j,\chi)}&=&\frac{1}{2}\,S^J_{ij}\,S^J_{ij} \\
S^{(J,\psi)}_{(i,\phi)\langle m,n\rangle}&=&S^J_{im}\,S^J_{in} \\
S^{(J,\psi)}_{(i,\phi)\widehat{(p,\chi)}}&=&C\,\frac{1}{2}\,e^{i\pi\phi} \,S_{ip}\,.
\end{eqnarray}
\end{subequations}
The notation in the ansatz is as follows. We denote by $\hat{Q}_J (m)$ the combination of weights $\hat{Q}_J (m)= h_J+h_m-h_{J\cdot m}$, while $Q_J(m)$ is the monodromy charge of the field $m$ w.r.t. the current $J$ in the mother theory which gives rise to the current $(J,\psi)$ in the permutation orbifold (independently of its symmetric or anti-symmetric representation). These two quantities are obviously related by $Q_J(m)=\hat{Q}_J(m) \,\,{\rm mod}\,\, \mathbb{Z}$. 
Using modular invariance, one can show that these phases satisfy the following relations (see \cite{Maio:2009tg}):
\begin{equation}
B=(-1)^\psi\,e^{3i\pi h_J}\,\,,\qquad A^2=C^2=(-1)^\psi\,e^{2 i\pi h_J}\,,
\end{equation}
$h_J$ being the weight of the simple current, which might depend on the central charge, rank and level of the original CFT. These relations come from modular invariance: so, we can see that $B$ is \textit{fully} fixed, while $A$ and $C$ are fixed \textit{up to a sign}. We could 
also have inserted a phase $E$ in the matrix element $S^{(J,\psi)}_{(i,\phi)\langle m,n\rangle}$. Modular invariance would then constrain it to $E^2=1$, hence $E$ would have been just a sign. As before in the simplified ansatz, these sign ambiguities are completely
understood in terms of the general sign ambiguities of fixed point resolution matrices. Within the three blocks
(diagonal, off-diagonal, twisted) they are fixed because we write all matrix elements in terms of $S^J$, $S$ and $P$, but this still
leaves three relative signs between the blocks. These signs are fixed by requiring that the result should recover the
BHS matrix. The latter has no free signs, because it is defined by a character representation. This therefore defines a
convenient canonical choice for the signs. The special case of the BHS formula corresponds to 
$h_J=\psi=0$ for the identity, hence $B=1$, while $A$ and $C$ are just signs, that must be taken positive.
However, we emphasize that any other sign
choice for $A$, $C$ or $E$ is equally valid; it is analogous to a gauge choice. Note that some of the matrices presented in the previous chapters use different sign conventions.

This ansatz more or less interpolates our previous results of chapters \ref{paper1} and \ref{paper2}, up to the above sign conventions. The phase in the twisted-twisted sector containing the hatted monodromy charge is necessary in order to make $S^J$ symmetric\footnote{In fact one can check that
\begin{equation}
e^{i\pi \hat{Q}_J(m)}\,P_{Jm,p}=e^{i\pi \hat{Q}_J(p)}\,P_{m,Jp} \equiv A_{mp}\nonumber
\end{equation}
with 
\begin{equation}
A_{mp}= e^{i\pi h_J}\,\sqrt{T}_{mm}\sum_l\left(e^{2 i\pi Q_J(l)}\,S_{ml}T^2_{ll}S_{lp}\right) \sqrt{T}_{pp}
\nonumber
\end{equation}
and $A_{mp}$ is symmetric.} as it should be since, for order-two currents,  $S^J_{ab}=S^{J^{-1}}_{ba}$ \cite{Fuchs:1996dd}. We need to put a hat on $Q_J$ in order to avoid ambiguities deriving from having the monodromy charge in the exponent, since it is defined only modulo integers. Similarly to what happens in the BHS formula (\ref{BHS}), the $P$ matrix enters the twisted-twisted sector.

A comment about the matrix element $S^{(J,\psi)}_{\langle m,n\rangle\widehat{(p,\chi)}}$ is in order. We can actually \textit{prove} that the quantity $S_{mp}$ vanishes when $J\cdot m=m$ and $\psi=1$ and use the second line of the ansatz also in this case. In fact, first of all, since $\widehat{(p,\chi)}$ is a twisted fixed point of $(J,\psi)$ and since $h_J$ must be (half-)integer in order for $m$ to be fixed by $J$, we can drop the $2h_J$ contribution from the monodromy of $p$, i.e. $Q_{J}(p)=\frac{\psi}{2}$. Secondly, using $S_{Jm,p}=e^{2 i\pi Q_J(p)}\,S_{mp}$ \cite{Schellekens:1990xy}, we have:
\begin{equation}
S_{mp}=S_{Jm,p}=e^{2 i\pi Q_J(p)}\,S_{mp}=e^{2 i\pi \frac{\psi}{2}}\,S_{mp}\,,
\end{equation}
implying that the non-identically-to-zero option of $S^{(J,\psi)}_{\langle m,n\rangle\widehat{(p,\chi)}}$ actually also vanishes when $J\cdot m=m$ and $\psi=1$. So in our ansatz we are claiming that $S^{(J,\psi)}_{\langle m,n\rangle\widehat{(p,\chi)}}$ vanishes also for $\psi=0$ when $\langle m,n\rangle$ is fixed-point-like. We also recall that for orbit-like off-diagonal fields there exists a similar relation between $S_{mp}$ and $S_{np}$:
\begin{equation}
S_{np}=S_{Jm,p}=e^{2 i\pi Q_J(p)}\,S_{mp}=e^{2 i\pi \frac{\psi}{2}}\,S_{mp}\,,
\end{equation}
but we cannot infer much from here. It is crucial in these manipulations that the field $p$ gives rise to a twisted field in the extended orbifold.

\subsection{Unitarity and modular invariance}
The proofs of unitarity and modular invariance of the ansatz are referred to \cite{Maio:2009tg}. The calculation is interesting since we are able to derive a few aside identities having to do with projected sums of selected elements of the unitary $S$ and $P$ matrices of the original theory.
In order to prove unitarity, we show that $S^{(J,\psi)}\cdot S^{(J,\psi)\dagger}=1$. 
Modular invariance is the statement that $(S^{(J,\psi)})^2= (S^{(J,\psi)} \cdot T^{(J,\psi)})^3$, where $T^{(J,\psi)}$ is the $T$ matrix of the permutation orbifold restricted to the fixed points of $(J,\psi)$. Using this relation to prove modular invariance would be computationally heavy, due to the double sum arising in the cube. Instead we re-write the constraint as
\begin{equation}
{T^{(J,\psi)}}^{-1} S^{(J,\psi)} {T^{(J,\psi)}}^{-1}=S^{(J,\psi)} T^{(J,\psi)} S^{(J,\psi)}\,,
\end{equation}
which is simpler since it involves only one sum on the r.h.s. and no sums at all on the l.h.s. Surprisingly enough, we find that the phases in the ansatz do not depend explicitly on the central charge $c$ of the mother CFT (the central charge of the permutation orbifold is $\hat{c}=2c$). The reason for this is that the $T$ matrices of the orbifold theory re-arrange themselves into suitable functions of $T$ matrices of the original theory. Explicitly (recall $T$ is diagonal: $T_{ij}=T_i\,\delta_{ij}$):
\begin{equation}
T^{(J,\psi)}_{\langle m,n\rangle}=T_m\,T_n\,,\qquad
T^{(J,\psi)}_{(i,\phi)}=T^2_i\,,\qquad
T^{(J,\psi)}_{\widehat{(p,\chi)}}=e^{i\pi\chi}\,\sqrt{T}_p\,.
\end{equation}
hence the central charge gets always re-absorbed in $T$. The phases $A$, $B$ and $C$ will be constrained by this calculation to be equal to the expressions given earlier.

\subsection{Checks}

Although we have an explicit proof that our results satisfy the conditions of modular invariance (see \cite{Maio:2009tg}), we do not have a general proof
that all other RCFT conditions are satisfied, although the simplicity and generality of the answer suggests that this is indeed
the right answer. The next issue one could check is the fusion rules of the extended CFT. Currents of order two that have
fixed points must have integer or half-integer spin. In the latter case there is no extension, but one may consider instead
the tensor product with an Ising model, extended with an integer spin product of currents. Indeed, also for integer spin currents
one can consider arbitrarily complicated tensor products and any integer spin product current. All of these should give sensible
fusion rules. We have built (\ref{ansatz with fixed points}) into the program {\tt kac} \cite{kac}, which computes fusion rules
for simple current extended WZW models and coset CFT's, and this gives us access to a huge number of explicit examples. We have checked many simple extensions,
and also combinations of permutation orbifolds. For example, denote by $X$ the permutation orbifold of $C(3)_2$. It has
85 primaries and 
four simple currents, the identity, the anti-symmetric component of the latter (which has spin 1) and two spin 3 currents $K$ and $L$ originating
from symmetric and anti-symmetric product of the simple current of $C(3)_2$. We can now tensor $X$ with itself, and extend the result
with $(K,K)$ or $(K,L)$ or $(L,L)$. This gives three distinct CFT's with 2578, 2284 and 2102 primaries respectively. Checking all their
fusion rules is very time-consuming, so we have just checked a large sample.
The fusion rules we have checked in these cases, and many others, have indeed integer coefficients. Note that our formalism allows us to consider
also the permutation orbifold of $X \times X$, and the simple current extensions thereof. For all these CFT's the fusion rules are
now explicitly available. Furthermore, for all these cases we can compute the boundary and crosscap coefficients as well as 
the annulus, Moebius and Klein bottle amplitudes using the formalism of \cite{Fuchs:2000cm} (generalizing earlier
works, such as  \cite{Cardy:1989ir,Pradisi:1996yd,Pradisi:1995pp}, and references cited in this paper).

\section{Conclusion}
In this chapter we have addressed the  problem of fixed point resolution in (extensions of) permutation orbifolds or equivalently the problem of finding the $S^J$ matrices for those classes of theories.

The results of this chapter allow us to make extensions of permutation orbifolds. We propose an ansatz for the $S^J$ matrices valid in the general case of simple currents of order 2. We have also shown how to get back the BHS formula when we extend the permutation orbifold by the identity current $(J,\psi)=(0,0)$. 
This ansatz is unitary and modular invariant. Moreover, unlike the results of the previous chapters,
it does not depend on any explicit details of the
particular CFT used in the mother theory, other than its modular properties.
It depends only on the weight $h_J$ of the current used in the extension (via phases) and on the matrices $S$ and $T$ (via the matrix $P$) of the mother theory. This implies that it can be used freely in any sequence of extensions and $Z_2$ permutations of CFT's, thus
leading to a huge set of possible applications.

There are still further generalizations possible: the extension of this result to higher order permutations and the extension to higher order
currents, and the combination of both. However, we will not discuss these problems here.

\part{STRING THEORY}
\chapter*{About Part II}

{\flushright
{\small 
\textit{I think the Moon is a world like this one,}\par
\textit{and the Earth is its moon.}\par
\emph{(E. Rostand, Cyrano de Bergerac)}\par
}
}

Part II focuses on String Theory. In particular, we address the problem of constructing four-dimensional string theories using the permutation orbifold. Not surprisingly, our approach will be based on CFT and we will apply the knowledge and the results of Part I to build modular invariant partition functions. 

Our method of generating spectra consists of several ingredients. First of all, we adapt Gepner's construction \cite{Gepner:1987vz,Gepner:1987qi} to include the permutation orbifold. Secondly, we look at the spectra generated by these permuted Gepner models obtained by extending it by a subset of all possible simple currents.

Gepner models are constructed out of tensor products of smaller CFT's, the so-called $N=2$ superconformal minimal models. Moreover, the total central charge of the tensor product CFT must add up to the particular value of nine. There are finitely-many (and in fact only 168) combinations of the minimal models that have the correct value for the central charge. Furthermore, in order to guarantee space-time and world-sheet supersymmetry, additional constraints must be imposed or, equivalently, the tensor product theory must be extended by a suitable set of specific integer-spin simple currents. Sometimes, it happens that two (or more) of the factors are identical. When this is the case, we can replace the full block by its permutation orbifold. Since the latter must also be supersymmetric, before being able to use it in the Gepner model we need to super-symmetrize it. This is done again by a simple current extension.

Already by looking at the sub-block of the permutation orbifold for two $N=2$ minimal models, a very interesting mathematical structure appears. For example, we learn how to make the orbifold supersymmetric and we discover that extended $N=2$ permutations generate sometimes ``exceptional'' simple currents, that were not expected a priori, because they have a completely different origin from standard orbifold currents, and whose existence is related to the presence of special relations involving $S$-matrix elements. In some cases, these exceptional currents have fixed points that remain currently unresolved.

Having this machinery ready, we can build, \textit{mutatis mutandis}, the supersymmetric orbifold of $N=2$ minimal models into Gepner's scheme. Using the simple current formalism, we are able to construct hundreds or thousands of spectra corresponding to each permutation orbifold of standard Gepner models. All these spectra will have Standard-Model structure, since we explicitly break the $SO(10)$ coming from the fermionic sector of the heterotic string in Gepner construction into $SU(3)\times SU(2) \times U(1)$. 

As far as the number of families is concerned, one then notices that the number three is strongly suppressed. This was already the case for conventional Gepner models. However, there exists a way to deal with this problem and make the number three as abundant as two or four, or at least of the same order of magnitude. This is the ``lifting'' procedure \cite{GatoRivera:2010gv,GatoRivera:2010xn,GatoRivera:2010fi}, which allows to replace a sub-block from the tensor product in Gepner models by an isomorphic CFT with identical modular properties. 

\chapter{Permutation orbifolds of $N=2$ minimal models}
\label{paper4}

{\flushright
{\small 
\textit{All we have to decide is what to do with the time that is given us.}\par
\textit{(J. R. R. Tolkien, The Lord of the Rings)}\par
}
}

\section{Introduction}

In this and the next chapter we consider applications of the previous results on fixed point resolution in extensions of permutation orbifolds to string theory phenomenology, where one is interested in computing four-dimensional particle spectra, possibly close to the Standard Model. Generically rational CFT's are very useful tools for computing features of phenomenological interest in perturbative string theory. However, the set of Rational CFT's at our disposal is disappointingly small. The only interacting rational CFT's that we can really use for building exact string theories are tensor products of $N=2$ minimal models, also known as ``Gepner models" \cite{Gepner:1987vz,Gepner:1987qi}. 
Historically the first area of application of rational CFT model building was the heterotic string. 

The full power of rational CFT model building only manifests itself
if one uses the complete set \cite{Kreuzer:1993tf}
of simple current modular invariant partition functions (MIPF's)  \cite{Schellekens:1989am,Schellekens:1989dq} 
(See  \cite{Schellekens:1990xy} for a review of simple current MIPF's. The underlying symmetries were discovered independently in \cite{Intriligator:1989zw}). 
Already basic physical constraints like world-sheet and space-time supersymmetry require a simple current MIPF. 
As we know by now, although the simple current symmetries can be read off from the modular transformation matrix $S$, and the corresponding MIPF's can be readily constructed, often additional information is required when the simple current action has fixed points \cite{Schellekens:1990xy,Schellekens:1989uf}. In order to make full use of the complete simple current formalism we need the following data of the CFT under consideration:
\begin{itemize}
\item{The exact conformal weights.}
\item{The exact ground state dimensions.}
\item{The modular transformation matrix $S$.}
\item{The fixed point resolution matrices $S^J$, for simple currents $J$ with fixed points.}
\end{itemize}
Not all of this information is needed in all cases. 
In heterotic spectrum computations all we need to know is the first two items, plus the
simple current orbits implied by $S$. To compute the Hodge numbers of heterotic
compactifications, we only need to know the exact ground state dimensions of the Ramond
ground states.  

In addition to Gepner models, for which all this information is available, there is at least another class that is potentially usable: the permutation orbifolds. 
For permutation orbifolds, it has been known for a long time how to compute their weights and ground state
dimensions, but there was no formalism for computing $S$ and $S^J$. In this case
it has been possible to compute the Hodge numbers and even the number of singlets for
the diagonal invariants \cite{Klemm:1990df,Fuchs:1991vu}. However, meanwhile it as become
clear that the values of Hodge numbers offer a rather poor road map to the heterotic string landscape.
In particular they lead to the wrong impression that the number of families is large and very often a 
multiple of 4 or 6. The former problem disappears if one allows breaking of the gauge group 
$E_6$ to phenomenologically more attractive subgroups (ranging from $SO(10)$ via $SU(5)$
or Pati-Salam to just $SU(3)\times SU(2) \times U(1)$ (times other factors) by allowing asymmetric
simple current invariants \cite{GatoRivera:2010gv,Schellekens:1989wx}, whereas the second
problem can be solved by modifying the bosonic sector of the heterotic string, for example
by means of heterotic weight lifting \cite{GatoRivera:2010xn}, B-L lifting \cite{GatoRivera:2010fi}.
All of these methods require knowledge of the full simple current structure of the building blocks.
This in its turn requires knowing $S$.

A first step towards the computation of $S$ for $\mathbb{Z}_2$ permutation orbifolds
was made in \cite{Borisov:1997nc}, almost ten years after permutation orbifolds were first  studied.
While this might seem sufficient for permutation orbifolds in heterotic string model building,
we will see that even in that case more is needed. The crucial ingredient is fixed point resolution. Therefore we expect that significant progress can be made by applying the results of chapters \ref{paper1}, \ref{paper2} and especially \ref{paper3}, extending the BHS formula \cite{Borisov:1997nc} to fixed point resolution matrices $S^J$, for currents $J$ of order 2. 
Since in $N=2$ minimal models all currents with fixed points have order 2, this seems
to be precisely what is needed. The purpose of this chapter is to determine which of the CFT data
listed above can now be computed for permutations orbifolds of $N=2$ minimal models, and provide algorithms for doing so.

\subsection{Basic concepts}

Following the discussion so far, throughout this work we will always consider the permutation orbifold:
\begin{equation}
(\mathcal{A}\times\mathcal{A})/\mathbb{Z}_2\,.
\end{equation}
Moreover, we look at its simple-current extensions and its simple current MIPF's. We have already seen that the orbifold currents always admit fixed points, that were resolved by the formula (\ref{ansatz with fixed points}) for the $S^J$ matrices. 

Here we want to apply the results of the previous chapters about fixed point resolution in simple current extensions of permutation orbifolds to the physically interesting case of $N=2$ minimal models. 
This may seem to be
straightforward, as a supersymmetric CFT is just an example of a CFT, and the aforementioned results hold for {\it any} CFT. 
However, the permutation orbifold obtained by applying the BHS formula (\ref{BHS}) turns out {\it not} to have world-sheet supersymmetry. This is related
to the fact that a straightforward Virasoro tensor product (the starting point for the permutation orbifold) does not have world-sheet supersymmetry
either, for the simple reason that tensoring produces combinations of R and NS fields. The solution to this problem in the case of the tensor product
is to extend the chiral algebra by a simple current of spin 3, the product of the world-sheet supercurrents of the two factors (or any two factors if there
are more than two). One might call this the supersymmetric tensor product.
However for this extended tensor product the BHS formalism of \cite{Borisov:1997nc} is not available. One can follow two paths to solve that problem: either one can 
try to generalize \cite{Borisov:1997nc} to supersymmetric tensor products (or more generally to extended tensor products) or one can try to
supersymmetrize the permutation orbifold. We will follow the second path.

One might expect that the chiral algebra of permutation orbifold has to be extended in order to restore world-sheet supersymmetry. That is indeed
correct, but it turns out that there are {\it two} plausible candidates for this extension: the symmetric and the anti-symmetric combination of the 
world-sheet supercurrent of the minimal model. Denoting the latter as $T_F$, the two candidates are the spin-3 currents $(T_F,0)$ and $(T_F,1)$.
Somewhat counter-intuitively, it is the second one that leads to a CFT with world-sheet supersymmetry. The first one,  $(T_F,0)$, gives rise
to a CFT that is similar, but does not have a spin-3/2 current of order 2. 

Both $(T_F,0)$ and $(T_F,1)$ have fixed points, but we know their resolution matrices from the general results of chapter \ref{paper3}. They
come in handy, because it turns out that one of these fixed points is the off-diagonal field $\langle 0,T_F \rangle$ of conformal weight $\frac32$. As stated above,
this is not a simple current of the permutation orbifold, but it is a well-known fact that chiral algebra extensions can turn primaries into
simple currents. This is indeed precisely what happens here. Since we know the fixed point resolution matrices of  $(T_F,0)$ and $(T_F,1)$ we
can work out the orbits of this new simple current. It turns out that in the former extension  $\langle 0,T_F \rangle$ has order 4, whereas in the latter it has
order 2. We conclude that the latter must be the supersymmetric permutation orbifold; we will refer to the former CFT as ``$X$".
The fixed point resolution also  determines the action
of the new world-sheet supercurrent $\langle 0,T_F \rangle$ on all other fields, combining them into world-sheet superfields of either NS or R type. 

The current $\langle 0,T_F \rangle$ has no fixed points, as one would expect in an $N=2$ CFT (because it has two supercurrents of opposite charge, 
and acting with either one changes the charge). However, there are in general  more off-diagonal fields that turn into simple currents.
Some of these do have fixed points, and since the simple currents originate from fields that were not simple currents in the
permutation orbifold, our previous results do not allow us to resolve these fixed points. We find that this problem only occurs if $k=2 \mod 4$, where $k=1\ldots \infty$ is the integer parameter labelling the $N=2$ minimal models.

To prevent confusion we list here all the CFT's that play a r\^ole  in the story:

\begin{itemize}
\item{The $N=2$ minimal models.} 
\item{The tensor product of two identical $N=2$ minimal models. We will refer to this as $(N=2)^2$. } 
\item{The BHS-orbifold of the above. This is the permutation orbifold as described in \cite{Borisov:1997nc}. It will be denoted
$(N=2)^2_{\rm orb}$.} 
\item{The supersymmetric extension of the tensor product. This is the extension of the tensor product by the spin-3 current $T_F\otimes T_F$. 
We will call this CFT $(N=2)^2_{\rm Susy}$.} 
\item{The supersymmetric permutation orbifold $(N=2)^2_{\rm Susy-orb}$. This is BHS orbifold extended by the spin-3 current $(T_F,1)$.} 
\item{The non-supersymmetric permutation orbifold $X$. This is BHS orbifold extended by the spin-3 current $(T_F,0)$.}
\end{itemize}

The plan of this chapter is as follows. \\
In section \ref{N=2 min mods} we review the theory of $N=2$ minimal models, their spectrum and $S$ matrix. As far as the characters are concerned, we recall the coset construction and state a few known results from parafermionic theories, in particular the string functions.
In section \ref{perm orb recap}, for convenience reasons, we recall relevant properties about general permutation orbifolds, the BHS formalism and its generalization to fixed point resolution matrices, that we have already described in the first part of this work. 
Then in section \ref{perm of N=2 mods} we move to the permutation orbifold of $N=2$ minimal models. We consider extensions by the various currents related to the spin-$\frac{3}{2}$ world-sheet supercurrent and explain how the exceptional off-diagonal currents appear.
We also work out the special extensions of the orbifold by the symmetric and anti-symmetric representation of the world-sheet current. 
In section \ref{except curr and fp} we study the exceptional simple currents and in particular the ones that have got fixed points. We give the structure of these off-diagonal currents as well as of their fixed points, in the case they have any. We illustrate the general ideas with the example of the minimal model at level two. 
In section \ref{CFT summary} we summarize the orbit and fixed
point structures for the various CFT's we consider, we present the analogous results for $N=1$ minimal models, where similar issues arise, and also some interesting differences. 
In section \ref{conclusion paper4} we give our conclusions.
We collect some technical details in appendix \ref{Appendix Paper4}. This chapter is based on \cite{Maio:2010eu}.

\section{$N=2$ minimal models}
\label{N=2 min mods}
In this section we review the minimal model of the $N=2$ superconformal algebra.

\subsection{The $N=2$ SCFT and minimal models}
The $N=2$ superconformal algebra (SCA) was first introduced in \cite{Ademollo:1976pp}. It contains the stress-energy tensor $T(z)$ (spin 2), a $U(1)$ current $j(z)$ (spin 1) and two fermionic currents $T_F^{\pm}(z)$ (spin $\frac{3}{2}$). 
Using the mode expansion
\begin{equation}
T(z)=\sum_{n\in\mathbb{Z}}\frac{L_n}{z^{n+2}}\,\,,\qquad
j(z)=\sum_{n\in\mathbb{Z}}\frac{J_n}{z^{n+1}}\,\,,\qquad
T_F^{\pm}(z)=\sum_{r\in\mathbb{Z}\pm\nu}\frac{G_{r}^{\pm}}{z^{r+\frac{3}{2}}}\,,
\end{equation}
the (anti-)commutator algebra is
\begin{subequations}
\label{SCA modes}
\begin{eqnarray}
[L_m,L_n] &=& (m-n)L_{m+n}+\frac{c}{12}(m^3-m)\delta_{m,-n}\,\\
{[}L_m,J_n{]} &=&-n J_{m+n}\,, \\
\{G^+_r,G^-_s\} &=& 2L_{r+s}+(r-s)J_{r+s}+\frac{c}{3}(r^2-\frac{1}{4})\delta_{r,-s}\, \\
\{G^+_r,G^+_s\} &=& \{G^-_r,G^-_s\} =0\,, \\
{[}J_m,G^\pm_r{]} &=& \pm\frac{1}{c}G^\pm_{r+n}\,,\\
{[}J_m,J_n{]} &=& \frac{c}{3}m\delta_{m,-n} \,.
\end{eqnarray}
\end{subequations}
The shift $\nu$ can in principle be real, but for our considerations we take it to be integer (NS sector) or half-integer (R sector). 
Unitary representations of the $N=2$ SCA can exists for values of the central charge $c\geq3$ (infinite-dimensional representations) and for the discrete series $c<3$ (finite-dimensional representations). The latter ones are discrete conformal field theories, the $N=2$ minimal models, whose central charge is specified by an integer number $k$, called the level, according to:
\begin{equation}
\label{N=2 central charge}
c=\frac{3k}{k+2}\,.
\end{equation}
The Cartan subalgebra is generated by $L_0$ and $J_0$, hence primary fields, denoted by 
\begin{equation}
\phi_{l,m,s}\equiv (l,m,s) \,,
\end{equation}
are labelled by their weights $h$ and charges $q$:
\begin{equation}
L_0 |h,q \rangle =h |h,q\rangle\,\,,\qquad J_0 |h,q\rangle=q |h,q\rangle\,.
\end{equation}
The allowed values for $h$ and $q$ are given by 
\begin{equation}
\label{N=2 weights}
h_{l,m,s}=\frac{l(l+2)-m^2}{4(k+2)} +\frac{s^2}{8}\,\,,\qquad 
q_{m,s}=-\frac{m}{k+2}+\frac{s^2}{2}\,,
\end{equation}
where $l,\,m,\,s$ are integer numbers with the property that
\begin{itemize}
\item $l=0,\,1,\dots,\,k$
\item $m$ is defined ${\rm mod}\,\,2(k+2)$ (we will choose the range $-k-1\leq m\leq k+2$)
\item $s$ is defined ${\rm mod}\,\,4$ (we will choose the range $-1\leq s\leq 2$, with $s=0,\,2$ for the NS sector and $s=\pm1$ for the R sector).
\end{itemize}
In addition, in order to avoid double-counting, one has to take into account that not all the fields are independent but are rather pairwise identified:
\begin{equation}
\phi_{l,m,s}\sim\phi_{k-l,m+k+2,s+2}\,.
\end{equation}
This identification is realized as a formal simple current extension.

In order to be able to say something about the characters of the minimal model, let us mention the coset construction. The $N=2$ minimal models can be described in terms of the coset
\begin{equation}
\label{coset}
\frac{SU(2)_k \times U(1)_4}{U(1)_{2(k+2)}}\,.
\end{equation}
Throughout this work, we use the convention that $U(1)_N$ contains $N$ primary fields (with $N$ always even). 
The characters of this coset are decomposed according to
\begin{equation}
\chi^{SU(2)_k}_l(\tau) \cdot \chi^{U(1)_4}_s (\tau)=
\sum_{m=-k-1}^{k+2}
\chi^{U(1)_{2(k+2)}}_m(\tau) \cdot \chi_{l,m,s} (\tau)\,,
\end{equation}
where $\chi_{l,m,s}$ are the characters (branching functions) of the coset theory. 
Their conformal dimension can be read off from the above decomposition and agrees with (\ref{N=2 weights}).

\subsection{Parafermions}
We will soon see that $\chi_{l,m,s}$ will be determined in terms of the so-called \textit{string functions}, which are related to the characters of the \textit{parafermionic theories} \cite{Fateev:1985mm,Qiu:1986zf}. In order to determine $\chi_{l,m,s}$, let us consider $SU(2)_k$ representations. Using the Weyl-Kac character formula \cite{Kac,Kac:1984mq}, $SU(2)_k$ characters are given by a ratio of generalized theta functions:
\begin{equation}
\label{su2 characters}
\chi^{SU(2)_k}_l(\tau,z)=
\frac{\Theta_{l+1,k+2}(\tau,z)+\Theta_{-l-1,k+2}(\tau,z)}{\Theta_{1,2}(\tau,z)+\Theta_{-1,2}(\tau,z)}\,,
\end{equation}
where by definition
\begin{equation}
\Theta_{l,k}(\tau,z)=\sum_{n\in\mathbb{Z}+\frac{l}{2k}}q^{k n^2}e^{-2 i \pi n k z}\,.
\end{equation}
Parafermionic conformal field theories are given by the coset
\begin{equation}
\frac{SU(2)_k}{U(1)_{2k}}\,, \qquad c=\frac{2(k-1)}{k+2}\,.
\end{equation}
We can decompose $SU(2)_k$ characters in term of $U(1)_{2k}$ and parafermionic characters as
\begin{equation}
\chi^{SU(2)_k}_l(\tau,z) =
\sum_{m=-k+1}^{k}
\chi^{U(1)_{2k}}_m(\tau,z) \cdot \chi^{{\rm para}_k}_{l,m} (\tau)\,.
\end{equation}
This decomposition also gives the weight of the parafermions:
\begin{equation}
h_{l,m}=\frac{l(l+2)}{4(k+2)}-\frac{m^2}{4k}\,\,,\qquad
l=0,1,\dots,k\,,\qquad m=-k+1,\dots,k\,.
\end{equation}
Using the fact that $U(1)_{2k}$ characters are just theta functions,
\begin{equation}
\chi^{U(1)_{2k}}_m(\tau,z)=
\frac{\Theta_{m,k}(\tau,z)}{\eta(\tau)}\,,
\end{equation}
the $SU(2)_k$ characters become
\begin{equation}
\label{string function definition}
\chi^{SU(2)_k}_m(\tau,z)=
\sum_{m=-k+1}^{k}
\frac{\Theta_{m,k}(\tau,z)}{\eta(\tau)}\cdot \chi^{{\rm para}_k}_{l,m} (\tau)
\equiv
\sum_{m=-k+1}^{k}
\Theta_{m,k}(\tau,z)\cdot C^{(k)}_{l,m} (\tau)\,,
\end{equation}
being $C^{(k)}_{l,m} (\tau)=\frac{1}{\eta(\tau)}\chi^{{\rm para}_k}_{l,m} (\tau)$ the $SU(2)_k$ string functions. Here, $\eta(\tau)$ is the Dedekind eta function, which is a modular form of weight $\frac{1}{2}$,
\begin{equation}
\eta(\tau)=q^{\frac{1}{24}}\prod_{k=1}^{\infty}(1-q^k)\,,\qquad
\eta(\tau)^{-1}=q^{-\frac{1}{24}}\sum_{n=0}^{\infty}P(n)q^n\,,\qquad
q=e^{2i\pi\tau}\,,
\end{equation}
with $P(n)$ the number of partitions of $n$.

As an example, consider the case with $k=1$. Since the characters of $\chi^{SU(2)_1}_m$ are the same as the characters of $\chi^{U(1)_2}_m$, we have
\begin{equation}
\chi^{{\rm para}_1}_{0,0}(\tau)=\chi^{{\rm para}_1}_{1,1}(\tau)=1\,\,,\qquad
\chi^{{\rm para}_1}_{0,1}(\tau)=\chi^{{\rm para}_1}_{1,0}(\tau)=0\,.
\end{equation}
These relations for $k=1$ generalize to arbitrary $k$ to give selection rules for the string functions. 
By decomposing $SU(2)$ representations into $U(1)$ representations, the branching functions (i.e. the parafermions) should not carry $U(1)$ charge, since they correspond to the coset (\ref{coset}) where the $U(1)$ part has been modded out. Bearing this observation in mind, the general $SU(2)_k$-character decomposition, including the selection rules, is
\begin{equation}
\chi^{SU(2)_k}_l(\tau,z) =
\sum_{
{\small
\begin{array}{c}
 m=-k+1\\
l+m=0\,{\rm mod}\,2
\end{array}
}
}^k
C^{(k)}_{l,m} (\tau) \cdot \Theta_{m,k}(\tau,z) \,.
\end{equation}
The selection rule is clearly $l+m=0\,{\rm mod}\,2$, hence
\begin{equation}
C^{(k)}_{l,m}\equiv0 \qquad {\rm if} \qquad l+m\neq0\,{\rm mod}\,2\,.
\nonumber
\end{equation}

\subsection{String functions and $N=2$ Characters}
The string functions of $SU(2)_k$ are Hecke modular forms \cite{Kac:1984mq}. They can be expanded as a power sum with integer coefficients as
\begin{equation}
C^{(k)}_{l,m} (\tau)=\exp{\left[2i\pi\tau\left(\frac{l(l+2)}{4(k+2)}-\frac{m^2}{4k}-\frac{c}{24}\right)\right]}
\sum_{n=0}^{\infty}p_n q^n\,,
\end{equation}
with $c=\frac{3k}{k+2}$, where $p_n$ is the number of states in the irreducible representation with highest weight $l$ for which the value of $J_0^3$ and $N$ are $m$ and $n$. These integer coefficients depend in general on the string function labels $l$ and $m$ and are most conveniently extracted from the following expression\footnote{There exist many different ways of determining the $SU(2)_k$ string functions. See for example \cite{Jacob:2000gw}, where a derivation is given in terms of representation theory of the parafermionic conformal models, or \cite{Jacob:2001rd}, where a new basis of states is provided for the parafermions. 
Our formula is the standard one, given in \cite{Kac:1984mq}. 
It also agrees with \cite{Nemeschansky:1989wg,Nemeschansky:1989rx}
For equivalent, but different-looking, expressions, see \cite{Fortin:2006dn,SchillingWarnaar}.
}:
\begin{equation}
\label{string function}
C^{(k)}_{l,m} (\tau)=\eta(\tau)^{-3} \sum_{-|x|<y\le |x|} {\rm sign}{(x)}\,e^{2i\pi\tau[(k+2)x^2-ky^2]}\,,
\end{equation}
where $x$ and $y$ belong to the range
\begin{equation}
\label{range for string function}
(x,y)\,{\rm or}\,\left(\frac{1}{2}-x,\frac{1}{2}+y\right)\in \left(\frac{l+1}{2(k+2)},\frac{m}{2k}\right)+\mathbb{Z}^2\,.
\end{equation}
Equation (\ref{string function}) is actually \textit{the} solution to (\ref{string function definition}), when the l.h.s. is given as in (\ref{su2 characters}).

The string functions satisfy a number of properties, that can be proved by looking at (\ref{string function}) and at the summation range (\ref{range for string function}):
\begin{itemize}
\item $C^{(k)}_{l,m}=0$, if $l+m\neq 0$ mod $2$;
\item $C^{(k)}_{l,m}=C^{(k)}_{l,m+2k}$ , i.e. $m$ is defined mod $2k$;
\item $C^{(k)}_{l,m}=C^{(k)}_{l,-m}$;
\item $C^{(k)}_{l,m}=C^{(k)}_{k-l,k+m}$.
\end{itemize}

Using theta function manipulations, the characters of the $N=2$ superconformal algebra can be expressed  in terms of the string functions as \cite{Gepner:1987qi,Blumenhagen:2009zz}
\begin{equation}
\label{MinChar}
\chi_{l,m,s}(\tau,z)=\sum_{j \,{\rm mod}\, k} C^{(k)}_{l,m+4j-s}(\tau)\cdot \Theta_{2m+(4j-s)(k+2),2k(k+2)}(\tau,kz)\,.
\end{equation}
This expression is invariant under any of the transformations $s\rightarrow s+4$ and $m\rightarrow m+2(k+2)$, which shows that $m$ is defined modulo $2(k+2)$ and $s$ modulo $4$. Also, $\chi_{l,m,s}=0$ if $l+m+s\neq 0$ mod 2 and moreover $\chi_{l,m,s}$ is invariant under the simultaneous interchange $l\rightarrow k-l$, $m\rightarrow m+k+2$ and $s\rightarrow s+2$. In the following, we will choose the standard range
\begin{equation}
l=0,\dots,k\,,\qquad m=-k-1,\dots,k+2\,,\qquad s=-1,\dots,2\
\end{equation}
for the labels of the $N=2$ characters. This range would actually produce an overcounting of states, since there is still the identification $\phi_{l,m,s}\sim \phi_{k-l,m+k+2,s+2}$ to take into account. For this purpose, it is more practical to consider the smaller range
\begin{itemize}
\item for k=odd:
\begin{equation}
\{0\le l<\frac{k}{2} \,,\,\forall m,\,\forall s\}
\end{equation}
\item for k=even:
\begin{eqnarray}
&&\{0\le l<\frac{k}{2} \,,\,\forall m,\,\forall s\}
\bigcup
\{l=\frac{k}{2} \,,\,m=1,\dots,k+1,\,\forall s\}\bigcup\\
&&\bigcup
\{l=\frac{k}{2} \,,\,m=0,\,s=0,1\}
\bigcup
\{l=\frac{k}{2} \,,\,m=k+2,\,s=0,1\}\nonumber
\end{eqnarray}
\end{itemize}
which automatically implements the above identification as well as the constraint $l+m+s=0$ mod $2$
\footnote{Observe however that formula (\ref{N=2 weights}) might give a negative weight for a field with labels $(l,m,s)$ in the range above. When this happens, we consider its identified primary with labels $(k-l,m+k+2,s+2)$, which is guaranteed to have positive weight.}. Taking this into account, the number of independent representations is given by
\begin{equation}
\#({\rm fields})=\underbrace{(k+1)}_{{\rm from}\,\,l}\cdot \underbrace{2(k+2)}_{{\rm from}\,\,m}
\cdot \underbrace{4}_{{\rm from}\,\,s}\cdot \underbrace{\frac{1}{2}}_{\rm ident.}
\cdot \underbrace{\frac{1}{2}}_{\rm constr.}=
2(k+1)(k+2)\,,
\end{equation}
while the number of simple currents is 
\begin{equation}
\#({\rm simple\,\,currents})=4\,(k+2)\,,
\end{equation}
in correspondence with all the fields having $l=0$ (as we will see in the next subsection).

To actually compute the minimal model characters using (\ref{MinChar}) is a complicated matter that can only
be done reliably using computer algebra. Results for the ground state dimensions are readily available in the literature, but as we will see, this is not sufficient to determine the conformal weights and ground state dimensions of the permutation orbifold. 

\subsection{Modular transformations and fusion rules}
The coset construction has the additional advantage of making clear what the modular $S$ matrix is for the minimal models. It is just the product of the $S$ matrix of $SU(2)$ at level $k$, the (inverse) $S$ matrix of $U(1)$ at level $2(k+2)$ and the $S$ matrix of $U(1)$ at level $4$:
\begin{eqnarray}
\label{N=2 S matrix}
S^{N=2}_{(l,m,s)(l',m',s')}&=& 
S^{SU(2)_k}_{l,l'} \left(S^{U(1)_{2(k+2)}}\right)_{m,m'}^{-1} S^{U(1)_4}_{s,s'}=\\
&=&\frac{1}{2(k+2)} \sin{\left(\frac{\pi}{k+2}(l+1)(l'+1)\right)}\,
e^{-i\pi\left(\frac{ss'}{2}-\frac{mm'}{k+2}\right)}\,.\nonumber
\end{eqnarray}
The corresponding fusion rules are
\begin{equation}
(l,m,s)\cdot(l',m',s')=\sum_{\lambda,\mu,\sigma}
N^{\lambda}_{\mu,\sigma}\,\delta^{(2(k+2))}_{m+m'-\mu,\,0}\,\delta^{(4)}_{s+s'-\sigma,\,0}\,\,
(\lambda,\mu,\sigma)\,,
\end{equation}
where $N^{\lambda}_{\mu,\sigma}$ are the $SU(2)_k$ fusion coefficients. Here,  $\delta^{(p)}_{x,\,0}$ is equal to $0$, except if $x=0$ mod $p$, in which case it is $1$. Since the $SU(2)_k$ current algebra has only two simple currents, namely the fields with $l=0$ (the identity) and with $l=k$, then all the fields $\phi_{0,m,s}$, and only these, are simple currents (recall the field identification of the $N=2$ minimal models).
In particular, the field $T_F\equiv(0,0,2)$ (with $l=0$) will be relevant in the sequel. It has spin $\frac{3}{2}$ and multiplicity two: it contains the (two) fermionic generators $T^\pm(z)$ of the $N=2$ superconformal algebra. Another field which will be relevant in chapter \ref{paper5} is the so-called spectral-flow operator $S_F\equiv(0,1,1)$, which is also a simple current and has spin $h=\frac{c}{24}$.

\section{Permutation orbifold}
\label{perm orb recap}
Before going into the details of the permutation orbifold of the $N=2$ minimal models, let us recall a few properties of the BHS permutation orbifold \cite{Borisov:1997nc}, restricted to the $\mathbb{Z}_2$ case
\begin{equation}
\mathcal{A}_{\rm perm} \equiv (\mathcal{A}\times \mathcal{A})/\mathbb{Z}_2\,.
\end{equation}
If $c$ is the central charge of $\mathcal{A}$, then the central charge of $\mathcal{A}_{\rm perm}$ is $2c$.
The typical (for exceptions see chapter \ref{paper1}) weights of the fields are:
\begin{itemize}
\item $h_{(i,\xi)}=2 h_i$
\item $h_{\langle i,j\rangle}=h_i+h_j$
\item $h\widehat{(i,\xi)}=\frac{h_i}{2}+\frac{c}{16}+\frac{\xi}{2}$
\end{itemize}
for diagonal, off-diagonal and twisted representations. 
Sometimes it can happen that the naive ground state has dimension zero: then one must go to its first non-vanishing descendant whose weight is incremented by integers. 

For the sake of this chapter, we are mostly interested in the orbifold characters. Let us recall the BHS expressions \cite{Borisov:1997nc} for the diagonal, off-diagonal and twisted $\mathbb{Z}_2$-orbifold characters. 
We denote by $\chi$ the characters of the original (mother) CFT $\mathcal{A}$ and by $X$ the characters of the permutation orbifold $\mathcal{A}_{\rm perm}$:
\begin{subequations}
\label{BHS characters}
\begin{eqnarray}
X_{\langle i,j \rangle}(\tau)&=&\chi_{i}(\tau)\cdot\chi_{j}(\tau) \\
X_{(i,\xi)}(\tau)&=&\frac{1}{2}\chi_{i}^2(\tau)+e^{i\pi\xi}\frac{1}{2}\chi_{i}(2\tau) \\
X_{\widehat{(i,\xi)}}(\tau)&=&\frac{1}{2}\chi_{i}(\frac{\tau}{2})+e^{-i\pi\xi}\,T_i^{-\frac{1}{2}}\,\frac{1}{2}\chi_{i}(\frac{\tau+1}{2})
\end{eqnarray}
\end{subequations}
where $T_i^{-\frac{1}{2}}=e^{-i\pi(h_i-\frac{c}{24})}$. 
Now, each character in the mother theory can be expanded as
\begin{equation}
\chi(\tau)=q^{h_{\chi}-\frac{c}{24}}\,\sum_{n=0}^\infty d_n q^n \qquad \qquad ({\rm with}\,\,q=e^{2i\pi\tau})
\end{equation}
for some non-negative integers $d_n$. Observe that the $d_n$'s can be extracted from
\begin{equation}
\label{q-derivative}
d_n=\frac{1}{n!}\frac{\partial^n}{\partial q^n}\left.\left(\sum_{k=0}^\infty d_k q^k\right)\right|_{q=0}\,.
\end{equation}
Similarly, each character of the permutation orbifold can be expanded as
\begin{equation}
X(\tau)=q^{h_X-\frac{c}{12}}\,\sum_{n=0}^\infty D_n q^n
\end{equation}
for some non-negative integers $D_n$. A relation similar to (\ref{q-derivative}) holds for the $D_n$'s.

Using (\ref{BHS characters}) and (\ref{q-derivative}), we can immediately find the relationships between the $d_n$'s and the $D_n$'s. Here they are:
\begin{subequations}
\label{d-D relations}
\begin{eqnarray}
D^{\langle i,j \rangle}_k &=& \sum_{n=0}^k d_n^{(i)}\,d_{k-n}^{(j)} \\
D^{(i,\xi)}_k &=&  \frac{1}{2} \sum_{n=0}^k d_n^{(i)}\,d_{k-n}^{(i)}+
\left\{
\begin{array}{lr}
0 & {\rm if}\,\,k={\rm odd} \\
\frac{1}{2}\,e^{i\pi\xi}\,d_{\frac{k}{2}}^{(i)} & {\rm if}\,\,k={\rm even}
\end{array}
\right. \\
D^{\widehat{(i,\xi)}}_k &=& d_{2k+\xi}^{(i)}
\end{eqnarray}
\end{subequations}
These expressions are particularly interesting because they tell us that, if we want to have an expansion of the orbifold characters up to order $k$, then it is not enough to expand the original characters up to the same order $k$ (it would be enough for the untwisted fields), but rather we should go up to the higher order $2k+1$, as it is implied by the third line of (\ref{d-D relations}). 
Using the characters (\ref{BHS characters}), one can compute their modular transformation and find the orbifold $S$ matrix, $S^{BHS}$ \cite{Borisov:1997nc}, that we have already given in (\ref{BHS}).

\section{Permutations of $N=2$ minimal models}
\label{perm of N=2 mods}
In this section we consider the permutation orbifold of two $N=2$ minimal models at level $k$. The CFT resulting from modding out the $\mathbb{Z}_2$ symmetry in the tensor product $(N=2)_k\otimes(N=2)_k$ is known from \cite{Klemm:1990df,Borisov:1997nc,Fuchs:1991vu}. Here we focus mostly on the new interesting features arising when one extends the theory with various simple currents. 

As already mentioned, each $N=2$ minimal model at level $k$ admits a supersymmetric current $T_F(z)$ with 
ground state multiplicity equal to two and spin $h=\frac{3}{2}$. In the coset language, it corresponds to the NS field partner of the identity, namely $(l,m,s)=(0,0,2)$. This current transforms each NS field into its NS partner (with different $s$) and each R field into its R conjugate (corresponding to the other value of $s$). In order to see this, note that the $m$ and $s$ indices are just $U(1)$ labels, hence in the fusion of two representations they simply add up: $(s)\times(s')=(s+s'\,{\rm mod}\, 4)$ and $(m)\times(m')=(m+m'\,{\rm mod} \,2(k+2))$.

The field $T_F(z)$ has simple fusion rules with any other field and it generates two integer-spin simple currents in the permutation orbifold, corresponding to the symmetric and anti-symmetric representations $(T_F,0)$ and $(T_F,1)$ of diagonal-type fields, both with spin $h=3$. Both these currents can be used to extend the permutation orbifold. They are both of order two and, interestingly (but not completely surprisingly), their product gives back the anti-symmetric representation of the identity:
\begin{equation}
(T_F,0)\cdot (T_F,1)=(0,1)\,,
\end{equation}
with all the other possible products obtained from this one by using cyclicity of the order two. In other words, the fields $(0,0), (T_F,0), (0,1), (T_F,1)$ form a $\mathbb{Z}_4$ group under fusion.

We will study the extensions in the next two subsections, where we will also see the new CFT structure coming from interchanging extensions and orbifolds. Before we do this, however, let us first mention some generic properties of the orbifold. 
Consider the permutation orbifold of two $N=2$ minimal models at level $k$ and extend it by either the symmetric or the anti-symmetric representation of $T_F(z)$. The resulting theory has the old standard simple currents coming from $\phi_{0,m,s}$ (or equivalently $\phi_{k,m+k+2,s+2}$, by the identification) in the mother theory (in number equal to the number of simple currents of the $(N=2)_k$ minimal model and corresponding to the \textit{orbits} of their diagonal representations according to the fusion rules given in the next two subsections) and an equal number of \textit{exceptional} simple currents that were not simple currents before the extension (since coming from fixed off-diagonal orbits of $\phi_{0,m,s}$, as we will see below). 

The structure of the exceptional simple current is very generic: it is the same for both $(T_F,0)$ and $(T_F,1)$, so we can consider both here. The word exceptional means that they are simple currents just because their extended $S$ matrix satisfies the relation $S_{0J}=S_{00}$ \cite{Dijkgraaf:1988tf}. 
First of all, note that the orbifold simple currents come from symmetric and anti-symmetric representations of the mother simple currents, hence there are as many as twice the number of simple currents of the mother minimal theory. 
Secondly, all the exceptional currents correspond to the label $l=0$ (or equivalently $l=k$) as it should be, since related to the $SU(2)_k$ algebra. This has the following consequence. Recall the orbifold (BHS) $S$ matrix in the untwisted sector:
\begin{eqnarray}
S^{BHS}_{(i,\psi)(j,\chi)}&=& \frac{1}{2}\,S_{ij}\,S_{ij}\nonumber\\
S^{BHS}_{(i,\psi)\langle m,n \rangle}&=& S_{im}\,S_{in}\nonumber
\end{eqnarray}
Using the minimal-model $S$ matrix (\ref{N=2 S matrix}) one has:
\begin{equation}
S_{(0,0,0)(0,0,0)}=
\frac{1}{2(k+2)}\,\sin{\left(\frac{\pi}{k+2}\right)}
=S_{(0,0,0)(0,m,s)}\nonumber
\end{equation}
and hence
\begin{equation}
\label{2S=S in BHS}
S^{BHS}_{((0,0,0),0),\langle(0,m,s),(0,m,s+2)\rangle}=2\,S^{BHS}_{((0,0,0),0),((0,0,0),0)}\,.
\end{equation}
This equality will soon be useful. In particular, the factor $2$ will disappear in the extension, promoting the off-diagonal fields $\langle(0,m,s),(0,m,s+2)\rangle$ into simple currents. We will come back later to these exceptional currents.

Let us show now that these exceptional simple currents of the $(T_F,\psi)$-extended orbifold correspond exactly to those particular \textit{off-diagonal fixed points} whose $(T_F,\psi)$-orbits ($\psi=0,\,1$) are generated from the simple currents of the mother $N=2$ minimal model.\\
Consider off-diagonal fields of the form $\langle(0,m,s),(0,m,s+2)\rangle$. They are fixed points\footnote{This is proved in the next subsections.} of $(T_F,\psi)$, since $T_F \cdot (0,m,s)=(0,m,s+2)$. The number of such orbits is equal to half the number of simple currents in the original minimal model (i.e. those fields with $l=0$). In the extension, they must be resolved. This means that each of them will give rise to two ``split'' fields in the extension. Hence their number gets doubled and one ends up with a number of split fields again equal to the number of simple currents of the original minimal model. Moreover, the extended $S$ matrix, $\tilde{S}$, will be expressed in terms of the $S^{J}$ matrix corresponding to $J\equiv (T_F,\psi)$, according to
\begin{equation}
\tilde{S}_{(a,\alpha)(b,\beta)}=
C\cdot[S^{BHS}_{ab}+(-1)^{\alpha+\beta}\,S^{(T_F,\psi)}_{ab}]\,.
\end{equation}
Recall that the $S^J$ matrix is non-zero only if the entries $a$ and $b$ are fixed points. The labels $\alpha$ and $\beta$ keep track of the two split fields ($\alpha,\,\beta=0\,,1$). The factor $C$ in front is a group theoretical quantity, that in case $a$ and $b$ are both fixed, is equal to $\frac{1}{2}$.

The generic formula for $S^J$ as given in \cite{Maio:2009tg} was recalled in (\ref{ansatz with fixed points}). In particular, the untwisted (i.e. diagonal and off-diagonal) entries of $S^J$ vanish, since $T_F$ does not have fixed points:
\begin{eqnarray}
S^{(T_F,\psi)}_{\langle m,n \rangle\langle p,q \rangle} &=& S^{T_F}_{mp}\,S^{T_F}_{nq} +(-1)^\psi S^{T_F}_{mq}\,S^{T_F}_{np} \equiv 0\nonumber\\
S^{(T_F,\psi)}_{(i,\phi)(j,\chi)} &=& \frac{1}{2}\,S^{T_F}_{ij}\,S^{T_F}_{ij}\equiv 0\nonumber\\
S^{(T_F,\psi)}_{(i,\phi)\langle m,n \rangle} &=& S^{T_F}_{im}\,S^{T_F}_{in}\equiv 0\nonumber\,.
\end{eqnarray}
This implies that
\begin{equation}
\tilde{S}_{(a,\alpha)(b,\beta)}=
C\cdot S^{BHS}_{ab}
\end{equation}
for each split field corresponding to untwisted fixed points $a,\,b$. If either $a$ or $b$ are not fixed points, then $S^{(T_F,\psi)}$ is  automatically zero and the $\tilde{S}$ is given directly by $S^{BHS}$, up to the overall group theoretical factor $C$ in front, which is equal to $2$ if both $a$ and $b$ are not fixed points and $1$ if only one entry is fixed. Using (\ref{2S=S in BHS}), this implies that after fixed point resolution one would have
\begin{equation}
\tilde{S}_{((0,0,0),0)\langle(0,m,s),(0,m,s+2)\rangle_\alpha} = \tilde{S}_{((0,0,0),0)((0,0,0),0)}
\qquad (\alpha=0,\,1)\,.
\end{equation}
This means that
\begin{equation}
\label{exceptional simple currents}
\langle(0,m,s),(0,m,s+2)\rangle_\alpha\qquad \alpha=0,\,1
\end{equation}
are the exceptional simple currents in the extended theory, being $((0,0,0),0)$ the identity of the permutation orbifold and $(0,m,s)$ simple currents in the mother theory. The label $m$ runs over all the possible values, $m\in [-k-1,k+2]$; the label $s$ is fixed by the constraint $l+m+s=0$ mod $2$. This is the origin of the exceptional currents in the extended permutation orbifold of two $N=2$ minimal models. Note that, since in the off-diagonal currents both fields appear with $s$ and $s+2$, we can fix once and for all the $s$-labels in the exceptional currents to be $s=0$ in the NS sector and $s=-1$ in the R sector. 

These exceptional simple currents may in principle have fixed points. However, it turns out to be not the case in general: in fact, we will see that only four of the several exceptional currents have fixed points and only if $k=2$ mod $4$. We will come back to this later.

\subsection{Extension by $(T_F,1)$}
Let us start by studying how the current under consideration, $(T_F,1)$, acts on different fields in the orbifold.
By looking at some specific examples or by computing the fusion rules, one can show that the orbits are given as in the following list. We denote the $N=2$ minimal representations as $i\equiv (l,m,s)$ and the ``shifted'' representations as $T_F\cdot i\equiv (l,m,s+2)$.
\begin{itemize}
\item Diagonal fields $(i,\xi)$ (recall that $\xi$ is defined mod $2$)
\begin{equation}
(T_F,1)\cdot (i,\xi) = (T_F\cdot i,\xi+1)
\end{equation}
\item Off-diagonal fields $\langle i,j \rangle$
\begin{equation}
(T_F,1)\cdot \langle i,j \rangle = \langle T_F\cdot i,T_F\cdot j \rangle
\end{equation}
\item Twisted fields $\widehat{(i,\xi)}$ (recall that $\xi$ is defined mod $2$)
\begin{eqnarray}
(T_F,1)\cdot \widehat{(i,\xi)} &=& \widehat{(i,\xi+1)}
\qquad {\rm if}\,\, i\,\,{\rm is}\,\, NS\,\,(s=0,\,2) \nonumber\\
&&\\
(T_F,1)\cdot \widehat{(i,\xi)} &=& \widehat{(i,\xi)}
\qquad {\rm if}\,\, i\,\,{\rm is}\,\, R\,\,(s=-1,\,1) \nonumber
\end{eqnarray}
\end{itemize}
A comment about possible fixed points is in order, since they get split in the extension and need to be resolved.
Observe that there cannot be any fixed points from the diagonal representations, 
since $T_F$ does not leave anything fixed. They will become all orbits and will all be kept in the extension, since they have integer monodromy:
\begin{equation}
Q_{(T_F,1)}(i,\xi)=2h_{T_F}+2h_i-2\left(h_i+\frac{1}{2}\right)\quad\in\mathbb{Z}\,.\nonumber
\end{equation}
The number of such orbits is equal to the number of fields in the mother minimal model.\\
On the other hand, there are in general fixed points for off-diagonal and twisted representations.
The off-diagonal fixed points arise when $j=T_F\cdot i$, i.e. in our notation when $\langle i,j \rangle$ is of the form $\langle(l,m,s),(l,m,s+2)\rangle$. When $l=0$, after splitting, these will be the exceptional simple currents in the extended theory. 
The remaining off-diagonal fields organize themselves into orbits, of which some are kept and some are projected out, depending on their monodromy. In particular, using
\begin{equation}
Q_{(T_F,1)}\langle i,j\rangle=2h_{T_F}+(h_i+h_j)-(h_{T_F\cdot i}+h_{T_F\cdot j})\quad{\rm mod}\,\,\mathbb{Z}\,,\nonumber
\end{equation}
and the fact that, from the term $\frac{s^2}{8}$ in (\ref{N=2 weights}), $h_i-h_{T_F\cdot i}$ is $\frac{1}{2}$ if $i$ is NS and $0$ if $i$ is R, we see that the orbit $(\langle i,j\rangle, \langle T_F\cdot i,T_F\cdot j\rangle)$ is kept only if $i$ and $j$ are both NS or both R, otherwise they are projected out.\\
The twisted fixed points come from all the R representations and are kept in the extension,
while the twisted fields coming from NS representations are not fixed and projected out in the extension, since their monodromy charge 
\begin{equation}
Q_{(T_F,1)}\widehat{(i,\xi)}=2h_{T_F}+h_{\widehat{(i,\xi)}}-h_{\widehat{(i,\xi+1)}}\quad{\rm mod}\,\,\mathbb{Z}\nonumber
\end{equation}
is half-integer, being $(T_F,1)$ of integer spin and the difference of weights between a ($\psi=0$)-twisted field and the corresponding ($\psi=1$)-twisted field equal to $\frac{1}{2}$.

For $k=2$ mod $4$ some of the exceptional currents in the extension have fixed points, either of the off-diagonal or twisted type, none of diagonal kind. They are specific $(T_F,1)$-orbits of off-diagonal fields plus all the twisted $(T_F,1)$-fixed points (necessarily in the Ramond sector of the original minimal model). At the moment of writing this work, we are not able to resolve them: in other words, their $S^J$ matrices are unknown, $J$ denoting any of the exceptional currents.

One important exceptional current of the permutation orbifold is the \textit{world-sheet supersymmetry} current, which is the only current of order two and spin $h=\frac{3}{2}$: it is the off-diagonal field coming from the tensor product of the identity with $T_F(z)$. It does not have fixed points, because $T_F$ does not. Let us denote it by $J^{w.s.}_{orb}\equiv \langle 0, T_F \rangle$. By the argument given above, $J^{w.s.}_{\rm orb}$ is guaranteed to be fixed by $(T_F,1)$. This means that in the extension it gets split into two fields, that we denote by $\langle 0, T_F \rangle_\alpha$, with $\alpha=0$ or $1$. In appendix \ref{Appendix Paper4} we check that indeed $\langle 0, T_F \rangle_\alpha$ has order two:
\begin{equation}
\langle 0, T_F \rangle_\alpha \cdot \langle 0, T_F \rangle_\alpha = (0,0)\,,
\end{equation}
where $(0,0)$ is the identity orbit.

Now consider the tensor product of two minimal models. We can either extend by $T_F(z)\otimes T_F(z)$ to make the product supersymmetric or we can mod out the $\mathbb{Z}_2$ symmetry and end up with the permutation orbifold. Let us start with the latter option. We know from the first part of this work that one can go back to the tensor product by extending the orbifold by the anti-symmetric representation of the identity, $(0,1)$. What we do instead is extending the orbifold by $(T_F,1)$. The resulting theory is the $N=2$ supersymmetric permutation orbifold which has the world-sheet spin-$\frac{3}{2}$ current in its spectrum.

Alternatively, we can change the order and perform the extension before orbifolding. Note that each $N=2$ factor is supersymmetric, but the product is not. In order to make it supersymmetric, we have to extend it by the tensor-product current $T_F(z)\otimes T_F(z)$. As a result, in the tensor product only those fields survive whose two factors are either both in the NS or both in the R sector. In this way, the fields in the product have factors that are aligned to be in the same sector. Now we still have to take the $\mathbb{Z}_2$ orbifold. Starting from the supersymmetric product, by definition, we look for $\mathbb{Z}_2$-invariant states/combinations and add the proper twisted sector.
We will refer to this mechanism which transforms the supersymmetric tensor product into the supersymmetric orbifold as \textit{super-BHS}, in analogy with the standard BHS from the tensor product to the orbifold. 
The following scheme summarizes this structure:
\begin{displaymath}
\xymatrix{
\boxed{(N=2)^2} \ar@/^/[d]^{BHS} \ar[r]^{T_F\otimes T_F} & \boxed{(N=2)^2_{\rm Susy}} \ar@/^/@2{->}[d]^{\rm super-BHS} \\
\boxed{(N=2)^2_{\rm orb}} \ar@/^/@{.>}[u]^{(0,1)} \ar[r]^{(T_F,1)\quad} & \boxed{(N=2)^2_{\rm Susy-orb}} \ar@/^/@{.>}[u]^{(0,1)}
}
\end{displaymath}

As a check, let us consider the following example. Take the case of level $k=1$. The $(N=2)_{1}$ minimal model has central charge equal to one and twelve primary fields (all simple currents). Its tensor product has central charge equal to two, as well as its $T_F\otimes T_F$-extension and $\mathbb{Z}_2$-orbifold. \\
By extending the tensor product by the current $T_F\otimes T_F$, one obtains the supersymmetric tensor product, with 36 fields. Instead, by going to the orbifold and extending by the current $(T_F,1)$, one obtains the supersymmetric orbifold with 60 fields. As a side remark, there is only one theory with this exact numbers of fields and same central charge and that is in addition supersymmetric, but only by working out the spectrum one can prove without any doubt that the theory in question is the $(N=2)_4$ minimal model, which is indeed supersymmetric.\\
We can continue now and extend the supersymmetric orbifold by the current $(0,1)$. This operation is the inverse of the $\mathbb{Z}_2$-orbifold (super-BHS). As expected, we end up to the supersymmetric tensor product. Equivalently, the $\mathbb{Z}_2$-orbifold of the supersymmetric tensor product gives back the supersymmetric orbifold, consistently.

\subsection{Extension by $(T_F,0)$}
Many things here are similar to the previous case. Let us start by giving the fusion rules of the current $(T_F,0)$ with any other field in the permutation orbifold.
\begin{itemize}
\item Diagonal fields $(i,\xi)$ (recall that $\xi$ is defined mod $2$)
\begin{equation}
(T_F,0)\cdot (i,\xi) = (T_F\cdot i,\xi)
\end{equation}
\item Off-diagonal fields $\langle i,j \rangle$
\begin{equation}
(T_F,0)\cdot \langle i,j \rangle = \langle T_F\cdot i,T_F\cdot j \rangle
\end{equation}
\item Twisted fields $\widehat{(i,\xi)}$ (recall that $\xi$ is defined mod $2$)
\begin{eqnarray}
(T_F,0)\cdot \widehat{(i,\xi)} &=& \widehat{(i,\xi)}
\qquad {\rm if}\,\, i\,\,{\rm is}\,\, NS\,\,(s=0,\,2) \nonumber\\
&&\\
(T_F,0)\cdot \widehat{(i,\xi)} &=& \widehat{(i,\xi+1)}
\qquad {\rm if}\,\, i\,\,{\rm is}\,\, R\,\,(s=-1,\,1) \nonumber
\end{eqnarray}
\end{itemize}
Again, the current $(T_F,0)$ does not have diagonal fixed points, but does have off-diagonal 
and twisted fixed points. The off-diagonal ones are like before, while the twisted ones come
this time from NS fields. Twisted fields coming from R representations are projected out in 
the extension. Each fixed point is split in two in the extended permutation orbifold and must be resolved. 
Moreover, there will also be orbits coming from the diagonal and off-diagonal fields.

Also for $(T_F,0)$-extensions, a few exceptional currents might have fixed points.  They are either off-diagonal $(T_F,0)$-orbits or all the twisted $(T_F,0)$-fixed points (necessarily of Neveu-Schwarz origin).

As before, consider now the tensor product of two minimal models and its permutation orbifold. Extend the orbifold with the current $(T_F,0)$, i.e. the symmetric representation $T_F(z)$. One obtains a new, for the moment mysterious, CFT that we denote by $X$. $X$ is not supersymmetric, since it does not contain the world-sheet supercurrent of spin $h=\frac{3}{2}$. To be more precise, $X$ does contain a spin $\frac{3}{2}$-current, which is again the off-diagonal field $\langle 0, T_F \rangle$. However, it is \textit{not} the world-sheet supersymmetry current. The reason is that in this case $\langle 0, T_F \rangle$ (or rather the two split fields $\langle 0, T_F \rangle_\alpha$, with $\alpha=0$ or $1$) has order $4$, instead of order $2$: acting twice with $J^{w.s.}_{\rm orb}(z)$ we should get back to the same field, but we do not. As we prove in appendix \ref{Appendix Paper4}:
\begin{equation}
\langle 0, T_F \rangle_\alpha \cdot \langle 0, T_F \rangle_\alpha = (0,1)\,,
\end{equation}
with $(0,1)\cdot(0,1)=(0,0)$. Hence there is no such a current as $J^{w.s.}_{\rm orb}(z)$ in $X$. Continuing extending this time by the current $(0,1)$ we get back to the familiar theory $(N=2)^2_{\rm Susy}$. 
The summarizing graph is below:
\begin{displaymath}
\xymatrix{
\boxed{(N=2)^2} \ar@/^/[d]^{BHS} \ar[r]^{T_F\otimes T_F} & \boxed{(N=2)^2_{\rm Susy}}  \\
\boxed{(N=2)^2_{\rm orb}}  \ar@/^/@{.>}[u]^{(0,1)} \ar[r]^{(T_F,0)\quad} & \boxed{{\rm Non-Susy}\,\,X} \ar[u]^{(0,1)}
}
\end{displaymath}

\subsection{Common properties}
By looking at the two graphs, we notice that there are two distinct ways of reproducing the behavior of the current $T_F\otimes T_F$ which makes the tensor product of two minimal models supersymmetric. We can go either through the supersymmetric permutation orbifold or through the non-supersymmetric CFT $X$, as shown below.
\begin{displaymath}
\xymatrix{
& (N=2)^2 \ar@{.>}[ddd]_{T_F}^{T_F}|\otimes \ar@/^/[d]^{BHS} & \\
& \ar@/^/[u]^{(0,1)} \ar[dl]_{(T_F,0)} (N=2)^2_{\rm orb} \ar[dr]^{(T_F,1)} &\\
{\rm Non-Susy}\,\,X \ar[dr]_{(0,1)} & & (N=2)^2_{\rm Susy-orb} \ar[dl]^{(0,1)} \\
& (N=2)^2_{\rm Susy} & \\
}
\end{displaymath}
We can summarize the commutativity of this diagram as:
\begin{equation}
(T_F\otimes T_F) \circ (0,1) = (0,1) \circ (T_F,\psi)
\end{equation}
when acting on $(N=2)^2_{\rm orb}$. The small circle $\circ$ means composition of extensions, e.g. $(J_2 \circ J_1) \mathcal{A}$ means that we start with the CFT $\mathcal{A}$, then we extend it by the simple current $J_1$ and finally we extend it again by the simple current $J_2$.

It is useful to ask what happens to the exceptional current $\langle 0, T_F \rangle$ (which coincides with $J^{w.s.}_{\rm orb}(z)$ for the $(T_F,1)$-extension). Using the fusion rules given earlier, it is easy to see that $\langle 0, T_F \rangle$ is fixed by both $(T_F,0)$ and $(T_F,1)$, because of the shift by $T_F$ in both the factors in off-diagonal fields and the symmetrization of the tensor product. As a consequence, the fixed point resolution is needed in both situations for the field $\langle 0, T_F \rangle$. 

Let us make a comment on the nature of the CFT $X$. We have already stressed enough that it is not supersymmetric. However, by looking at it more closely, it is quite similar to the supersymmetric orbifold $(N=2)^2_{\rm Susy-orb}$. For example, they contain the same number of fields and in particular they have the same diagonal and off-diagonal fields. They only differ for their twisted fields, being of R type in the supersymmetric orbifold and of NS type in $X$.

Another interesting point is that the $(0,1)$ extension of both $X$ and $(N=2)^2_{\rm Susy-orb}$ gives back the same answer, namely the $(N=2)^2_{\rm Susy}$. One could ask how this happens in detail. The reason is that, after the $(T_F,\psi)$-extension (either $\psi=0$ or $1$) of the orbifold, one is left with orbits and/or fixed points corresponding to orbifold fields of diagonal, off-diagonal and twisted type. In particular, as we already mentioned before, from the twisted fields only the fixed points survive, with the difference that for $\psi=1$ they come from the Ramond sector and for $\psi=0$ from the NS sector. However, they are completely projected out by the $(0,1)$-extension, which leaves only untwisted (i.e. off-diagonal and diagonal -both symmetric and anti-symmetric-) fields in the supersymmetric tensor product\footnote{The reason is that the current $(0,1)$ always couples a twisted field $\widehat{(p,0)}$ to its partner $\widehat{(p,1)}$, as it is shown in appendix \ref{Appendix Paper4}. Since these fields have weights which differ by $\frac{1}{2}$, then their monodromy will be half-integer and they will be projected out in the $(0,1)$-extension.}.

\section{Exceptional simple currents and fixed points}
\label{except curr and fp}
Let us be a bit more precise on the exceptional simple currents which admit fixed points. There are four of them and they are always related to the following mother-theory simple currents 
\begin{equation}
J_+\equiv(l,m,s)\equiv(0,\frac{k+2}{2},s)\equiv(k,-\frac{k+2}{2},s+2)
\end{equation}
and
\begin{equation}
J_-\equiv (0,-\frac{k+2}{2},s)\equiv(k,\frac{k+2}{2},s+2)
\end{equation}
(with $s=0$ in the NS sector, $s=-1$ in the R sector). We will soon prove that $s$ must be in the NS sector. i.e. $s=0$, otherwise there are no fixed points. Using the facts that $m$ is defined mod $2(k+2)$ and that $s$ is defined mod $4$, together with the identification $(l,m,s)=(k-l,m+k+2,s+2)$, it is easy to show that $J_+$ and $J_-$ are of order four, i.e. $J_+^4=J_-^4=1$. 
Moreover, we will soon show that off-diagonal fixed points of the exceptional currents originate from fields in the mother $N=2$ theory with $l$-label equal to $l=\frac{k}{2}$. One can easily check that, on these fields, the square of $J_\pm$, $J^2_\pm$, acts as follows. 
For $J_\pm$ in the R sector, $J_\pm^2$ fixes any other field (either R or NS) of the original minimal model:
\begin{equation}
(J_\pm\in R)\quad J_\pm^2:\,\, (l=\frac{k}{2},m,s) \longrightarrow (l=\frac{k}{2},m,s)
\Longrightarrow J_\pm^2\simeq 0\equiv (0,0,0)\,,
\end{equation}
acting on them effectively as the identity; for $J_\pm$ in the NS sector, $J_\pm^2$ takes an R (NS) field into its conjugate R (NS) field:
\begin{equation}
(J_\pm\in NS)\quad J_\pm^2:\,\, (l=\frac{k}{2},m,s) \longrightarrow (l=\frac{k}{2},m,s+2)
\Longrightarrow J_\pm^2\simeq T_F\equiv (0,0,2)\,,
\end{equation}
acting effectively as the supersymmetry current.

Having introduced now the currents $J_\pm$ in the mother theory, we can write down the four simple currents in the orbifold theory extended by $(T_F,\psi)$ which admit fixed points. Recalling that $T_F=(0,0,2)$ acts by shifting by two the $s$-labels in the original minimal model, we can consider the following off-diagonal fields in the permutation orbifold:
\begin{equation}
\label{formual for simple currents in orb}
\langle J_\pm,T_F\cdot J_\pm \rangle\,.
\end{equation}
The two off-diagonal combinations above satisfy the condition (\ref{2S=S in BHS}); hence, after fixed point resolution, each of them generates two exceptional simple currents (for a total of four) in the $(T_F,\psi)$-extended theory:
\begin{equation}
\label{formual for exceptional simple currents}
\langle J_\pm,T_F\cdot J_\pm \rangle_\alpha\,\,,\qquad \alpha=0,\,1\,,
\end{equation}
being $T_F\cdot J_\pm =(0,\pm \frac{k+2}{2},s+2)$. This is another way of re-writing (\ref{exceptional simple currents}), specialized to the exceptional currents that have fixed points.

If one wants to be very precise about the fixed points, one should study the fusion coefficients, which is in the present case very complicated, but in principle doable. However, we can still make some preliminary progress using intuitive arguments. First of all, since the resolved currents (\ref{exceptional simple currents}) carry an index $\alpha$ which distinguishes them, but are very similar otherwise, it is reasonable to expect that they might have the same fixed points and that hence the fixed-point conformal field theories corresponding to the exceptional currents might be pairwise identical.\\
Secondly, observe that in (\ref{exceptional simple currents}) the field $(0,m,s)$ is equivalent to $(k,m+k+2 \,\,{\rm mod}\,\,2(k+2),s+2\,\, {\rm mod}\,\, 4)$. From the $SU(2)_k$ algebra, the field labelled by $l=k$ is the only non-trivial simple current with fusion rules given by
\begin{equation}
(k)\cdot(j)=(k-j)\,,
\end{equation}
so in order for it to have fixed points, $k$ must be at least even. Moreover, $j$ is a fixed point of the $SU(2)_k$ algebra if and only if $j=\frac{k}{2}$. This argument tells us that off-diagonal fixed points of (\ref{exceptional simple currents}) must be orbits whose component fields have $l$-label equal to $l=\frac{k}{2}$. 

Actually there are only four (coming from the above two resolved) exceptional simple currents which have fixed points and the corresponding four fixed-point conformal field theories are pairwise identical. Indeed, the exceptional simple currents have $m$-label equal to $m=\pm\frac{k+2}{2}$, even $s$-label and hence the generic constraints $l+m+s=0 \,\,{\rm mod} \,\,2$ implies that $k=2$ mod $4$. 

Let us describe more in detail the exceptional simple currents with fixed points. Consider again (\ref{formual for exceptional simple currents}) and study the fusion rules of (\ref{formual for simple currents in orb}). We are most interested in off-diagonal fixed points, because they have an interesting structure; as far as the other kind (namely twisted) of fixed points is concerned, they are as already reported in the previous section (namely of NS type for $(T_F,0)$ and of R type for $(T_F,1)$). Compute the fusion rule of the current $(J_\pm,T_F J_\pm)$ with any field of the form:
\begin{equation}
\label{offdiag f.p. of except curr}
\langle f,J_\pm f'\rangle\,,
\end{equation}
where $f'$ has either the same $s$-label as $f$ or different; in other words, either $f'=f$ or $f'=T_F\cdot f$. Here, $f$ and $f'$ label primaries of the original $N=2$ minimal model which might be fixed points of (\ref{formual for exceptional simple currents}), having their $l$-values equal to $l=\frac{k}{2}$. Explicitly, $f=(\frac{k}{2},m,s)$ and $f'=(\frac{k}{2},m,s')$, with $s'=s$ or $s'=s+2$.

We would like to show that the fields $\langle f,J_\pm f'\rangle$ constitute the subset of off-diagonal fixed points for the exceptional currents. For most of them, this subset will be empty, but not for (\ref{formual for exceptional simple currents}). As a remark, note that not all the fields in (\ref{offdiag f.p. of except curr}) are independent, since they are identified pairwise by the extension. We will come back to this at the end of this subsection.

Now let us compute the fusion rules. Naively:
\begin{eqnarray}
\langle J_\pm,T_F J_\pm\rangle \cdot \langle f,J_\pm f'\rangle&\propto&
(J_\pm\otimes T_F J_\pm+T_F J_\pm\otimes J_\pm)\cdot
(f\otimes J_\pm f'+J_\pm f'\otimes f)\nonumber\\
&=&
(J_\pm f\otimes T_F f'+J_\pm^2 f'\otimes T_F J_\pm f +\nonumber\\
&&
\qquad\qquad\qquad +T_F J_\pm f\otimes J_\pm^2 f'+T_F J_\pm^2 f'\otimes J_\pm f)\,.\nonumber
\end{eqnarray}
For currents in the R sector, $J_\pm^2=1$, while $J_\pm^2=T_F$ in the NS sector; hence the above expression simplifies in both cases:
\begin{equation}
\langle J_\pm,T_F J_\pm\rangle\cdot\langle f,J_\pm f'\rangle \propto \dots=
\left\{
\begin{array}{cc}
(J_\pm f\otimes T_F f'+ f'\otimes T_F J_\pm f +& {\rm R\,\,sector}\\
\qquad+T_F J_\pm f\otimes f'+T_F f'\otimes J_\pm f)\,.\nonumber & \\
&\\
(J_\pm f\otimes f'+ T_F f'\otimes T_F J_\pm f +& {\rm NS\,\,sector}\\
\qquad+T_F J_\pm f\otimes T_F f'+f'\otimes J_\pm f) &
\end{array}
\right.
\nonumber
\end{equation}
In terms of representations, we can decompose the r.h.s. in two pieces corresponding to the following symmetric representations:
\begin{eqnarray}
\label{temp fp of excep curr}
({\rm R})\quad \langle J_\pm,T_F J_\pm\rangle \cdot \langle f,J_\pm f'\rangle&=& 
\langle J_\pm f,T_F f'\rangle+\langle f',T_F J_\pm f\rangle\nonumber\\
({\rm NS})\quad \langle J_\pm,T_F J_\pm\rangle \cdot\langle f,J_\pm f'\rangle&=& 
\langle f',J_\pm f\rangle+\langle T_F f',T_F J_\pm f\rangle 
\end{eqnarray}
We have replaced here the proportionality symbol with an equality: a more accurate calculation of the fusion coefficients would show that the proportionality constant is indeed one. 
It is crucial that none of the two pieces in the first line (R sector) reduces to $\langle f,J_\pm f'\rangle$ as on the l.h.s.; on the contrary, either of them does, respectively if $f=f'$ and $f'=T_F\cdot f$, in the second line (NS sector). For example, in the NS situation, this is obvious in the case $f=f'$; if $f'=T_F\cdot f$ instead, we  must remember that the  brackets means symmetrization and that off-diagonal fields that are equal up to the action of $(T_F,\psi)$ are actually \textit{identified} by the extension. Similar arguments hold for the R situation as well.

Note here that the two pieces in (\ref{temp fp of excep curr}) are related by the application of $T_F$: if we talked about tensor product fields then the relation would be given by the tensor product $T_F\otimes T_F$, but since we are working in the orbifold, it is actually provided by the diagonal representation $(T_F,\psi)$. 

Let us move now to the extended orbifold. 
From the fusion rules given earlier, in the permutation orbifold extended by $(T_F,\psi)$, off-diagonal fields belong to the same orbit if and only if
\begin{equation}
(T_F,\psi)\cdot\langle i,j \rangle=\langle T_F\cdot i,T_F\cdot j\rangle\,.
\end{equation}
Since
\begin{equation}
(T_F,\psi)\cdot \langle f,J_\pm f'\rangle=\langle T_F f,T_F J_\pm f'\rangle\,,
\end{equation}
then the two quantities appearing on the r.h.s. of (\ref{temp fp of excep curr}) are identified by the extension and add up to give
\begin{eqnarray}
\label{fp of excep curr}
({\rm R})\qquad \langle J_\pm,T_F J_\pm\rangle\cdot\langle f,J_\pm f'\rangle&=& \langle J_\pm f,T_F f'\rangle\,,\nonumber\\
({\rm NS})\qquad \langle J_\pm,T_F J_\pm\rangle\cdot\langle f,J_\pm f'\rangle &=& \langle f',J_\pm f\rangle\,.
\end{eqnarray}
As a consequence, exceptional currents coming from R fields never have fixed points (neither if $f=f'$ nor if $f'=T_F\cdot f$), while NS fields do have. 
This shows that the exceptional simple currents with fixed points arise only for NS fields in the mother theory and they are exactly of the desired form.

As a consistency check, let us give the following argument about the currents (\ref{formual for exceptional simple currents}) (equivalently, identify $l\rightarrow k-l\,,\dots$ etc). We have already established that $k$ must be even in order for the currents to have fixed points, so we can discuss the two options of $k=4p$ and $k=2+4p$ (for $p\in\mathbb{Z}$) separately. In the former case, $k=4p$,
\begin{equation}
h_{\langle J_\pm,T_F\cdot J_\pm\rangle_\alpha}=h_{J_\pm}+h_{T_F\cdot J_\pm}=2\cdot\frac{3k}{16}=\frac{3p}{2}\,.
\end{equation}
This is either integer or half-integer, depending on $p$, so the currents might admit fixed points. However, the current $m$-label is equal to $2p+1\in\mathbb{Z}_{\rm odd}$; since the $l$-label is even, then the $N=2$ constraint forces the $s$-label to be $\pm 1$. As a consequence, the currents (\ref{formual for exceptional simple currents}) are of Ramond-type and hence cannot have fixed points.
In the latter case, $k=2+4p$,
\begin{equation}
h_{\langle J_\pm,T_F\cdot J_\pm\rangle_\alpha}=h_{J_\pm}+h_{T_F\cdot J_\pm}=
\left(\frac{3k}{16}-\frac{1}{8}\right)+\left(\frac{3k}{16}+\frac{3}{8}\right)=1+\frac{3p}{2}\,.
\end{equation}
This is either integer or half-integer, depending on $p$, so the current can have fixed points. Moreover, since the $m$-label of the exceptional current is now equal to $2p+2\in\mathbb{Z}_{\rm even}$, the currents (\ref{formual for exceptional simple currents}) are now of NS-type, hence they \textit{will have} fixed points.

Needless to say, we do expect all a priori possible fields of the form (\ref{offdiag f.p. of except curr}) to survive the $(T_F,\psi)$-extension, the reason being that their $(T_F,\psi)$-orbits must have zero monodromy charge with respect to the current $(T_F,\psi)$. As an exercise, let us compute this charge and prove that it vanishes (mod integer). For this purpose, we need to know the weight of (\ref{offdiag f.p. of except curr}). Since
\begin{equation}
h_{J_\pm f}=h_f-\frac{1}{16}(k+2\pm 4m)
\end{equation}
$m$ being the $m$-label of the field $f$, then 
\begin{equation}
h_{\langle f,J_\pm f'\rangle}=h_f+h_{J_\pm f'}=2h_f-\frac{1}{8}(k+2\pm 4m)+\frac{1}{2}\,\delta_{f',T_F f}\,.
\end{equation}
Similarly, we need to compute $h_{\langle T_F f,T_F J_\pm f'\rangle}$. Since 
\begin{equation}
h_{T_F J_\pm f}=h_{T_F f}-\frac{1}{16}(k+2\pm 4m)
\end{equation}
then again 
\begin{equation}
h_{\langle T_F f,T_F J_\pm f'\rangle}=h_{T_F f}+h_{T_F J_\pm f'}=2h_{T_F f}-\frac{1}{8}(k+2\pm 4m)+\frac{1}{2}\,\delta_{f',T_F f}\,.
\end{equation}
Hence:
\begin{equation}
Q_{(T_F,\psi)}\big(\langle f,J_\pm f'\rangle\big)=h_{(T_F,\psi)}+h_{\langle f,J_\pm f'\rangle}- h_{\langle T_F f,T_F J_\pm f'\rangle}=0\quad({\rm mod}\,\,\mathbb{Z})\,,
\end{equation}
i.e. these fields are kept in the extension and  organize themselves into orbits. Still, some fields seem not to appear among the off-diagonal fields that we would expect. The solutions to this problem is provided by the extension: \textit{fields are pairwise identified}. In fact, as a consequence of (\ref{temp fp of excep curr}), two fields related by the action of (\ref{formual for simple currents in orb}) are mapped into each other by $(T_F,\psi)$ and hence are identified by the currents (\ref{formual for exceptional simple currents}) in the extension.

What happens in determining the fixed points of the exceptional currents is the following. Start with a field $f$ which has $l$-label equal to $\frac{k}{2}$ and apply $J_\pm$ on $f$, recalling that $J^4_\pm=1$ and $J^2_\pm=T_F$ for NS-type currents,
\begin{displaymath}
\xymatrix{
& f \ar@/^/[dr]^{J_\pm} & \\
J_\pm T_F f \ar@/^/[ur]^{J_\pm} & & J_\pm f \ar@/^/[dl]^{J_\pm} \\
& T_F f \ar@/^/[ul]^{J_\pm} &
}
\end{displaymath}
as shown in the graph. The four fields organize themselves pairwise into two $J_\pm$-orbits which are related by the action of $T_F$, or better of $(T_F,\psi)$. In fact, from the fusion rules of $(T_F,\psi)$ with off-diagonal fields it follows that
\begin{equation}
(T_F,\psi)\cdot\langle f,J_\pm f\rangle=\langle T_F f, J_\pm T_F f\rangle\,.
\end{equation}
Each $J_\pm$-orbit has the same form as (\ref{offdiag f.p. of except curr}). 
In the $(T_F,\psi)$-extension they are identified and become fixed points of the exceptional simple currents (\ref{formual for exceptional simple currents}). \\
Similarly, we can organize the fields differently. For instance, by starting from the $J_\pm$-orbit $\langle f,J_\pm T_F f\rangle$, we have
\begin{equation}
(T_F,\psi)\cdot\langle f,J_\pm T_F f\rangle=\langle T_F f, J_\pm f\rangle\,,
\end{equation}
where we used $T_F^2=1$. The same argument holds if we start from any $J_\pm$-orbit of two consecutive fields in the graph above: the $(T_F,\psi)$-extension will always identify it with the remaining orbit.

In the next subsection we give and explicit example corresponding to the ``easy'' case of minimal models at level two.

\subsection{$k=2$ Example}
In order to better visualize the structure of the exceptional simple currents and their fixed points, let us consider the $k=2$ case, where we permute two $N=2$ minimal models at level two. This case is easy enough to be worked out explicitly, but complicated enough to show all the desired properties. This minimal model has 24 fields (12 in the R sector and 12 in the NS sector), of which 16 simple currents. Its permutation orbifold has got 372 fields, of which 32 simple currents coming from diagonal (symmetric and anti-symmetric) combinations of the original simple currents. The ones with (half-)integer spin have generically got fixed points which we know how to resolve from chapter \ref{paper3}.

In the $(T_F,\psi)$-extended orbifold theory, the exceptional currents with fixed points are
\begin{equation}
\langle J_\pm,T_F\cdot J_\pm\rangle_\alpha\,\,,\qquad \alpha=0,\,1\,,
\end{equation}
with
\begin{equation}
J_+=(0,2,0)\qquad{\rm and}\qquad J_-=(0,-2,0)\,.
\end{equation}
Their off-diagonal fixed points are of the form
\begin{equation}
\langle f,J_\pm f'\rangle\,,
\end{equation}
with $f$ and $J_\pm f'$ given by
\begin{eqnarray}
f=(1,1,0) &{\rm and}& J_\pm f'=(1,-1,0)\nonumber\\
f=(1,2,1) &{\rm and}& J_\pm f'=(1,0,1)\nonumber\\
f=(1,-1,0) &{\rm and}& J_\pm f'=(1,1,2)\nonumber\\
f=(1,2,1) &{\rm and}& J_\pm f'=(1,0,-1)\nonumber
\end{eqnarray}
To these, we still have to add the twisted fixed points, but we know already exactly what they are. 
One can observe that some fields appear twice, e.g. $(1,2,1)$, and other fields never appear, e.g. $(1,2,-1)$. This can be easily explained.
The reason why some of them appear more than once is because $f$ and $f'$ can have either equal or different $s$-values ($J_\pm$ only acts on the $m$-values).\\
Similarly, some fields are identified by the $(T_F,\psi)$-extension and hence they \textit{seem} never to appear. For example, the off-diagonal field $\langle(1,2,-1),(1,0,1)\rangle$ seems not to be there, but it is actually identified  with $\langle(1,2,1),(1,0,-1)\rangle$, which appears in the last line of the list above; similarly $\langle(1,2,-1),(1,0,-1)\rangle$ seems again not to be there as well, but it is identified with $\langle(1,2,1),(1,0,1)\rangle$ which is there in the second line of the same list.

More in general, this is a consequence of (\ref{temp fp of excep curr}). In the present situation we see this explicitly. Let us look at the  current 
\begin{equation}
\langle(0,2,0),(0,2,2)\rangle
\end{equation}
in the permutation orbifold and compute its fusion rules with the off-diagonal field $\langle(1,2,-1),(1,0,1)\rangle$:
\begin{equation}
\label{example fusion}
\langle(0,2,0),(0,2,2)\rangle\cdot \langle(1,2,-1),(1,0,1)\rangle= 
\langle(1,2,-1),(1,0,1)\rangle+\langle(1,2,1),(1,0,-1)\rangle\,.
\end{equation}
We see the appearance of the second term on the r.h.s., which is also an off-diagonal field, so we are led to ask about its fusion as well:
\begin{equation}
\label{example fusion 2}
\langle(0,2,0),(0,2,2)\rangle\cdot \langle(1,2,1),(1,0,-1)\rangle= 
\langle(1,2,-1),(1,0,1)\rangle+\langle(1,2,1),(1,0,-1)\rangle\,,
\end{equation}
which is exactly the same as the first one. However, observe that the current $(T_F,\psi)$ relates the two terms on both r.h.s.'s:
\begin{eqnarray}
(T_F,\psi)\cdot \langle(1,2,-1),(1,0,1)\rangle&=&\langle(1,2,1),(1,0,-1)\rangle\nonumber\\
(T_F,\psi)\cdot \langle(1,2,1),(1,0,-1)\rangle&=&\langle(1,2,-1),(1,0,1)\rangle\,.
\end{eqnarray}
Then, they form one orbit in the $(T_F,\psi)$-extension and, since they have integer monodromy charge, this off-diagonal orbit survives the projection. Due to (\ref{example fusion}) and (\ref{example fusion 2}), this orbit becomes an off-diagonal fixed point of the exceptional current in the $(T_F,\psi)$-extended orbifold.

As a comment, we remark that it is not known at the moment how to resolve these fixed points. The reason is that they are fixed points of an off-diagonal current for which there is no solution yet, unlike for the fixed points of diagonal currents for which the solution exists and was provided by our ansatz in chapter \ref{paper3}.

\section{Orbit structure for $N=2$ and $N=1$}
\label{CFT summary}
In this section we want to summarize the simple current orbits for the theories considered here, and give the analogous results
for $N=1$ minimal models for comparison. Most of the construction, and in particular the definition of the six kinds
of CFT listed in the introduction works completely analogously for $N=2$ and $N=1$.
The world-sheet supercurrent, originating from the diagonal field $\langle 0,T_F \rangle$, comes in both cases from a 
fixed point. However, a novel feature occurring for $N=1$ but not for $N=2$ is that this supercurrent itself has  fixed points 
whose resolution requires additional data. 

Another important difference between the $N=2$ and $N=1$ permutation orbifolds is that in the latter case the supersymmetric and
the non-supersymmetric orbifold (the extensions of the BHS orbifold by $(T_F,1)$ or $(T_F,0)$ respectively) have a different
number of primaries, whereas for $N=2$ this is the same.

The simple current groups of all these theories are as described below. A few currents always play a special r\^ole, namely
\begin{itemize}
\item{The ``un-orbifold" current. This is the current that undoes the permutation orbifold. In the BHS orbifold this is
the anti-symmetric diagonal field $(0,1)$, which has spin-1. If the theories are extended by $(T_F,1)$ or $(T_F,0)$ this
field becomes part of a larger module, but is still the ground state of that module.}
\item{The world-sheet supercurrent(s). This has always weight $\frac32$, and can have fixed points only for $N=1$ (and then
it usually does). The supersymmetric permutation orbifolds always have two of them, which originate from the
split fixed points of the off-diagonal field $\langle 0, T_F \rangle$. Note that this multiplicity, two, has nothing to do with
the number of supersymmetries. The latter is given by the dimension of the ground state of the supercurrent module.
The fusion product of the two supercurrents is always the un-orbifold current. These spin-$\frac32$ currents also occur in the non-supersymmetric
theory $X$, except in that case they generate a $\mathbb{Z}_4$ group, whereas in the supersymmetric case the discrete group they generate is
 $\mathbb{Z}_2\times \mathbb{Z}_2$.}
\item{The Ramond ground state simple currents. These exist only for the $N=2$ and not for the $N=1$ superconformal models.}
\end{itemize}

In the following we call a fixed point ``resolvable" if we have explicit formulas for the fixed point resolution matrices,
and unresolvable otherwise. Therefore, ``unresolvable" does not mean that the fixed points cannot be resolved in principle, but
simply that it is not yet known how to do it.
Note that the choices of generators of discrete
groups described below are not unique, but we made convenient choices.  As much as possible, we try to choose the special
currents listed above as generators of the discrete group factors.
\begin{itemize}
\item{$N=2,\,\, k=1\mod 2$. 
\begin{itemize}
\item{The minimal models have a simple current group $\mathbb{Z}_{4k+8}$. As its generator one can take the
Ramond ground state simple current. The power $2k+4$ of this generator is the world-sheet supercurrent. None of the
simple current has fixed points.}
\item{The supersymmetric permutation orbifold has a group structure $\mathbb{Z}_{4k+8} \times \mathbb{Z}_2$. The first factor is generated by
the Ramond ground state simple current. The power $2k+4$ of this generator is the un-orbifold current. This is the only current that has fixed points, which are resolvable. The factor $\mathbb{Z}_2$
is generated by the world-sheet supercurrent.}
\item{The non-supersymmetric permutation orbifold $X$  also has a group structure $\mathbb{Z}_{4k+8} \times \mathbb{Z}_2$. The spin-$\frac32$
fields originating from the diagonal field $\langle 0,T_F \rangle$ have order 4, and generate a $\mathbb{Z}_{k+2}$ subgroup of  $\mathbb{Z}_{4k+8}$.
The order-two element of $\mathbb{Z}_{4k+8}$ is, just as above, the un-orbifold current. Also in this case it has resolvable fixed points. 
}
\end{itemize}
}
\item{$N=2, \,\,k=0\mod 4$. 
\begin{itemize}
\item{The minimal models have a simple current group $\mathbb{Z}_{2k+4}\times \mathbb{Z}_2$. As the generator of the first factor one can take the
Ramond ground state  simple current, and the world-sheet supercurrent can be used as the generator of the second. 
The middle element of the  $\mathbb{Z}_{2k+4}$ factor is an integer spin current with resolvable fixed points.}
\item{The supersymmetric permutation orbifold has a group structure $\mathbb{Z}_{2k+4} \times \mathbb{Z}_2 \times \mathbb{Z}_2$. The first factor is generated by
the Ramond ground state simple current. The second factor by the 
un-orbifold current. The last factor is generated by the world-sheet supercurrent. The middle element of the first factor and
the generator of the second factor, as well as their product have resolvable fixed points.}
\item{The non-supersymmetric permutation orbifold $X$  has a group structure $\mathbb{Z}_{2k+4} \times \mathbb{Z}_4$. The spin-$\frac32$
fields originating from the diagonal field $\langle 0,T_F \rangle$ have order 4 and
can be chosen as generators of the $\mathbb{Z}_4$ factor. 
There are three non-trivial currents with resolvable fixed points, which have the same origin (in terms of minimal model fields)
 as the ones in the supersymmetric orbifold.
}
\end{itemize}
}
\item{$N=2,\,\, k=2\mod 4$. 
\begin{itemize}
\item{The minimal models have a simple current group $\mathbb{Z}_{2k+4}\times \mathbb{Z}_2$. The structure is exactly as for $k=0\mod4$.}
\item{The supersymmetric permutation orbifold has a group structure $\mathbb{Z}_{2k+4} \times \mathbb{Z}_2 \times \mathbb{Z}_2$. One can choose
the same generators as above for $k=0\mod4$. The fixed point structure is also identical, except that there are four
additional currents with unresolvable fixed points. These four currents are the two order 4 currents of $\mathbb{Z}_{2k+4}$ multiplied with
each of the two world-sheet supercurrents.} 
\item{The non-supersymmetric permutation orbifold $X$  has a group structure $\mathbb{Z}_{2k+4} \times \mathbb{Z}_4$. As in the
supersymmetric case, there are three non-trivial currents with resolvable fixed points, and four with unresolvable fixed points.  These currents have
the same origin as those of the supersymmetric orbifold.
}
\end{itemize}
}
\item{$N=1,\,\, k=1\mod 2$. 
\begin{itemize}
\item{The minimal models have a simple current group $\mathbb{Z}_{2}$, generated by the world-sheet supercurrent. This current has
resolvable fixed points.}
\item{The supersymmetric permutation orbifold has a group structure $\mathbb{Z}_{2} \times \mathbb{Z}_2$. The two factors can be
generated by the un-orbifold current and by the world-sheet current. The fourth element also has spin-$\frac32$, and is an alternative world-sheet
supercurrent. The un-orbifold current has resolvable fixed points, the supercurrents have unresolvable fixed points. }
\item{The non-supersymmetric permutation orbifold $X$  has a group structure $\mathbb{Z}_{8}$. The order-2 element in this
subgroup is the un-orbifold current, which has resolvable fixed points. None of the other currents have fixed points.}
\end{itemize}
}
\item{$N=1,\,\, k=0\mod 2$. 
\begin{itemize}
\item{The minimal models have a simple current group $\mathbb{Z}_{2}\times \mathbb{Z}_2$. All currents have resolvable fixed points.
One of them is the world-sheet supercurrent.}
\item{The supersymmetric permutation orbifold has a group structure $\mathbb{Z}_{2} \times \mathbb{Z}_2 \times \mathbb{Z}_2$. Two of the three factors
are generated by the un-orbifold current and one of the world-sheet supercurrents. All currents have fixed points, and for four of them,
including the supersymmetry generators, they are unresolvable.}
\item{The non-supersymmetric permutation orbifold $X$  has a group structure $\mathbb{Z}_{4} \times \mathbb{Z}_2$. All currents have
fixed points, and for four of them they are unresolvable. 
}
\end{itemize}
}
\end{itemize}

\section{Conclusion}
\label{conclusion paper4}
In this chapter we have studied permutations and extensions of $N=2$ minimal models at arbitrary level $k$. These models are very interesting for several reason: not only because they are non-trivial solvable conformal field theories, but also because they are the building blocks of Gepner models which have some relevance in string theory phenomenology.

Our main points are two. First of all, a new structure arises relating conformal field theories built out of minimal models. Starting from the tensor product we perform $\mathbb{Z}_2$-orbifold and extension in both possible orders, generating in this way new CFT's. Some of them are easily recognizable, such as the $N=2$ supersymmetric orbifold obtained by extending the standard permutation orbifold by the current $(T_F,1)$. Some others are however not known, like the CFT that we have denoted by $X$, obtained by extending the orbifold by $(T_F,0)$. 
Secondly, unexpected off-diagonal simple currents appear due to the interplay of the orbifold and the extension procedure. Sometimes they have fixed points that need to be resolved. However, because they are related to off-diagonal currents, we do not know how to resolve them at the moment.

\chapter{Permutation orbifolds of heterotic Gepner models}
\label{paper5}

{\flushright
{\small 
\textit{What is it that breathes fire into the equations}\par
\textit{and makes a universe for them to describe?}\par
\textit{The usual approach of science}\par
\textit{of constructing a mathematical model}\par
\textit{cannot answer the questions}\par
\textit{of why there should be a universe for the model to describe.}\par
\textit{(S. Hawking, A Brief History of Time)}\par
}
}

\section{Introduction}

We are finally able to apply our previous results on permutation orbifolds to the phenomenologically interesting case of four-dimensional string model building. The traditional way of constructing particle spectra is due to Gepner, who used special tensor products of $N=2$ minimal models on which space-time and world-sheet supersymmetries can be imposed by suitable simple current extensions. The models that we are going to construct can be called permuted Gepner models, since the $N=2$ building blocks will be replaced, when possible, by their $N=2$ supersymmetric permutation orbifolds, described in the last chapter. Moreover, we will deal with heterotic Gepner models, where Gepner's construction is carried on only on the right supersymmetric sector of the string. In fact, heterotic string theory \cite{Gross:1984dd} is the oldest approach towards the construction of the standard model in string theory. It
owes its success to the fact that the gross features of the standard model appear to come out nearly automatically:
families of chiral fermions in representations that are structured as in $SO(10)$-based GUT models.

In constructing spectra, CFT's \cite{Belavin:1984vu} turn out to be very useful. A general heterotic CFT consists of a right-moving sector that has $N=2$ world-sheet supersymmetry and a non-supersymmetric
left-moving sector. Most existing work has been limited either to free CFT's (bosons, fermions or orbifolds) for these two sectors, 
or to interacting CFT's where the bosonic sector is essentially a copy of the fermionic one. Furthermore the interacting CFT's
themselves have mostly been limited to tensor products of $N=2$ minimal models \cite{Gepner:1987vz,Gepner:1987qi}.

Already in the late eighties of last century ideas were implemented to reduce some of these limitations of interacting CFT's.
Instead of minimal models, Kazama-Suzuki  models were used \cite{Kazama:1988qp}. Another extension was to consider permutation orbifolds
of $N=2$ minimal models \cite{Klemm:1990df,Fuchs:1991vu}. But both of these ideas could only be analyzed in a very limited way themselves. The real power of
interacting CFT construction comes from the use of simple current invariants  
\cite{Schellekens:1989am,Intriligator:1989zw,Schellekens:1989dq,Kreuzer:1993tf,GatoRivera:1991ru},
which greatly enhance the number and scope
of the possible constructions. In particular
the left-right symmetry of the original Gepner models could be broken by
considering asymmetric simple current invariants \cite{Schellekens:1989wx}, allowing for example a breaking of the canonical $E_6$ subgroup to
$SO(10)$, $SU(5)$, Pati-Salam models  or even just the standard model (with some additional factors in the
gauge group). However, precisely this powerful tool is not available at present in either Kazama-Suzuki models or
permutation orbifolds. The original computations were limited to diagonal invariants, where with a combination of a variety
of tricks the spectrum could be obtained. Up to now, all that is available in the literature is a very short list of Hodge numbers
and singlets for $(2,2)$ spectra with $E_6$ gauge groups \cite{Klemm:1990df,Fuchs:1991vu,Font:1989qc,Schellekens:1991sb,Fuchs:1993av,Aldazabal:1994rm} (the last paper discusses permutation orbifolds of Kazama-Suzuki models).
To use the full power of simple current methods we need to know the exact CFT spectrum
and the fusion rules of the primary fields of the  building blocks. 
The former has never been worked out for Kazama-Suzuki models, and the latter
was not available for permutation orbifolds until recently. 

Using pioneering work by Borisov, Halpern and Schweigert  \cite{Borisov:1997nc}, in chapters \ref{paper1}-\ref{paper3} we have extended their results to fixed point resolution matrices \cite{Maio:2009kb,Maio:2009cy,Maio:2009tg}, while in chapter \ref{paper4} we have
constructed the $\mathbb{Z}_2$ permutation orbifolds of $N=2$ minimal models \cite{Maio:2010eu}. These can now be used as building blocks
in heterotic CFT constructions, on equal footing, and in combination with all other building blocks, such as  the minimal
models themselves and free fermions. Furthermore we can now for the first time apply the full simple current machinery in exactly
the same way as for the minimal models. 

Meanwhile, another method was added to this toolbox, allowing us to advance a bit more deeply into the heterotic landscape, and
away from free or symmetric CFT's. This is called ``heterotic weight lifting" \cite{GatoRivera:2009yt}, a replacement of $N=2$ building blocks in the bosonic
sector by isomorphic (in the sense of the modular group) $N=0$ building blocks (more precisely, replacing $N=2$ building blocks together
with the extra $E_8$ factor). This method requires knowledge of the exact CFT spectrum, which indeed we have. A variant of this idea is
the replacement of the $U(1)_{B-L}$ factor (times $E_8$) by an isomorphic CFT. This has been called ``B-L lifting".

The purpose of this chapter is to put all these ingredients together using permutation orbifolds of $N=2$ minimal models as building 
blocks in combination with minimal models. We want to do this for the following reasons:
\begin{itemize}
\item{Check the consistency of the permutation orbifold CFT's we presented in chapter \ref{paper4}. Chiral heterotic spectra are very sensitive to the correctness of conformal weights and ground state dimensions of the CFT, as well as the
correctness of the simple current orbits. This is especially true for weight-lifted spectra, because
they have non-trivial Green-Schwarz anomaly cancellations.}
\item{Compare our results with those of previous work on permutation orbifolds \cite{Klemm:1990df,Fuchs:1991vu}. These results
were obtained using a rather different method, by applying permutations directly to complete heterotic string spectra.}
\item{Check if the generic trends on fractional charges and family number are confirmed also in the class of permutation orbifolds.}
\item{Add a few more items to the growing list of potentially interesting three-family interacting CFT models.}
\end{itemize}

The key ingredient of the present discussion is our previous chapter \ref{paper4}, where we have studied permutations, together with extensions in all possible order, and found very interesting novelties. For example, we have determined how to construct a supersymmetric permutation of minimal models: in particular, the world-sheet supersymmetry current in the supersymmetric orbifold turns out to be related to the anti-symmetric representation of the world-sheet supersymmetry current of the original minimal model. When the symmetric representation is used, instead, one ends up with a conformal field theory, which is isomorphic to the supersymmetric orbifold, but it is not supersymmetric itself. 

In the extended permuted orbifolds so-called exceptional simple currents appear, which originate from off-diagonal representations. Generically, there are many of them, depending on the particular model under consideration, and they do not have fixed points. However, if and only if the ``level'' of the minimal model is equal to $k=2$ mod $4$, four of all these exceptional currents do admit fixed points. As a consequence, in those cases the knowledge of the modular $S$ matrix is plagued by the existence of non-trivial and unknown $S^J$ matrices (one $S^J$ matrix for each exceptional current $J$). The full set of $S^J$ matrices is available for standard $\mathbb{Z}_2$ orbifolds (see \cite{Maio:2009kb,Maio:2009cy,Maio:2009tg}), but not for their (non-)supersymmetric extensions, due to these four exceptional currents with fixed points 
\cite{Schellekens:1990xy,Schellekens:1989uf,Fuchs:1996dd,Fuchs:1995zr,Schellekens:1999yg,Fuchs:1995tq}.

Here we consider permutations in Gepner models. One starts with Gepner's standard construction where the internal CFT is a product of $N=2$ minimal models. Sometimes there are (at least) two $N=2$ identical factors in the tensor product. When it is the case, we can replace these two factors with their permutation orbifold. Moreover, one also has to impose space-time and world-sheet supersymmetry, which is achieved by suitable simple-current extensions.

This chapter is organized as follows. 
In section \ref{Section: Heterotic Gepner models} we review the standard construction of heterotic Gepner models. In section \ref{Section: Permutation orbifold of N=2 minimal models} we review the main ingredients and the most relevant results of $\mathbb{Z}_2$ permutation orbifolds when applied to $N=2$ minimal models. 
In section \ref{Section: Lifts} we describe the heterotic weight lifting and the B-L lifting procedures, which allow us to replace the trivial $E_8$ factor plus either one $N=2$ minimal model or the $U(1)_{B-L}$ with a different CFT, which has identical modular properties, in the bosonic (left) sector. 
In section \ref{Section: Comparison} we compare our results on (2,2) spectra with the known literature. 
In section \ref{Section: Results} we present our phenomenological results concerning the family number distributions,
gauge groups, fractional charges and other relevant data. 
In appendix \ref{Appendix Paper5} we derive a few facts about simple current invariants.  
Appendix \ref{Appendix Paper5: tables} contains
tables summarizing the main results for the four cases (standard Gepner models and the three kinds of lifts). The content of this chapter is based on \cite{Maio:2011qn}.

\section{Heterotic Gepner models}
\label{Section: Heterotic Gepner models}
In this section we review the construction of four-dimensional heterotic string theory. The starting point is a set of bosons $X^\mu$ ($\mu=0,\dots,3$) for both the right and left movers, a right-moving set of NSR fermions $\psi^\mu$, plus corresponding ghosts, and an internal CFT with central charges $(c_L,c_R)=(22,9)$, that we denote by $\mathscr{C}_{22,9}=\mathscr{C}_{22}\times\mathscr{C}_{9}$. Observe that the right-moving superconformal field theory $(X,\psi)$+ghosts has central charge $c=3$. Equivalently, one can think of it as the conformal field theory of two bosons $X^i$ and their fermionic superpartners $\psi^i$ in light-cone gauge. The fermions $\psi^i$ form an $SO(2)_1$ abelian algebra, with central charge $c=1$.

The next step is to replace the NSR $SO(2)_1$ fermions by a set of $13$ bosonic fields living in the maximal torus of an $(E_8)_1\times SO(10)_1$ affine Lie algebra. This is the bosonic string map \cite {Lerche:1986cx}, which transforms the fermionic CFT into a bosonic one with same modular properties. The total right-moving CFT has now central charge equal to $c_R=2+9+13=24$, as the left-moving bosonic theory. Hence, all four-dimensional heterotic strings correspond to all compactified bosonic strings with an internal sector:
\begin{equation}
\mathscr{C}_{22,9}\times \left( (E_8)_1\times SO(10)_1 \right)_R \,.
\end{equation}
To summarize:
\begin{eqnarray}
\hbox{Left-moving}&&(X^\mu,{\rm ghost})\times \mathscr{C}_{22}\nonumber\\
\hbox{Right-moving}&&(X^\mu,{\rm ghost})\times \mathscr{C}_{9}\times (E_8)_1\times SO(10)_1\nonumber
\end{eqnarray}
with $\mu=0,\dots,3$. Equivalently, in light-cone gauge one uses $X^i$ instead of $(X^\mu,{\rm ghost})$.

In the right-moving sector, all the CFT building blocks have  $N=2$ worldsheet supersymmetry. This implies the existence of two operators
with simple fusion rules: the worldsheet supercurrent $T_F$ and the spectra flow operator $S_F$. In general, the
internal CFT in the fermionic sector is itself built out of $N=2$ building blocks, that have such currents as well. 

In order to preserve right-moving world-sheet supersymmetry, the total supercurrent $T^{\rm st}_F+T_F^{\rm int}$ must have a well-defined periodicity, since it couples to the gravitino. Here, $T^{\rm st}_F=\psi^\mu \partial X_\mu$ is the world-sheet supercurrent in space-time and $T_F^{\rm int}$ is the supercurrent of the internal sector. Hence the allowed states will have the same spin structure in all the subsectors of the tensor product, namely the R (NS) sector of $SO(10)_1$ must be coupled to the R (NS) sector of the internal CFT. This result is achieved by an integer-spin simple current extension of the full right-moving algebra, where the current is given by the product of the supercurrents $T^{\rm st}_F\cdot T_F^{\rm int}$: it corresponds to projecting out all the combinations of mixed spin structures. When the internal CFT is a product of many sub-theories, as in the case of Gepner models, each with its own world-sheet supercurrent $T_{F,i}$, then one has to extend the full right-moving algebra by all the currents $T^{\rm st}_F\cdot  T_{F,i}$. In simple current language this means that we extend the chiral
algebra by all currents
\begin{equation}
\label{WS}
W_i= (0,\ldots,0,T_{F,i},0,\ldots,0;V)\,,
\end{equation}
where we use a semi-colon to separate the internal and space-time part, and we use the standard notation $0,V,S,C$ for $SO(10)_1$ simple currents (or conjugacy classes).

A sufficient and necessary condition for space-time supersymmetry is the presence of a right-moving spin-$1$ chiral current transforming as an $SO(10)_1$ spinor. Hence this current must be equal to the product of the spinor $S$ of the $SO(10)_1$, which has spin $h=\frac{5}{8}$, times an operator $S^{\rm int}$ from the Ramond sector of the internal CFT $\mathscr{C}_9$, which must then have spin $h=\frac{3}{8}$. This last value saturates the chiral bound $h\geq\frac{c}{24}$ for the internal right-moving CFT which has central charge $c=9$, hence $S^{\rm int}$ corresponds to a Ramond ground state. 

Among the Ramond ground states, one is very special. $N=2$ supersymmetry possesses a one-parameter continuous automorphism of the algebra, known as spectral flow, which, when restricted to half-integer values of the parameter, changes the spin structures and maps Ramond fields to NS fields, hence uniquely relating fermionic to bosonic fields. In particular, under spectral flow, the NS field corresponding to the identity is mapped to a Ramond ground state which has $h=\frac{c}{24}$ and is called the spectral-flow operator. Not surprisingly, the spectral flow operator is related to the $N=1$ space-time supersymmetry charge. We will denote it as $S_F$.

In our set-up of four dimensional heterotic string theories, $N=1$ space-time supersymmetry is achieved again by a simple current extension. The current in question is the product of the space-time spin field $S_F^{\rm st}$ with $S_F^{\rm int}$, where $S_F^{\rm int}$ is the spectral-flow operator. If the internal CFT is built out of many factors, then $S_F^{\rm int}=\bigotimes_i S_{F,i}$, where $S_{F,i}$ is the spectral-flow operator in each factor. 
In simple current language, the space-time supersymmetry condition amounts to extending the chiral algebra of the CFT by
the simple current
\begin{equation}
\label{SS}
S_{\rm susy}=(S_{F,1},\ldots,S_{F,r};S)\,,
\end{equation}
where $r$ denotes the number of factors.
Obviously these simple current extensions must be closed under fusion, in combination with all world-sheet supersymmetry extensions
discussed above. Modular invariance of the final theory is then guaranteed by the simple current construction.

So far everything holds for any combination of superconformal $N=2$ building blocks. The only ones available in practice
(prior to this work) are suitable combinations of free bosons and/or fermions, and $N=2$ minimal models. 
We have already discussed $N=2$ minimal models in chapter \ref{paper4}. These are 
unitary finite-dimensional representations of the $N=2$ superconformal algebra, which exist only for  $c\leq3$. They are labelled by an integer $k$, in terms of which the central charge is 
\begin{equation}
c=\frac{3k}{k+2}\,.
\end{equation}
Using the coset description of the $N=2$ minimal models
\begin{equation}
\frac{SU(2)_k \times U(1)_4}{U(1)_{2(k+2)}}\,
\end{equation}
one can label representations by three integers $(l,m,s)$, where $l$ is an $SU(2)_k$ quantum number and $m$ and $s$ are $U(1)$ labels. The range is: $l=0,\dots,k$, $m=-k-1,\dots,k+2$, $s=-1,\dots,2$ ($s=0,2$ for NS sector, $s=\pm1$ for R sector). Moreover, fields satisfy the constraint $l+m+s=$ even and are pairwise identified according to $\phi_{l,m,s}\sim\phi_{k-l,m+k+2,s+2}$, which is realized as a formal simple current extension.

Now consider the right-moving algebra of the heterotic string. The internal CFT $\mathscr{C}_9$ can be built as a product of a sufficient number of $N=2$ minimal models such that 
\begin{equation}
\label{MinSum}
\sum_i^r\frac{3k_i}{k_i+2}=9\,,
\end{equation}
so the full algebra is
\begin{equation}
\bigotimes_i (N=2)_i \otimes (E_8)_1\otimes SO(10)_1
\end{equation}
and representations are labelled by
\begin{equation}
\bigotimes_i (l_i,m_i,s_i) \otimes (0) \otimes (s_0)\,.
\end{equation}
Observe that the $(E_8)_1$ algebra has only one representation, i.e. the identity, and it is often omitted in the product. Here $s_0$ denotes one of the four $SO(10)_1$ representations, $s_0=O,V,S,C$. As discussed above, we impose world-sheet and space-time supersymmetry by simple-current extensions. The world-sheet supercurrent for each $N=2$ minimal model is labelled by $T_{F,i}=(0,0,2)$ and the spectral-flow operator is $S_{F,i}=(0,1,1)$.  These are used in the world-sheet and space-time chiral algebra extensions (\ref{WS}) and (\ref{SS}).

These chiral algebra extensions are mandatory only in the fermionic sector. However, modular invariance does not allow
an extension in just one chiral sector. The most common way of dealing with this is to use exactly the same CFT in the 
left-moving sector, with exactly the same extensions. Of course any $N=2$ CFT is a special example of an $N=0$ CFT.
This construction leads to $(2,2)$ theories, with spectra analogous to Calabi-Yau compactifications, characterized by Hodge
number pairs and with a certain number of families in the $(27)$ of $E_6$. 
On the other hand, modular invariance is blind to most features of the CFT spectrum. It only
sees the modular group representations. This makes it possible to use in the left, bosonic, sector a different set of extension
currents than on the right. In particular one can replace the image of the space-time current by something else, thus breaking
$E_6$ to $SO(10)$. Furthermore one can break world-sheet supersymmetry in the bosonic sector. One can even go a step
further and break $SO(10)$ and $E_8$ to any subgroup, as long as this breaking can be restored by means of simple
currents. Those currents are then mandatory in the fermionic sector (since otherwise the bosonic string map cannot be used),
but can be replaced by isomorphic alternatives in the left sector. In general, we will call this class $(0,2)$ models. 

All the aforementioned possibilities will be considered in this chapter, except $E_8$ breaking. The $SO(10)$ breaking we consider is to $SU(3)\times SU(2) \times U(1)_{30} \times U(1)_{20}$, where the first three factors are the standard model gauge groups with the standard $SU(5)$-GUT normalization for the $U(1)$. The fourth factor corresponds in certain cases to $B-L$. It is known that under such a breaking fractionally-charged particles may arise \cite{Schellekens:1989qb,Wen:1985qj,Athanasiu:1988uj}. They can be either chiral or non-chiral, or even absent in the massless sector. We will investigate when these options occur.

\section{Orbifolds of $N=2$ minimal models}
\label{Section: Permutation orbifold of N=2 minimal models}
In chapter \ref{paper4} the permutation orbifold of $N=2$ minimal models was studied. Extensions and permutations were performed in all possible orders and a nice structure was seen to arise, together with exceptional off-diagonal simple currents appearing in the extended orbifolds. In this section we recall the procedure of how to build a supersymmetric permutation orbifolds starting from $N=2$ minimal models. We will restrict ourselves to $\mathbb{Z}_2$ permutations, because a formalism to build permutation orbifold CFT's for higher
cyclic orders is not yet  available.

Consider the internal CFT $\mathscr{C}_9$ to be a tensor product of $r$ minimal models such that the total central charge is equal to $9$. We denote such a theory as\footnote{Note that here we mean the unextended tensor product. In particular, world-sheet supersymmetry 
extensions are not implied.}
\begin{equation}
(k_1,k_2,k_3\dots,k_r)\,,
\end{equation}
each $k_i$ parametrizing the $i^{\rm th}$ minimal model. Suppose that two of the $k_i$'s are equal: then the two corresponding minimal models are also identical and one can apply the orbifold mechanism to interchange them. We will use brackets to label the block corresponding to the orbifold CFT: e.g. if $k_2=k_3$, then the permutation orbifold will be denoted by 
\begin{equation}
(k_1,\langle k_2,k_3\rangle\dots,k_r)\,.
\end{equation}
Multiple permutations are of course also possible. For convenience, we will follow the standard notation, used extensively in literature, of ordering the minimal models according to increasing level, namely $k_i\leq k_{i+1}$. Consequently, identical factors will always appear next to each other.
The orbifolded theory has the same central charge of the original one, namely $\sum_i^r c_i=9$, and hence can be used to build four dimensional string theories. 

Note that by $\langle k,k\rangle$ we mean the {\it supersymmetric} permutation orbifold, which, as explained in chapter \ref{paper4}, is obtained from the minimal model with level $k$ by first constructing the non-supersymmetric BHS orbifold (which we will denote as
$[ k,k]$), extending this CFT by
the anti-symmetric combination of the world-sheet supercurrent $(T_F,1)$, and resolving the fixed points occurring as a result
of that extension. This fixed point resolution promotes some fields to simple currents. All these simple currents will be used
to build MIPF's, using the general formalism presented in \cite{Kreuzer:1993tf}. 

Fixed point resolution enters the discussion at various points, and to prevent confusion we summarize here some relevant
facts. In the following we consider chains of extensions of the chiral algebra of a {\rm CFT}, and denote them as $({\rm CFT})_n$. Here
$({\rm CFT})_0$ is the original {\rm CFT}, $({\rm CFT})_1$ a first extension, $({\rm CFT})_2$ a second extension etc. In this process the chiral
algebra is enlarged in each step. The number of primary fields can decrease because some are projected out and others are
combined into new representations, but it can also increase due to fixed point resolution (apart from some special cases
the decrease usually wins over the increase). We are not assuming that each extension is itself ``indecomposable" ({\it i.e.} not
the result of several smaller extensions), but in practice the case of most interest will be a chain of extensions of order 2.
The following facts are important.
\begin{itemize}
\item{Simple currents $J$ are characterized by the identity $S_{0J}=S_{00}$, where $S$ is the modular transformation matrix.
For all other fields $i$, $S_{0i} > S_{00}$.}
\item{In an extension by a simple current of order $N$, the matrix elements $S_{0f}$ of fixed point fields are reduced by a factor of $N$. 
For this reason a fixed point field of $({\rm CFT})_n$ can be a simple current of $({\rm CFT})_{n+1}$. We will call these ``exceptional simple currents".}
\item{Exceptional simple currents can be used to build new MIPF's in $({\rm CFT})_{n+1}$, but such MIPF's are not simple
current MIPF's of $({\rm CFT})_n$. They are exceptional MIPF's.}
\item{If the fixed point resolution matrices of $({\rm CFT})_n$ are known, 
we can promote the exceptional simple currents of $({\rm CFT})_{n+1}$ to ordinary ones. This makes it possible to treat them on equal footing with
all other simple currents of $({\rm CFT})_{n+1}$.}
\item{Obviously, this process can be iterated: exceptional simple currents of $({\rm CFT})_{n+1}$ can themselves have fixed points, which can become
simple currents of $({\rm CFT})_{n+2}$.}
\item{If we know the fixed point resolution matrices of $({\rm CFT})_n$, we also know all the fixed point resolution matrices
of the ordinary simple currents of $({\rm CFT})_{n+1}$, but if the exceptional simple currents have fixed points, there is currently
no formalism available to determine their fixed point resolution matrices.}
\end{itemize}

In the previous chapters \ref{paper1}-\ref{paper3} we have developed a formalism for all fixed point resolution matrices
of the BHS permutation orbifolds. This plays the r\^ole of $({\rm CFT})_0$ in the foregoing. The supersymmetric permutation orbifold
$\langle k, k\rangle$ is $({\rm CFT})_1$. It always has exceptional simple currents, but only for $k=2~\hbox{mod}~4$ they have fixed 
points. As explained above, we cannot resolve these fixed points, but in heterotic spectrum computations
this is not necessary. This would be necessary if we want to go beyond spectrum computations to determine couplings. In spectrum
computations, fixed point fields $f$ appear in the partition function as character combinations of the form
\begin{equation}
N_f\, \bar\chi_f(\bar\tau)\chi_f(\tau), \  \  \  \  \ N_f>1 ,
\end{equation}
which is resolved into a certain number of distinct representations $(f,\alpha)$ that contribute to the partition function as in (\ref{paper1: fp mipf contribution}). 
Note that for $N_f\geq4$ the last condition has several solutions, and to find out which one is the right one the
twist on the stabilizer of the fixed point must be determined \cite{Fuchs:1996dd}. However, here
we merely want to add up the values of $N_f$ for a left-right combination of interest, and the individual values of $m_\alpha$ do not matter.

A few fields of the supersymmetric orbifold will be relevant in the following, all of untwisted type. They are:
\begin{itemize}
\item The symmetric representation of the spectral flow operator $(S_F,0)$, with $S_F=(0,1,1)$. It will be relevant to make the whole theory supersymmetric.
\item The world-sheet supercurrent of the supersymmetric orbifold, that we denote by $\langle 0,T_F \rangle$\footnote{Actually, since $\langle 0,T_F \rangle$ is a fixed point of $(T_F,1)$ in the unextended orbifold, there exist two fields $\langle 0,T_F \rangle_\alpha$ (with $\alpha=0,1$) in the supersymmetric orbifold corresponding to the two resolved fixed points. One can use any of them, since they produce the same CFT.}.
\item The anti-symmetric representation of the identity, denoted by $(0,1)$. We will call it  the ``un-orbifold current" since the extension by this current undoes the orbifold, giving back the original tensor product. 
\end{itemize}

The un-orbifold current exists in the BHS orbifold $[k,k]$ as well as in the 
supersymmetric orbifold $\langle k, k\rangle$. Denoting extension currents by means of a subscript, we have
the following CFT relations
\begin{eqnarray*}
(k,k) &= [ k, k]_{\rm unorb}\\
(k,k)_{(T_F,T_F)} &= \langle k, k\rangle_{\rm unorb}
\end{eqnarray*}
that can be checked using the box diagrams given in chapter \ref{paper4}.

In general, the full set of simple current MIPF's obtained from the permutation orbifold CFT  $(k_1,\langle k_2,k_3\rangle\dots,k_r)$ will
have a partial overlap with those of straight tensor product $(k_1, k_2,k_3,\ldots,k_r)$. Since the set of simple currents
of  $(k_1,\langle k_2,k_3\rangle\dots,k_r)$ includes the un-orbifold current one might expect that the latter set is entirely included in the
former. However, this is not quite correct, since the supersymmetric permutation orbifold has fewer simple currents than the
tensor product from which it originates, as explained above. In the extension chain,  $\langle k, k\rangle_{\rm unorb}$
is $({\rm CFT})_2$. In both steps in the chain

\begin{eqnarray*}
\hfill ({\rm CFT})_0 &= & [k,k] \\ &\downarrow\\  ({\rm CFT})_1 &=& \langle k, k\rangle\\ & \downarrow\\  ({\rm CFT})_2& =& \langle k, k\rangle_{\rm unorb} = (k,k)_{(T_F,T_F)}
\end{eqnarray*}
exceptional simple currents appear. Those of the first step are promoted to ordinary simple currents using
fixed point resolution in the BHS orbifold. We then work directly with $\langle k, k\rangle$ as a building block, but
by doing so we cannot use the exceptional simple currents emerging in the second step. In this case the
exceptional simple currents could be used by working with  $(k,k)_{(T_F,T_F)}$ directly, but then we are back in the unpermuted theory.
So the point is not that these MIPF's are unreachable, just that they cannot be reached using the simple currents of $\langle k, k\rangle $. Obviously, if we were to use a different exceptional simple current in the second extension, such that $({\rm CFT})_2$ is
a new, not previously known CFT  with exceptional simple currents, some of its MIPF's cannot be reached using simple current methods  
neither from $({\rm CFT})_1$ nor from $({\rm CFT})_2$. In all cases, one can try to derive such MIPF's explicitly as exceptional invariants, and
they can then be taken into account in heterotic spectrum computations, but this requires tedious and strongly case-dependent 
calculations. But in this chapter we only consider simple current invariants, without any claim regarding completeness of the
set of MIPF's we obtain.

The phenomenon of  exceptional simple currents is nothing new, and occurs for example in the D-invariant
of $A_{1,4}$ (which is isomorphic to $A_{2,1}$), or the extension of the tensor products of two Ising models extended by the product of the fermions (turning it into a free boson).

The simplest explicit example occurs for $k=1$. In this case the discussion can be made a bit more explicit, since
the permutation orbifold is itself a minimal model, namely the one with $k=4$: 
\begin{eqnarray*}
 \langle 1,1\rangle  &=& (4)\\
 (1,1)_{(T_F,T_F)} &=&(4)_{\rm unorb} = (4)_D
\end{eqnarray*}
The minimal $k=1$ model has 12 primaries, all simple currents, and hence the tensor product $(1,1)$ has 144 simple currents. To make the tensor product world-sheet supersymmetric we have to extend it by $(T_F,T_F)$, reducing the number of simple currents by a factor of four\footnote{Of the 12 simple currents of the minimal $k=1$ model, 6 are in the Ramond and 6 in the NS sector. In the extended tensor product, only fields with factors both in the R or in the NS sector survive (thus reducing their number by a factor of two) and they are moreover pairwise identified by the extension (thus giving another factor of two).} to 36. 
The $k=4$ minimal model has 24 simple currents. If we extend the $k=4$ minimal model by the un-orbifold current (which can be identified as such in the $\langle 1,1\rangle$
interpretation), these 24 original simple currents are reduced
to 12. Since the resulting CFT is isomorphic to $(1,1)_{(T_F,T_F)}$ there must be 24 additional simple currents. Indeed there are, but they 
are exceptional. They are related to the aforementioned exceptional currents in the D-invariant of $A_{1,4}$. This is
also the only example of exceptional simple currents in $N=2$ minimal models, and clearly in this case no MIPF's are
missed, since we can explicitly consider $(1,1)$ as well as $(4)_D$.  There might exist additional examples of exceptional
simple currents in tensor products of $N=2$ minimal models.

If the chiral algebra contains the un-orbifold current of a permutation orbifold, we obviously get nothing new. Therefore we
demand that this current is not in the chiral algebra. In general, it would be possible  to forbid it in either the left or the
right chiral algebra. This is already sufficient to find new cases. We do this, for example, with the $SU(5)$ extension
currents of the standard model, which are required in the right (fermionic string) chiral algebra, but not in the left one. However,
it turns out that the un-orbifold current is local with respect to all other simple currents.

In appendix \ref{Section: Simple current invariants} we prove a small theorem about simple current invariants. Consider a simple current  modular invariant partition function
\begin{equation}
Z(\tau,\bar{\tau})=\sum_{k,\,l}\bar{\chi}_k(\bar{\tau})M_{kl}\chi_l(\tau)\,.
\end{equation}
In the theorem it is shown that:
if a current $J$ that is local with respect to all currents used to construct the modular invariant appears on the right hand side (holomorphic sector) of the algebra, then
it will also appear on the left hand side (anti-holomorphic sector):
\begin{equation}
M_{0J}\neq 0 \qquad\Leftrightarrow \qquad M_{J0}\neq 0\,.
\end{equation}
Furthermore we show that the un-orbifold current is local with respect to all other currents.
Therefore the existence of the un-orbifold current on one side implies its existence also on the other side.
Hence it is sufficient to forbid its occurrence in either the left or the right sector.

However, there are a few cases where it cannot be forbidden at all, because it is generated by combinations
of world-sheet and space-time supersymmetry in the right (fermionic) sector, where such chiral algebra extensions 
are required. In general, a tensor product is extended by the currents $S_{\rm susy}$ and $W_i$, as explained in the previous 
section. 

If $k$ is even, the un-orbifold current does not appear on the orbit of the Ramond spinor current $S_F$, and hence can never
be generated. For arbitrary $k$ we have in the supersymmetric permutation orbifold
\begin{equation}
(S_F,0)^{2(k+2)}=
\left\{
\begin{array}{lc}
(0,0) & {\rm if}\,\,k\,\,{\rm even}\\
(0,1) & {\rm if}\,\,k\,\,{\rm odd}
\end{array}
\right.
\,,
\end{equation}
so that for $k$ odd one can obtain the un-orbifold current as a power of $(S_F,0)$. Note that instead of $2(k+2)$ one could use
any odd multiple of $2(k+2)$. In the tensor product $S_F$ is combined with the spinor currents of all the other factors,
which will be raised to the same power. Now note that in minimal models of level $k$ the following is true
\begin{equation}
S_F^{\phantom{F}2(k+2)}=
\left\{
\begin{array}{lc}
0 & {\rm if}\,\,k\,\,{\rm even}\\
T_F & {\rm if}\,\,k\,\,{\rm odd}
\end{array}
\right.
\,.
\end{equation}
Furthermore, the value $2(k+2)$ is the first non-trivial power for which either the identity or the world-sheet supercurrent is
reached. It follows that if the tensor product contains a factor with $k_i$ even, 
the complete susy current  $(S_{F,1},\ldots,S_{F,r};S)$ must be raised to a power that is a multiple of four in order
to reach either the identity or a world-sheet supercurrent. This is true for minimal model factors as well as supersymmetric permutation
orbifolds $\langle k_i, k_i\rangle$. 

Consider then a tensor product $(k_1,\ldots,k_{m-1},\langle k_m,k_m\rangle,k_{m+2},\ldots,k_r)$. 
Take the susy current 
$$(S_{F,1},\ldots ,S_{F,m-1},(S_F,0),S_{F,m+2},\ldots S_{F,r};S)$$ 
to the power $2 M$, where $M$ is the smallest common multiple of $k_i+2$, for all $i$ (including $i=m$). If all $k_i$ are odd, 
this yields 
\begin{equation}
\label{Power}
(T_{F,1},\ldots,T_{F,m-1},(0,1),T_{F,m+2},\ldots T_{F,r};V)
\end{equation}
Since this is a power of an integer spin current, the
susy current, it must have integer spin. Therefore the number of $T_{F,i}$ must be odd. Indeed, it is not hard to show that
eqn. (\ref{MinSum}) can only be satisfied with all $k_i$ odd if the total number of factors, $r$, is odd. It then follows that all entries
$T_{F,i}$  as well as the representation $V$ of $SO(10)_1$ can be nullified by world-sheet supersymmetry. Hence it follows that
the un-orbifold current of $\langle k_m,k_m\rangle$ is automatically in the chiral algebra. It also follows that if one of the $k_i$ 
is even the un-orbifold current is {\it not} in the chiral algebra generated by $S_{\rm susy}$ and $W_i$. The same reasoning can be
applied to tensor products containing more than one permutation orbifold. The conclusion is that the un-orbifold currents of each factor separately are not generated by $S_{\rm susy}$ and $W_i$, but if all $k_i$ (in minimal models as well as the permutation
orbifolds) are odd, the combination $(0,\ldots,(0,1),\ldots,(0,1),\ldots,0;0)$, with an un-orbifold component in each permutation
orbifold, will automatically appear. Obviously, if there is more than one permutation orbifold factor this does not undo the permutation. 

The set of tensor combinations with only odd factors is rather limited, namely
\begin{eqnarray*}
&(1,1,1,1,1,1,1,1,1)\\
 &(3,3,3,3,3)\\
 &(1,3,3,3,13)\\
 &(1,1,7,7,7)\\
 &(1,1,5,5,19)\\
 &(1,1,3,13,13)
 \end{eqnarray*}
We will not consider permutations of $k=1$, because $\langle1,1\rangle=4$, and hence nothing new can be found by allowing $\langle1,1\rangle$. Furthermore there is no need to consider any single permutations in the foregoing tensor products. However, we do expect the combinations 
$(3,\langle 3,3\rangle,\langle 3,3\rangle)$, 
$(\langle 1,1\rangle,\langle 7,7\rangle,7)\equiv(4,\langle 7,7\rangle,7)$ and 
$(\langle 1,1\rangle,\langle 5,5\rangle,19)\equiv(4,\langle 5,5\rangle,19)$ 
to yield something new.

For technical reasons in this work we consider only permutations of minimal models having level $k \leq 10$: computing time and memory use become just too large for large $k$. Nevertheless, the interval $k\in[2,10]$ still covers almost all the standard Gepner models where at least two factors can be permuted.

\subsection{Permutations of permutations}
An additional thing that one could try to do (and which we can in principle do with our formalism, since we know all the relevant data that are needed) is to consider permutations of permutations. Permutations of permutations are possible only for a few Gepner models, because one would need to have a number of factors in the tensor product which is larger than four and with at least four identical minimal models. Out of the 168 possibilities, there are only a few combinations that have these properties. They are:
\begin{eqnarray}
 &(6,6,6,6)&\nonumber\\
 &(1,4,4,4,4)&\nonumber\\
 &(3,3,3,3,3)&\nonumber\\
 &(1,2,2,2,2,4)&\nonumber\\
 &(2,2,2,2,2,2)&
\end{eqnarray}
As before we restrict the $k>1$. Observe that the maximal level is $k=6$, so these cases are actually all the possibilities that one can consider and one can relax here our previous restriction to $k\leq10$.

The approach one should take is the following. Consider a block of four identical minimal models. As before we can permute the factors pairwise and obtain a tensor product of two larger blocks, but again identical. Hence we can permute them again and end up with only one big block which replaces the four ones that we started with:
\begin{displaymath}
\xymatrix{
(k,k,k,k) \ar[d] \\
(\langle k,k \rangle, \langle k,k\rangle) \ar[d]\\
\big\langle \langle k,k \rangle, \langle k,k\rangle\big\rangle
}
\end{displaymath}

Although straightforward, we have not performed this calculation here. There are only very few cases to analyze, namely the five listed above, but, on the one hand, it is a pretty lengthy computation and on the other hand we do not expect drastically different spectra in comparison with normal permutations.

\section{Lifts}
\label{Section: Lifts}
In \cite{GatoRivera:2009yt} the authors describe a new method for constructing heterotic Gepner-like four-dimensional string theories out of $N=2$ minimal models. The method consists of replacing one $N=2$ minimal model together with the $E_8$ factor by a non-supersymmetric CFT with identical modular properties. Generically this method produces a spectrum with fewer massless states. Surprisingly, it is possible to get chiral spectra and gauge groups such as $SO(10)$, $SU(5)$ and other subgroups including the Standard Model. However, the most interesting feature is probably the abundant appearance of three-family models, which are very rare in standard Gepner models \cite{GatoRivera:2010gv}. Let us review how it is done in more detail, at least in the simplest case.

Start from the coset representation of the minimal model:
\begin{equation}
\frac{SU(2)_k\times U(1)_4}{U(1)_{2(k+2)}}\,,
\end{equation}
subject to field identification by the simple current $(J,2,k+2)$. Here $J$ is the simple current of the $SU(2)_k$ factor and the $U_N$ fields are labelled by their charges as $0,\dots,N-1$. 
The product of the $N=2$ minimal model and the $E_8$ factor is then
\begin{equation}
\left(\frac{SU(2)_k\times U(1)_4}{U(1)_{2(k+2)}}\right)_{(J,2,k+2)}\times E_8\,,
\end{equation}
where the brackets denote this identification. The next step is to remove the identification and mod out $E_8$ by $U(1)_{2(k+2)}$: the new CFT is then
\begin{equation}
SU(2)_k\times U(1)_4\times \frac{E_8}{U(1)_{2(k+2)}}\,.
\end{equation}
Finally we restore the identification by a standard order-2 current extension of the resulting CFT. This procedure works provided we can embed the $U(1)_{2(k+2)}$ factor into $E_8$. Some examples of how to embed $U(1)_{2(k+2)}$ into $E_8$ are given in \cite{GatoRivera:2009yt}. Finally, one can check explicitly that the modular $S$ and $T$ matrices are the same as for the $N=2$ 
minimal model times $E_8$, as they must be by construction. 
The resulting CFT is $SU(2)_k\times U(1)_4\times X_7$, where $X_7$ is the reminder of $E_8$ after dividing out $U(1)_{2(k+2)}$. $X_7$ has central charge $7$ and modular matrices $S$ and $T$ given by the complex conjugates of those of $U(1)_{2(k+2)}$ (since the ones of $E_8$ are trivial). Generically, this procedure raises the weights of the primaries in the new CFT, hence the name ``weight lifting''.

As it appears from above, the lifting of Gepner models is achieved by only a slight modification of standard Gepner models. All one has to
do is to shift the weights of some fields in the left-moving CFT by a certain integer, and replace the ground state dimensions by another, usually larger, value.  
In \cite{GatoRivera:2009yt} a list of possible lifts is given for $N=2$ minimal models at level $k$. Not for any level there exists a lift and sometimes for fixed $k$ there are more lifts. When applied to standard Gepner models, a lot of new ``lifted'' Gepner models are generated. Notationally, if a Gepner model is denoted by $(k_1,\dots,k_i,\dots, k_r)$, the corresponding lifted model will be denoted by $(k_1,\dots,\hat{k}_i,\dots, k_r)$, where the lift is done on the $i^{\rm th}$ $N=2$ factor. If for a given $k$ there
exists more than one lift, we use a tilde to denote it.

In \cite{GatoRivera:2010fi} a different class of lifts was considered, the so called B-L lifts. In this case one replaces the $U(1)_{20}$ (with 20 primaries), that is the remainder of $SU(3)\times SU(2)\times U(1)$ embedded in $SO(10)$. In the Standard Model the abelian factor is the $U(1)_{Y}$ hypercharge (denoted also as $U(1)_{30}$, with 30 primaries). The $U(1)_{20}$ that we replace here corresponds to $B-L$, hence the name ``B-L lifting''. 
It is not possible to simply replace the $U(1)_{20}$ by an isomorphic CFT with 20 primaries, central charge $c=1$ and same modular $S$ and $T$ matrices, since all the $c=1$ CFT's are classified. Again, what one can do is to add the $E_8$ factor and replace the $E_8\times U(1)_{20}$ block, which has central charge $c=9$. 
As it turns out, there are only two possible B-L lifts, that we denote by $A$ and $B$. In terms of 
compactifications from ten dimensions these possibilities can be distinguished as follows. If one compactifies the $E_8\times E_8$ heterotic string one gets $SO(10)\times E_8$ in
four dimensions. The standard model can be embedded in $SO(10)$ (trivial lift, {\it i.e.} standard, unlifted $B-L$) or $E_8$ (lift A). If one compactifies the $SO(32)/\mathbb{Z}_2$ heterotic
string, one gets $SO(26)$, in which the standard model can then be embedded via an $SO(10)$ subgroup; this yields lift B. As explained in  \cite{GatoRivera:2010fi} both lift A and
lift B yield, perhaps counter-intuitively, chiral spectra. In the unlifted case, the number of families is typically a multiple of 6, and sometimes 2; for lift A, the family number quantization unit was found to be usually 1, whereas for lift B it was usually 2. 

In this chapter we will apply all these kinds of lifts to permuted Gepner models. This means that we make, when possible, all sorts of known lifts (namely, standard weight lifting and B-L lifting) for the $N=2$ factors that do \textit{not} belong to the sub-block(s) of the permutation orbifold. Note that permutations and lifting act independently: a given minimal
model factor is either unchanged, or lifted, or interchanged with another, identical factor. It may well be possible to construct lifted CFT's for the permutation orbifolds themselves,
but no examples are known, and they are in any case not obtainable by the methods of  \cite{GatoRivera:2009yt}, because there only a single minimal model factor is lifted. There
is one exception to this: there is one known simultaneous lift of two minimal model factors with $k=1$. Conceivably one could apply a permutation to those two identical
factors in combination with this lift. We have however not investigated this possibility.

\section{Comparison with known results}
\label{Section: Comparison}

To compare our results with previous ones on permutation orbifolds \cite{Klemm:1990df,Fuchs:1991vu}, it is  important to understand the
differences in these approaches. These authors first construct the basic Gepner model with all world-sheet and space-time supersymmetry
projections already in place in the left- as well as the right-moving sector.
 
They start from either the diagonal (A-type) invariants of all the minimal models, or the D and E-type
(exceptional) invariants. They then apply a cyclic permutation to the minimal model factors that are identical. They allow for additional
phase symmetries occurring in combination with the permutations. This combined operation
is applied to the full partition function. 

By contrast, we first build an $N=2$ permutation orbifold, then tensor it with other building blocks (either minimal models or other
$N=2$ permutation orbifolds), then impose world-sheet and space-time supersymmetry, but only on the fermionic sector, and consider general simple current modular invariants. 

So the differences can be summarized as follows
\begin{itemize}
\item{In \cite{Klemm:1990df,Fuchs:1991vu} general cyclic $\mathbb{Z}_L$ permutations are considered, while our results are limited to $L=2$.}
\item{In \cite{Klemm:1990df,Fuchs:1991vu} extra phases are modded out in combination with the permutations.}
\item{In \cite{Klemm:1990df,Fuchs:1991vu} permutations of D and E-invariants are considered.}
\item{We only consider permutations of factors with $2 \leq k \leq 10$.}
\item{We consider general simple current invariants.}
\item{We consider not only $(2,2)$ but also $(0,2)$ invariants and breaking of $SO(10)$.}
\end{itemize}

In order to make a comparison we will ignore the last point and focus on $(2,2)$ models. 
Since simple current invariants
include D-invariants as special cases, and because they involve monodromy phases of currents with respect to fields, one might expect
that at least part of the limitations in the second and third point are overcome. Exceptional invariants can be taken into account
in our method by multiplying the simple current modular matrix with an explicit exceptional modular matrix. Indeed, in standard Gepner
models we {\it have} taken them into account, and analysed the class of $(1,16^*,16^*,16^*)$ three-family models \cite{Gepner:1987hi}.
In the present case one could easily use exceptional invariants in those factors that are not permuted. To use permutations of exceptional
invariants we would first have to construct the exceptional MIPF explicitly in the permutation orbifold CFT, which can be done in principle
with a tedious computation. The first point is, however, much harder to overcome, because it would involve extending the BHS construction
to higher cyclic orders.  

Now let us see how the comparison works out in practice. In \cite{Fuchs:1991vu} a table is presented with
all models where cyclic permutations, phase symmetries and cyclic permutations together with phase symmetries have been modded out. For each model the authors give the number of generations $\rm n_{27}$, anti-generations $\rm n_{\overline{27}}$ and singlets $\rm n_{1}$. 
The first two numbers are equal to Hodge numbers of Calabi-Yau manifolds, namely $h_{21}= \rm n_{27}$ and 
$h_{11}=\rm n_{\overline{27}}$.
These quantities are first obtained by using modular invariance of the partition function of the cyclically-orbifolded Gepner models and are then compared with the same quantities derived by using topological arguments applied to the smooth Calabi-Yau manifold after that the singularities have been resolved. 
The number of families is specified by $\rm n_{gen}=n_{27}-n_{\overline{27}}$. The total number of singlets is strongly dependent on the multiplicities of the (descendants) states of the $N=2$ minimal models, which can be read off directly from the character expansions. The singlet number $\rm n_{1}$ turns out to be crucial for differentiating different models with equal $\rm n_{27}$ and $\rm n_{\overline{27}}$. Our comparison is based on these three numbers.  In table (\ref{HodgeNumbers}) we list the values 
we obtained for these three numbers in the cases we considered. Note that these are the numbers obtained without any
simple current extensions or automorphisms. The cases marked with a $*$ are $K_3 \times T_2$ type compactifications
with an $E_7$ spectrum; the numbers that are indicated are the ones obtained after decomposing $E_7$ to $E_6$.

In comparing the A-type invariants without phase symmetries, we get agreement, but in a somewhat unexpected way.
In \cite{Fuchs:1991vu} one-permutation models are not considered, because the authors argue that they always produce the same spectra as unpermuted Gepner models. However, we do manage to build one-permutation models as explained in section \ref{Section: Permutation orbifold of N=2 minimal models}. The only Gepner combinations for which the one-permutation models yield nothing new are the purely-odd combinations. 
Furthermore, the one-permutation orbifolds do indeed yield new results. For example, for the combinations $(2,2,2,2,\langle 2,2\rangle)$
the three numbers are $(90,0,284)$ as opposed to $(90,0,285)$ for the unpermuted case; for $(6,6,\langle6,6\rangle)$ we find 
$(106,2,364)$ as opposed to $(149,1,503)$; for $(\langle 3,3\rangle,10,58)$ we get $(75,27,392)$ as opposed to $(85,25,425)$. These
three example illustrate three distinct situations. In the first example, the only difference with the unpermuted case is that the number
of singlets is reduced by one. In the second example, the Hodge pair $(106,2)$ does occur for a non-trivial simple current invariant
of the tensor product $(6,6,6,6)$, namely $(6_A,6_A,6_A,6_D)$, but with 365 singlets instead of 364. In the
last example the Hodge pair $(75,27)$ does not occur for any simple current MIPF of $(3,3,10,58)$ (the only other combination that
occurs is $(53,41,401)$ plus the mirrors, so that even the Euler number of the permutation orbifold is new)\rlap.\footnote{For a complete
list of Hodge number and singlets of Gepner models see \cite{hodge}.}

In order to make a non-trivial comparison between
our spectra and those of \cite{Fuchs:1991vu} we have to look at Gepner models with {\it two} permutations. It turns out that our spectra (specified by $\rm n_{27}$, $\rm n_{\overline{27}}$ and $\rm n_{1}$) do agree with those of \cite{Fuchs:1991vu}. However, to get the full match, we always have to extend the model by one current. This current is (see section \ref{Section: Permutation orbifold of N=2 minimal models}) the double un-orbifold current, which has the un-orbifold current in each of the two factors corresponding to the permutation orbifold and the identity current in the remaining factors. 
Also in this case we already get new spectra even if we do {\it not} extend by this current. Consider for example $(\langle 6,6\rangle,\langle 6,6\rangle)$. As mentioned above, the $(6,6,6,6)$ gives $(149,1,503)$; the completely unextended spectrum we
get for the $(\langle 6,6\rangle,\langle 6,6\rangle)$ yields $(77,1,269)$; if we extend the two permutation orbifold CFT's by
the current combination $((0,1),(0,1))$ (where $(0,1)$ is the un-orbifold current) we find $(83,3,301)$, which is precisely the result quoted in \cite{Fuchs:1991vu}  for
the permutation orbifold. It is noteworthy that \cite{Fuchs:1991vu} lists a triplet $(77,1,271)$ for the combination
$(6_D,6_D,6_A,6_A)$, which from our perspective is a simple current invariant of $(6,6,6,6)$. Again we see two spectra with a minor
difference only in the number of singlets, which we will comment on below. In  one case we could not make
a comparison, because in \cite{Fuchs:1991vu} no result is listed for $(2,2,\langle 2,2\rangle,\langle 2,2\rangle)$ without extra phases.
In all other cases our results agree with \cite{Fuchs:1991vu}.
The need for extending by a combination of un-orbifold currents suggests that such currents are automatically generated or implicitly
present
in the
formalism used in \cite{Fuchs:1991vu}, for reasons we do not fully understand, but which are presumably related to an interchange
in the order of two operations: permutation and simple current extension. This is also consistent with the fact that these authors
find no new results for single permutations: if an un-orbifold current is automatically present in that case, one inevitably returns to the
unpermuted case. Note that for $(3,\langle3,3\rangle,\langle3,3\rangle)$ we have seen before that the separate un-orbifold current of
each permutation orbifold is automatically present in the chiral algebra, and hence so is the combination of the two. Therefore in this
case we do not have to extend by $(0,(0,1),(0,1))$ to find agreement with \cite{Fuchs:1991vu} because the extension is already automatically present.

Let us now compare the cases with extra phase symmetries.
In almost all cases, using the simple-current formalism, we recover for a given suitably-extended model the same Hodge numbers and the same singlet number as in those spectra where both the phase symmetry and the permutation symmetry have been modded out. In a sense, these phase symmetries correspond to simple current extensions or automorphisms. The only two exceptions, out of the many successful instances, both coming from the $2^6$ Gepner model (nr. 21 of Table II in \cite{Fuchs:1991vu}) with two permutations and phase symmetries, are
\begin{itemize}
\item (21)(43)56, 111100 ($\rm n_{27}=21$, $\rm n_{\overline{27}}=21$, $\rm n_{1}=180$, $\chi=0$),
\item (21)3(54)6, 333111 ($\rm n_{27}=44$, $\rm n_{\overline{27}}=8$, $\rm n_{1}=199$, $\chi=-72$),
\end{itemize}
where the first entry is the permutation orbifold and the second one is the phase symmetry. We were not able to find these two cases using our procedure.

There are a few other cases that we do not have, but for reasons that are easy to understand. Consider model nr. 168 in the same table. It corresponds to the $6^4$ Gepner model. The double permutation that we reproduce is the one labelled as
\begin{itemize}
\item $6_A 6_A 6_A 6_A$: (21)(43) ($\rm n_{27}=83$, $\rm n_{\overline{27}}=3$, $\rm n_{1}=301$, $\chi=-160$).
\end{itemize}
The other two, with D invariants
\begin{itemize}
\item $6_A 6_A 6_D 6_D$: (21)(43) ($\rm n_{27}=45$, $\rm n_{\overline{27}}=1$, $\rm n_{1}=181$, $\chi=-88$),
\item $6_D 6_D 6_D 6_D$: (21)(43) ($\rm n_{27}=35$, $\rm n_{\overline{27}}=3$, $\rm n_{1}=154$, $\chi=-64$),
\end{itemize}
are not present. However these are not comparable with our $6^4$ since they come out of a different construction. In fact, the D invariant is obtained as a simple current automorphism of the $k=6$ Gepner models by the $SU(2)_k$ current $(k,0,0)$ (with $k=6$). This current has spin $h=\frac{k}{4}=\frac32$. In \cite{Fuchs:1991vu} the authors 
consider the permutation of two such $k=6$ models, each with such a simple current automorphism. 
This is different from what happens here. Here, we immediately replace the block by its permutation orbifold; moreover, when we extend it by the current $(T_F,1)$ to build the supersymmetric permutation orbifold, the off-diagonal field $\langle (0,0,0)(6,0,0)\rangle$ with spin $h=\frac{3}{2}$ (the obvious candidate for creating the automorphism invariant) is not a simple current. We expect that the permutation
orbifold of two $6_D$ models is present as an exceptional invariant of $\langle6,6\rangle$.

The spectra mentioned in the last two paragraphs, that were present in \cite{Fuchs:1991vu} but absent in our results, might also be understood
as follows.  As explained in \ref{Section: Permutation orbifold of N=2 minimal models}, one may consider simple current extension
chains of the form
$$ {\rm (CFT)}_0 \rightarrow {\rm (CFT)}_1 \rightarrow {\rm (CFT)}_2 \rightarrow \ldots $$
In this chain, the supersymmetric permutation orbifold is ${\rm (CFT)}_1$. We can use all its simple currents to build
MIPF's, and in particular we find all simple current extensions ${\rm (CFT)}_2$. However there are situations where ${\rm (CFT)}_2$
itself has new simple currents that are exceptional, and whose orbits cannot be fully resolved because we do not have
the complete fixed point resolution formalism for ${\rm (CFT)}_1$ available. Therefore MIPF's generated by such second order
exceptional simple currents cannot be obtained. At best, one could try to get them by explicit computation as exceptional MIPF's
of ${\rm (CFT)}_1$. The problem of unresolvable fixed points occurs precisely for supersymmetric permutation orbifolds when
$k=2\mod 4$, and therefore might be relevant precisely in these examples. 

\LTcapwidth=14truecm
\begin{center}
\vskip .7truecm
\begin{longtable}{|c||c|c|c|}
\caption{{Hodge data for permutation orbifolds of Gepner models.}}\\
\hline
 \multicolumn{1}{|c||}{Tensor product}
& \multicolumn{1}{c|}{$h_{21}$}
& \multicolumn{1}{l|}{$h_{11}$ }
& \multicolumn{1}{c|}{Singlets}\\
\hline
\endfirsthead
\multicolumn{4}{c}%
{{\bfseries \tablename\ \thetable{} {\rm-- continued from previous page}}} \\
\hline 
 \multicolumn{1}{|c||}{model}
& \multicolumn{1}{c|}{$h_{21}$}
& \multicolumn{1}{l|}{$h_{11}$}
& \multicolumn{1}{c|}{Singlets}\\
\hline
\endhead
\hline \multicolumn{4}{|r|}{{Continued on next page}} \\ \hline
\endfoot
\hline \hline
\endlastfoot\hline
\label{HodgeNumbers}
$(1,1,1,1,1,1,\langle 2,2\rangle)$ &    $23^*$&     $23^*$&   177 \\ 
$(1,1,1,1,1,\langle 4,4\rangle)$ &    84 &      0 &   249 \\ 
$(1,1,1,1,\langle 10,10\rangle)$ &    57 &      9 &   248 \\ 
$(1,1,1,1,\langle 2,2\rangle,4)$ &    35 &     11 &   229 \\ 
$(1,1,1,\langle 2,2\rangle,2,2)$ &    $23^*$&     $23^*$&   175 \\ 
$(1,1,1,2,\langle 6,6\rangle)$ &    $23^*$&     $23^*$&   173 \\ 
$(1,1,1,\langle 4,4\rangle,4)$ &    73 &      1 &   242 \\ 
$(1,1,1,\langle 3,3\rangle,8)$ &    $23^*$&     $23^*$&   173 \\ 
$(1,1,1,\langle 2,2\rangle,\langle 2,2\rangle)$ &    $23^*$&     $23^*$&   173 \\ 
$(1,1,2,2,\langle 4,4\rangle)$ &    35 &     11 &   211 \\ 
$(1,1,\langle 2,2\rangle,2,10)$ &    46 &     10 &   234 \\ 
$(1,1,4,\langle 10,10\rangle)$ &    75 &      3 &   279 \\ 
$(1,1,\langle 6,6\rangle,10)$ &    37 &     13 &   211 \\ 
$(1,1,\langle 2,2\rangle,4,4)$ &    51 &      3 &   250 \\ 
$(1,1,\langle 2,2\rangle,\langle 4,4\rangle)$ &    35 &     11 &   209 \\ 
$(1,2,2,\langle 10,10\rangle)$ &    61 &      1 &   251 \\ 
$(1,\langle 2,2\rangle,2,2,4)$ &    61 &      1 &   260 \\ 
$(1,2,4,\langle 6,6\rangle)$ &    51 &      3 &   235 \\ 
$(1,2,\langle 4,4\rangle,10)$ &    62 &      2 &   241 \\ 
$(1,2,\langle 3,3\rangle,58)$ &    41 &     17 &   273 \\ 
$(1,\langle 4,4\rangle,4,4)$ &    84 &      0 &   279 \\ 
$(1,\langle 2,2\rangle,10,10)$ &    89 &      5 &   343 \\ 
$(1,\langle 3,3\rangle,4,8)$ &    41 &      5 &   219 \\ 
$(1,\langle 2,2\rangle,5,40)$ &    35 &     35 &   329 \\ 
$(1,\langle 2,2\rangle,6,22)$ &    68 &      8 &   297 \\ 
$(1,\langle 2,2\rangle,7,16)$ &    43 &     19 &   289 \\ 
$(1,\langle 2,2\rangle,8,13)$ &    27 &     27 &   249 \\ 
$(1,\langle 2,2\rangle,\langle 2,2\rangle,4)$ &    61 &      1 &   259 \\ 
$(\langle 2,2\rangle,2,2,2,2)$ &    90 &      0 &   284 \\ 
$(2,2,2,\langle 6,6\rangle)$ &    73 &      1 &   251 \\ 
$(2,2,\langle 4,4\rangle,4)$ &    51 &      3 &   242 \\ 
$(2,2,\langle 3,3\rangle,8)$ &    41 &      5 &   218 \\ 
$(2,2,\langle 2,2\rangle,\langle 2,2\rangle)$ &    90 &      0 &   283 \\ 
$(2,\langle 10,10\rangle,10)$ &   105 &      3 &   380 \\ 
$(2,\langle 8,8\rangle,18)$ &    79 &      7 &   322 \\ 
$(\langle 2,2\rangle,2,3,18)$ &    65 &      5 &   279 \\ 
$(2,\langle 7,7\rangle,34)$ &    63 &     15 &   312 \\ 
$(\langle 2,2\rangle,2,4,10)$ &    69 &      3 &   265 \\ 
$(\langle 2,2\rangle,2,6,6)$ &    86 &      2 &   297 \\ 
$(2,\langle 2,2\rangle,\langle 6,6\rangle)$ &    73 &      1 &   250 \\ 
$(3,\langle 6,6\rangle,18)$ &    51 &     11 &   254 \\ 
$(3,\langle 5,5\rangle,68)$ &    53 &     29 &   328 \\ 
$(3,\langle 8,8\rangle,8)$ &    99 &      3 &   346 \\ 
$(3,\langle 3,3\rangle,\langle 3,3\rangle)$ &    59 &      3 &   228 \\ 
$(4,4,\langle 10,10\rangle)$ &    94 &      4 &   334 \\ 
$(4,\langle 6,6\rangle,10)$ &    55 &      7 &   238 \\ 
$(4,\langle 5,5\rangle,19)$ &    41 &     17 &   238 \\ 
$(4,\langle 7,7\rangle,7)$ &    66 &      6 &   270 \\ 
$(\langle 5,5\rangle,5,12)$ &    83 &      5 &   308 \\ 
$(\langle 6,6\rangle,6,6)$ &   106 &      2 &   364 \\ 
$(\langle 4,4\rangle,10,10)$ &   101 &      5 &   370 \\ 
$(\langle 3,3\rangle,10,58)$ &    75 &     27 &   392 \\ 
$(\langle 3,3\rangle,12,33)$ &    47 &     31 &   306 \\ 
$(\langle 3,3\rangle,13,28)$ &    97 &     13 &   404 \\ 
$(\langle 3,3\rangle,18,18)$ &   125 &      9 &   490 \\ 
$(\langle 2,2\rangle,3,3,8)$ &    39 &     15 &   249 \\ 
$(\langle 2,2\rangle,4,4,4)$ &    60 &      6 &   285 \\ 
$(\langle 4,4\rangle,5,40)$ &    65 &     17 &   334 \\ 
$(\langle 4,4\rangle,6,22)$ &    70 &     10 &   304 \\ 
$(\langle 4,4\rangle,7,16)$ &    79 &      7 &   308 \\ 
$(\langle 4,4\rangle,8,13)$ &    48 &     12 &   242 \\ 
$(\langle 3,3\rangle,9,108)$ &    69 &     49 &   466 \\ 
$(\langle 6,6\rangle,\langle 6,6\rangle)$ &    77 &      1 &   269 \\ 
$(\langle 2,2\rangle,\langle 4,4\rangle,4)$ &    51 &      3 &   240 \\ 
$(\langle 2,2\rangle,\langle 3,3\rangle,8)$ &    41 &      5 &   216 \\ 
$(\langle 2,2\rangle,\langle 2,2\rangle,\langle 2,2\rangle)$ &    90 &      0 &   282 \\ 
$(1,\langle 2,2\rangle,\langle 10,10\rangle)$ &    61 &      1 &   250 \\ 
$(\langle 4,4\rangle,\langle 10,10\rangle)$ &    75 &      3 &   273 \\ 
$(1,\langle 4,4\rangle,\langle 4,4\rangle)$ &    73 &      1 &   234 \\ 
\end{longtable}
\end{center}

As already mentioned, the list of Hodge numbers and singlets in table (\ref{HodgeNumbers}) is obtained
without any simple current extensions other than those required to get a $(2,2)$ model. The complete list obtained with arbitrary simple currents can be found on the website \cite{hodge}. 

Although the 
results in table (\ref{HodgeNumbers}) are for $(2,2)$ models, the focus of the present chapter was on $(0,2)$ models.
We can compare the results with those of \cite{GatoRivera:2010gv} and ask what permutation orbifolds add. Consider
first the set of $(0,2)$ models closest to $(2,2)$ models, namely those with an $E_6$ gauge symmetry. They are
characterized by the same three numbers ${\rm n}_{27}, {\rm n}_{\overline{27}}$ and  ${\rm n}_1$, but since there
is not necessarily a world-sheet supersymmetry in the bosonic sector they may not have a Calabi-Yau interpretation. For simplicity
we will refer to these  as ``pseudo Hodge pairs" and ``pseudo Hodge triplets".
In
the complete set of standard Gepner models without exceptional invariants we obtained a total of 1418\footnote{For the standard, unpermuted Gepner models, the number of genuine 
Hodge number pairs with world-sheet  supersymmetry in both sectors is 906. A list can be found on the website \cite{hodge}.}
different pseudo Hodge  pairs
and 9604 different pseudo Hodge triplets. For the genuine permutation orbifolds (without extensions by un-orbifold currents)
these numbers are respectively  498 and 3830. Note that some permutation orbifolds with $k > 10$ were not considered.
How many of the permutation orbifold numbers are new? If we combine the data for pseudo Hodge pairs and remove identical ones,
we obtain a total of 1447 pseudo Hodge pairs, so that the total has increased by a mere 29. But if we look at pseudo Hodge triplets, the
increase is much more substantial. This number increases from 9604 to 12145, an increase of 2541 or about $26\%$. We tentatively
conclude that permutation orbifolds mainly give new points in existing moduli spaces. The following observation is further evidence
in that direction.

One remarkable feature of the permutation orbifold spectra is the occurrence of identical Hodge numbers and a number of singlets
that is almost the same. For example, in the set of permutation orbifolds obtained from the $(2,2,2,2,2,2)$ tensor product
we find spectra with (genuine) Hodge numbers $(90,0)$, and either 282, 283, 284 or 285 singlets.  A closer look at the spectrum
reveals what is going on here. We also compute the number of massless vector bosons in these spectra, and it turns out that
this is respectively 2,3,4 and 5 (in addition to those of $E_6$) in these cases. This is consistent with the occurrence of a Higgs mechanism that has made
one or more of the vector bosons heavy by absorbing the corresponding number of singlets. So apparently we are finding
points in the same moduli space, but with a vev for certain moduli fields so that some of the $U(1)$'s are removed. 
This is expected to occur in Gepner models, but it is nice to see this happen entirely within RCFT. The same observation was
made in \cite{Fuchs:1991vu}.
The reduction of the number of $U(1)$'s
by itself has a straightforward reason: each $N=2$ model has an intrinsic $U(1)$, and replacing two minimal models by a
permutation reduces the number of $U(1)$'s by 1. Hence the $(2,2,2,2,2,2)$ model generically has five $U(1)$'s (six, minus one 
combination
that becomes an $E_6$ Cartan-subalgebra generator), and $(\langle 2,2\rangle,\langle 2,2\rangle,\langle 2,2\rangle)$ generically has
only two. However, the number of vector bosons can be larger than that because the simple current MIPF's add extra
generators to the chiral algebra. Indeed, among the MIPF's of $(\langle 2,2\rangle,\langle 2,2\rangle,\langle 2,2\rangle)$ we
do not only find $(90,0,282,2)$ (where the last entry is the number of $U(1)$'s), but also $(90,0,283,3)$ and $(90,0,284,4)$.

\section{Results}
\label{Section: Results}

The CFT approach, based on simple currents extensions, turns out to be extremely powerful. Although we have considered in this work only order-two permutations, the number of new modular invariant partition functions or, equivalently, the number of new spectra for each model is huge, in the order of a few thousands. Simple currents allow us to generate a huge number of four dimensional spectra.

Here we discuss the more phenomenological aspects of our results, considering the breaking of $SO(10)$ into subgroups, including the Standard Model. 
Conceptually this is very similar to work on unpermuted Gepner models presented in 
\cite{GatoRivera:2010gv,GatoRivera:2010xn,GatoRivera:2010fi}, to which we refer for  more detailed  descriptions. In these papers several, mostly empirical, observations were made regarding the resulting
spectra. The main question of interest here is if these observations continue to hold as we extend the
scope of RCFT's considered. 

\subsection{Gauge groups}

Within $SO(10)$, all the simple currents of the conformal field theory constructed out of the Standard Model in the left (bosonic) sector extend the algebra to one of the following gauge groups: $SO(10)$ itself and any of the seven rank-5 subgroups, namely the  Pati-Salam group 
$SU(4)\times SU(2)\times SU(2)$, the Georgi-Glashow GUT group $SU(5)\times U(1)$, two  global realizations of left-right symmetric algebra 
$SU(3)\times SU(2)\times SU(2)\times U(1)$, and three global realizations of the standard model algebra $SU(3)\times SU(2)\times U(1)\times U(1)$. Counted as Lie-algebras there are just five of them, but the last two come in several varieties when we describe them as CFT chiral algebras. These are distinguished by the fractionally charged (here ``charge" refers to unconfined electric charge)
representations that are allowed. For the left-right algebra this can be either $\frac13$ or $\frac16$, 
(we call these ``LR, Q=1/3" and ``LR, Q=1/6" respectively) 
and for the standard model this can be $\frac12$, $\frac13$ or $\frac16$
(SM, Q=1/2, 1/3 or 1/6).  In the string chiral algebra  these different global realizations are distinguished by the presence of certain integer spin currents. If these
currents have conformal weight one, they manifest themselves in the massless spectrum as extra gauge bosons. 
This happens in particular for the highly desirable global group corresponding to the standard model with only integer unconfined electric charge. In this class of heterotic strings this necessarily implies an extension of the standard model to (at least) $SU(5)$. Furthermore, if the standard model gauge group is extended to $SU(5)$, this group cannot be broken by a field-theoretic Higgs mechanism, because the required Higgs scalar, a $(24)$, cannot be massless in the heterotic string spectrum. A heterotic string spectrum contains either these massless vector bosons, or fractionally charged states that forbid the former because they are non-local with respect to them \cite{Schellekens:1989qb} (see also \cite{Wen:1985qj,Athanasiu:1988uj}).

These eight gauge groups are obtained as extensions of the affine Lie algebra $SU(3)_1\times SU(2)_1\times U_{30}$, with a $U(1)$ normalization that gives rise to $SU(5)$-GUT type unification. In general, there is an additional $U(1)$ factor
that corresponds to a gauged $B-L$ symmetry in certain cases. In B-L lifted spectra this $U(1)$ is replaced by a non-abelian group.
In addition, the gauge group consists out of a $U(1)$ factor for each superconformal building block, which sometimes is extended
to a larger group, depending on the MIPF considered. There may also be extensions of the standard model gauge group outside
$SO(10)$, such as $E_6$ or trinification, $SU(3)^3$.
In standard Gepner models there is furthermore an unbroken $E_8$ factor, which
in lifted Gepner models is replaced by certain combinations of abelian and non-abelian groups. In scanning spectra we focus only
on the aforementioned  eight (extended) standard model groups. 

\subsection{MIPF scanning}

Since it is essentially impossible to construct the complete set of distinct MIPF's, we use a random scan. This is done
by choosing 10.000 randomly chosen simple current subgroups ${\cal H}$ (see appendix \ref{Appendix Paper5}) generated by at most three simple currents.  
Furthermore, if the number of distinct torsion matrices $X$ is larger than 100, we make 100 random choices. The entire set
is guaranteed to be mirror symmetric, because for every given spectrum one  can always construct a mirror by multiplying the
MIPF with a simple current MIPF of $SO(10)_1$ that flips the chirality of all spinors. Note that this does not imply anything
about mirror symmetry of an underlying geometrical interpretation. It is a trivial operation on the spectrum that can however be used
to get some idea on the completeness of the scan. 

\subsection{Fractional Charges} 

Fractional charges can appear either in the form of chiral particles, or as vector-like particles  (where ``vector-like" is defined
with respect to the standard model gauge group) or only as massive particles, with masses of order the string scale. If a spectrum has
chiral fractionally charged particles, we reject it after counting it. In nearly all remaining cases the spectrum contains 
massless vector-like fractional charges (unless there is GUT unification). 
We regard such spectra as acceptable at this stage. Since no evidence for fractionally charged
particles exists in nature, with a limit of less than $10^{-20}$ per nucleon \cite{Perl:2009zz}, clearly these vector-like particles will have to acquire a mass.
Furthermore this will almost certainly have to be a huge (GUT scale or string scale) mass, since otherwise their abundance cannot
be credibly expected to be below the experimental limit. This can in principle happen if the vector-like particles couple to moduli that
get a vev. An analysis of existence of couplings is in principle doable in this class of models, although there may be some technical complications
in those cases where no fixed point resolution procedure is available at present (namely  the permutation orbifolds with
$k=2\ {\rm mod}\ 4$).  However,  this analysis is beyond the scope of this work, and we treat spectra with 
vector-like fractional charges as valid candidates, for the time being. Just as in previous work \cite{GatoRivera:2010gv,GatoRivera:2010xn,GatoRivera:2010fi,FF2}, there are extremely rare
occurrences of spectra without any massless fractionally charged particles at all, but we only found examples with an
even number of families. Examples with three families were found in \cite{FF2} by scanning part of the
free-fermion landscape. In the context of orbifold models and Calabi-Yau compactifications, it is known that GUT breaking 
by modding out freely acting discrete symmetries leads to spectra without massless fractional charges (\cite{Witten:1985xc}; see
\cite{Blaszczyk:2009in} for a recent implementation of this idea in the context of the ``heterotic mini-landscape" \cite{Buchmuller:2005jr,Lebedev:2006kn,Lebedev:2008un}). While these models do fit the data on fractional charges, 
the question remains for which fundamental reasons such vacua are preferred over all others, especially if they are much rarer.

In table (\ref{FreqTable}) we display how often four mutually exclusive types of spectra occur in the total sample,
before distinguishing MIPF's. The types are: spectra with chiral, fractionally charged exotics, 
chiral spectra with a GUT gauge group $SU(5)$ or $SO(10)$, non-chiral spectra (no exotics and no families), spectra with $N$ families and massless $SU(5)$ vector bosons and vector-like fractionally charged exotics, and  the same without massless fractionally charged exotics.  For comparison, we include some results based on the data of \cite{GatoRivera:2010gv,GatoRivera:2010xn,GatoRivera:2010fi}\rlap.\footnote{We thank the authors for making their raw data available to us.} All lines refer to Gepner models, except the one labelled ``free fermions". 
The results on free fermions are based on a special class that can be analysed with simple current in a way analogous to Gepner models, as explained in \cite{GatoRivera:2010gv}. It does not represent the entire class of free fermionic models. For other work on this kind of construction, including three family models, see \cite{FF2,FF1} and references therein.

  \vskip 1.truecm
\begin{table}[h]
\begin{center}
~~~~~~~~~~~~~~~~\begin{tabular}{|l|c|c|c|c|c|cl} \hline
Type &  Chiral Exotics & GUT & Non-chiral & $N > 0$ fam. & No frac.\\ \hline \hline
Standard$^*$ & 37.4\% &32.7\% & 20.5 \% &  9.3\% & 0 \\
Standard, perm. & 29.7\% & 33.4 \% & 27.9 \% & 8.9\% & 0 \\ 
Free fermionic &  1.5\% &  2.9\% &   94.4\% &  1.1\% & 0.072\% \\ 
Lifted  & 28.3\% &  18.7\% &   51.9\% & 1.1\% & 0.00051\%\\ 
Lifted, perm. & 26.8\% & 8.9\% & 62.7 \% & 1.6\%  & 0.00078\% \\ 
$({\hbox{B-L}})_{\hbox{\footnotesize Type-A}}^*$ &  21.3\% &  28.0\% &  50.4 \% & 0.3\% & 0.00017\% \\ 
$({\hbox{B-L}})_{\hbox{\footnotesize Type-A}}$, perm. & 22.8\% & 8.1 \% & 69.1 \% & 0.03\% & 0 \\ 
$({\hbox{B-L}})_{\hbox{\footnotesize Type-B}}^*$ &  38.5\% &   8.7\% &  52.1\% &0.6\% & 0 \\ 
$({\hbox{B-L}})_{\hbox{\footnotesize Type-B}}$, perm. & 27.6\% & 7.3 \% & 65.0 \% & 0.1\% & 0 \\ 
\hline
 \end{tabular}
\caption[Relative frequency of various types of spectra]{Relative frequency of various types of spectra. An asterisk indicates that exceptional minimal model MIPF's are included.}
\label{FreqTable}
\end{center}
\end{table}

In table (\ref{MIPFTable}) we specify the absolute number of distinct MIPF's (more precisely, distinct spectra, based on the criteria spelled out in \cite{GatoRivera:2010gv,GatoRivera:2010xn,GatoRivera:2010fi}) with
non-chirally-exotic spectra. The column marked ``Total" specifies the total number of distinct spectra without chiral exotics, the third column lists the number of distinct 3-family spectra and the last column the number of distinct $N$-family spectra, in both cases regardless of gauge group and without modding out mirror symmetry.

 \vskip 1.truecm
\begin{table}[h]
\begin{center}
~~~~~~~~~~~~~~~~\begin{tabular}{|l|c|c|c|c|} \hline
Type &  Total & 3-family & $N$ family, $N>0$ \\ \hline \hline
Standard$^*$ & 927.100 & 1.220 & 369.089   \\
Standard, perm. & 245.821 & 0 & 64.085   \\ 
Free fermionic &  504.312 &   0 &  19.655   \\ 
Lifted  &     3.177.493 &  85.864 &  537.581 \\ 
Lifted, perm. & 601.452 & 4.702   &   54.926 \\ 
$({\hbox{B-L}})_{\hbox{\footnotesize Type-A}}^*$ &  445.978 &  24.203 & 155.425  \\ 
$({\hbox{B-L}})_{\hbox{\footnotesize Type-A}}$, perm. & 155.784 & 778 & 6.758   \\ 
$({\hbox{B-L}})_{\hbox{\footnotesize Type-B}}^*$ & 206949 &0 &   55917 \\ 
$({\hbox{B-L}})_{\hbox{\footnotesize Type-B}}$, perm. & 156.309 & 0 & 6.861  \\ 
\hline
 \end{tabular}
\caption{Total numbers of distinct spectra. }
\label{MIPFTable}
\end{center}
\end{table}

\subsection{Family number}
In this subsection we would like to say something about the distributions of the number of families emerging from the spectra of permuted Gepner models. The common features of all the different cases is that an even number of families is always more favourable than an odd one and these distributions decrease exponentially when the number of families increases.

Figure \ref{famplot_standard} shows the distribution of the number of families for permutation orbifolds of standard Gepner Models. 
All family numbers are even, as is the case for unpermuted Gepner models (we did not include exceptional MIPF's, which provides
the only way to get three families in standard Gepner models). The greatest common denominator $\Delta$ of the number of families
for a given tensor combination displays a similar behavior as  observed in \cite{GatoRivera:2010gv}. Two classes can be distinguished.
Either $\Delta=6$  (or in a few cases a multiple of $6$), or $\Delta=2$ (sometimes 4), but there are no MIPF's with a number of families 
that is a multiple of three. In other words, the set of family numbers occurring in these two classes have no overlap whatsoever. It follows
that in the second class there are no spectra with zero families.  An interesting example is\
$(3,\langle6,6\rangle,18)$. It has no spectra with chiral exotics, all spectra are chiral and have 4, 8, 16, 20, 28, 32, 40 or 56 families, of types
$SO(10)$, Pati-Salam, $SU(5)\times U(1)$ or $SM, Q=1/2$. If we compare this with the unpermuted Gepner model we find some
striking differences. In that case the same group types occur, but now there are spectra with chiral exotics, and the family quantization is
in units of 2, not 4. In  \cite{GatoRivera:2010gv} an intriguing observation was made regarding the occurrence of these two classes. 
The second class was found to occur if all values $k_i$ of the factors in the tensor product are divisible by 3. This observation
also holds for permutation orbifolds, if one uses the values of $k_i$ of the unpermuted theory.

In figure \ref{famplot_lift} we show the family distribution for lifted Gepner models. As expected, this distribution
looks a lot more favourable for three family models.
The number three appears with more or less the same order of magnitude as two or four. However, there are some clear peaks at even
family numbers, which were not visible in the analogous distribution for unpermuted Gepner models presented in  
\cite{GatoRivera:2010xn}. For this reason three families are still suppressed by a factor of 3 to 4 with respect to 2 or 4 families.

B-L lifts give similar results to those presented in  \cite{GatoRivera:2010fi}.
Figure \ref{famplot_liftA} contains the distribution of the number of families for permutation orbifolds of B-L lifted (lift A) Gepner models. Figure \ref{famplot_liftB} contains the same, but for the lift B. Here, odd numbers are completely absent. 
Note that certain group types (namely those without a ``$B-L$" type $U(1)$ factor) cannot occur in chiral spectra in these models, and
that in the  type that do occur the $U(1)$ is replaced by a non-abelian group.

\section{Conclusions}

In this chapter we have considered $\mathbb{Z}_2$ permutation orbifolds of heterotic Gepner models. This should be viewed as an application of the previous chapter \ref{paper4} where $\mathbb{Z}_2$ permutations were studied for $N=2$ minimal models, which are the building blocks in Gepner construction.

Our main conclusion is that these new building blocks work as they should. They can be used on completely equal
footing with all other available ones, which are the $N=2$ minimal models and some free-fermionic building blocks.  We have
checked the combination with minimal models and found full agreement with previous results on permutation
orbifolds whenever they were available. The comparison did bring a few surprises, especially the fact that we were able
to get new spectra for single permutations,  where the old method of \cite{Fuchs:1991vu} gave nothing new.

We were able to go far beyond the old approach by finding many more $(2,2)$ models, as well as new $(0,2)$ models with
$SO(10)$ breaking. We combined all this with heterotic weight lifting and B-L lifting. The main conclusion is that in most respects
all observations 
concerning family number and fractional charges made for minimal models
continue to hold in this new class. Also in this case weight lifting greatly enhances  the
set of three family models in comparison to neighboring numbers. Although this appears to give some entirely new models
(Hodge number pairs that were not seen before), we found additional evidence for the observation of \cite{Fuchs:1991vu} that
many of these models look like additional rational points in existing moduli spaces.

\begin{figure}[p]
\begin{center}
\includegraphics[angle=90,scale=0.70]{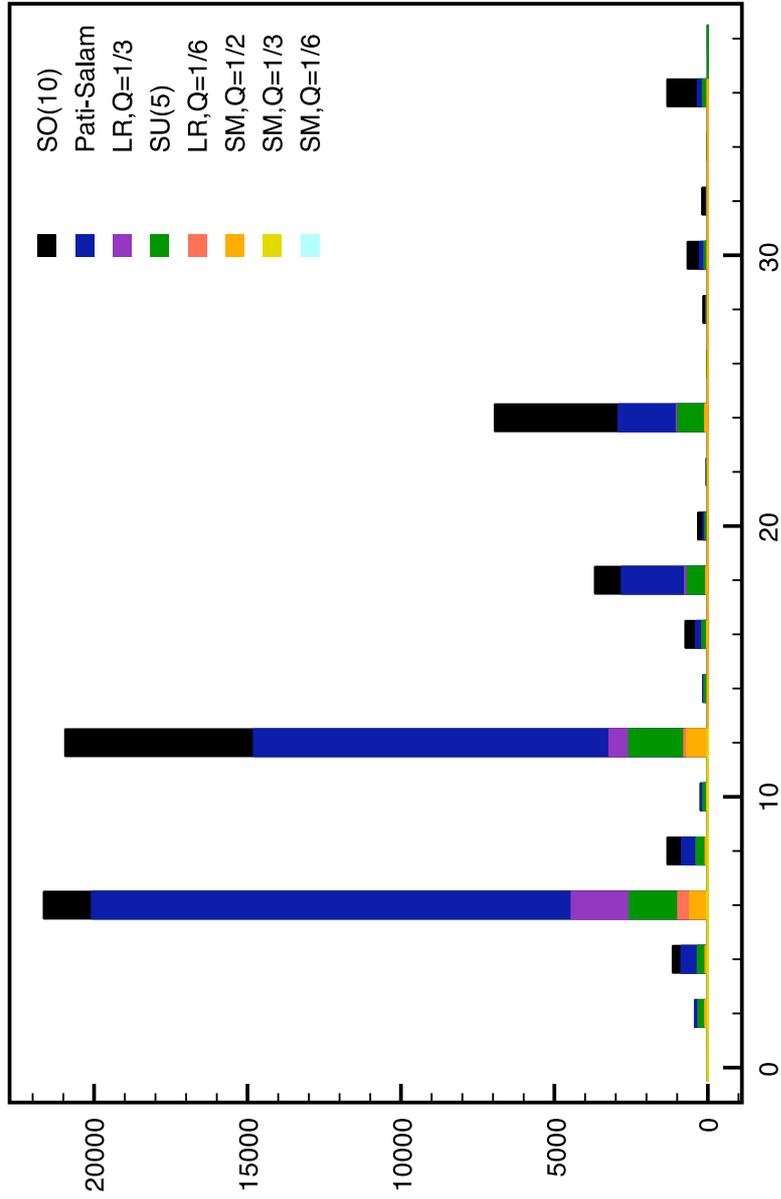}
\caption[Families in standard Gepner models]{\small Distribution of the number of families for permutation orbifolds of standard Gepner Models.}
\label{famplot_standard}
\end{center}
\end{figure}

\begin{figure}[p]
\begin{center}
\includegraphics[angle=90,scale=0.70]{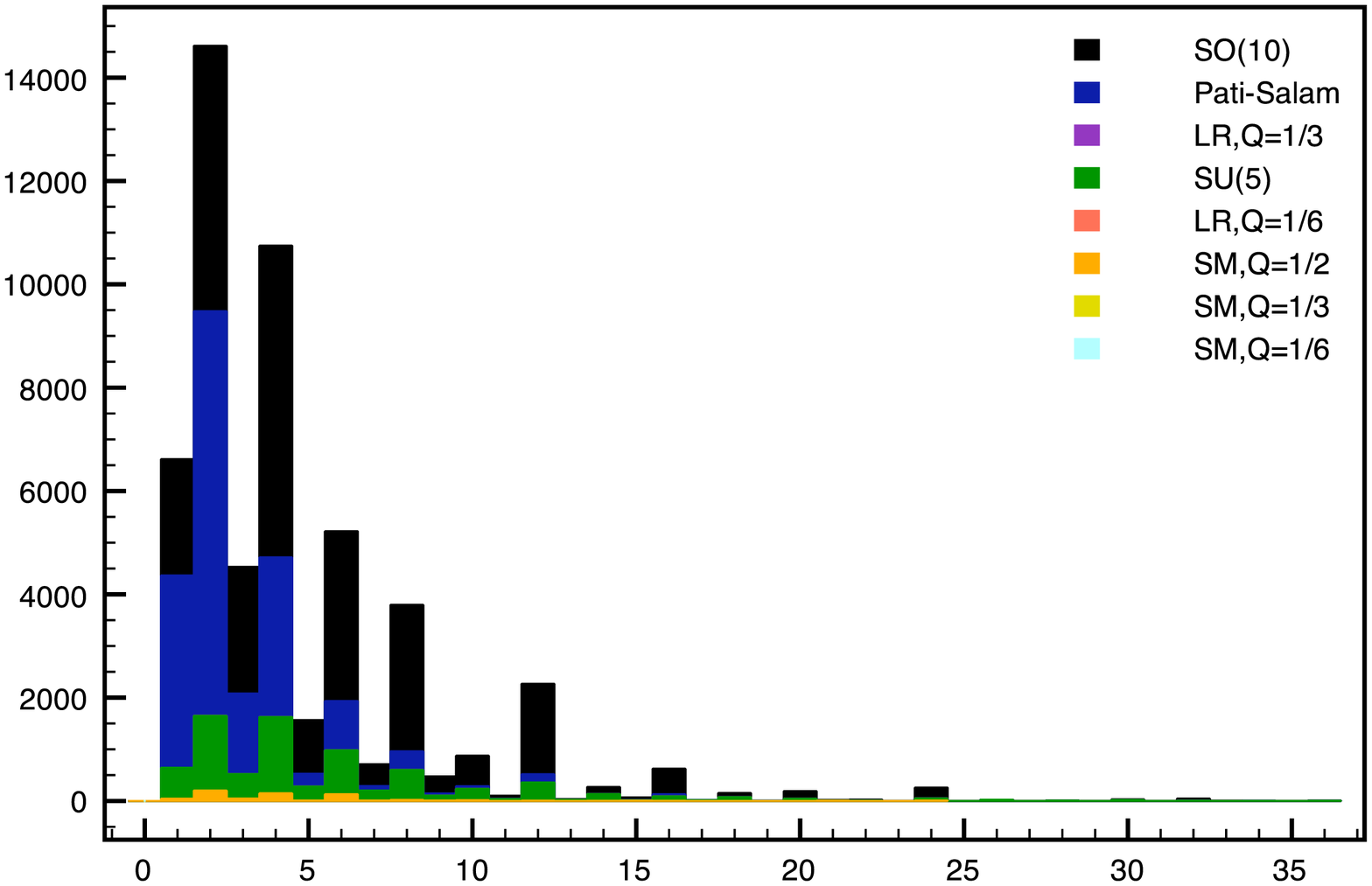}
\caption[Families in permuted Gepner models]{\small Distribution of the number of families for permutation orbifolds of lifted Gepner Models.}
\label{famplot_lift}
\end{center}
\end{figure}

\begin{figure}[p]
\begin{center}
\includegraphics[angle=90,scale=0.70]{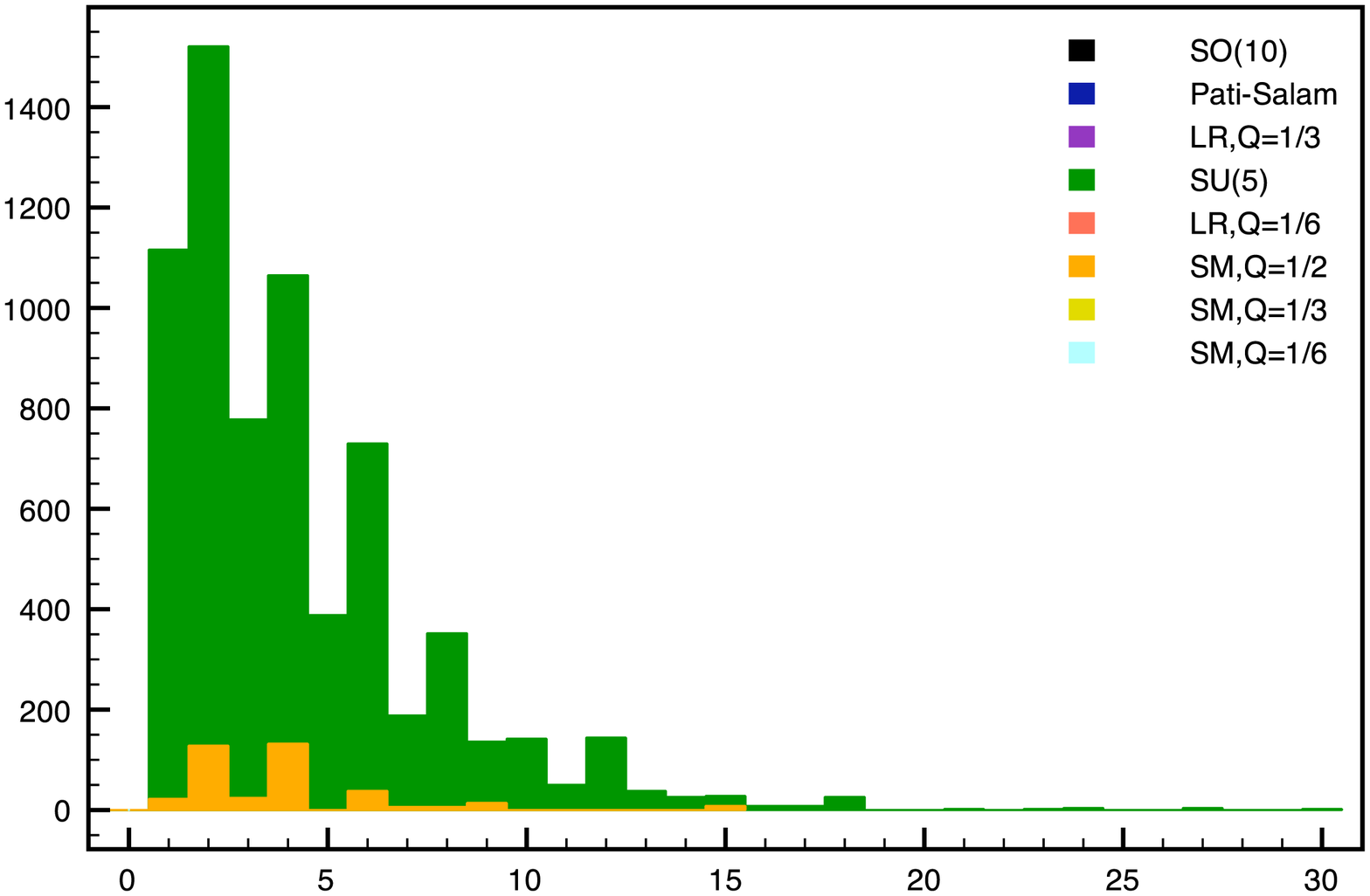}
\caption[Families in lifted permuted Gepner models (Lift A)]{\small Distribution of the number of families for permutation orbifolds of B-L lifted (lift A) Gepner Models.}
\label{famplot_liftA}
\end{center}
\end{figure}

\begin{figure}[p]
\begin{center}
\includegraphics[angle=90,scale=0.70]{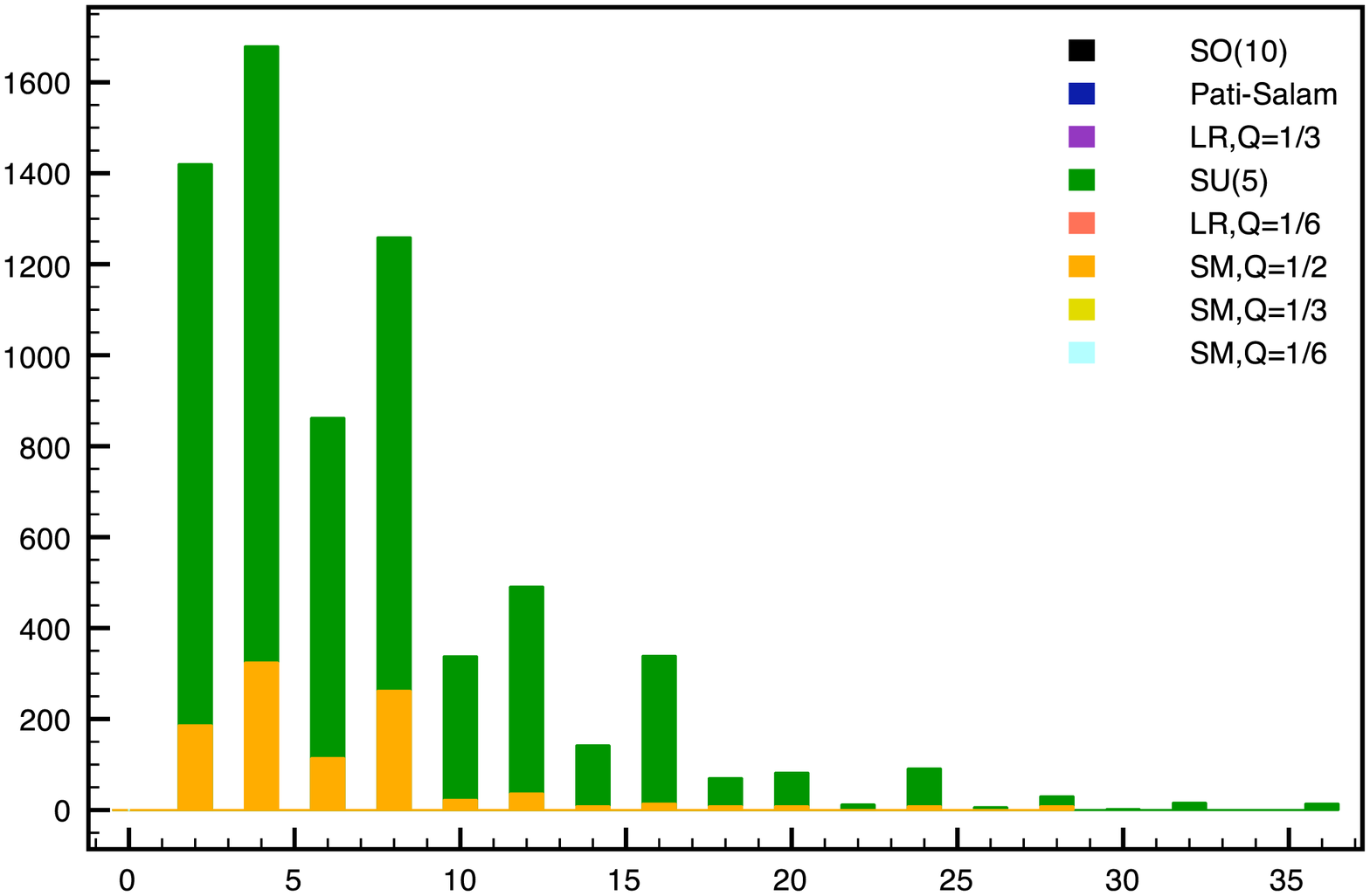}
\caption[Families in lifted permuted Gepner models (Lift B)]{\small Distribution of the number of families for permutation orbifolds of B-L lifted (lift B) Gepner Models.}
\label{famplot_liftB}
\end{center}
\end{figure}

\part{DISCUSSION}
\chapter{Conclusion}
\label{conclusions}

{\flushright
{\small 
\textit{But break, my heart,}\par
\textit{for I must hold my tongue.}\par
\textit{(W. Shakespeare, Hamlet)}\par
}
}

In this thesis we have considered topics in two-dimensional conformal field theory that have relevance in string theory, in particular concerning the phenomenological purposes of describing the low-energy four-dimensional physics as we know it. 

This work consists mainly of two parts. The first part deals with mathematical aspects of $2d$ CFT's. Although technical, the results obtained here have general validity and are applicable to many contexts. 
In details, we start from CFT's which admit permutation symmetries that can be modded out. This happens for example when a CFT is built as a tensor product of smaller CFT's with some identical factors. Then, we look at all possible extensions and provide the full $S$ matrix of the full resulting CFT. Here the word extension is not just a mere undefined mechanism that could in principle be performed in several arbitrary ways, but is instead a very well defined and powerful procedure that allows us to generate many new CFT's out of a single existing one. The crucial ingredient is the existence of particular fields, known as simple currents, in the original CFT: the more simple currents there are, the more new theories can be generated. Simple current extensions exhibit the full power of CFT's. 

All the quantities, in particular characters and modular matrices, of the extended theories can be derived in terms of analogous quantities of the original theory. However this is not straightforward at all when simple currents leave some fields, known as fixed points, unchanged under fusion rules. When this is the case, one has to go through a non-trivial formalism which eventually leads to the determination of the desired quantities. This problem is known as the fixed point resolution. In the first part of this work we have showed how to accomplish this goal in the case of extensions of permutation orbifolds, at least when two factors are identical. This means that we have solved the problem for permutations of two factors, or equivalently for the $\mathbb{Z}_2$ orbifold. Generalizations to any number of factors are much more complicated and a full formalism including fixed point resolution is not available at present. 
The main results of this work were obtained in chapter \ref{paper3}, where in particular the fixed point resolution matrices for the $\mathbb{Z}_2$ orbifold were given (see formula (\ref{ansatz with fixed points})) in full generality, for any CFT $\mathcal{A}$ and simple current $J$. They are expressed in terms of the weight $h_J$ of the simple current and the modular $S$ and $T$ matrices of the original $\mathcal{A}$. Formula (\ref{ansatz with fixed points}) also includes the results of chapters \ref{paper1}-\ref{paper2} as particular cases.

The second part deals with the physical implications of the mathematical methods. Eventually one is interested in computing particle spectra, with maybe $N=1$ space-time supersymmetry in four dimensions. As it is known, in order to have $N=1$ supersymmetry in space-time, four-dimensional string theories must have an internal sector with $N=2$ world-sheet supersymmetry. This is normally achieved by taking tensor products of $N=2$ minimal models and adding some extra constraints. It can happen that this product shows an explicit permutation symmetry: this is the case when some of the factors are identical. Modding out this symmetry is equivalent to replacing the block of identical factors by its permutation orbifold. Also in these physical applications we have considered again the $\mathbb{Z}_2$ orbifold.

The first thing we have done is to look at permutations of $N=2$ minimal models and their possible extensions. In particular, although the $N=2$ factors are by definition supersymmetric, their tensor product is not, since all the fields in a supersymmetric theory should have a well-defined periodicity. In order to make the tensor product supersymmetric one has to extend it by a particular simple current. Similar considerations apply to the $\mathbb{Z}_2$ orbifold: it is not supersymmetric, but a particular simple current extension is enough to make it so. Moreover, a third different simple current extension allows us to recover both the standard tensor product of $N=2$ minimal models from their standard permutation orbifold and the supersymmetric tensor product from their supersymmetric permutation orbifold. 
These facts are illustrated by the box diagrams in section \ref{perm of N=2 mods}. Many surprises show up here, in particular concerning the existence of exceptional simple currents. These were a priori not expected to be there: they have a completely different origin from standard simple currents and arise as a consequence of the extension procedure. Sometimes they also admit fixed points that must be resolved.

Because of their tensor product structure, these CFT's have in general a very large number of simple currents, that in turn can be used to generate a huge amount of new theories. In this spirit, we have constructed thousands of theories with associated particle spectra and studied relevant properties, such as the number of families, fractional charges and gauge groups. One can also modify the construction in several ways, for example by introducing suitable lifts of one of the factors in the tensor product. This in general improves the results about the family number, but leaving undesired fractional charges in the matter content.

\part{APPENDIX}
\appendix
\addcontentsline{toc}{chapter}{Appendix}
\chapter{Facts on $\langle 0,T_F\rangle$-fusions}
\label{Appendix Paper4}

{\flushright
{\small 
\textit{To conceited men, all other men are admirers.}\par
\textit{(A. de Saint-Exup$\acute{e}$ry, The Little Prince)}\par
}
}

\section{Twisted-fields orbits of the $(0,1)$-current}
In this appendix we compute the fusion rules of $\langle 0,T_F\rangle$. Before doing that however we need to prove, as an intermediate result, that in any permutation orbifold the simple current $(0,1)$ (anti-symmetric representation of the identity) always couples a twisted field to its own (un)excited partner, i.e.
\begin{equation}
\widehat{(p,0)} \stackrel{(0,1)}{\leftrightarrow} \widehat{(p,1)}\,.
\end{equation}
To prove this, let us compute the fusion coefficients:
\begin{equation}
(0,1)\cdot \widehat{(p,\xi)} = \sum_K N_{(0,1)\widehat{(p,\xi)}}^{\phantom{(0,1)\widehat{(p,\xi)}}K} (K)\,,
\end{equation}
where the sum runs over all the fields $K$ in the orbifold. By Verlinde's formula \cite{Verlinde:1988sn}:
\begin{eqnarray}
N_{(0,1)\widehat{(p,\xi)}}^{\phantom{(0,1)\widehat{(p,\xi)}}K} 
&=&  \sum_{N} 
\frac{S_{(0,1)N}S_{\widehat{(p,\xi)}N}S_{\phantom{\dagger}N}^{\dagger \phantom{N}K}}{S_{(0,0)N}}=
\nonumber\\  &=&  
\sum_{\langle i,j\rangle}
\frac{S_{(0,1)\langle i,j \rangle}S_{\widehat{(p,\xi)}\langle i,j \rangle}S_{\phantom{\dagger}\langle i,j \rangle}^{\dagger\phantom{\langle i,j \rangle}K}}{S_{(0,0)\langle i,j \rangle}}+
\nonumber\\ &+&  
\sum_{(j,\chi)}
\frac{S_{(0,1)(j,\chi)}S_{\widehat{(p,\xi)}(j,\chi)}S_{\phantom{\dagger}(j,\chi)}^{\dagger\phantom{(j,\chi)}K}}{S_{(0,0)(j,\chi)}}+
\nonumber\\ &+& 
\sum_{\widehat{(j,\xi)}}
\frac{S_{(0,1)\widehat{(j,\chi)}}S_{\widehat{(p,\xi)}\widehat{(j,\chi)}}S_{\phantom{\dagger}\widehat{(j,\chi)}}^{\dagger\phantom{\widehat{(j,\chi)}}K}}{S_{(0,0)\widehat{(j,\chi)}}}
\nonumber\,.
\end{eqnarray}
Now use the orbifold $S$ matrix (\ref{BHS}): the first line automatically vanishes, since $S^{BHS}$ vanishes when one entry is a twisted field and the other one is off-diagonal. The other two lines give
\begin{equation}
N_{(0,1)\widehat{(p,\xi)}}^{\phantom{(0,1)\widehat{(p,\xi)}}K} =
\frac{1}{2}\sum_{\chi=0}^1\sum_j e^{i\pi\chi}\, S_{pj} \cdot S_{(j,\chi)}^{\star\phantom{(j,\chi)}K} -
\frac{1}{2}\sum_{\chi=0}^1\sum_j e^{i\pi(\xi+\chi)}\, P_{pj} \cdot S_{\widehat{(j,\chi)}}^{\star\phantom{\widehat{(j,\chi)}}K}\,.
\nonumber
\end{equation}
The two contributions both vanish if $K$ is of diagonal type or of off-diagonal type, as one can easily verify by using (\ref{BHS}). On the other hand, if $K$ is of twisted type, we find a non-vanishing answer that can be written as
\begin{equation}
N_{(0,1)\widehat{(p,\xi)}}^{\phantom{(0,1)\widehat{(p,\xi)}}\widehat{(k,\eta)}} =
\frac{1}{2}\, \delta_p^k\,(1-e^{i\pi(\xi-\eta)})=\delta_p^k\,\delta_{\xi+1}^\eta\,.
\end{equation}
Here we have used unitarity of the $S$ and $P$ matrices. In other words,
\begin{equation}
(0,1)\cdot \widehat{(p,0)} = \widehat{(p,1)}\,,
\end{equation}
as well as the other way around, being the current $(0,1)$ of order two.

\section{Fusion rules of $\langle 0,T_F\rangle$}
In this section we would like to show that the fusion coefficients of $\langle 0,T_F\rangle$ with itself, before and after the $(T_F,\psi)$-extension, do not depend on the sign choice for the coefficients $A$ and $C$ appearing in the $S^J$ ansatz (\ref{ansatz with fixed points}). In particular, the intrinsic ambiguity related to the freedom of ordering twisted fields (i.e. which one we label by $\chi=0$ and which one by $\chi=1$) should not make any difference in the calculation of the fusion rules. 
The calculation is straightforward and relatively short before making the $(T_F,\psi)$-extension, since it involves only the BHS $S$ matrix: we will describe it in detail. 

However, after taking the $(T_F,\psi)$-extension, the full extended $S$ matrix must be used. This means that the BHS $S$ matrix appears together with the $S^{(T_F,\psi)}$ matrix; moreover, fixed point resolution implies that the fixed points of $(T_F,\psi)$ are split, hence there will be twice their number, while non-fixed points form orbits and only half of them will be independent.
The calculation in this case is lengthy and more involved, so we will only point out where the sign ambiguities mentioned above could (but will not) play a role.

\subsection{Before $(T_F,\psi)$-extension}
The quantity that we want to compute is
\begin{equation}
\langle 0,T_F\rangle \cdot \langle 0,T_F\rangle =
\sum_{K} N_{\langle 0,T_F\rangle\langle 0,T_F\rangle}^{\phantom{\langle 0,T_F\rangle\langle 0,T_F\rangle}K} (K)\,,
\end{equation}
where the sum runs over all the fields $K$ of the permutation orbifold. The quantity $N_{\langle 0,T_F\rangle\langle 0,T_F\rangle}^{\phantom{\langle 0,T_F\rangle\langle 0,T_F\rangle}K}$ is given by Verlinde's formula \cite{Verlinde:1988sn}
\begin{equation}
N_{\langle 0,T_F\rangle\langle 0,T_F\rangle}^{\phantom{\langle 0,T_F\rangle\langle 0,T_F\rangle}K}=
\sum_{N} \frac{S_{\langle 0,T_F\rangle N}S_{\langle 0,T_F\rangle N}S_{\phantom{\dagger}N}^{\dagger \phantom{N}K}}{S_{(0,0)N}}\,.
\end{equation}
Let us start with the case that $K$ is a diagonal field, $K=(k,\chi)$, and use the BHS expression for the orbifold $S$ matrix:
\begin{eqnarray}
N_{\langle 0,T_F\rangle\langle 0,T_F\rangle}^{\phantom{\langle 0,T_F\rangle\langle 0,T_F\rangle}(k,\chi)}&=&
\sum_{m<n}\frac{(S_{0m}S_{T_F, n}+S_{0n}S_{T_F, m})^2\cdot(S^\star_{mk}S^\star_{nk})}{S_{0m}S_{0n}}+\nonumber\\
&+&
\sum_{\phi=0}^1\sum_i\frac{(S_{0i}S_{T_F, i})^2\cdot(\frac{1}{2}S^{\star 2}_{ik})}{(\frac{1}{2}S^2_{0i})}+
0\,.\nonumber
\end{eqnarray}
The zero in the second line comes from the twisted contribution, since from the BHS formula $S_{\langle mn\rangle \widehat{(i,\chi)}}=0$. The sum over $\phi$ gives a factor of $2$ in the diagonal contribution. In the first sum we can use
\begin{equation}
\sum_{m,\,n}=2\sum_{m<n}+\sum_{m=n}\,.
\end{equation}
The sum $\sum_{m=n}$ will cancel the diagonal contribution. Eventually we are left only with three terms coming from expanding the square in the sum over $m$ and $n$. The two sums are now independent and factorize:
\begin{eqnarray}
N_{\langle 0,T_F\rangle\langle 0,T_F\rangle}^{\phantom{\langle 0,T_F\rangle\langle 0,T_F\rangle}(k,\chi)}&=&
\frac{1}{2}\sum_m S^\star_{mk}S_{0m}\sum_n \frac{S^\star_{nk}S^2_{T_F, n}}{S_{0n}}+\nonumber\\
&&
\frac{1}{2}\sum_n S^\star_{nk}S_{0n}\sum_m \frac{S^\star_{mk}S^2_{T_F, m}}{S_{0m}}+\nonumber\\
&&
+\sum_m S^\star_{mk}S_{0m}\sum_n S^\star_{nk}S_{0n}=\nonumber\\
&=&
\delta_{k,0}N_{T_F T_F}^{\phantom{T_F T_F}k} +\delta_{k,T_F}=\nonumber\\
&=&
\delta_{k,0} +\delta_{k,T_F}\,,
\end{eqnarray}
where we have used the fact that $T_F$ has order two, i.e. $N_{T_F T_F}^{\phantom{T_F T_F}k}=\delta_{k,0}$. Note that the answer does not depend on $\chi$.

We can now repeat the same steps in the case that $K$ is off-diagonal, $K=\langle k_1,k_2\rangle$ (with $k_1<k_2$). 
We get:
\begin{equation}
N_{\langle 0,T_F\rangle\langle 0,T_F\rangle}^{\phantom{\langle 0,T_F\rangle\langle 0,T_F\rangle}\langle k_1,k_2\rangle}\propto \delta_{0,k_1}\cdot\delta_{0,k_2}=0\,,
\end{equation}
since $k_1\neq k_2$.

Similarly, for $K$ twisted, $K=\widehat{(k,\chi)}$:
\begin{equation}
N_{\langle 0,T_F\rangle\langle 0,T_F\rangle}^{\phantom{\langle 0,T_F\rangle\langle 0,T_F\rangle}\widehat{(k,\chi)}}=0+0+0=0\,,
\end{equation}
where the first and third contributions vanish because $S_{\langle mn\rangle\widehat{(i,\chi)}}=0$ in the BHS $S$ matrix, while the second one vanishes because $\sum_{\phi=0}^1 e^{i\pi\phi}=0$.

Putting everything together we have the following fusion rules for $\langle 0,T_F\rangle$ with itself before the $(T_F,\psi)$-extension:
\begin{equation}
\langle 0,T_F\rangle \cdot \langle 0,T_F\rangle=
(0,0)+(0,1)+(T_F,0)+(T_F,1)\,.
\end{equation}

\subsection{After $(T_F,\psi)$-extension}
After the extension by $(T_F,\psi)$, the off-diagonal field $\langle 0,T_F\rangle$ becomes a simple current. Moreover, since it is fixed by $(T_F,\psi)$, as well as $(0,\psi)$, it gets split and originates two simple currents, $\langle 0,T_F\rangle_\alpha$ with $\alpha=0,\,1$.

In order to compute the fusion rules between $\langle 0,T_F\rangle_\alpha$ and $\langle 0,T_F\rangle_\beta$, we need to know the full $S$ matrix of the extension. It is given by \cite{Fuchs:1996dd}
\begin{equation}
\tilde{S}_{a_\alpha b_\beta}= {\rm Const} \cdot (S_{ab}+(-1)^{\alpha+\beta}S^{(T_F,\psi)}_{ab})\,.
\end{equation}
Here, $S_{ab}$ is the BHS $S$ matrix and $S^{(T_F,\psi)}_{ab}$ is the fixed-point resolution matrix $S^J$ corresponding to the current $J=(T_F,\psi)$. The overall constant is a group-theoretical factor such that
\begin{equation}
{\rm Const}=
\left\{
\begin{array}{lc}
\frac{1}{2}  & {\rm if\,\,both\,\,}a\,\&\,b\,\,{\rm are\,\,fixed\,\,points}\\
1 & {\rm if\,\,either\,\,}a\,{\rm or}\,b\,\,({\rm not\,\,both})\,\,{\rm is\,\,fixed\,\,point}\\
2 & {\rm if\,\,neither\,\,}a\,\&\,b\,\,{\rm are\,\,fixed\,\,points}
\end{array}
\right.
\end{equation}
As mentioned in chapter \ref{paper4}, $S^{(T_F,\psi)}_{ab}$ in the untwisted sector vanishes, because $T_F$ does not have fixed points.

We want to compute:
\begin{equation}
\langle 0,T_F\rangle_\alpha\cdot\langle 0,T_F\rangle_\beta=
\sum_Q N_{\langle 0,T_F\rangle_\alpha \langle 0,T_F\rangle_\beta}^{\phantom{\langle 0,T_F\rangle_\alpha \langle 0,T_F\rangle_\beta}Q} (Q)\,,
\end{equation}
where
\begin{equation}
N_{\langle 0,T_F\rangle_\alpha \langle 0,T_F\rangle_\beta}^{\phantom{\langle 0,T_F\rangle_\alpha \langle 0,T_F\rangle_\beta}Q}=
\sum_{N} \frac{\tilde{S}_{\langle 0,T_F\rangle_\alpha N}\tilde{S}_{\langle 0,T_F\rangle_\beta N}\tilde{S}_{\phantom{\dagger}N}^{\dagger \phantom{N}Q}}{\tilde{S}_{(0,0)N}}\,.
\end{equation}

Consider $Q$ to be diagonal, $Q=(q,\chi)$. Diagonal fields are never fixed points of $(T_F,\psi)$, hence if the $S^{(T_F,\psi)}$ has at least one diagonal entry it vanishes. Then we have: 
\begin{eqnarray}
N_{\langle 0,T_F\rangle_\alpha \langle 0,T_F\rangle_\beta}^{\phantom{\langle 0,T_F\rangle_\alpha \langle 0,T_F\rangle_\beta}(q,\chi)}
&=&
\sum_{N} \frac{\tilde{S}_{\langle 0,T_F\rangle_\alpha N}\tilde{S}_{\langle 0,T_F\rangle_\beta N}\tilde{S}_{\phantom{\dagger}N}^{\dagger \phantom{N}(q,\chi)}}{\tilde{S}_{(0,0)N}}=\\
&=&
\sum_{\langle mn \rangle} \frac{\tilde{S}_{\langle 0,T_F\rangle_\alpha \langle mn \rangle}\tilde{S}_{\langle 0,T_F\rangle_\beta \langle mn \rangle}\tilde{S}_{\phantom{\dagger}\langle mn \rangle}^{\dagger \phantom{\langle mn \rangle}(q,\chi)}}{\tilde{S}_{(0,0)\langle mn \rangle}}+\nonumber\\
&+&
\sum_{(p,\phi)} \frac{\tilde{S}_{\langle 0,T_F\rangle_\alpha (p,\phi)}\tilde{S}_{\langle 0,T_F\rangle_\beta (p,\phi)}\tilde{S}_{\phantom{\dagger}(p,\phi)}^{\dagger \phantom{(p,\phi)}(q,\chi)}}{\tilde{S}_{(0,0)(p,\phi)}}+\nonumber\\
&+&
\sum_{\gamma=0}^1\sum_{\widehat{(p,\phi)}_\gamma} \frac{\tilde{S}_{\langle 0,T_F\rangle_\alpha \widehat{(p,\phi)}_\gamma}\tilde{S}_{\langle 0,T_F\rangle_\beta \widehat{(p,\phi)}_\gamma}\tilde{S}_{\phantom{\dagger}\widehat{(p,\phi)}_\gamma}^{\dagger \phantom{\widehat{(p,\phi)}_\gamma}(q,\chi)}}{\tilde{S}_{(0,0)\widehat{(p,\phi)}_\gamma}}\,.\nonumber
\end{eqnarray}
Let us stress a few points here. First, the sum over $\langle m,n \rangle$ is symbolic: we must consider both the situations when $\langle m,n \rangle$ is a fixed point of $(T_F,\psi)$ (in which case it will carry an extra label $\langle m,n \rangle_\gamma$, with $\gamma=0$ or $1$) and when it is just an orbit representative (in which case we should not include its partner $\langle T_F\cdot m,T_F\cdot n \rangle$ in the sum in order to avoid double counting). \\
Diagonal fields are always orbit representatives, while twisted fields are always fixed points. In principle, the $S^{(T_F,\psi)}$ matrix can appear in the sums over $\langle m,n \rangle$ and over $\widehat{(p,\phi)}$, but in practice it only appear in the latter, since it vanishes for untwisted-untwisted entries. So the possible ambiguity might play a role only in the last line. Hence let us have a closer look there. For off-diagonal-twisted entries, the BHS $S$ matrix is identically zero, so we can replace $\tilde{S}$ with $S^{(T_F,\psi)}$, up to the overall constant. Using the ansatz (\ref{ansatz with fixed points}), the contribution to the fusion rules from the last line is then
\begin{eqnarray}
&& 2
\sum_{\begin{array}{c}
\widehat{(p,\phi)}\\
\widehat{(p,\phi)}\,\,{\rm f.p.\,of\,\,}(T_F,\psi)
\end{array}}
\frac{
\frac{1}{2}(-1)^{\alpha+\gamma}A S_{0p} \cdot \frac{1}{2}(-1)^{\beta+\gamma}A S_{0p} \cdot C^\star\frac{1}{2}e^{-i\pi\chi} S^{\star}_{pq}
}{
C\frac{1}{2}S_{0p}
}=\nonumber \\
&=&
\frac{1}{2}\,A^2\,\frac{C^\star}{C}\,(-1)^{\alpha+\beta}\,e^{-i\pi\chi}\,\,\cdot
\sum_{\begin{array}{c}
\widehat{(p,\phi)}\\
\widehat{(p,\phi)}\,\,{\rm f.p.\,of\,\,}(T_F,\psi)
\end{array}}
S_{0p}S^\star_{pq}\nonumber\,.
\end{eqnarray}
The sum over the twisted fixed points $\widehat{(p,\phi)}$ of $(T_F,\psi)$ contains the $\psi$ dependence. What is relevant for our discussion here is the prefactor: there is no ambiguity related to different choices for the coefficients $A$ and $C$, since changing $A\rightarrow -A$ and/or $C\rightarrow -C$ would not alter the result.

The full and exact calculation of the fusion rules  after the $(T_F,\psi)$-extension is too lengthy to be repeated and we will not do it here. In particular, the cases when $Q$ is off-diagonal or twisted are not very relevant, since then the fusion coefficients vanish identically, as one can check numerically. We simply state the outcome of the complete calculation:
\begin{itemize}
\item For the $(T_F,0)$-extension:
\begin{eqnarray}
\langle 0,T_F\rangle_\alpha \cdot \langle 0,T_F\rangle_\alpha &=& (0,1) \qquad \alpha=0,\,1\nonumber\\
\langle 0,T_F\rangle_\alpha \cdot \langle 0,T_F\rangle_\beta &=& (0,0) \qquad \alpha\neq\beta\,;
\end{eqnarray}
hence $\langle 0,T_F\rangle_\alpha$ is of order four, being $(0,1)\cdot(0,1)=(0,0)$, so it cannot be a supersymmetry current.
\item For the $(T_F,1)$-extension:
\begin{eqnarray}
\langle 0,T_F\rangle_\alpha \cdot \langle 0,T_F\rangle_\alpha &=& (0,0) \qquad \alpha=0,\,1\nonumber\\
\langle 0,T_F\rangle_\alpha \cdot \langle 0,T_F\rangle_\beta &=& (0,1) \qquad \alpha\neq\beta\,;
\end{eqnarray}
hence $\langle 0,T_F\rangle_\alpha$ is of order two, as a supersymmetry current should be.
\end{itemize}
Note that in both cases only a particular diagonal field contributes to the fusion rules, namely the identity, as one could have expected because of the order two of $T_F$.


\chapter{MIPF's and tables}
\label{Appendix Paper5}

{\flushright
{\small 
\textit{But the conceited man did not hear him.}\par
\textit{Conceited people never hear anything but praise.}\par
\textit{(A. de Saint-Exup$\acute{e}$ry, The Little Prince)}\par
}
}

\section{Simple current invariants}
\label{Section: Simple current invariants}
Consider a simple current $J$ or order $N$, i.e. $J^N=1$. Define the \textit{monodromy parameter} $r$ as
\begin{equation}
h_J=\frac{r(N-1)}{2N}\quad {\rm mod}\,\,\mathbb{Z}\,.
\end{equation}
Also, define the \textit{monodromy charge} $Q_J(\Phi)$ of $\Phi$ w.r.t. $J$ as
\begin{equation}
Q_J(\Phi)=h_J+h_{\Phi}-h_{J\phi}\quad {\rm mod}\,\,\mathbb{Z}\,.
\end{equation}
The monodromy charge takes values $t/N$, with $t\in\mathbb{Z}$. The current $J$ organizes fields into orbits $(\Phi,J\Phi,\dots,J^d\Phi)$, where $d$ (not necessarily equal to $N$) is a divisor of $N$. On each orbit the monodromy charge is
\begin{equation}
Q_J(J^n\Phi)=\frac{t+rn}{N}\quad {\rm mod}\,\,\mathbb{Z}\,.
\end{equation}

If a simple current $J$ exists in a (rational) CFT, and if it satisfies the condition that $N$ times its conformal weight is an integer\rlap,\footnote{This is sometimes called the ``effective center condition" and eliminates for example the odd level simple currents of $A_1$, which have
order two, but quarter-integer spins.}
then it is known how to associate a modular invariant partition function to it. Suppose that the current $J$ has integer spin and order $N$. Then a MIPF is given by
\begin{equation}
\label{MIPF}
Z(\tau,\bar{\tau})=\sum_{k,\,l}\bar{\chi}_k(\bar{\tau})M_{kl}(J)\chi_l(\tau)\,.
\end{equation}
One way of expressing $M_{kl}(J)$ is \cite{Kreuzer:1993tf,Schellekens:1989wx}:
\begin{equation}
\label{M_kl}
M_{kl}(J)=\sum_{p=1}^N \delta(\Phi_k,J^p \Phi_l)\cdot\delta^1(\hat{Q}_J(\Phi_k)+\hat{Q}_J(\Phi_l))\,
\end{equation}
where $\delta^1(x)=1$ for $x=0$ mod $\mathbb{Z}$ and $\hat{Q}$ is defined on $J$-orbits as
\begin{equation}
\hat{Q}_J(J^n\Phi)=\frac{(t+rn)}{2N} \quad{\rm mod}\,\,\mathbb{Z}\,.
\end{equation}
Morally speaking, $\hat{Q}$ is half the monodromy charge. Formula (\ref{M_kl}) defines a modular invariant partition function, since it commutes with the $S$ and $T$ modular matrices, as shown in \cite{Schellekens:1989dq}.
The set of all the simple currents forms an abelian group $\mathcal{G}$ under fusion multiplication. It is always a product of cyclic factors generated by a (conventionally chosen) complete subset of independent simple currents.

The foregoing associates a modular invariant partition function with a single simple current. One can construct even more of them
by multiplying the matrices $M$. The most general simple current MIPF associated with a given subgroup of $\mathcal{G}$ can be
obtained as follows \cite{GatoRivera:2010gv,Kreuzer:1993tf}. Choose a subgroup of $\mathcal{G}$ denoted $\mathcal{H}$, such that
each element satisfies the effective center condition $N h_J \in \mathbb{Z}$.  Its generators are simple currents $J_s$, $s=1,\dots,k$ for some $k$. 
They have relative monodromies $Q_{J_s}(J_t)=R_{st}$. Take any matrix $X$ that satisfies the equation
\begin{equation}
X+X^T=R\,.
\end{equation}
The matrix $X$ (called the {\it torsion matrix}) determines the multiplicities $M_{ij}$ according to
\begin{equation}
\label{M=nr of sols}
M_{ij}(\mathcal{H},X)={\rm nr.\,\,of\,\,solutions}\,\,K\,\,{\rm to\,\,the\,\,conditions:}
\end{equation}
\begin{itemize}
\item $j=Ki$, $K\in\mathcal{H}$.
\item $Q_M(i)+X(M,K)=0$  mod 1 for all $M\in\mathcal{H}$\,.
\end{itemize}
Here $X(K,J)$ is defined in terms of the generating current $J_s$ as
\begin{equation}
 X(K,J)\equiv \sum_{s,t}n_s m_t X_{st}\,,
\end{equation}
with $K=(J_1)^{n_1}\dots(J_k)^{n_k}$ and $J=(J_1)^{m_1}\dots(J_k)^{m_k}$.

\subsection{A small theorem}
In this subsection we prove the following theorem.

\begin{theorem}
The following statements are true.\\
i) If a simple current $J$ is local w.r.t. any other current $K$, i.e. $Q_K(J)=0$ (mod $\mathbb{Z}$), then $M_{JJ}(K)\neq 0$.\\
ii) For a simple current $J$, which is local w.r.t. any other current $K$, $M_{J0}(K)=M_{0J}(K)$. 
In particular, if $M_{J0}(K)\neq 0$, then also $M_{0J}(K)\neq 0$.
\end{theorem}
\begin{proof}
For the proof we use the statement (\ref{M=nr of sols}).\\
Let us start with {\it i)} and consider $M_{JJ}(K,X)$. The first condition in (\ref{M=nr of sols}) has only one solution, namely $K=0$. The second condition reads
\begin{equation}
Q_M(J)+X(M,0)=0 \qquad \forall M
\end{equation}
and is always true, because the two terms vanish separately. This proves that $M_{JJ}(K)\neq 0$.

Point {\it ii)} goes as follows. Consider $M_{0J^c}(K,X)$. There is again only one solution to the first condition, namely $K=J^c$. The second condition reads
\begin{equation}
Q_M(0)+X(M,J^c)=0\,.
\end{equation}
The first term vanishes by hypothesis, while the second is either zero (in which case $M_{0J^c}(K,X)\neq0$) or non-zero (in which case $M_{0J^c}(K,X)=0$). \\
Similarly, look at $M_{J0}(K,X)$. There is again only one solution to the first condition, namely $K=J^c$. The second condition reads
\begin{equation}
Q_M(0)+X(M,J^c)=0\,.
\end{equation}
The first term vanishes by hypothesis, while the second is either zero or non-zero. In any case, the same condition holds for both $M_{0J^c}(K,X)$ and $M_{J0}(K,X)$. This implies that $M_{J0}(K,X)=M_{0J^c}(K,X)$ (note that these matrix elements can only be 0 or 1).
By closure of the algebra, and because $J^c$ is always a
power of $J$, we may replace $J^c$ by $J$ in this relation.
\end{proof}

Consider now the permutation orbifold. This theorem applies in particular to the un-orbifold current $J=(0,1)$ when coupled to any other current $K$, which is either a standard (diagonal) or an exceptional (off-diagonal) one. In fact, using the same procedure as we did in chapter \ref{paper1} to compute the simple current and fixed point structure of the permutation orbifold, one can show that
\begin{eqnarray}
N_{(J,\psi)\langle p,q\rangle}^{\phantom{(J,\psi)\langle p,q\rangle}\langle p',q'\rangle} &=&
N_{Jp}^{\phantom{Jp}p'}N_{Jq}^{\phantom{Jq}q'}+N_{Jp}^{\phantom{Jp}q'}N_{Jq}^{\phantom{Jq}p'}\,,\nonumber\\
N_{(J,\psi)(i,\chi)}^{\phantom{(J,\psi)(i,\chi)}(i',\chi')} &=&
\frac{1}{2} N_{Ji}^{\phantom{Ji}i'} (N_{Ji}^{\phantom{Ji}i'}+e^{i\pi(\psi+\chi-\chi')})\,.\nonumber
\end{eqnarray}
Hence, for the current $(J,\psi)=(0,1)$, we have
\begin{equation}
N_{(0,1)\langle p,q\rangle}^{\phantom{(0,1)\langle p,q\rangle}\langle p',q'\rangle} =
\delta_p^{p'}\delta_q^{q'}+\delta_p^{q'}\delta_q^{p'}=\delta_p^{p'}\delta_q^{q'}\,,\nonumber
\end{equation}
namely $\langle p,q\rangle$ must be fixed by $(0,1)$ in order for this to be non-zero (recall that $p<q$ and $p'<q'$), and
\begin{equation}
N_{(0,1)(i,\chi)}^{\phantom{(0,1)(i,\chi)}(i',\chi')} =
\frac{1}{2} \delta_i^{i'}(\delta_i^{i'}-e^{i\pi(\chi-\chi')})\,\nonumber
\end{equation}
which is non-zero only if $i=i'$ and $\chi\neq\chi'$ (recall that we can think of $\chi$ as defined mod $2$). Equivalently, in the fusion language:
\begin{equation}
(0,1)\cdot \langle p,q\rangle=\langle p,q\rangle\quad,\qquad
(0,1)\cdot (i,\chi)=(i,\chi+1)\,.
\end{equation}
This implies that $(0,1)$ has zero monodromy charge w.r.t. any other current, since
\begin{eqnarray}
Q_{\langle p,q\rangle}\big( (0,1) \big) &=&
h_{\langle p,q\rangle}+h_{(0,1)}-h_{\langle p,q\rangle}=0\quad{\rm mod}\,\,\mathbb{Z}\,,\nonumber\\
Q_{(i,\chi)}\big( (0,1) \big) &=&
h_{(i,\chi)}+h_{(0,1)}-h_{(i,\chi+1)}=0\quad{\rm mod}\,\,\mathbb{Z}\,.\nonumber
\end{eqnarray}

Now, the un-orbifold current $(0,1)$ has order two, hence $J^c=J$ and $M_{J0}(K,X)=M_{0J}(K,X)$. This also implies that its existence on left-moving sector is guaranteed by its existence on the right-moving sector (and vice-versa).

\subsection{Summary of results}
\label{Appendix Paper5: tables}

Here we present four tables summarizing the results on the number of families for the
standard, heterotic weight lifted, B-L lifted (lift A) and B-L lifted (lift B) cases. These tables contain
information about spectra in which the un-orbifold current is not allowed in the chiral algebra. This
means that these are genuine permutation orbifold spectra. By inspection, we do indeed find that these
spectra are usually different than those obtained in the unpermuted case.
In the columns we specify
respectively the tensor combination, the greatest common divisor $\Delta$ of the number of families for all MIPF's of
the tensor product and the maximal net number of families encountered. In the next column we indicate which of the seven
$SO(10)$ subgroups occur, with the labelling 
\begin{itemize}
\item{ 0: SM,  Q=1/6} 
\item{ 1: SM,  Q=1/3} 
\item{ 2: SM,  Q=1/2} 
\item{ 3: LR, Q=1/6} 
\item{ 4: $SU(5)\times U(1)$} 
\item{ 5: LR, Q=1/3} 
\item{ 6: Pati-Salam. }
\end{itemize}
Since $SO(10)$ can always occur there is no need to indicate it. In \cite{GatoRivera:2010gv} a simple criterion was derived
to determine which subgroups can occur in each standard Gepner model. The allowed subgroups for permutation orbifolds
of Gepner models are a subset of these. In some cases, such as $(\langle 5,5\rangle,5,12)$, some of the subgroups cannot be realized.
In the column labelled ``Exotics" we indicate if, for a given tensor product, spectra with chiral exotics occur. Note that 
in most cases the absence of such spectra is a consequence of the fact that only GUT gauge groups occur. In the
next columns we list
the number of distinct  three family and in column 6
the number of distinct $N$-family $(N>0)$ spectra.
In these tables
only cases with $\Delta > 0$ are shown. If a permutation orbifold seems to be missing, than either it is a  permutation
for $k > 10$, or it is a purely odd tensor product for which all permutations are trivial, or it has only non-chiral spectra and
hence $\Delta=0$ and there are no chiral exotics.  
In column 1 of the second table, $\langle \cal{A},\cal{A}\rangle$  denotes the permutation
orbifold of CFT ${\cal A}$, a hat indicates the lifted CFT, and a tilde indicates the second lift of a CFT.
It turns out that in the only permutation orbifold with two distinct
lifts of the same factor, $(\hat 5,\langle5,5\rangle,12)$ and  $(\tilde 5,\langle5,5\rangle,12)$, $\Delta=0$ in both cases,
which is why a tilde never occurs in the tables. The last column indicates which percentage of the spectra has no mirror.
Since mirror symmetry is exact in the full set, this gives an indication of how close our random scan is to a full enumeration.

\LTcapwidth=14truecm
\begin{center}
\vskip .7truecm
\begin{longtable}{|c||c|c|c|c|c|c|c|}
\caption{{Results for standard Gepner models}}\\
\hline
 \multicolumn{1}{|c||}{model}
& \multicolumn{1}{c|}{$\Delta$}
& \multicolumn{1}{l|}{Max. }
& \multicolumn{1}{c|}{Groups }
& \multicolumn{1}{l|}{Exotics }
& \multicolumn{1}{c|}{3 family}
& \multicolumn{1}{c|}{$N$ fam.}  
& \multicolumn{1}{c|}{Missing} \\ 
\hline
\endfirsthead
\multicolumn{8}{c}%
{{\bfseries \tablename\ \thetable{} {\rm-- continued from previous page}}} \\
\hline 
 \multicolumn{1}{|c||}{model}
& \multicolumn{1}{c|}{$\Delta$}
& \multicolumn{1}{l|}{Max. }
& \multicolumn{1}{c|}{Groups }
& \multicolumn{1}{l|}{Exotics}
& \multicolumn{1}{c|}{3 family}
& \multicolumn{1}{c|}{$N$ fam.}  
& \multicolumn{1}{c|}{Missing} \\ 
\hline
\endhead
\hline \multicolumn{8}{|r|}{{Continued on next page}} \\ \hline
\endfoot
\hline \hline
\endlastfoot\hline
\label{StandardSummary}
$(1,1,1,1,1,\langle 4,4\rangle)$ & 6 & 84 & 3,5,6 & Yes & 0 & 342 &    6.14\% \\ 
$(1,1,1,1,\langle 10,10\rangle)$ & 6 & 48 & 3,5,6 & Yes & 0 & 124 &    4.84\% \\ 
$(1,1,1,1,\langle 2,2\rangle,4)$ & 6 & 48 & 3,5,6 & Yes & 0 & 75 &    6.67\% \\ 
$(1,1,1,\langle 4,4\rangle,4)$ & 6 & 84 & 3,5,6 & Yes & 0 & 2717 &   22.89\% \\ 
$(1,1,2,2,\langle 4,4\rangle)$ & 6 & 24 & 3,5,6 & Yes & 0 & 106 &    0.00\% \\ 
$(1,1,\langle 2,2\rangle,2,10)$ & 6 & 48 & 3,5,6 & Yes & 0 & 662 &    7.70\% \\ 
$(1,1,4,\langle 10,10\rangle)$ & 6 & 72 & 3,5,6 & Yes & 0 & 493 &    7.10\% \\ 
$(1,1,\langle 6,6\rangle,10)$ & 12 & 24 & 3,5,6 & Yes & 0 & 63 &    0.00\% \\ 
$(1,1,\langle 2,2\rangle,4,4)$ & 6 & 48 & 3,5,6 & Yes & 0 & 226 &    6.19\% \\ 
$(1,1,\langle 2,2\rangle,\langle 4,4\rangle)$ & 12 & 24 & 3,5,6 & Yes & 0 & 73 &    6.85\% \\ 
$(1,2,2,\langle 10,10\rangle)$ & 6 & 60 & 3,5,6 & Yes & 0 & 191 &    4.71\% \\ 
$(1,\langle 2,2\rangle,2,2,4)$ & 12 & 60 & 3,5,6 & Yes & 0 & 363 &    3.31\% \\ 
$(1,2,4,\langle 6,6\rangle)$ & 12 & 48 & 3,5,6 & Yes & 0 & 57 &    3.51\% \\ 
$(1,2,\langle 4,4\rangle,10)$ & 6 & 60 & 3,5,6 & Yes & 0 & 1605 &   14.08\% \\ 
$(1,2,\langle 3,3\rangle,58)$ & 6 & 24 & 0,1,2,3,4,5,6 & Yes & 0 & 102 &    0.00\% \\ 
$(1,\langle 4,4\rangle,4,4)$ & 6 & 84 & 3,5,6 & Yes & 0 & 5605 &    6.57\% \\ 
$(1,\langle 2,2\rangle,10,10)$ & 6 & 84 & 3,5,6 & Yes & 0 & 989 &    6.47\% \\ 
$(1,\langle 3,3\rangle,4,8)$ & 6 & 36 & 0,1,2,3,4,5,6 & Yes & 0 & 37 &    0.00\% \\ 
$(1,\langle 2,2\rangle,6,22)$ & 6 & 60 & 3,5,6 & Yes & 0 & 985 &    3.25\% \\ 
$(1,\langle 2,2\rangle,7,16)$ & 12 & 48 & 3,5,6 & Yes & 0 & 41 &    0.00\% \\ 
$(1,\langle 2,2\rangle,\langle 2,2\rangle,4)$ & 12 & 60 & 3,5,6 & Yes & 0 & 165 &    0.61\% \\ 
$(\langle 2,2\rangle,2,2,2,2)$ & 6 & 90 & 6 & Yes & 0 & 1849 &    5.19\% \\ 
$(2,2,2,\langle 6,6\rangle)$ & 12 & 72 & 6 & Yes & 0 & 245 &    0.00\% \\ 
$(2,2,\langle 4,4\rangle,4)$ & 6 & 48 & 3,5,6 & Yes & 0 & 250 &    0.00\% \\ 
$(2,2,\langle 3,3\rangle,8)$ & 6 & 36 & 2,4,6 & Yes & 0 & 55 &    0.00\% \\ 
$(2,2,\langle 2,2\rangle,\langle 2,2\rangle)$ & 6 & 90 & 6 & Yes & 0 & 1580 &    1.58\% \\ 
$(2,\langle 10,10\rangle,10)$ & 6 & 102 & 3,5,6 & Yes & 0 & 328 &    0.00\% \\ 
$(2,\langle 8,8\rangle,18)$ & 6 & 72 & 2,4,6 & Yes & 0 & 316 &    0.00\% \\ 
$(\langle 2,2\rangle,2,3,18)$ & 6 & 60 & 2,4,6 & Yes & 0 & 780 &    4.36\% \\ 
$(2,\langle 7,7\rangle,34)$ & 12 & 48 & 3,5,6 & Yes & 0 & 9 &    0.00\% \\ 
$(\langle 2,2\rangle,2,4,10)$ & 6 & 66 & 3,5,6 & Yes & 0 & 1550 &    3.81\% \\ 
$(\langle 2,2\rangle,2,6,6)$ & 6 & 84 & 6 & Yes & 0 & 1735 &    3.80\% \\ 
$(2,\langle 2,2\rangle,\langle 6,6\rangle)$ & 12 & 72 & $SO(10)$ only & No & 0 & 219 &    0.00\% \\ 
$(3,\langle 6,6\rangle,18)$ & 4 & 56 & 2,4,6 & No & 0 & 232 &    0.00\% \\ 
$(3,\langle 5,5\rangle,68)$ & 24 & 24 & 4 & No & 0 & 18 &    0.00\% \\ 
$(3,\langle 8,8\rangle,8)$ & 6 & 96 & 2,4,6 & Yes & 0 & 1909 &    1.41\% \\ 
$(3,\langle 3,3\rangle,\langle 3,3\rangle)$ & 2 & 56 & 4 & No & 0 & 126 &    0.00\% \\ 
$(4,4,\langle 10,10\rangle)$ & 6 & 90 & 5 & Yes & 0 & 188 &    0.00\% \\ 
$(4,\langle 6,6\rangle,10)$ & 12 & 48 & 5 & Yes & 0 & 70 &    0.00\% \\ 
$(4,\langle 5,5\rangle,19)$ & 12 & 24 & 5 & Yes & 0 & 6 &    0.00\% \\ 
$(4,\langle 7,7\rangle,7)$ & 12 & 60 & 5 & Yes & 0 & 11 &    0.00\% \\ 
$(\langle 5,5\rangle,5,12)$ & 6 & 78 & $SO(10)$ only & No & 0 & 44 &    0.00\% \\ 
$(\langle 6,6\rangle,6,6)$ & 2 & 104 & 6 & Yes & 0 & 1230 &    0.00\% \\ 
$(\langle 4,4\rangle,10,10)$ & 6 & 96 & 3,5,6 & Yes & 0 & 693 &    0.72\% \\ 
$(\langle 3,3\rangle,10,58)$ & 6 & 60 & 0,1,2,3,4,5,6 & Yes & 0 & 97 &    0.00\% \\ 
$(\langle 3,3\rangle,12,33)$ & 2 & 20 & 4 & No & 0 & 30 &    0.00\% \\ 
$(\langle 3,3\rangle,13,28)$ & 6 & 84 & 1,4,5 & Yes & 0 & 587 &    0.00\% \\ 
$(\langle 3,3\rangle,18,18)$ & 2 & 116 & 2,4,6 & Yes & 0 & 681 &    0.00\% \\ 
$(\langle 2,2\rangle,3,3,8)$ & 6 & 48 & 2,4,6 & Yes & 0 & 332 &    3.61\% \\ 
$(\langle 2,2\rangle,4,4,4)$ & 6 & 54 & 5 & Yes & 0 & 75 &    0.00\% \\ 
$(\langle 4,4\rangle,5,40)$ & 6 & 48 & 3,5,6 & Yes & 0 & 98 &    0.00\% \\ 
$(\langle 4,4\rangle,6,22)$ & 6 & 60 & 3,5,6 & Yes & 0 & 440 &    0.00\% \\ 
$(\langle 4,4\rangle,7,16)$ & 6 & 72 & 3,5,6 & Yes & 0 & 271 &    0.00\% \\ 
$(\langle 4,4\rangle,8,13)$ & 6 & 48 & 0,1,2,3,4,5,6 & Yes & 0 & 180 &    0.00\% \\ 
$(\langle 3,3\rangle,9,108)$ & 2 & 28 & 4 & No & 0 & 30 &    0.00\% \\ 
$(\langle 6,6\rangle,\langle 6,6\rangle)$ & 4 & 80 & $SO(10)$ only & No & 0 & 152 &    0.00\% \\ 
$(\langle 2,2\rangle,\langle 4,4\rangle,4)$ & 6 & 48 & 5 & Yes & 0 & 103 &    0.00\% \\ 
$(\langle 2,2\rangle,\langle 3,3\rangle,8)$ & 6 & 36 & 4 & No & 0 & 37 &    0.00\% \\ 
$(\langle 2,2\rangle,\langle 2,2\rangle,\langle 2,2\rangle)$ & 6 & 90 & $SO(10)$ only  & No & 0 & 224 &    1.34\% \\ 
$(1,\langle 2,2\rangle,\langle 10,10\rangle)$ & 12 & 60 & 3,5,6 & Yes & 0 & 155 &    0.00\% \\ 
$(\langle 4,4\rangle,\langle 10,10\rangle)$ & 6 & 72 & 5 & Yes & 0 & 142 &    0.00\% \\ 
$(1,\langle 4,4\rangle,\langle 4,4\rangle)$ & 6 & 84 & 3,5,6 & Yes & 0 & 848 &    0.83\% \\ 
\end{longtable}
\end{center}

\LTcapwidth=14truecm
\begin{center}
\vskip .7truecm
\begin{longtable}{|c||c|c|c|c|c|c|c|}
\caption{{Results for lifted Gepner models}}\\
\hline
 \multicolumn{1}{|c||}{model}
& \multicolumn{1}{c|}{$\Delta$}
& \multicolumn{1}{l|}{Max. }
& \multicolumn{1}{c|}{Groups }
& \multicolumn{1}{l|}{Exotics }
& \multicolumn{1}{c|}{3 family}
& \multicolumn{1}{c|}{$N$ fam.}  
& \multicolumn{1}{c|}{Missing} \\ 
\hline
\endfirsthead
\multicolumn{8}{c}%
{{\bfseries \tablename\ \thetable{} {\rm-- continued from previous page}}} \\
\hline 
 \multicolumn{1}{|c||}{model}
& \multicolumn{1}{c|}{$\Delta$}
& \multicolumn{1}{l|}{Max. }
& \multicolumn{1}{c|}{Groups }
& \multicolumn{1}{l|}{Exotics}
& \multicolumn{1}{c|}{3 family}
& \multicolumn{1}{c|}{$N$ fam.}  
& \multicolumn{1}{c|}{Missing} \\ 
\hline
\endhead
\hline \multicolumn{8}{|r|}{{Continued on next page}} \\ \hline
\endfoot
\hline \hline
\endlastfoot\hline
\label{HWLSummary}
$(\widehat{1},1,1,1,1,\langle 4,4\rangle)$ & 3 & 33 & 3,5,6 & Yes & 45 & 205 &   16.10\% \\ 
$(\widehat{1},1,1,1,\langle 10,10\rangle)$ & 3 & 24 & 3,5,6 & Yes & 0 & 39 &    2.56\% \\ 
$(\widehat{1},1,1,1,\langle 2,2\rangle,4)$ & 3 & 18 & 3,5,6 & Yes & 16 & 50 &   14.00\% \\ 
$(\widehat{1},1,1,\langle 4,4\rangle,4)$ & 3 & 33 & 3,5,6 & Yes & 549 & 1016 &   28.54\% \\ 
$(\widehat{1},1,2,2,\langle 4,4\rangle)$ & 3 & 12 & 3,5,6 & Yes & 17 & 60 &    0.00\% \\ 
$(\widehat{1},1,\langle 2,2\rangle,2,10)$ & 3 & 24 & 3,5,6 & Yes & 123 & 283 &    7.42\% \\ 
$(\widehat{1},1,4,\langle 10,10\rangle)$ & 3 & 24 & 3,5,6 & Yes & 33 & 206 &    7.77\% \\ 
$(\widehat{1},1,\langle 6,6\rangle,10)$ & 6 & 6 & 3,5,6 & Yes & 0 & 15 &    0.00\% \\ 
$(\widehat{1},1,\langle 2,2\rangle,4,4)$ & 3 & 24 & 3,5,6 & Yes & 34 & 237 &    4.64\% \\ 
$(\widehat{1},1,\langle 2,2\rangle,\langle 4,4\rangle)$ & 6 & 12 & 3,5,6 & Yes & 0 & 38 &    0.00\% \\ 
$(\widehat{1},2,2,\langle 10,10\rangle)$ & 12 & 24 & 3,5,6 & Yes & 0 & 18 &    0.00\% \\ 
$(\widehat{1},\langle 2,2\rangle,2,2,4)$ & 6 & 24 & 3,5,6 & Yes & 0 & 71 &    7.04\% \\ 
$(\widehat{1},2,\langle 3,3\rangle,58)$ & 1 & 8 & 0,1,2,3,4,5,6 & Yes & 2 & 40 &    0.00\% \\ 
$(\widehat{1},\langle 2,2\rangle,10,10)$ & 6 & 24 & 3,5,6 & Yes & 0 & 105 &    5.71\% \\ 
$(\widehat{1},\langle 3,3\rangle,4,8)$ & 2 & 8 & 0,1,2,3,4,5,6 & Yes & 0 & 23 &    0.00\% \\ 
$(\widehat{1},\langle 2,2\rangle,6,22)$ & 3 & 24 & 3,5,6 & Yes & 58 & 281 &    5.69\% \\ 
$(\widehat{1},\langle 2,2\rangle,7,16)$ & 6 & 12 & 3,5,6 & Yes & 0 & 13 &    0.00\% \\ 
$(\widehat{2},\langle 2,2\rangle,2,2,2)$ & 1 & 36 & 6 & Yes & 587 & 10481 &    6.91\% \\ 
$(\widehat{2},2,2,\langle 6,6\rangle)$ & 2 & 36 & 6 & Yes & 0 & 595 &    0.17\% \\ 
$(\widehat{2},2,\langle 3,3\rangle,8)$ & 1 & 10 & 2,4,6 & Yes & 3 & 51 &    0.00\% \\ 
$(\widehat{2},2,\langle 2,2\rangle,\langle 2,2\rangle)$ & 2 & 24 & 6 & Yes & 0 & 807 &    1.73\% \\ 
$(\widehat{2},\langle 10,10\rangle,10)$ & 4 & 8 & 3,5,6 & Yes & 0 & 24 &    0.00\% \\ 
$(\widehat{2},\langle 8,8\rangle,18)$ & 1 & 12 & 2,4,6 & Yes & 6 & 85 &    0.00\% \\ 
$(\widehat{2},\langle 2,2\rangle,3,18)$ & 1 & 24 & 2,4,6 & Yes & 26 & 225 &    4.00\% \\ 
$(\widehat{2},\langle 2,2\rangle,4,10)$ & 2 & 24 & 3,5,6 & Yes & 0 & 89 &    3.37\% \\ 
$(\widehat{2},\langle 2,2\rangle,6,6)$ & 1 & 24 & 6 & Yes & 9 & 305 &    1.97\% \\ 
$(\widehat{3},\langle 6,6\rangle,18)$ & 2 & 20 & 2,4,6 & No & 0 & 85 &    0.00\% \\ 
$(\widehat{3},\langle 5,5\rangle,68)$ & 12 & 12 & 4 & No & 0 & 4 &    0.00\% \\ 
$(\widehat{3},\langle 8,8\rangle,8)$ & 1 & 48 & 2,4,6 & Yes & 146 & 1709 &    0.06\% \\ 
$(\widehat{3},\langle 3,3\rangle,\langle 3,3\rangle)$ & 1 & 28 & 4 & No & 11 & 80 &    0.00\% \\ 
$(\widehat{4},4,\langle 10,10\rangle)$ & 2 & 16 & 5 & Yes & 0 & 47 &    0.00\% \\ 
$(\widehat{4},\langle 7,7\rangle,7)$ & 1 & 3 & 5 & Yes & 2 & 6 &    0.00\% \\ 
$(\widehat{6},\langle 6,6\rangle,6)$ & 1 & 8 & 6 & Yes & 7 & 85 &    0.00\% \\ 
$(\langle 3,3\rangle,10,\widehat{58})$ & 1 & 6 & 0,1,2,3,4,5,6 & Yes & 2 & 11 &    0.00\% \\ 
$(\langle 3,3\rangle,\widehat{12},33)$ & 2 & 6 & 4 & No & 0 & 5 &    0.00\% \\ 
$(\langle 3,3\rangle,\widehat{13},28)$ & 1 & 30 & 1,4,5 & Yes & 29 & 170 &    0.00\% \\ 
$(\langle 2,2\rangle,\widehat{3},3,8)$ & 1 & 24 & 2,4,6 & Yes & 14 & 479 &    4.38\% \\ 
$(\langle 2,2\rangle,\widehat{4},4,4)$ & 2 & 24 & 5 & Yes & 0 & 111 &    0.90\% \\ 
$(\langle 4,4\rangle,\widehat{6},22)$ & 8 & 8 & 3,5,6 & Yes & 0 & 5 &    0.00\% \\ 
$(\langle 4,4\rangle,\widehat{8},13)$ & 4 & 16 & 0,1,2,3,4,5,6 & Yes & 0 & 17 &    0.00\% \\ 
$(\langle 3,3\rangle,\widehat{9},108)$ & 2 & 6 & 4 & No & 0 & 6 &    0.00\% \\ 
$(\widehat{4},\langle 2,2\rangle,\langle 4,4\rangle)$ & 4 & 20 & 5 & Yes & 0 & 53 &    0.00\% \\ 
$(\langle 2,2\rangle,\langle 3,3\rangle,\widehat{8})$ & 2 & 6 & 4 & No & 0 & 11 &    0.00\% \\ 
$(1,1,1,\widehat{4},\langle 4,4\rangle)$ & 1 & 24 & 3,5,6 & Yes & 78 & 859 &   20.84\% \\ 
$(1,1,\widehat{2},2,\langle 4,4\rangle)$ & 1 & 6 & 3,5,6 & Yes & 0 & 75 &    4.00\% \\ 
$(1,1,\widehat{2},\langle 2,2\rangle,10)$ & 1 & 24 & 3,5,6 & Yes & 20 & 323 &    8.05\% \\ 
$(1,1,\widehat{4},\langle 10,10\rangle)$ & 2 & 16 & 3,5,6 & Yes & 0 & 191 &    4.71\% \\ 
$(1,1,\langle 2,2\rangle,\widehat{4},4)$ & 2 & 32 & 3,5,6 & Yes & 0 & 297 &    5.05\% \\ 
$(1,\widehat{2},2,\langle 10,10\rangle)$ & 1 & 16 & 3,5,6 & Yes & 28 & 262 &    0.00\% \\ 
$(1,\langle 2,2\rangle,2,2,\widehat{4})$ & 4 & 32 & 3,5,6 & Yes & 0 & 118 &    8.47\% \\ 
$(1,\widehat{2},4,\langle 6,6\rangle)$ & 2 & 24 & 3,5,6 & Yes & 0 & 77 &    0.00\% \\ 
$(1,\widehat{2},\langle 4,4\rangle,10)$ & 1 & 16 & 3,5,6 & Yes & 147 & 1160 &    8.71\% \\ 
$(1,\widehat{2},\langle 3,3\rangle,58)$ & 1 & 8 & 0,1,2,3,4,5,6 & Yes & 3 & 56 &    0.00\% \\ 
$(1,\widehat{4},\langle 4,4\rangle,4)$ & 1 & 24 & 3,5,6 & Yes & 425 & 3219 &    6.28\% \\ 
$(1,\langle 3,3\rangle,\widehat{4},8)$ & 2 & 10 & 0,1,2,3,4,5,6 & Yes & 0 & 27 &    0.00\% \\ 
$(1,\langle 2,2\rangle,\widehat{6},22)$ & 1 & 48 & 3,5,6 & Yes & 46 & 645 &    3.41\% \\ 
$(2,2,\widehat{4},\langle 4,4\rangle)$ & 4 & 20 & 3,5,6 & Yes & 0 & 93 &    0.00\% \\ 
$(2,2,\langle 3,3\rangle,\widehat{8})$ & 2 & 8 & 2,4,6 & Yes & 0 & 29 &    0.00\% \\ 
$(\langle 2,2\rangle,2,\widehat{3},18)$ & 2 & 30 & 2,4,6 & Yes & 0 & 380 &    2.63\% \\ 
$(\langle 2,2\rangle,2,\widehat{4},10)$ & 4 & 24 & 3,5,6 & Yes & 0 & 107 &    0.93\% \\ 
$(\langle 2,2\rangle,2,\widehat{6},6)$ & 2 & 48 & 6 & Yes & 0 & 477 &    2.73\% \\ 
$(3,\widehat{8},\langle 8,8\rangle)$ & 1 & 32 & 2,4,6 & Yes & 24 & 480 &    0.00\% \\ 
$(\langle 5,5\rangle,5,\widehat{12})$ & 1 & 6 & $SO(10)$ only & No & 2 & 8 &    0.00\% \\ 
$(\langle 2,2\rangle,3,3,\widehat{8})$ & 1 & 18 & 2,4,6 & Yes & 6 & 116 &    0.00\% \\ 
$(\langle 4,4\rangle,8,\widehat{13})$ & 2 & 14 & 0,1,2,3,4,5,6 & Yes & 0 & 20 &    0.00\% \\ 
$(1,\widehat{2},\langle 2,2\rangle,2,4)$ & 2 & 24 & 3,5,6 & Yes & 0 & 1092 &    4.85\% \\ 
$(1,2,\langle 3,3\rangle,\widehat{58})$ & 2 & 12 & 0,1,2,3,4,5,6 & Yes & 0 & 11 &    0.00\% \\ 
$(1,\langle 3,3\rangle,4,\widehat{8})$ & 2 & 6 & 0,1,2,3,4,5,6 & Yes & 0 & 10 &    0.00\% \\ 
\end{longtable}
\end{center}

\LTcapwidth=14truecm
\begin{center}
\vskip .7truecm
\begin{longtable}{|c||c|c|c|c|c|c|c|}
\caption{{Results for B-L lifted (lift A) Gepner models}}\\
\hline
 \multicolumn{1}{|c||}{model}
& \multicolumn{1}{c|}{$\Delta$}
& \multicolumn{1}{l|}{Max. }
& \multicolumn{1}{c|}{Groups }
& \multicolumn{1}{l|}{Exotics }
& \multicolumn{1}{c|}{3 family}
& \multicolumn{1}{c|}{$N$ fam.}  
& \multicolumn{1}{c|}{Missing} \\ 
\hline
\endfirsthead
\multicolumn{8}{c}%
{{\bfseries \tablename\ \thetable{} {\rm-- continued from previous page}}} \\
\hline 
 \multicolumn{1}{|c||}{model}
& \multicolumn{1}{c|}{$\Delta$}
& \multicolumn{1}{l|}{Max. }
& \multicolumn{1}{c|}{Groups }
& \multicolumn{1}{l|}{Exotics}
& \multicolumn{1}{c|}{3 family}
& \multicolumn{1}{c|}{$N$ fam.}  
& \multicolumn{1}{c|}{Missing} \\ 
\hline
\endhead
\hline \multicolumn{8}{|r|}{{Continued on next page}} \\ \hline
\endfoot
\hline \hline
\endlastfoot\hline
\label{BLASummary}
$(1,2,\langle 3,3\rangle,58)$ & 1 & 6 & 0,1,2,3,4,5,6 & Yes & 4 & 33 &    0.00\% \\ 
$(1,\langle 3,3\rangle,4,8)$ & 2 & 6 & 0,1,2,3,4,5,6 & Yes & 0 & 12 &    0.00\% \\ 
$(2,2,\langle 3,3\rangle,8)$ & 2 & 8 & 2,4,6 & Yes & 0 & 25 &    0.00\% \\ 
$(2,\langle 8,8\rangle,18)$ & 1 & 12 & 2,4,6 & Yes & 14 & 90 &    0.00\% \\ 
$(\langle 2,2\rangle,2,3,18)$ & 2 & 18 & 2,4,6 & Yes & 0 & 364 &    3.85\% \\ 
$(3,\langle 6,6\rangle,18)$ & 2 & 14 & 2,4,6 & No & 0 & 84 &    0.00\% \\ 
$(3,\langle 5,5\rangle,68)$ & 6 & 6 & 4 & No & 0 & 12 &    0.00\% \\ 
$(3,\langle 8,8\rangle,8)$ & 1 & 30 & 2,4,6 & Yes & 238 & 1799 &    0.11\% \\ 
$(3,\langle 3,3\rangle,\langle 3,3\rangle)$ & 1 & 18 & 4 & No & 15 & 84 &    0.00\% \\ 
$(\langle 3,3\rangle,10,58)$ & 1 & 10 & 0,1,2,3,4,5,6 & Yes & 1 & 19 &    0.00\% \\ 
$(\langle 3,3\rangle,12,33)$ & 2 & 4 & 4 & No & 0 & 6 &    0.00\% \\ 
$(\langle 3,3\rangle,13,28)$ & 1 & 9 & 1,4,5 & Yes & 57 & 346 &    0.00\% \\ 
$(\langle 3,3\rangle,18,18)$ & 1 & 14 & 2,4,6 & Yes & 30 & 200 &    0.00\% \\ 
$(\langle 2,2\rangle,3,3,8)$ & 1 & 12 & 2,4,6 & Yes & 30 & 246 &    0.00\% \\ 
$(\langle 4,4\rangle,8,13)$ & 2 & 8 & 0,1,2,3,4,5,6 & Yes & 0 & 49 &    0.00\% \\ 
$(\langle 3,3\rangle,9,108)$ & 2 & 4 & 4 & No & 0 & 6 &    0.00\% \\ 
$(\langle 2,2\rangle,\langle 3,3\rangle,8)$ & 2 & 6 & 4 & No & 0 & 12 &    0.00\% \\ 
\end{longtable}
\end{center}

\LTcapwidth=14truecm
\begin{center}
\vskip .7truecm
\begin{longtable}{|c||c|c|c|c|c|c|c|}
\caption{{Results for B-L lifted (lift B) Gepner models}}\\
\hline
 \multicolumn{1}{|c||}{model}
& \multicolumn{1}{c|}{$\Delta$}
& \multicolumn{1}{l|}{Max. }
& \multicolumn{1}{c|}{Groups }
& \multicolumn{1}{l|}{Exotics }
& \multicolumn{1}{c|}{3 family}
& \multicolumn{1}{c|}{$N$ fam.}  
& \multicolumn{1}{c|}{Missing} \\ 
\hline
\endfirsthead
\multicolumn{8}{c}%
{{\bfseries \tablename\ \thetable{} {\rm-- continued from previous page}}} \\
\hline 
 \multicolumn{1}{|c||}{model}
& \multicolumn{1}{c|}{$\Delta$}
& \multicolumn{1}{l|}{Max. }
& \multicolumn{1}{c|}{Groups }
& \multicolumn{1}{l|}{Exotics}
& \multicolumn{1}{c|}{3 family}
& \multicolumn{1}{c|}{$N$ fam.}  
& \multicolumn{1}{c|}{Missing} \\ 
\hline
\endhead
\hline \multicolumn{8}{|r|}{{Continued on next page}} \\ \hline
\endfoot
\hline \hline
\endlastfoot\hline
\label{BLBSummary}
$(1,2,\langle 3,3\rangle,58)$ & 2 & 10 & 0,1,2,3,4,5,6 & Yes & 0 & 32 &    0.00\% \\ 
$(1,\langle 3,3\rangle,4,8)$ & 2 & 10 & 0,1,2,3,4,5,6 & Yes & 0 & 10 &    0.00\% \\ 
$(2,2,\langle 3,3\rangle,8)$ & 2 & 14 & 2,4,6 & Yes & 0 & 34 &    0.00\% \\ 
$(2,\langle 8,8\rangle,18)$ & 2 & 16 & 2,4,6 & Yes & 0 & 108 &    0.00\% \\ 
$(\langle 2,2\rangle,2,3,18)$ & 2 & 36 & 2,4,6 & Yes & 0 & 476 &    3.99\% \\ 
$(3,\langle 6,6\rangle,18)$ & 4 & 28 & 2,4,6 & No & 0 & 82 &    0.00\% \\ 
$(3,\langle 5,5\rangle,68)$ & 8 & 16 & 4 & No & 0 & 12 &    0.00\% \\ 
$(3,\langle 8,8\rangle,8)$ & 2 & 56 & 2,4,6 & Yes & 0 & 1781 &    0.00\% \\ 
$(3,\langle 3,3\rangle,\langle 3,3\rangle)$ & 2 & 32 & 4 & No & 0 & 81 &    0.00\% \\ 
$(\langle 3,3\rangle,10,58)$ & 2 & 18 & 0,1,2,3,4,5,6 & Yes & 0 & 18 &    0.00\% \\ 
$(\langle 3,3\rangle,12,33)$ & 2 & 8 & 4 & No & 0 & 6 &    0.00\% \\ 
$(\langle 3,3\rangle,13,28)$ & 2 & 18 & 1,4,5 & Yes & 0 & 322 &    0.00\% \\ 
$(\langle 3,3\rangle,18,18)$ & 2 & 26 & 2,4,6 & Yes & 0 & 191 &    0.00\% \\ 
$(\langle 2,2\rangle,3,3,8)$ & 2 & 24 & 2,4,6 & Yes & 0 & 226 &    0.00\% \\ 
$(\langle 4,4\rangle,8,13)$ & 4 & 16 & 0,1,2,3,4,5,6 & Yes & 0 & 45 &    0.00\% \\ 
$(\langle 3,3\rangle,9,108)$ & 2 & 10 & 4 & No & 0 & 6 &    0.00\% \\ 
$(\langle 2,2\rangle,\langle 3,3\rangle,8)$ & 2 & 10 & 4 & No & 0 & 10 &    0.00\% \\ 
\end{longtable}
\end{center}

\cleardoublepage

\pagestyle{plain}

\cleardoublepage


\cleardoublepage

\phantomsection
\addcontentsline{toc}{chapter}{Summary}
\chapter*{Summary}

{\flushright
{\small 
\textit{Had the routine of our life at this}\par
\textit{place been known to the world,}\par
\textit{we should have been regarded as madmen}\par
\textit{- although, perhaps, as madmen of a harmless nature.}\par
\textit{(E. A. Poe, The Murders In The Rue Morgue)}\par
}
}

\section*{Quantum Field Theory and the Standard Model}
The general lesson in Physics to be learned from the XX century is that the world at extremely large distances as well as at extremely short distances is not described anymore by the Galilean physics that we are used to since the XVI-XVII centuries. In these regimes, in fact, other effects become relevant, either of geometrical or of probabilistic origin.

Large distance physics is captured by Einstein's theory of gravity. Within this framework, the gravitational constant is still $G$, but  a few unifications occur: space and time are unified into a single entity, the space-time; mass and energy become equivalent; the speed of light $c$ becomes a universal constant, with the same value in any reference frame. Einstein's theory also tells us how matter propagates in curved space-times and at the same time how space-time is curved by matter. It is also used to study the motion of planets, various kinds of black holes and the universe itself.

Short distance physics on the other hand is governed by quantum mechanics. The unity of action is $\hbar$ and another unification occurs: particles are described as waves. Consequently, strange and classically-impossible phenomena can happen. For example, particles cannot be localized anymore at a given position, instead they only have a probability of being around that position and are in principle spread out through large domains with some probability distribution; similarly, one cannot give them a specific momentum, since also momenta follow probability distributions; they can tunnel through walls and move from one side to the other side of the barrier with some probability; they can also interfere and produce a fringe pattern for their probability densities. A generic state is then characterized by a superposition of pure states each coming with a given probability. Determinism is lost and replaced by probabilism, and classical results are obtained as expectation values.

When very many (ideally, infinitely many) particles start interacting at microscopic scales, quantum mechanics is upgraded to Quantum Field Theory. In this framework, fields are the fundamental objects and particles are excitations of fields. For example, in this way, one thinks of photons as excitations of the electro-magnetic field. Interactions are described perturbatively in terms of higher order corrections, that can also be visualized pictorially via Feynman diagrams, where forces are carried by intermediate bosonic particles, and it is possible to compute scattering amplitudes and probabilities for particular processes involving in principle any number of particles at a given order in perturbation theory.

When computing amplitudes, consistent cancellations of infinities appear order by order, yielding a finite answer for the event probabilities. However, this is completely true if and only if gravity is not included in the picture. Hence, quantum theories of electro-magnetism, weak and strong interactions all make perfect sense. In accelerator physics this is just what is needed. When gravity is taken into account, the infinities do not cancel anymore, instead divergent amplitudes appear and new counter-terms must be added at any order of the perturbation expansion. Hence a quantum theory of gravity based purely on quantum field theory does not make sense.

Luckily most of the times gravity can indeed be reasonably neglected. In fact, at microscopic scales its strength is so small compared to the other forces that it can be safely ignored. However, this is still quite unsatisfactory for several reasons. First of all, as a matter of principle, two theories, namely Einstein's gravity and quantum mechanics, which work perfectly well in their own regimes, respectively large and short distances, seem to be incompatible with each other. Secondly, there exist instances where in order to fully understand physics, both theories together must be used. One example is a black hole, where gravity is strong and the matter is localized in a very small region around the singularity; another example is the universe in its first moments, when it was an extremely dense plasma of matter and radiation with dominant quantum effects.

We will come back to this problem later, but for the moment let us ignore gravity and discuss the current theory of particle physics tested everyday in accelerators (energy $\sim$ TeV). It is known as the Standard Model of particle physics (figure\footnote{Figure taken from the website $http://en.wikipedia.org/wiki/Standard\_Model$.} \ref{Standard_Model_of_Elementary_Particles}).

\begin{figure} [ht]
\centering
\includegraphics[scale=1]{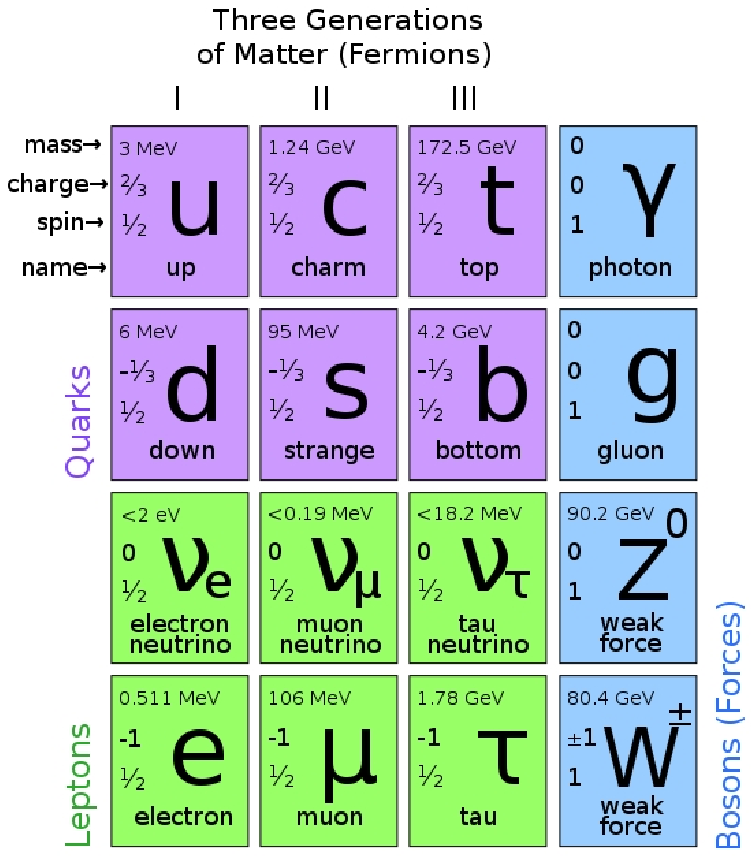}
\caption{Standard Model of particle physics.}
\label{Standard_Model_of_Elementary_Particles}
\end{figure} 

Fields arise as representations of particular Lie (gauge) groups. The Standard Model (gauge) group is 
$SU(3)\times SU(2)\times U(1)$, where the $SU(3)$ factor refers to the strong interactions and the
$SU(2)\times U(1)$ to the unified electro-weak interactions. Matter fields are fermions (quarks and leptons), while the force-mediating fields are (gauge) bosons (gluons, photon, $Z^0$, $W^{\pm}$).
Quarks interact strongly, leptons weakly, but both of them appear in three generations (families). 
Still missing in the picture is an additional particle, the Higgs boson, which is needed to explain the origin of mass for all the other elementary particles. 

The Standard Model is a very good theory in describing elementary particle physics. However, it still presents some obvious problem. First of all, gravity is not included. Secondly, all the couplings in the theory are free and can in principle have any value. Hence, one would like to have a quantum theory that at the same time can describe gravity and justify the value of the Standard Model coupling constants, maybe as vacuum expectation value of more fundamental fields. A candidate for such a theory exists and it is String Theory.

\section*{String Theory and Conformal Field Theory}
The idea behind the theory is very simple: elementary particles are not point-like, instead they are small oscillating one-dimensional strings which look zero-dimensional when observed from a distance. 
The theory was born initially as an attempt to describe phenomena such as flux tubes in strong interactions. It was then abandoned, due to the appearance in the spectrum of a spin-two particle which had been never observed and the simultaneous advent of the $SU(3)$ gauge theory (QCD) which was very successful since the very beginning in describing strong interactions. 
The renaissance of the theory arrived in the eighties, after reinterpreting the spin-two particle as the graviton. In this way String Theory became a theory of gravity. In the nineties, more ingredients were added to the framework, in particular branes, higher dimensional surfaces inside a ten-dimensional space-time.

Strings come in two types: open and closed. 
Open strings are characterized by the fact that they have endpoints. However, the endpoints are not free to move in space, but must be attached to branes. The reason for this is charge conservation: strings carry charges and the charge cannot simply disappear when it reaches the endpoints. The quantization of a string implies the existence of a set of creation and annihilation operators that are used to construct the Hilbert space of states out of the vacuum. Each excited state corresponds to a particle with given mass, charge and spin. Generically, among these states, one finds the gauge bosons.
Closed strings have no endpoints and are free to move inside the whole ten-dimensional space. The quantization implies the existence of two sets of creation and annihilation operators, since both left-movers and right-movers can propagate through the loop. Generically, among the exited states, one finds the spin-two particle that is interpreted as the graviton.

A string moving in space-time sweeps out a two-dimensional surface, the world-sheet. The field theory defined on the world-sheet is conformal, namely admits additional symmetries. Conformal Field Theories in two dimensions are very special, since the symmetry group is infinite dimensional and makes it possible to exactly solve them. The main ingredients of a Conformal Field Theory are the conformal fields, called primary fields. Theories with a finite number of primaries are called rational. By acting with the symmetry generators on the primaries, one builds the whole Hilbert space corresponding to a given representation. All the information about a given Hilbert space is summarized into character functions, which are then used to write down modular invariant partition functions and hence generate particle spectra. A partition function tells us how left-movers are coupled to right-movers.

Conformal Field Theories are used in String Theory in many places. The use that we have focused on here is the construction of four-dimensional string theories and corresponding particle spectra. The standard way of constructing such theories is to start with a four-dimensional theory and add an internal sector that takes care of the extra dimensions. The internal sector must have very special properties, that are guaranteed by carefully choosing the building blocks and by imposing specific projections. 

The full power of a Conformal Field Theory shows itself in the production of a huge numbers of partition functions. Each of them correspond to a particle spectrum with features that vary from one to another: generically, they will all have different predictions about the number of families and gauge groups, and furthermore they will admit fractionally-charged particles. The abundance of spectra makes it impossible to pick the right one: there is no reason and no selection principle why we would have to live in a Standard-Model-like world as we observe it. In particular, from the study of the family-number distributions that one gets with these methods it appears that the number three is more disfavored than other small numbers, for example two or four. One could then ask why do we not observe two or four families, instead of three; or maybe our tools are still too primitive to produce three-family models in abundance. Probably the landscape of possible four-dimensional vacua of String Theory is too big, it will never be fully explored and we will never know; or maybe in the long term it will be possible to find all the solutions and get a complete understanding of the matter, but when, and if, this will happen, we will not be around to taste the end.

\cleardoublepage

\phantomsection
\addcontentsline{toc}{chapter}{Samenvatting}
\chapter*{Samenvatting}

{\flushright
{\small 
\textit{And I find it kind of funny, I find it kind of sad,}\par
\textit{the dreams in which I'm dying are the best I've ever had.}\par
\textit{(R. Orzabal, Mad World)}\par
}
}

\section*{Snaartheorie ontmoet de echte wereld}
Onze huidige kennis van de wereld op zijn fundamentele niveau dateert van de jaren zeventig toen het \textit{Standaard Model van de deeltjesfysica} werd gebouwd, met behulp van quantum velden theorie, als een fusie tussen de twee belangrijkste ontdekkingen van de vorige eeuw, namelijk kwantummechanica en de relativiteitstheorie. Het Standaard Model is de best werkende theorie van de deeltjesfysica die we op dit moment hebben. Niet alleen omvat het elementaire deeltjes, zoals elektronen, en de fundamentele krachten, zoals de elektro-magnetische kracht, in een uiterst elegante wiskundige formulering, maar het voldoet ook aan de experimenten op een ongelooflijk hoog niveau van precisie. Daarom heeft het een zeer sterke voorspellende kracht.

Het Standaard Model beschrijft echter \textit{niet} de zwaartekracht. De pogingen van de afgelopen jaren 
om het Standaard Model uit te breiden door ook de zwaartekracht op te nemen zijn allemaal jammerlijk mislukt. Een compleet nieuwe en meer fundamentele theorie is nodig, die \textit{zowel} het Standaard Model als limiet moet bevatten \textit{en} het moet veralgemeniseren om de zwaartekracht op te nemen. Een kandidaat bestaat en heet \textit{Snaartheorie}.

Snaartheorie werd in de jaren zeventig geboren. Het idee erachter is dat elementaire deeltjes niet puntvormig zijn, maar snaarvormig: het zijn heel kleine filamenten, die in de ruimte bewegen en trillen. Tijdens het verplaatsen bestrijken ze een twee-dimensionaal oppervlak in ruimte en tijd, het zogenaamde wereldoppervlak.

In haar eerste formulering bleek Snaartheorie te veel problemen te hebben, waarbij de lastigste van allemaal de voorspelling van extra dimensies is: precies zes meer dan we in het dagelijks leven observeren. De uitdaging was vervolgens uit te leggen waarom we een wereld met vier dimensies (drie ruimtelijke dimensies plus tijd) ervaren, terwijl de theorie tien voorspelt. Het antwoord op deze vraag staat bekend als \textit{compactificatie}.

Door middel van een oud mechanisme uit het begin van de vorige eeuw kon dit probleem worden aangepakt: de reden waarom we de zes extra dimensies niet zien is omdat ze zo ``klein'' zijn  dat het onmogelijk is om ze met versnellers waar te nemen. Hun wiskundige structuur is onderhevig aan verschillende technische beperkingen en is in het algemeen zeer gecompliceerd. Er zijn er veel mogelijke keuzes, en elk van hen geeft een totaal andere vier-dimensionale fysica.

Naast alle mogelijke constructies waarin de zes dimensies een direct meetkundige interpretatie hebben, zijn er andere wiskundige manieren om een hoger-dimensionale theorie tot vier dimensies te compactificeren. Helaas staan deze methodes niet altijd een geometrische interpretatie toe als compactificatieruimte en zijn er zeer abstracte concepten mee gemoeid. In ons onderzoek gebruiken we juist deze concepten die nauwelijks kunnen worden gevisualiseerd, maar zorgen voor zeer krachtige gereedschappen om vier-dimensionale modellen van de echte wereld vanuit een tien-dimensionale Snaartheorie te bouwen.

De bouw van deze vier-dimensionale modellen gaat als volgt. Men begint met een vier-dimensionale Snaartheorie en, omwille van de consistentie, voegt men er een \textit{interne sector} met specifieke eigenschappen aan toe. In de interne sector zal rekening moeten worden gehouden met de vrijheidsgraden uit de zes extra dimensies die men in het begin verwaarloosd heeft.

De interne theorie is verantwoordelijk voor de fysica die we in vier dimensies observeren. Ondanks haar ingewikkeldheid is zij onder goede controle, met name dankzij haar symmetrie$\ddot{e}$n. Manipulaties van de bouwstenen staan ons toe om een ​​groot aantal fenomenologisch aantrekkelijke modellen te produceren. Sommigen van hen hebben eigenschappen die heel dicht bij het Standaard Model van de deeltjesfysica liggen. Met ons onderzoek bouwen we dergelijke modellen. De meerderheid van hen is echter heel verschillend van wat we ervaren: er worden extra deeltjes verwacht en nieuwe symmetrie$\ddot{e}$n van de natuur voorspeld. Dit roept onmiddellijk een andere vraag op: als Snaartheorie correct is, waarom leven we in deze bijzondere wereld met zijn eigenaardige deeltjes en symmetrie$\ddot{e}$n, terwijl er vele andere werelden mogelijk zijn in dit \textit{landschap}? Het doel van ons onderzoek is ook om deze vraag te proberen te beantwoorden. Het woord landschap hier wordt vaak gebruikt in de literatuur om naar een dergelijk groot aantal mogelijkheden te verwijzen.

Het hoofdthema van dit proefschrift is het besturen van permutatiesymmetrie$\ddot{e}$n van de bouwstenen in de interne theorie te bestuderen. Dit project heeft een wiskundig en een natuurkundig aspect. De wiskundige kant omvat de definitie van het probleem en de oplossing voor zeer technische (en niet-triviale) kwesties. We konden met succes een cruciale formule ontdekken die geldig is in veel belangrijke situaties. Deze wees de weg naar de meest generieke oplossing. Vanuit fysisch oogpunt hebben we de bovengenoemde formule toegepast op de bouwstenen van de interne theorie. Daarna hebben wij deze resultaten toegepast op snaarfenomenologie om vier-dimensionale modellen te bouwen.

Snaartheorie is een spannend gebied en het is nu het juiste moment voor conceptuele doorbraken, die nodig zijn om het hele gebied vooruit te brengen en misschien voorspellingen te geven. Precies nu zal de LHC in Geneve, de grootste deeltjesversneller ooit, de eerste resultaten produceren en inzicht geven in fysica voorbij het Standaard Model. We zullen misschien binnenkort erachter komen of nieuwe symmetrie$\ddot{e}$n, deeltjes en extra dimensies werkelijk bestaan.

\cleardoublepage

\phantomsection
\addcontentsline{toc}{chapter}{Curriculum Vit\ae}
\chapter*{Curriculum Vit\ae}

{\flushright
{\small 
\textit{I am what I am.}\par
\textit{-M.}\par
}
}

Michele Maio was born in Avellino (Italy) on $22^{\rm nd}$ March 1981. He lived in the nearby village of Montella until the age of 19, where he was educated, and then moved to Bologna, where he studied Physics at the University. He received his \textit{Laurea} degree in 2005 with a thesis on black hole radiation. In the following September, Michele moved to the Netherlands and attended a two-year master's programme at the University of Amsterdam, with a scholarship from the Institute of Theoretical Physics. Here he was introduced to the subject of string theory. In December 2007, he started his Ph.D. at the National Institute for Subatomic Physics (Nikhef, Amsterdam), studying topics in (conformal) field theories and string theory, and formally graduating in October 2011. In June of the same year, he started a junior postdoc position at Nikhef.


\end{document}